\documentclass[11pt,twoside,leqno]{article}

\usepackage{amsmath,amsfonts,amsthm,latexsym}

\usepackage{tikz,graphicx,subfigure}
\usetikzlibrary{arrows}
\usepgflibrary{decorations.markings}
\usetikzlibrary{decorations.markings}
\usepackage{overpic}
\usepackage{psfrag}
\usepackage{comment}

\newcommand{\ds}{\displaystyle}
\newcommand{\OO}{\mathcal{O}}
\newcommand{\diag}{\mathop{\mathrm{diag}}}

\newcommand{\supp}{\mathop{\mathrm{supp}}}
\newcommand{\R}{\mathbb{R}}
\newcommand{\C}{\mathbb{C}}
\renewcommand{\Re}{\mathop{\mathrm{Re}}}
\renewcommand{\Im}{\mathop{\mathrm{Im}}}
\renewcommand{\mod}{\mathop{\mathrm{mod}}}
\newcommand{\eps}{\varepsilon}
\DeclareMathOperator{\Tr}{Tr}
\DeclareMathOperator{\Bal}{Bal}

\DeclareMathOperator*{\Res}{Res}
\DeclareMathOperator{\Ai}{Ai}

\newtheorem{theorem}{Theorem}[section]
\newtheorem{lemma}[theorem]{Lemma}
\newtheorem{corollary}[theorem]{Corollary}
\newtheorem{proposition}[theorem]{Proposition}
\newtheorem{Remark}[theorem]{Remark}

\newtheorem{Example}[theorem]{Example}
\newtheorem{Assumption}[theorem]{Assumption}
\newtheorem{definition}[theorem]{Definition}
\newenvironment{remark}{\begin{Remark}\rm}{\end{Remark}}

\numberwithin{equation}{section}
\numberwithin{figure}{section}

\title{The Hermitian two matrix model with an even quartic potential}
\date{}
\author{Maurice Duits\footnote{Department of Mathematics, California Institute of Technology, 1200 E. California Blvd, Pasadena CA 91125, USA. E-mail: mduits@caltech.edu},   Arno B.J. 
Kuijlaars\footnote{Department of Mathematics, Katholieke Universiteit Leuven, Celestijnenlaan 200B, 3001 Leuven, Belgium. E-mail: arno.kuijlaars@wis.kuleuven.be}, and Man Yue 
Mo\footnote{Department of Mathematics,
University of Bristol, Bristol BS8 1TW, UK. E-mail:
m.mo@bristol.ac.uk. } }

\begin{document}
\maketitle

\pagestyle{myheadings} \thispagestyle{plain}
\markboth{M. Duits, A.B.J. Kuijlaars, and M.Y. Mo}{Two matrix model with quartic potential}

\allowdisplaybreaks

\begin{abstract}
We consider the two matrix model with an even quartic potential $W(y)=y^4/4+\alpha y^2/2$ and an even polynomial potential $V(x)$. The main result of the paper is the formulation of a vector equilibrium problem for the limiting mean density for the eigenvalues of one of the matrices $M_1$. The vector equilibrium problem is defined for  three measures, with external fields on the first and  third measures and an upper constraint on the second measure. The proof is based on a steepest descent analysis of a $4\times4$ matrix valued  Riemann-Hilbert problem that characterizes the correlation kernel for the eigenvalues of $M_1$. Our results generalize earlier results for the case $\alpha=0$, where the external field on the third measure was not present.
\end{abstract}

\setcounter{tocdepth}{2} \tableofcontents

\section{Introduction and statement of results}
\subsection{Hermitian two matrix model}

The Hermitian two-matrix model is a probability measure of the form
\begin{equation}\label{eq:2MM}
    \frac{1}{Z_n} \exp\left(-n \Tr(V(M_1)+W(M_2)-\tau M_1 M_2 )\right) dM_1 d M_2,
\end{equation}
defined on the space of pairs $(M_1,M_2)$ of $n\times n$ Hermitian
matrices. The constant $Z_{n}$ in \eqref{eq:2MM} is a normalization constant,
$\tau\in\mathbb{R}\setminus\{0\}$ is the coupling constant
and $d M_1 d M_2$ is  the flat Lebesgue measure on the
space of pairs of Hermitian matrices. In \eqref{eq:2MM}, $V$ and $W$ are
the potentials of the matrix model. In this paper, we assume $V$
to be a general even polynomial and we take
$W$ to be the even quartic polynomial
\begin{equation} \label{eq:defW}
    W(y) = \frac{1}{4} y^4  + \frac{\alpha}{2} y^2, \qquad \alpha \in \mathbb R.
    \end{equation}
Without loss of generality we may (and do) assume that
\begin{equation} \label{eq:taupositive}
    \tau > 0.
    \end{equation}
We are interested in describing the eigenvalues of $M_1$ in the large $n$ limit.

In \cite{DuiKu2} the case $\alpha = 0$ was studied in detail.
An important ingredient in the analysis of \cite{DuiKu2} was
a vector equilibrium problem that describes
the limiting mean eigenvalue distribution of $M_1$. In this
paper we extend the vector equilibrium problem to the case $\alpha \neq 0$.

\subsection{Background}

The two-matrix model \eqref{eq:2MM} with polynomial potentials $V$ and $W$
was introduced in \cite{ItzZub, Meh2} as a model for quantum gravity and string theory.
The interest is in the double scaling limit for critical potentials. It is
generally believed that the two-matrix model is able to describe all $(p,q)$
conformal minimal models, whereas the one-matrix model is limited to $(p,2)$
minimal models \cite{DaKaKo, Dou, Eyn2}.
In \cite{Kaz} the two-matrix model was proposed for the study
of the Ising model on a random surface, where the logarithm of the
partition function (i.e., the normalizing constant $Z_n$ in \eqref{eq:2MM}) is
expected to be the generating function in the enumeration of graphs on surfaces.
For more information and background on the physical interest we refer to the
the surveys \cite{DiF,DiGiZi}, and more recent physical papers \cite{BergEyn, Eyn3, EynOra1, EynOra2}

The two matrix model have a very rich integrable structure that is connected to
biorthogonal polynomials, isomonodromy deformations, Riemann-Hilbert problems
and  integrable equations,  see e.g.\ \cite{AdlvMo1, Ber1, BertEyn, BeEyHa1, BeEyHa2, ErcMcL, EynMeh, Kap, KuiMcL}. This is the basis of the mathematical treatment of the two matrix model,
see also the survey \cite{Ber2}.

% The two matrix model was introduced in \cite{ItZu,Meh2}. Since its introduction, it has
% become a very active research area
% \cite{AdlvMo,BergEyn,Bert,BeEy,BeEyHa1,BeEyHa2,DaKaKo,DuGeKu,DuKu,ErcMcL,Ey1,Ey2,Ey3,EyMe,EyOr1,EyOr2,Kap,Kaz,KuMcL,MeSh,Mo}.
% Good reviews on the subject can be found in \cite{DiF,DiGiZi}.

The eigenvalues of the matrices $M_1$ and $M_2$ in the two-matrix model are a determinantal
point process with correlation kernels that are expressed in terms of biorthogonal polynomials.
These are two families of monic polynomials $\{p_{k,n}(x)\}_{k=0}^\infty$ and
$\{q_{l,n}(y)\}_{y=0}^\infty$, where $p_{k,n}$ has degree $k$ and $q_{l,n}$ has degree $l$, satisfying the condition
\begin{equation} \label{eq:biorthogonality}
     \int_{-\infty}^{\infty} \int_{-\infty}^{\infty}
    p_{k,n}(x)q_{l,n}(y)  e^{-n\left(V(x)+W(y)-\tau x y \right)} dx dy = h_{k,n}^2 \delta_{k,l}.
    \end{equation}
The polynomials are well-defined by \eqref{eq:biorthogonality} and have simple and real zeros \cite{ErcMcL}.
Moreover, the zeros of $p_{k,n}$ and $p_{k+1,n}$, and those of $q_{l,n}$ and $q_{l+1,n}$ are interlacing \cite{DuGeKu}.

The kernels are expressed in terms of these biorthogonal polynomials and their transformed functions
\begin{align*}
    Q_{l,n}(x) &=e^{-n V(x)} \int_{-\infty}^{\infty} q_{l,n}(y) e^{-n\left( W(y) -\tau x y\right) }  dy,\\
    P_{k,n}(y) &=e^{-n W(y)} \int_{-\infty}^{\infty} p_{k,n}(x) e^{-n\left(V(x)-\tau x y \right)}  d x,
\end{align*}
as follows:
\begin{align}
K_{11}^{(n)}(x_1,x_2) & =\sum_{k=0}^{n-1} \frac{1}{h_{k,n}^2} p_{k,n}(x_1) Q_{k,n}(x_2),\label{eq:defK11}\\
K_{12}^{(n)}(x,y) & =\sum_{k=0}^{n-1} \frac{1}{h_{k,n}^2} p_{k,n}(x) q_{k,n}(y), \label{eq:defK12} \\
K_{21}^{(n)}(y,x) & =\sum_{k=0}^{n-1} \frac{1}{h_{k,n}^2} P_{k,n}(y) Q_{k,n}(x) -e^{-n\left(V(x)+W(y)-\tau x y\right)}, \label{eq:defK21} \\
K_{22}^{(n)}(y_1,y_2) & =\sum_{k=0}^{n-1} \frac{1}{h_{k,n}^2} P_{k,n}(y_1) q_{k,n}(y_2). \label{eq:defK22}
\end{align}
Then Eynard and Mehta \cite{EynMeh, MehShu}, see also \cite{BorRai, Del, Meh1},
showed that the joint probability density function
for the eigenvalues $x_1,\ldots,x_n$ of $M_1$ and $y_1,\ldots, y_n$ of $M_2$ is given by
\begin{multline*}
\mathcal P(x_1,\ldots,x_n,y_1,\ldots,y_n) \\ =
\frac{1}{(n!)^2}
\det
\begin{pmatrix}
\left(K_{11}^{(n)} (x_i,x_j)\right)_{i,j=1}^n& \left(K_{12}^{(n)} (x_i,y_j)\right)_{i,j=1}^n\\
\left(K_{21}^{(n)} (y_i,x_j)\right)_{i,j=1}^n &
\left(K_{22}^{(n)} (y_i,y_j)\right)_{i,j=1}^n\end{pmatrix},
\end{multline*}
and the marginal densities take the form
\begin{multline} \label{eq:integrateout}
    \underbrace{\int \cdots \int}_{n-k+n-l \textrm{ times}}
    \mathcal P(x_1,\ldots,x_n,y_1,\ldots,y_n)  dx_{k+1} \cdots d x_n \ dy_{l+1} \cdots dy_n \\ =
    \frac{ (n-k)! (n-l)!}{(n!)^2}
    \det
    \begin{pmatrix}
    \left(K_{11}^{(n)} (x_i,x_j)\right)_{i,j=1}^k& \left(K_{12}^{(n)} (x_i,y_j)\right)_{i,j=1}^{k,l}\\
    \left(K_{21}^{(n)} (y_i,x_j)\right)_{i,j=1}^{l,k} &
    \left(K_{22}^{(n)} (y_i,y_j)\right)_{i,j=1}^l\end{pmatrix}.
\end{multline}
In particular, by taking $l=0$, so that we average over the eigenvalues $y_1,\ldots,y_n$ of $M_2$,
we  find that the eigenvalues of $M_1$ are a determinantal point process with kernel $K_{11}^{(n)}$, see \eqref{eq:defK11}.
This kernel is constructed out of the biorthogonal family $\{p_{k,n}\}_{k=0}^\infty$ and $\{Q_{l,n}\}_{l=0}^\infty$
and the associated determinantal point process is an example of  a biorthogonal ensemble \cite{Bor}.
It is also an example of a multiple orthogonal polynomial ensemble in the sense of \cite{Kui1}.

In order to describe the behavior of the eigenvalues in the large $n$ limit, one needs to
control the kernels \eqref{eq:defK11}--\eqref{eq:defK22} as $n \to \infty$. Due to
special recurrence relations satisfied by the biorthogonal polynomials, there
exist Christoffel-Darboux-type formulas that express the $n$-term sums \eqref{eq:defK11} and \eqref{eq:defK22}
into a finite number (independent of $n$) of biorthogonal polynomials and transformed functions,
see \cite{BeEyHa1}. This paper also gives differential
equations and  a remarkable duality between spectral curves, see also \cite{BertEyn, BeEyHa2}.

A Riemann-Hilbert problem for biorthogonal polynomials was first formulated in \cite{ErcMcL}.
The Riemann-Hilbert problem in \cite{ErcMcL} is of size $2 \times 2$ but it is non-local
and one has not been able to apply an asymptotic analysis to it. Local
Riemann-Hilbert problems were formulated in \cite{BeEyHa2, Kap, KuiMcL}, but these
Riemann-Hilbert problems are of larger size, depending on the degrees
of the potentials $V$ and $W$. The formulation of a local Riemann-Hilbert problem for
biorthogonal polynomials, however, opens up the way for the application of
the Deift-Zhou \cite{DeiZho} steepest descent method, which was applied very successfully
to the Riemann-Hilbert problem for orthogonal polynomials, see \cite{BleIts1, Dei, DKMVZ1, DKMVZ2}
and many later papers.

In \cite{DuiKu2} the Deift-Zhou steepest descent method was indeed applied to the
Riemann-Hilbert problem from \cite{KuiMcL} for the case where $W$ is given by \eqref{eq:defW}
with $\alpha=0$. It gave a precise asymptotic analysis of the kernel $K_{11}^{(n)}$ as $n \to \infty$,
leading in particular to the local universality results that are well-known in one-matrix
models \cite{DKMVZ1}.
The analysis in \cite{DuiKu2} was restricted to the genus zero case.
The extension to higher  genus was done in \cite{Mo2}.

\subsection{Vector equilibrium problem}
As already stated, it is the purpose of the present paper to extend the results
of \cite{DuiKu2, Mo2} to the case of general $\alpha$.

An important role in the analysis in \cite{DuiKu2} is played by a vector equilibrium problem that
characterizes the limiting mean density for the eigenvalues of $M_1$ (and also gives the limiting zero
distribution of the biorthogonal polynomials $p_{n,n}$). One of the main contributions
of the present paper is the formulation of the appropriate  generalization to
general $\alpha \in \mathbb R$. We refer to the standard reference \cite{SafTot}
for notions of logarithmic potential theory and equilibrium problems with
external fields.

\subsubsection{Case $\alpha = 0$}
Let us first recall the vector equilibrium problem from \cite{DuiKu2},
which involves the minimization of an energy functional over three measures.
For a measure $\mu$ on $\mathbb C$ we define the logarithmic energy
\[ I(\mu) = \iint \log \frac{1}{|x-y|} d\mu(x) d\mu(y) \]
and for two measures $\mu$ and $\nu$ we define the mutual
energy
\[ I(\mu,\nu) = \iint \log \frac{1}{|x-y|} d\mu(x) d\nu(y). \]
The energy functional in \cite{DuiKu2} then takes the form
\begin{equation} \label{eq:energy0}
          \sum_{j=1}^3 I(\nu_j) - \sum_{j=1}^2 I(\nu_j,\nu_{j+1})
          + \int  \left( V(x) - \frac{3}{4} |\tau x|^{4/3} \right)  d \nu_1(x)
          \end{equation}
and the vector equilibrium problem is to minimize \eqref{eq:energy0}
among all measures $\nu_1$, $\nu_2$ and $\nu_3$ such that
\begin{enumerate}
\item[(a)] the measures have finite logarithmic energy;
\item[(b)] $\nu_1$ is a measure on $\mathbb R$ with $\nu_1(\mathbb R) = 1$;
\item[(c)] $\nu_2$ is a measure on $i\mathbb R$ with $\nu_2(i \mathbb R) = 2/3$;
\item[(d)] $\nu_3$ is a  measure on $\mathbb R$ with $\nu_3(\mathbb R) = 1/3$;
\item[(e)] $\nu_2 \leq \sigma_2$ where $\sigma_2$ is the unbounded measure
with density
\begin{equation} \label{eq:sigma2alpha0}
    \frac{d\sigma_2}{|dz|} = \frac{\sqrt{3}}{2\pi} \tau^{4/3} |z|^{1/3}, \qquad z \in i \mathbb R
    \end{equation}
on the imaginary axis.
\end{enumerate}

A main feature of the vector equilibrium problem is that it involves
an external field  acting on the first measure as well as an upper constraint
\eqref{eq:sigma2alpha0} acting on the second measure. Note that an upper constraint
arises typically in the asymptotic analysis of discrete orthogonal polynomials, see e.g.\
\cite{BKMM, BleLie, DraSaf, KuiRak, Rak}. The interaction between the measures
in \eqref{eq:energy0} is of the Nikishin type where consecutive measures
attract each other, but there is no direct interaction between measures $\nu_i$ and $\nu_j$
if $|i-j| \geq 2$. The notion of a Nikishin system originated in works
on Hermite-Pad\'e rational approximation,  see \cite{Apt2, GonRak, Kui1, NikSor}.
Vector equilibrium problems also played a role in the recent papers
\cite{BalBer, BeGeSz, BlDeKu, KuMaWi} that are  related
to random matrix theory and \cite{ApKaSa, DuGeKu, DuiKu1, KuiRom, ZhaRom} that are related
to recurrence relations.

\subsubsection{General $\alpha$}

For general $\alpha \in \mathbb R$, the relevant energy functional
 takes the form
\begin{multline} \label{eq:energy}
        E(\nu_1,\nu_2,\nu_3)
         = \sum_{j=1}^3 I(\nu_j) - \sum_{j=1}^2 I(\nu_j,\nu_{j+1}) \\
          + \int  V_1(x) d \nu_1(x)  + \int V_3(x) d \nu_3(x),
\end{multline}
where $V_1$ and $V_3$ are certain external fields acting
on $\nu_1$ and $\nu_3$, respectively.
The vector equilibrium problem is to minimize $E(\nu_1, \nu_2, \nu_3)$
among all measures $\nu_1$, $\nu_2$, $\nu_3$, such that
\begin{enumerate}
\item[(a)] the measures have finite logarithmic energy;
\item[(b)] $\nu_1$ is a measure on $\mathbb R$ with $\nu_1(\mathbb R) = 1$;
\item[(c)] $\nu_2$ is a measure on $i\mathbb R$ with $\nu_2(i \mathbb R) = 2/3$;
\item[(d)] $\nu_3$ is a  measure on $\mathbb R$ with $\nu_3(\mathbb R) = 1/3$;
\item[(e)] $\nu_2 \leq \sigma_2$ where $\sigma_2$ is a certain measure on the imaginary axis.
\end{enumerate}

Comparing with \eqref{eq:energy0} we see that there is an external field
$V_3$ acting on the third measure as well.
The vector equilibrium problem depends on the input data $V_1$, $V_3$,
and $\sigma_2$ that will be described next. Recall that $V$ is an
even polynomial and that $W$ is the quartic polynomial given by \eqref{eq:defW}.

\paragraph{External field $V_1$:}
The external field $V_1$ that acts on $\nu_1$ is defined by
\begin{equation} \label{eq:V1}
    V_1(x) = V(x) + \min_{s \in \mathbb R} \left(W(s) - \tau xs\right).
    \end{equation}
The minimum is attained at a value $s = s_1(x) \in \mathbb R$ for
which $W'(s) = \tau x$, that is
\begin{equation} \label{eq:saddlepoint}
    s^3 + \alpha s = \tau x.
    \end{equation}
For $\alpha \geq 0$, this value of $s$ is uniquely defined by \eqref{eq:saddlepoint}.
For $\alpha < 0$ there can be more than one real solution of \eqref{eq:saddlepoint}.
The relevant value is the one that has the same sign as $x$ (since $\tau > 0$, see
\eqref{eq:taupositive}). It is uniquely defined, except for $x=0$.

\paragraph{External field $V_3$:}
The external field $V_3$ that acts on $\nu_3$ is not present
if $\alpha \geq 0$. Thus
\begin{equation} \label{eq:V3a}
    V_3(x) \equiv 0 \qquad \text{if } \alpha \geq 0.
    \end{equation}

For $\alpha < 0$, the external field $V_3(x)$ is non-zero only for
$x \in (-x^*(\alpha), x^*(\alpha))$ where
\begin{align} \label{eq:xstar}
     x^*(\alpha) = \begin{cases}
     \ds \frac{2}{\tau} \left( \frac{-\alpha}{3} \right)^{3/2},
        & \alpha < 0, \\[10pt]
        0, & \alpha \geq 0.
        \end{cases}
    \end{align}
For those $x$, the equation
\eqref{eq:saddlepoint} has three real solutions
$s_1 = s_1(x)$, $s_2 = s_2(x)$, $s_3 = s_3(x)$ which we take
to be ordered such that
\[ W(s_1) - \tau xs_1 \leq W(s_2) - \tau xs_2 \leq W(s_3) - \tau xs_3. \]
Thus the global minimum of $s \in \mathbb R \mapsto W(s) = \tau xs$ is attained
at $s_1$, and this global minimum played a role in the definition \eqref{eq:V1} of $V_1$.
The function has another local minimum at $s_2$ and a local maximum at $s_3$,
and these are used in the definition of $V_3$. We define $V_3 : \mathbb R \to \mathbb R$ by
\begin{equation} \label{eq:V3b}
    V_3(x) =
        \begin{cases}
        \left(W(s_3(x)) - \tau xs_3(x)\right) - \\
        \qquad \qquad \left( W(s_2(x)) - \tau xs_2(x)\right),  & \text{for } x \in (-x^*(\alpha), x^*(\alpha)), \\
        0 & \text{elsewhere.}
        \end{cases}
            \end{equation}
Thus $V_3(x)$ is the difference
between the local maximum and the other local minimum of $s \in \mathbb R \mapsto W(s) = \tau xs$,
which indeed exist if and only if $x \in (-x^*(\alpha), x^*(\alpha))$,
where $x^*(\alpha)$ is given by \eqref{eq:xstar}. In particular $V_3(x) > 0$ for
$x \in (-x^*(\alpha), x^*(\alpha))$.

\paragraph{The constraint $\sigma_2$:}
To describe the measure $\sigma_2$ that acts as a constraint on $\nu_2$, we consider the
equation
\begin{equation} \label{eq:saddlepointimaginary}
    s^3 + \alpha s = \tau z, \qquad \textrm{with } z \in i \mathbb R.
    \end{equation}
There is always a solution $s$ on the imaginary axis. The other
two solutions are either on the imaginary axis as well, or
they are off the imaginary axis, and lie symmetrically with
respect to the imaginary axis. We define
\begin{equation} \label{eq:sigma2}
    \frac{d\sigma_2(z)}{|dz|} = \frac{\tau}{\pi} \Re s(z)
    \end{equation}
where $s(z)$ is the solution of \eqref{eq:saddlepointimaginary} with largest
real part. We then have for the support $S(\sigma_2)$ of $\sigma_2$,
\begin{equation}
     S(\sigma_2) =
        i \mathbb R \setminus (-i y^*(\alpha), i y^*(\alpha)),
    \end{equation}
where
\begin{align} \label{eq:ystar}
    y^*(\alpha)  =
        \begin{cases} \ds
        \frac{2}{\tau} \left( \frac{\alpha}{3} \right)^{3/2}, &   \alpha > 0, \\[10pt]
        0, &   \alpha \leq 0.
        \end{cases}
\end{align}

This completes the description of the vector equilibrium problem for
general $\alpha$. It is easy to check that for $\alpha = 0$ it reduces to the
vector equilibrium  described before.

\subsection{Solution of vector equilibrium problem}
Our first main theorem deals with the solution of
the vector equilibrium problem. We use $S(\mu)$ to denote
the support of a measure $\mu$. The logarithmic
potential of $\mu$ is the function
\begin{equation} \label{eq:Umu}
    U^{\mu}(x) = \int \log \frac{1}{|x-s|} d\mu(s), \qquad x \in \mathbb C,
\end{equation}
which is a harmonic function on $\mathbb C \setminus S(\mu)$
and superharmonic on $\mathbb C$.

\begin{theorem} \label{theo:theo1}
The above vector equilibrium problem has a unique minimizer $(\mu_1, \mu_2, \mu_3)$
that satisfies the following.
\begin{enumerate}
\item[\rm (a)]
There is a constant $\ell_1 \in \mathbb R$ such that
\begin{equation} \label{eq:varmu1}
\left\{
\begin{aligned}
    2U^{\mu_1}(x) & = U^{\mu_{2}}(x)-V_1(x) + \ell_1, \qquad x\in S(\mu_1), \\
    2U^{\mu_1}(x) & \geq U^{\mu_{2}}(x)-V_1(x)+ \ell_1, \qquad x\in \mathbb R\setminus S(\mu_1).
\end{aligned}
\right.
\end{equation}
If $0 \not\in S(\mu_1)$ or $0 \not\in S(\sigma_2 - \mu_2)$ then
\begin{equation} \label{eq:Smu1}
    S(\mu_1) = \bigcup_{j=1}^N [a_j, b_j],
\end{equation}
for some $N \in \mathbb N$ and $a_1 < b_1 < a_2 < \cdots < a_N < b_N$, and
on  each of the intervals $[a_j,b_j]$ in $S(\mu_1)$ there is a density
\begin{equation} \label{eq:rho1}
    \frac{d\mu_1}{dx} = \rho_1(x)
    = \frac{1}{\pi} h_j(x) \sqrt{(b_j-x)(x-a_j)}, \qquad
        x \in [a_j,b_j]
        \end{equation}
and $h_j$ is non-negative and real analytic on $[a_j,b_j]$.
\item[\rm (b)]
We have
 \begin{equation} \label{eq:varmu2}
\left\{
\begin{aligned}
    2U^{\mu_2}(x) & = U^{\mu_1}(z)+U^{\mu_3}(x), \qquad x\in S(\sigma_2-\mu_2), \\
    2U^{\mu_2}(x) & < U^{\mu_1}(z)+U^{\mu_3}(x), \qquad x\in i \mathbb R \setminus S(\sigma_2 - \mu_2),
\end{aligned}
\right.
\end{equation}
and there is a constant $c_2 \geq 0$ such that
\begin{equation}    \label{eq:Smu2}
    S(\mu_2)  = S(\sigma_2), \quad \textrm{ and } \quad
    S(\sigma_2 - \mu_2)  = i \mathbb R \setminus (-ic_2, ic_2).
 \end{equation}
Moreover, $\sigma_2 - \mu_2$ has a density
\begin{equation} \label{eq:rho2}
    \frac{d(\sigma_2 -\mu_2)}{|dz|} = \rho_2(z), \qquad z \in i \mathbb R
    \end{equation}
that is positive and real analytic on $i \mathbb R \setminus [-ic_2, ic_2]$.
If $c_2 > 0$, then $\rho_2$ vanishes as a square root at $z= \pm ic_2$.
If $\alpha \geq 0$, then $c_2 > y^*(\alpha)$,
where $y^*(\alpha)$ is given by \eqref{eq:ystar}.
\item[\rm (c)]
We have
\begin{equation} \label{eq:varmu3}
\left\{
\begin{aligned}
    2U^{\mu_3}(x) & = U^{\mu_2}(x) - V_3(x), \qquad x\in S(\mu_3), \\
    2U^{\mu_3}(x) & > U^{\mu_2}(x) - V_3(x), \qquad x \in \mathbb R \setminus S(\mu_3),
\end{aligned}
\right.
\end{equation}
and there is a constant $c_3 \geq 0$ such that
\begin{equation} \label{eq:Smu3}
     S(\mu_3) = \mathbb R \setminus (-c_3, c_3).
     \end{equation}
Moreover, $\mu_3$ has a density
\begin{equation} \label{eq:rho3}
    \frac{d\mu_3}{dx} = \rho_3(x), \qquad x \in \mathbb R,
    \end{equation}
that is positive and real analytic on $\mathbb R \setminus [-c_3,c_3]$.
If $\alpha \geq 0$, then $c_3 = 0$.
If $\alpha < 0$, then $c_3 < x^*(\alpha)$ where
$x^*(\alpha)$ is given by \eqref{eq:xstar}.
If $c_3 > 0$, then $\rho_3$ vanishes as a square root at $x = \pm c_3$.

\item[\rm (d)] All three measures are symmetric with respect to $0$,
so that for $j=1,2,3$ we have $\mu_j(A) = \mu_j(-A)$ for every Borel set $A$.
\end{enumerate}
\end{theorem}

In part (a) of the theorem it is stated that $S(\mu_1)$ is a finite union
of intervals under the condition that $S(\mu_1)$ and $S(\sigma_2 - \mu_2)$
are disjoint. If this condition is not satisfied then we are in one
of the singular cases that will be discussed in Section \ref{subsec:cases}
below. However, the condition is not necessary as will be explained
in  Remark \ref{rem:spectralcurve} below. We chose to include the condition in Theorem \ref{theo:theo1}
since the focus of the present paper is on the regular cases.

The conditions \eqref{eq:varmu1}, \eqref{eq:varmu2}, and \eqref{eq:varmu3}
are the Euler-Lagrange variational conditions associated with the vector
equilibrium problem. We note the strict inequalities in \eqref{eq:varmu2} and
\eqref{eq:varmu3}. These are consequences of special properties of
the constraint $\sigma_2$ and the external field $V_3$
that are listed in parts (b) and (c) of the following lemma.

\begin{lemma} \label{lem:convex}
The following hold.
\begin{enumerate}
\item[\rm (a)] Let $\nu_2$ be a measure on $i \mathbb R$ such that
$\nu_2 \leq \sigma_2$. If $0 \not\in S(\sigma_2 - \nu_2)$
then $x \mapsto V_1(x) - U^{\nu_2}(x)$
is real analytic on $\mathbb R$.
\item[\rm (b)] The density $\frac{d \sigma_2}{|dz|}(iy)  =  \frac{\tau}{\pi} \Re s(iy)$
(see \eqref{eq:sigma2}) is an increasing function for $y > 0$.
\item[\rm (c)] Let $\alpha < 0$.
Let $\nu_2$ be a measure on $i \mathbb R$ of finite logarithmic energy such that
$\nu_2 \leq \sigma_2$. Then $x \mapsto V_3(\sqrt{x}) - U^{\nu_2}(\sqrt{x})$
is a decreasing and convex function on $(0, (x^*(\alpha))^2)$.
\end{enumerate}
\end{lemma}

Lemma \ref{lem:convex} is proved in Section \ref{sec:preliminaries}
and the proof of Theorem \ref{theo:theo1} is given in Section \ref{sec:proofTheorem1}.

A major role in what follows will be played by functions defined on
a compact four-sheeted Riemann surface that we will introduce in Section \ref{sec:RiemannSurface}.
The sheets are connected along the supports $S(\mu_1)$, $S(\sigma_2-\mu_2)$ and $S(\mu_3)$
of the minimizing measures for the vector equilibrium problem.
The main result of Section \ref{sec:RiemannSurface} is Proposition \ref{prop:globalmeromorphic}
which says that the function defined by
\[ V'(z) - \int \frac{d\mu_1(s)}{z-s}  \]
on the first sheet has an extension to a globally meromorphic function on the full Riemann surface.
This very special property is due to the special forms of the
external fields $V_1$ and $V_3$ and the constraint $\sigma_2$, which interact in
a very precise way.

\subsection{Classification into cases} \label{subsec:cases}
% \subsection{Regular cases}

According to Theorem \ref{theo:theo1} the structure of the supports
is the same for $\alpha > 0$ as it was for $\alpha = 0$ in \cite{DuiKu2},
that is, $S(\mu_3) =  \mathbb R$ and $S(\sigma_2 - \mu_2) = i \mathbb R \setminus
(-i c_2, i c_2)$ for some $c_2 > 0$. The supports determine the
underlying Riemann surface, and so  the case $\alpha > 0$
is very similar to the case $\alpha = 0$. There are no phase
transitions in case $\alpha > 0$, except for the possible closing or opening
of gaps in the support of $\mu_1$. These type of transitions already occur in the one-matrix
model.

For $\alpha < 0$, however, certain new phenomena occur which come from
the fact that the external field $V_3$ on $\mu_3$ (defined in \eqref{eq:V3b})
has its maximum at $0$
and therefore tends to move $\mu_3$  away from $0$. As a result there
are cases where $S(\mu_3)$ is no longer the full real axis, but a
strict subset \eqref{eq:Smu3} with $c_3 > 0$.

In addition, it is also possible that $c_2 = 0$ in \eqref{eq:Smu2}
such that $S(\sigma_2 - \mu_2)$ is the full imaginary axis and the constraint
$\sigma_2$ is not active. These new phenomena already occur for
the simplest case
\[ V(x) = \frac{1}{2} x^2 \]
for which explicit calculations were done in \cite{DuGeKu} based on
the coefficients in the recurrence relations satisfied
by the biorthogonal polynomials.
These calculations lead to the phase diagram shown in Figure~\ref{fig:phasediagram}
which is taken from \cite{DuGeKu}. There are four phases corresponding
to the following four cases that are determined by the fact whether $0$
is in the support of the measures $\mu_1$, $\sigma_2 - \mu_2$, $\mu_3$ or not:
\begin{description}
\item[Case I:] $0 \in S(\mu_1)$,  $0 \not\in S(\sigma_2 - \mu_2)$,  and $0 \in S(\mu_3)$,
\item[Case II:] $0 \not\in S(\mu_1)$, $0 \not\in S(\sigma_2 - \mu_2)$,  and $0 \in S(\mu_3)$,
%\item[Case II:] $0 \in \supp(\mu_1)$, $c_2 > 0$ and $c_3 > 0$.
\item[Case III:] $0 \not\in S(\mu_1)$, $0 \in S(\sigma_2 - \mu_2)$,  and $0 \not\in S(\mu_3)$,
\item[Case IV:] $0 \in S(\mu_1)$, $0 \not\in S(\sigma_2 - \mu_2)$,  and $0 \not\in S(\mu_3)$.
%\item[Case IV:] $0 \not\in \supp(\mu_1)$, $c_2 > 0$ and $c_3 = 0$.
\end{description}

The four cases correspond to regular behavior of the supports at $0$. There is another
regular situation (which does not occur for $V(x) = \frac{1}{2}x^2$),
namely
\begin{description}
\item[Case V:] $0 \not\in S(\mu_1)$, $0 \not\in S(\sigma_2 - \mu_2)$,  and $0 \not\in S(\mu_3)$.
\end{description}
The five cases determine the cut structure of the Riemann surface and
we will use this classification throughout the paper.

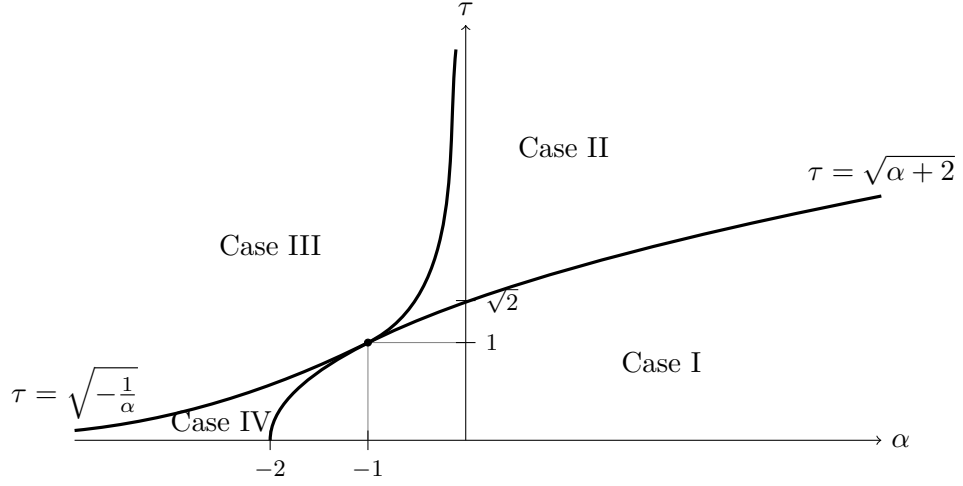
\begin{figure}[t]
\begin{center}
\begin{tikzpicture}[scale=1.3]
\draw[->](0,0)--(0,4.25) node[above]{$\tau$};
\draw[->](-4,0)--(4.25,0) node[right]{$\alpha$};
\draw[help lines] (-1,0)--(-1,1)--(0,1);
\draw[very thick,rotate around={-90:(-2,0)}] (-2,0) parabola (-4.5,6.25) node[above]{$\tau=\sqrt{\alpha+2}$};
\draw[very thick] (-1,1)..controls (0,1.5) and (-0.2,3).. (-0.1,4)
             (-1,1)..controls (-2,0.5) and (-3,0.2).. (-4,0.1) node[above]{$\tau=\sqrt{-\frac{1}{\alpha}}$};
\filldraw  (-1,1) circle (1pt);
\draw (0.1,1) node[font=\footnotesize,right]{$1$}--(-0.1,1);
\draw (-1,0.1)--(-1,-0.1) node[font=\footnotesize,below]{$-1$};
\draw (-2,0.1)--(-2,-0.1) node[font=\footnotesize,below]{$-2$};
\draw (0.1,1.43) node [font=\footnotesize,right]{$\sqrt 2$}--(-0.1,1.43);
\draw[very thick] (2,0.8) node[fill=white]{Case I}
                  (-2.5,0.2) node{Case IV} % {Case II}
                  (-2,2) node[fill=white]{Case III}
                  (1,3) node[fill=white]{Case II}; %{Case IV};
\end{tikzpicture}
\end{center}
\caption{Phase diagram in the $\alpha$-$\tau$ plane for the case $V(x) = \frac{1}{2} x^2$:
the curves $\tau = \sqrt{\alpha+2}$ and
$\tau = \sqrt{-1/\alpha}$ separate the phase diagram
into four regions. The four regions correspond to the cases:
Case I: $N=1$, $c_2 > 0$ and $c_3 = 0$,
Case II: $N=2$, $c_2 > 0$ and $c_3 = 0$, % Case II: $N=1$, $c_2 > 0$ and $c_3 > 0$,
Case III: $N=2$, $c_2 = 0$ and $c_3 > 0$, and
Case IV: $N=1$, $c_2 > 0$ and $c_3 > 0$.} % Case IV: $N=2$, $c_2 > 0$ and $c_3 = 0$}
\label{fig:phasediagram}
\end{figure}

Singular
behavior occurs when two consecutive supports intersect at $0$.
\begin{description}
\item[Singular supports I:] $0 \in S(\mu_1) \cap S(\sigma_2 - \mu_2)$, $0 \not\in S(\mu_3)$,
\item[Singular supports II:] $0 \not\in S(\mu_1)$, $0 \in  S(\sigma_2 - \mu_2) \cap S(\mu_3)$.
\end{description}
There is a multisingular case, when all three supports meet at $0$:
\begin{description}
\item[Singular supports III:] $0 \in  S(\mu_1) \cap S(\sigma_2 - \mu_2) \cap S(\mu_3)$.
\end{description}

Besides a singular cut structure for the Riemann surface, we can also have a singular
behavior of the first measure $\mu_1$. These singular cases also appear in
the usual equilibrium problem for the one-matrix model, see \cite{DKMVZ1}, and
they are as follows.
\begin{description}
\item[Singular interior point for $\mu_1$:]  The density of $\mu_1$ vanishes at an interior point of $S(\mu_1)$.
\item[Singular endpoint for $\mu_1$:] The density of $\mu_1$ vanishes to higher order than square
root at an endpoint of $S(\mu_1)$.
\item[Singular exterior point for $\mu_1$:] Equality holds in the variational inequality in \eqref{eq:varmu1}
at a point $x \in \mathbb R \setminus S(\mu_1)$.
\end{description}

The measures $\sigma_2 - \mu_2$ and $\mu_3$ cannot have singular endpoints, singular exterior points,
or singular interior points, except at $0$. Singular interior points  of these measures at $0$ are as follows.
\begin{description}
\item[Singular interior point for $\sigma_2 - \mu_2$:]  The density of $\sigma_2 -  \mu_2$ vanishes at $0 \in S(\sigma_2 -\mu_2)$.
\item[Singular interior point for $\mu_3$:] The density of $\mu_3$ vanishes at $0 \in S(\mu_3)$.
\end{description}

While there is great interest in the singular cases we restrict the analysis in
this paper to the regular cases, for which we give the following precise definition.
\begin{definition}
The triplet $(V, W, \tau)$ is regular if the supports of the minimizers from
the vector equilibrium problem satisfy
\begin{equation} \label{eq:regularsupports}
    S(\mu_1) \cap S(\sigma_2 - \mu_2) = \emptyset \qquad \text{and} \qquad
    S(\mu_3) \cap S(\sigma_2 - \mu_2) = \emptyset \end{equation}
and if in addition, the measure $\mu_1$ has no singular interior points, singular endpoints, or
singular exterior points, and the measures $\sigma_2-\mu_2$ and $\mu_3$ do not have a singular
interior point at $0$.
\end{definition}
The condition \eqref{eq:regularsupports} may be reformulated as
\[ c_2 = 0 \quad \Longrightarrow \quad 0 \not\in S(\mu_1) \cup S(\mu_3). \]

\subsection{Limiting mean eigenvalue distribution}

The measure $\mu_1$ is the limiting mean eigenvalue distribution of the matrix $M_1$
in the two-matrix model as $n \to \infty$. In this paper we prove this only for regular cases.
To prove it for singular cases, one would have to analyze the nature of the singular behavior
which is beyond the scope of what we want to do in this paper.

\begin{theorem} \label{theo:theo2}
Suppose $(V,W, \tau)$ is regular.
Let $\mu_1$ be the first component of the minimizer $(\mu_1, \mu_2, \mu_3)$
of the vector equilibrium problem.
Then $\mu_1$ is the limiting mean distribution of the eigenvalues of $M_1$
as $n \to \infty$ with $n \equiv 0 (\mod 3)$.
\end{theorem}

We recall that the eigenvalues of $M_1$ after averaging
over $M_2$ are a determinantal point process on $\mathbb R$ with
a kernel $K^{(n)}_{11}$ as given in \eqref{eq:defK11}.
The statement of Theorem \ref{theo:theo2} comes down to the statement
that
\begin{equation} \label{eq:kerneltendstorho1}
    \lim_{n \to \infty} \frac{1}{n} K_{11}^{(n)}(x,x) = \rho_1(x), \qquad x \in \mathbb R
    \end{equation}
where $\rho_1$ is the density of the measure $\mu_1$.

The restriction to $n \equiv 0 (\mod 3)$ is for convenience only,
since it simplifies the expressions in the steepest descent analysis of the
Riemann-Hilbert problem that we are going to do.

The existence of the limiting mean eigenvalue distribution was proved
by Guionnet \cite{Gui} in much more general context. She characterized
the minimizer by a completely different variational
problem, and also connects it with a large deviation principle. It would
be very interesting to see the connection with our vector equilibrium problem.
A related question would be to ask if it is possible to establish a large deviation principle with
the energy functional \eqref{eq:energy} as a good rate function.

We are going to prove \eqref{eq:kerneltendstorho1}
by applying the Deift-Zhou steepest descent analysis to the
Riemann-Hilbert problem \eqref{eq:RHforY} below. Without too much extra effort
we can also obtain the usual
universal local scaling limits that are typical for unitary
random matrix ensembles.
Namely, if  $\rho_1(x^*) > 0$ then the scaling limit is the sine kernel
\[ \lim_{n \to \infty} \frac{1}{n \rho_1(x^*)} K_{11}^{(n)} \left( x^* + \frac{x}{n \rho_1(x^*)},
    x^* + \frac{y}{n \rho_1(x^*)} \right) = \frac{\sin \pi(x-y)}{\pi(x-y)}, \]
while if $x^* \in \{a_1,b_1, \ldots, a_N, b_N \}$ is an end point of $S(\mu_1)$ then
the scaling is the Airy kernel, i.e., for some $c > 0$, we have
\[ \lim_{n \to \infty}
    \frac{1}{(cn)^{2/3}} K_{11}^{(n)} \left( x^* \pm \frac{x}{(cn)^{2/3}}, x^* \pm \frac{y}{(cn)^{2/3}} \right)
    = \frac{\Ai(x) \Ai'(y) - \Ai'(x) \Ai(y)}{x-y} \]
with $+$ if $x^* = b_j$ and $-$ if $x^* = a_j$ for some $j=1,\ldots, N$. Recall that we are
in the regular case so that the density of
$\rho_1$ vanishes as a square root at $x^*$.
The proofs of these local scaling limits will be omitted here, as they
are very similar to the proofs in \cite{DuiKu2}.

\subsection{About the proof of Theorem \ref{theo:theo2}}

The first step in the proof of Theorem \ref{theo:theo2} is the setup of the Riemann-Hilbert (RH)
problem for biorthogonal polynomial $p_{n,n}$ and its connection with the correlation
kernel $K^{(n)}_{11}$. We use the RH problem of \cite{KuiMcL} which we briefly recall.

The RH problem of \cite{KuiMcL} is based on the observation that the polynomial
$p_{n,n}$ that is characterized by the biorthogonality conditions \eqref{eq:biorthogonality}
can alternatively be characterized by the conditions (we assume $W$ is quartic and $n$ is a multiple of three)
\begin{equation} \label{eq:multipleorthogonal}
    \int_{-\infty}^{\infty} p_{n,n}(x) x^k w_{j,n}(x) dx = 0,
    \qquad k = 0, \ldots, n/3-1, \ j =0,1,2,
    \end{equation}
which involves three varying (i.e., $n$-dependent) weight functions
\begin{equation} \label{eq:weightswj}
    w_{j,n}(x)= {e}^{-n V(x)} \int_{-\infty}^{\infty} y^j {e}^{-n(W(y)-\tau xy)} dy, \qquad j=0,1,2.
\end{equation}
The conditions \eqref{eq:multipleorthogonal} are known as multiple orthogonality conditions
of type II, see e.g.\ \cite{Apt1, Kui2, NikSor, VAs}.

A RH problem for multiple orthogonal polynomials was given by Van Assche, Geronimo and
Kuijlaars in \cite{VAGeKu} as an extension
of the well-known RH problem for orthogonal polynomials of
Fokas, Its, and Kitaev \cite{FoItKi}.
For the multiple orthogonality  \eqref{eq:multipleorthogonal} the RH problem is
of size $4 \times 4$ and it asks for
$Y : \mathbb C \setminus \mathbb R \to \mathbb C^{4 \times 4}$ satisfying
\begin{equation}\label{eq:RHforY}
\left\{\begin{array}{ll}\multicolumn{2}{l}{Y  \textrm{ is analytic
in } \mathbb C\setminus \mathbb R,} \\
Y_+(x)=Y_-(x) \begin{pmatrix}   1& w_{0,n}(x)& w_{1,n}(x)& w_{2,n}(x)\\
                                0&1&0&0\\
                                0&0&1&0\\
                                0&0&0&1
\end{pmatrix},& x\in \mathbb R,\\
Y(z)=(I+\OO(1/z))\begin{pmatrix} z^n & 0&0&0\\
                                0&z^{-n/3}&0&0\\
                                0&0&z^{-n/3}&0\\
                                0&0&0&z^{-n/3}\\
\end{pmatrix}, & z\to \infty.\end{array}\right.
\end{equation}

The RH problem has a unique solution. The first row of $Y$ is given in terms
of the biorthogonal polynomial $p_{n,n}$ as follows
\[ Y_{1,1}(z) = p_{n,n}(z), \qquad
        Y_{1,j+2}(z) = \frac{1}{2\pi i} \int_{-\infty}^{\infty} \frac{p_{n,n}(x) w_{j,n}(x)}{x-z} dx, \quad j=0,1,2, \]
and the other rows are built out of certain polynomials of degree $n-1$ in a similar way,
see \cite{KuiMcL, VAGeKu} for details.

Multiple orthogonal polynomials have a Christoffel-Darboux formula \cite{DaeKui}
which implies that the correlation kernel \eqref{eq:defK11} can be rewritten in the integrable form
\[ \frac{f_1(x) g_1(y) + f_2(x) g_2(y) + f_3(x) g_3(y) + f_4(x) g_4(y)}{x-y} \]
for certain functions $f_j, g_j$, for $j=1,\ldots, 4$, and in fact it has the
following representation
\begin{multline} \label{eq:kernelintermsofY}
    K^{(n)}_{11}(x,y) \\  =
    \frac{1}{2\pi i (x-y)}\begin{pmatrix} 0 & w_{0,n}(y)& w_{1,n}(y) & w_{2,n}(y)
    \end{pmatrix} Y_+^{-1}(y)  Y_+(x) \begin{pmatrix}
    1\\ 0 \\ 0\\ 0 \end{pmatrix},
\end{multline}
for $x, y \in \mathbb R$,
in terms of the solution $Y$ of the RH problem \eqref{eq:RHforY}, see \cite{DaeKui}.

The proof  of Theorem \ref{theo:theo2} is an involved and lengthy steepest descent analysis
of the RH problem \ref{eq:RHforY} in which the vector equilibrium problem is
used in an essential way. This is similar to \cite{DuiKu2} which deals with the case $\alpha = 0$.
Certain complications arise because the formulas for the external field and the constraint
in the vector equilibrium problem are less explicit as in the case $\alpha = 0$.
This not only complicates the analysis of the vector equilibrium problem
in Sections \ref{sec:preliminaries} and \ref{sec:proofTheorem1}, but it will
continue to play a role via the functions $\theta_j$ defined in
Section \ref{subsec:saddlevalues}
and  $\lambda_j$ defined in Section \ref{subsec:lambda} throughout the paper.

We also note that the analysis in \cite{DuiKu2} was restricted to the one-cut
case, which leads to an underlying Riemann surface of genus $0$. This
restriction was removed in \cite{Mo2}. The problem in the higher genus
case is in the construction of the global parametrix.
In Section \ref{sec:global} we give a self-contained account that is
based on the ideas developed in \cite{Mo2} and \cite{KuiMo}, which we think is
of independent interest.

We also wish to stress that in Case IV the Riemann surface
always has genus $\geq 1$, even if $S(\mu_1)$ consists of one interval, see \eqref{eq:genusR} below.
This phenomenon did not appear for $\alpha = 0$.

\subsection{Singular cases}

Although we do not treat the singular cases in this paper
we wish to make some comments about the possible critical behaviors that
we see in the two-matrix model with the quartic potential $W(y) = \frac{1}{4} y^4 + \frac{\alpha}{2} y^2$.

As already discussed in Section \ref{subsec:cases} the singular behavior is
associated with either a singular behavior in the measures $\mu_1$, $\sigma_2 -\mu_2$, or $\mu_3$,
or a singular behavior in the supports. The singular behavior in the measure $\mu_1$ also
appears in the one-matrix model that is described by orthogonal polynomials.
It is known that the critical behavior at a
singular interior point where the density vanishes quadratically is described by the Hastings-McLeod solution of the
Painlev\'e II equation, see \cite{BleIts2, ClaKui, Shc}. This Painlev\'e II transition
is the canonical mechanism by which a gap opens up in the support in the one-matrix model.

The critical behavior at a singular endpoint where the density vanishes with exponent $5/2$
is described by a special
solution of the Painlev\'e I$_2$ equation (the second member of the
Painlev\'e I hierarchy), see \cite{ClaVan}. The critical behavior at
a singular exterior point is described by Hermite functions \cite{BerLee, Cla, Mo1}
and this describes an opening of a new band of eigenvalues (birth of a cut).

We see these critical behaviors also in the two-matrix model with
an even quartic $W$. In particular, the opening of a gap at $0$ in the support
of $\mu_1$ is a Painlev\'e II transition. In our classification of regular cases,
this is a transition from Case I to Case II, or a transition from Case IV to Case~V.
In the phase diagram  of Figure \ref{fig:phasediagram} for $V(y) = \frac{1}{2} y^2$,
this transition is on the part of the parabola $\tau = \sqrt{\alpha + 2}$, with $\alpha > -1$.

A Painlev\'e II transition also appears when either $\sigma_2 - \mu_2$ or $\mu_3$ has
a density that vanishes quadratically at $0$. Then $0$ is a singular interior point
and again a gap can open but now in the support of the measures ``that are on the other sheets''
and have no direct probabilistic meaning. If the density of $\sigma_2 - \mu_2$ vanishes at $0$ then the transition
is from Case III to Case~V. If the density of $\mu_3$ vanishes at $0$ then the
transition is from Case I to Case IV or from Case II to Case~V.
In the phase diagram of Figure \ref{fig:phasediagram} the transition from Case I to
Case IV takes place on the part of the parabola $\tau = \sqrt{\alpha + 2}$, with $-2 < \alpha < -1$.

The cases of singular supports represent critical phenomena that do not
appear in the one-matrix model. What we called Singular Supports I in Section \ref{subsec:cases}
corresponds to a transition from Case III to Case IV. This is a transition when the gap around $0$
in the support of $S(\mu_1)$ closes and simultaneously the gap in the support of
$S(\sigma_2 - \mu_2)$ opens up (or vice versa). On the level of the Riemann surface
it means that the two branch points on the real line that are the endpoints of
the gap in $S(\mu_1)$ come together at $0$, and then split again to become a pair
of complex conjugate branch points. These branch points are then on the imaginary axis
and are the endpoints $\pm ic_2$ of $S(\sigma_2-\mu_2)$. A transition of this type
does not change the genus of the Riemann surface.

This type of transition was observed first in the context of random matrices with external source
and non-intersecting Brownian motions, see \cite{AdlvMo2, BleKu3, BreHik, TraWid},i
where it was described in terms of  Pearcey integrals.  The Pearcey transition is a second
mechanism by which a gap in the support may open up (or close). As it involves
three sheets of the Riemann surface it cannot take place in the one-matrix model
which is connected to a two-sheeted Riemann surface.

The case of Singular Supports II gives a transition from Case II to Case III.
This is a situation where the gap in the support of $\sigma_2 - \mu_2$ closes and simultaneously
the gap in $S(\mu_3)$ opens. This also typically corresponds to a Pearcey transition, but it does not involve the
first sheet of the Riemann surface, which means that this transition is not
visible in the eigenvalue distribution of the random matrix.
In the phase diagram of Figure \ref{fig:phasediagram} the Pearcey transitions
are on the curve $\tau = \sqrt{-1/\alpha}$, $\alpha \neq -1$.

The case of Singular Supports III represents a new critical phenomenon. Here the supports
of all three measures $\mu_1$, $\sigma_2 - \mu_2$ and $\mu_3$ are closed at $0$.
In Figure \ref{fig:phasediagram} this is the case at the multi-critical point $\alpha = -1$
and $\tau=1$ where the Painlev\'e transitions and Pearcey transitions come together.
One may approaches the multi-critical point from the  Case III region, where there
is a gap around $0$ in the supports of both $\mu_1$ and $\mu_3$, while the support of
$\sigma_2 - \mu_2$ is the full imaginary axis. At the multi-critical point the supports
of $\mu_1$ and $\mu_3$ close simultaneously, while also the support of $\sigma_2 - \mu_2$ opens
up, which results in a transition from Case III to Case I.

We conjecture that the case of Singular Supports III is of similar nature as was studied
very recently \cite{AdFevM, DeKuZh} for a critcal case of non-intersecting Brownian motions (or random walks)
with two starting and two ending
points. By fine-tuning the starting and ending points one may create a situation where two
groups of non-intersecting Brownian motions fill out two ellipses which are tangent to each
other at one point. Our conjecture is that the local eigenvalue behavior around $0$
in the multi-critical case is the same as that for the non-intersecting Brownian
motions at the point of tangency.
The conjecture is supported by preliminary calculations that suggest that
the local parametrix of \cite{DeKuZh} can also be used if one tries to extend the
RH analysis of the present paper to the multi-critical situation.

\section{Preliminaries and the proof of Lemma \ref{lem:convex}}
\label{sec:preliminaries}

Before coming to the proof of Theorem \ref{theo:theo1} we study the
equation \eqref{eq:saddlepoint} in more detail.
This equation will also play a role in the proof of Theorem \ref{theo:theo2},
where in the first step of the steepest descent analysis,
we will use functions defined by integrals
\begin{equation} \label{eq:integrals}
    \int_{\Gamma} {e}^{-n(W(s) - \tau z s)} ds, \qquad W(s) = \frac{1}{4}s^4 + \frac{\alpha}{2} s^2
\end{equation}
where $\Gamma$ is an unbounded contour in the complex $z$-plane.

\subsection{Saddle point equation and functions $s_j$} \label{subsec:saddlepoint}

The large $n$ asymptotics of the integrals \eqref{eq:integrals} is determined by the
solutions of the saddle point equation $W'(s) - \tau z = 0$, that is
\begin{equation} \label{eq:saddlepoint2}
    s^3 + \alpha s = \tau z.
    \end{equation}
In \eqref{eq:saddlepoint} we considered this equation for $z = x \in \mathbb R$.
We defined a solution $s_1(x)$ for every $x \in \mathbb R$, and for $\alpha < 0$ and $ |x| < x^*(\alpha)$
we also defined $s_2(x)$ and $s_3(x)$.

We define solution $s_1(z)$, $s_2(z)$ and $s_3(z)$ of \eqref{eq:saddlepoint2}
for complex $z$ as follows. We distinguish between the two cases $\alpha > 0$ and $\alpha < 0$.

\paragraph{Case $\alpha > 0$.}
In case $\alpha > 0$ the saddle point equation \eqref{eq:saddlepoint2}
has branch points $\pm i y^*(\alpha) \in i \mathbb R$ where $y^*(\alpha)$ is given
by \eqref{eq:ystar}. The Riemann surface $\mathcal S$
for the equation \eqref{eq:saddlepoint2} then has three sheets that we choose
as follows
\begin{equation} \label{eq:sheetsS1}
\left\{
\begin{aligned}
    \mathcal S_1 & = \mathbb C \setminus \left( (-i \infty, -i
    y^*(\alpha)]
        \cup [i y^*(\alpha), i \infty) \right), \\
    \mathcal S_2 & = \mathbb C \setminus \left(\mathbb R \cup (-i \infty, -i
    y^*(\alpha)]
        \cup [i y^*(\alpha), i \infty)  \right), \\
    \mathcal S_3 & = \mathbb C \setminus \mathbb R.
    \end{aligned}
    \right.
\end{equation}

We already defined $s_1(x)$ for $x \in \mathbb R$ as the unique real saddle point.
This function has an analytic continuation to $\mathcal S_1$ that we also
denote by $s_1$. Then $s_2$ and $s_3$ are defined by analytic continuation onto $\mathcal S_2$
and $\mathcal S_3$, respectively.

\paragraph{Case $\alpha < 0$.}
In case $\alpha < 0$ the saddle point equation \eqref{eq:saddlepoint2}
has two real branch points $\pm x^*(\alpha)$ with $x^*(\alpha)$ given
by \eqref{eq:xstar}. The three sheets of the Riemann surface $\mathcal S$
for the equation \eqref{eq:saddlepoint2}  are now chosen as follows
\begin{equation} \label{eq:sheetsS2}
\left\{
\begin{aligned}
    \mathcal S_1 & = \mathbb C \setminus i \mathbb R,  \\
    \mathcal S_2 & = \mathbb C \setminus \left((-\infty, -x^*(\alpha)]
        \cup [x^*(\alpha), \infty) \cup i \mathbb R \right), \\
    \mathcal S_3 & = \mathbb C \setminus \left( (-\infty, - x^*(\alpha)]
        \cup [x^*(\alpha), \infty) \right).
    \end{aligned}
    \right.
\end{equation}

In case $\alpha < 0$, we have that $s_1(x)$ is defined for $x \in \mathbb R \setminus \{0\}$.
It is the real saddle point for which $W(s) - \tau x s$ is minimal.
The function $s_1$ has an analytic continuation to $\mathcal S_1$ that we also
denote by $s_1$. Then $s_2$ and $s_3$ are defined by analytic continuation onto $\mathcal S_2$
and $\mathcal S_3$, respectively. It is a straightforward check that
for $x \in (- x^*(\alpha), x^*(\alpha))$
this definition of $s_2(x)$ and $s_3(x)$ coincides with the one earlier given.

\begin{lemma} \label{lem:Res1}
The functions $s_j$ are have the symmetries
\begin{equation} \label{eq:symmetries}
    s_j(-z) = -s_j(z), \qquad s_j(\overline{z}) = \overline{s_j(z)}, \qquad j=1,2,3.
    \end{equation}
In addition we have that
\begin{equation} \label{eq:Res1}
    \Re s_1(z) > 0  \qquad \text{if } \Re z > 0.
    \end{equation}
\end{lemma}
\begin{proof}
The symmetries \eqref{eq:symmetries} are clear.

For $z = x \in \mathbb R$ with $x > 0$ we have that $s_1(x) > 0$. Therefore,
by continuity, $\Re s_1(z) > 0$ for $z$ in a neighborhood of the positive
real axis. If $\Re s_1(z) = 0$ for some $z$, so that $s_1(z)$ is purely
imaginary, then
\[ \tau z =  s_1(z)^3 + \alpha s_1(z)  \]
is purely imaginary as well. The inequality $\Re s_1(z) > 0$ therefore
extends into the full right half-plane as claimed in \eqref{eq:Res1}.
\end{proof}

From the lemma it follows that in both
cases the constraint $\sigma_2$, see \eqref{eq:sigma2}, is given
by
\begin{align} \nonumber
    \frac{d\sigma_2(z)}{dz} & = \frac{\tau}{\pi i} \Re s_{1,-}(z) \\
    & = \frac{\tau}{2\pi i} \left(s_{1,-}(z) - s_{1,+}(z) \right),
    \qquad z \in i \mathbb R.
    \label{eq:sigma2withs1}
    \end{align}
The imaginary axis is oriented upwards, so
that $s_{1,-}(z)$ ($s_{1,+}(z)$) for $z \in i\mathbb R$ denotes the limiting value of $s_1$
as we approach $z \in i \mathbb R$ from the right (left) half-plane.

\subsection{Values at the saddles and functions $\theta_j$} \label{subsec:saddlevalues}

We define
\begin{equation} \label{eq:thetaj}
    \theta_j(z) = - W(s_j(z)) + \tau z s_j(z), \qquad j=1,2,3
    \end{equation}
as the value of $-(W(s) -\tau zs)$ at the saddle $s= s_j(z)$.
Note that
\begin{align}
    \theta_j'(z) & = \left( - W'(s_j(z)) + \tau z \right) s_j'(z) + \tau s_j(z)   = \tau s_j(z) \label{eq:thetajder}
        \end{align}
    so that, up to a factor $\tau$, $\theta_j$ is a primitive
    function of $s_j$.

Then $\theta_j$ is defined and analytic on $\mathcal S_j$, see \eqref{eq:sheetsS1} and \eqref{eq:sheetsS2},
and
\begin{equation} \label{eq:thetacontinuations}
\left\{
\begin{aligned}
    \theta_{1,\pm} = \theta_{2,\mp} & \quad \text{ on } (-i \infty, -i y^*(\alpha)]
        \cup [i y^*(\alpha), i \infty),  \\
    \theta_{2,\pm} = \theta_{3,\mp} & \quad \text{ on } (- \infty, -x^*(\alpha)]
        \cup [x^*(\alpha), \infty).
    \end{aligned}
    \right.
    \end{equation}
Recall from \eqref{eq:xstar} and \eqref{eq:ystar} that
we have put $x^*(\alpha) = 0$ if $\alpha > 0$ and $y^*(\alpha) = 0$ if $\alpha < 0$,
so that we can treat the two cases simultaneously in \eqref{eq:thetacontinuations}.

The jumps for the $\theta_j$ functions from \eqref{eq:thetacontinuations},
are taken together in terms of the jumps of the diagonal matrix
\begin{equation} \label{eq:defTheta}
    \Theta(z) = \begin{pmatrix} \theta_1(z) & 0 & 0 \\
    0 & \theta_2(z) & 0 \\ 0 & 0 & \theta_3(z) \end{pmatrix},
        \qquad z \in \mathbb C \setminus (\mathbb R \cup i \mathbb R)
        \end{equation}
as follows.

\begin{corollary} \label{cor:jumpTheta}
 For $x \in \mathbb R$ we have
\begin{equation} \label{eq:jumpThetaonR}
    \left\{
    \begin{aligned}
    \Theta_+(x) & = \Theta_-(x),     && |x| < x^*(\alpha), \\
    \Theta_+(x) & = \begin{pmatrix} 1 & 0 & 0 \\  0 & 0 & 1 \\ 0 & 1 & 0 \end{pmatrix}
     \Theta_-(x) \begin{pmatrix} 1 & 0 & 0 \\  0 & 0 & 1 \\ 0 & 1 & 0 \end{pmatrix},
     && |x| > x^*(\alpha).
     \end{aligned} \right.
     \end{equation}

For $z = iy \in i \mathbb R$ we have
\begin{equation} \label{eq:jumpThetaoniR}
    \left\{
    \begin{aligned}
    \Theta_+(z) & = \Theta_-(z), &&  |y| < y^*(\alpha), \\
    \Theta_+(z) & = \begin{pmatrix} 0 & 1 & 0 \\  1 & 0 & 0 \\ 0 & 0 & 1 \end{pmatrix}
     \Theta_-(z) \begin{pmatrix} 0 & 1 & 0 \\  1 & 0 & 0 \\ 0 & 0 & 1 \end{pmatrix},
      && |y| > y^*(\alpha).
     \end{aligned} \right.
     \end{equation}
\end{corollary}

Also note that by \eqref{eq:V1},  the definition of $s_1$ and \eqref{eq:thetaj}, we have
\begin{equation} \label{eq:V1bis}
    V_1(x) = V(x) - \theta_1(x), \qquad x \in \mathbb R,
    \end{equation}
and by \eqref{eq:V3a}-\eqref{eq:V3b}, the definition of $s_2$ and $s_3$,
and \eqref{eq:thetaj}
\begin{equation} \label{eq:V3bis}
    V_3(x) = \left\{ \begin{array}{cl}
    \theta_2(x) - \theta_3(x), &  \text{for } x \in (-x^*(\alpha), x^*(\alpha)), \\
    0, & \text{elsewhere.}
    \end{array} \right.
    \end{equation}

\subsection{Large $z$ asymptotics}

In what follows we will need the behavior of $s_j(z)$ and $\theta_j(z)$ as $z \to \infty$.

Throughout the paper we define fractional exponents with a branch cut
along the negative real axis. We use $I$, $II$, $III$ and $IV$ to
denote the four quadrants of the complex $z$-plane.
We also put
\[ \omega = {e}^{2\pi i/3}. \]

We state the following lemma without proof. It follows easily
from the saddle point equation \eqref{eq:saddlepoint2}.
\begin{lemma} \label{lem:sjasymptotics}
We have as $z \to \infty$
\begin{align*}
   s_1(z) & = \begin{cases}
   (\tau z)^{1/3} - \frac{\alpha}{3} (\tau z)^{-1/3} + \frac{\alpha^3}{81} (\tau z)^{-5/3} + \OO(z^{-7/3}), & \text{in } I \cup IV, \\
    \omega (\tau z)^{1/3} - \frac{\alpha}{3} \omega^2 (\tau z)^{-1/3} + \frac{\alpha^3}{81} \omega (\tau z)^{-5/3} +\OO(z^{-7/3}),
        & \text{in } II, \\
    \omega^2 (\tau z)^{1/3} - \frac{\alpha}{3} \omega (\tau z)^{-1/3} + \frac{\alpha^3}{81} \omega^2 (\tau z)^{-5/3} +\OO(z^{-7/3}),
        & \text{in } III,
        \end{cases} \\
   s_2(z) & = \begin{cases}
   \omega (\tau z)^{1/3} - \frac{\alpha}{3} \omega^2 (\tau z)^{-1/3} + \frac{\alpha^3}{81} \omega (\tau z)^{-5/3} +\OO(z^{-7/3}), & \text{in } I \\
    (\tau z)^{1/3} - \frac{\alpha}{3} (\tau z)^{-1/3} + \frac{\alpha^3}{81} (\tau z)^{-5/3} +\OO(z^{-7/3}),
        & \text{in } II \cup III, \\
    \omega^2 (\tau z)^{1/3} - \frac{\alpha}{3} \omega (\tau z)^{-1/3} + \frac{\alpha^3}{81} \omega^2 (\tau z)^{-5/3} +\OO(z^{-7/3}),
        & \text{in } IV,
        \end{cases} \\
   s_3(z) & = \begin{cases}
   \omega^2 (\tau z)^{1/3} - \frac{\alpha}{3} \omega (\tau z)^{-1/3} + \frac{\alpha^3}{81} \omega^2 (\tau z)^{-5/3} +\OO(z^{-7/3}), & \text{in } I \cup II \\
   \omega (\tau z)^{1/3} - \frac{\alpha}{3} \omega^2 (\tau z)^{-1/3} + \frac{\alpha^3}{81} \omega (\tau z)^{-5/3} +\OO(z^{-7/3}),
        & \text{in } III \cup IV.
        \end{cases}
\end{align*}
\end{lemma}

We have a similar result for the asymptotics of $\theta_j$.
Note that the following asymptotic behaviors are consistent
with the property that $\theta_j' = \tau s_j$, see \eqref{eq:thetajder}.

\begin{lemma} \label{lem:thetajasymptotics}
We have as $z \to \infty$
\begin{align*}
   \theta_1(z) & = \begin{cases}
   \frac{3}{4} (\tau z)^{4/3} - \frac{\alpha}{2} (\tau z)^{2/3} + \frac{\alpha^2}{6} - \frac{\alpha^3}{54} (\tau z)^{-2/3} + \OO(z^{-4/3}), & \text{in } I \cup IV \\
   \frac{3}{4} \omega (\tau z)^{4/3} - \frac{\alpha}{2} \omega^2 (\tau z)^{2/3} + \frac{\alpha^2}{6} - \frac{\alpha^3}{54} \omega (\tau z)^{-2/3} + \OO(z^{-4/3}),
        & \text{in } II, \\
   \frac{3}{4} \omega^2 (\tau z)^{4/3} - \frac{\alpha}{2} \omega (\tau z)^{2/3} + \frac{\alpha^2}{6} - \frac{\alpha^3}{54} \omega^2 (\tau z)^{-2/3} + \OO(z^{-4/3}),
        & \text{in } III,
        \end{cases} \\
   \theta_2(z) & = \begin{cases}
   \frac{3}{4} \omega (\tau z)^{4/3} - \frac{\alpha}{2} \omega^2 (\tau z)^{2/3} + \frac{\alpha^2}{6} - \frac{\alpha^3}{54} \omega (\tau z)^{-2/3} + \OO(z^{-4/3}), & \text{in } I \\
   \frac{3}{4} (\tau z)^{4/3} - \frac{\alpha}{2} (\tau z)^{2/3} + \frac{\alpha^2}{6} - \frac{\alpha^3}{54}  (\tau z)^{-2/3} + \OO(z^{-4/3}),
        & \text{in } II \cup III, \\
   \frac{3}{4} \omega^2 (\tau z)^{4/3} - \frac{\alpha}{2} \omega (\tau z)^{2/3} + \frac{\alpha^2}{6} - \frac{\alpha^3}{54} \omega^2 (\tau z)^{-2/3} + \OO(z^{-4/3}),
        & \text{in } IV,
        \end{cases} \\
   \theta_3(z) & = \begin{cases}
   \frac{3}{4} \omega^2 (\tau z)^{4/3} - \frac{\alpha}{2} \omega (\tau z)^{2/3} + \frac{\alpha^2}{6} - \frac{\alpha^3}{54} \omega^2 (\tau z)^{-2/3} + \OO(z^{-4/3}), & \text{in } I \cup II \\
   \frac{3}{4} \omega (\tau z)^{4/3} - \frac{\alpha}{2} \omega^2 (\tau z)^{2/3} + \frac{\alpha^2}{6} - \frac{\alpha^3}{54} \omega (\tau z)^{-2/3} + \OO(z^{-4/3}),
        & \text{in } III \cup IV.
        \end{cases}
\end{align*}
\end{lemma}

\subsection{Two special integrals}

As a final preparation for the proof of Lemma \ref{lem:convex}
we need the  evaluation of  the following two definite integrals.
\begin{lemma}
We have for $x > 0$
\begin{equation} \label{eq:integral1}
    \int_{i \mathbb R} \frac{d\sigma_2(z)}{(x-z)^2}
        = - \tau s_1'(x), \qquad \text{and} \qquad
    \int_{i \mathbb R} \frac{d\sigma_2(z)}{x - z^2}
        = \frac{\tau s_1(\sqrt{x})}{\sqrt{x}}.
\end{equation}
\end{lemma}

\begin{proof}
Because of the formula \eqref{eq:sigma2withs1} for $\sigma_2$ we have
\[ \int_{i \mathbb R} \frac{d\sigma_2(z)}{(x-z)^2}
    = \frac{\tau}{2\pi i} \int_{i\mathbb R}
        \frac{s_{1,-}(z) - s_{1,+}(z)}{(x-z)^2} dz \]
Since $s_1$ is analytic in $\mathbb C \setminus i \mathbb R$
and $s_1(z) = \OO(z^{1/3})$ as $z \to \infty$, see Lemma \ref{lem:sjasymptotics},
we can evaluate the integral using contour integration
and residue calculus. It
follows that
\[ \frac{\tau}{2\pi i} \int_{i\mathbb R}
        \frac{s_{1,-}(z)}{(x-z)^2} dz = - \tau s_1'(x), \qquad \text{and} \qquad
        \frac{\tau}{2\pi i} \int_{i\mathbb R}
        \frac{s_{1,+}(z)}{(x-z)^2} dz = 0, \]
and the first integral in \eqref{eq:integral1} is proved.

The second integral follows by a similar calculation,
where we also use the fact that $s_1$ is an odd function.
\end{proof}

\subsection{Proof of Lemma \ref{lem:convex}}

Now we come to the proof of Lemma \ref{lem:convex}.
\subsubsection{Proof of part (a)}  %Lemma \ref{lem:convex} (a)}

\begin{proof}

Integrating the first formula in \eqref{eq:integral1}
two times with respect to $x$, and using
the fact that $\theta_1' = \tau s_1$, we find that for some
constants $A$ and $B$,
\[ \int_{i\mathbb R} \left( \log |z-x| - \log |z| \right)
    d\sigma_2(z) = \theta_1(x) + Ax + B, \qquad x > 0. \]
Thus
\begin{align} \nonumber
    V_1(x) - U^{\nu_2}(x)  & =
    V(x) - \theta_1(x) + \int \log |z-x| \, d\nu_2(x)  \\
    & = \label{eq:lemconvexa}
        V(x) -
    \int_{i \mathbb R} \left( \log |z-x| - \log |z| \right)
        d(\sigma_2 - \nu_2)(z) - Ax - B'
        \end{align}
with a different constant $B' = B - \int \log |z| \, d\nu_2(z)$.

Since $0 \not\in S(\sigma_2 - \nu_2)$ there exists
a $c > 0$ such that $S(\sigma_2 - \nu_2) \subset i \mathbb R \setminus (-ic, ic)$
and so the integral in the right-hand side of \eqref{eq:lemconvexa}
defines a real analytic function of $x$.
Then \eqref{eq:lemconvexa} proves that $V_1 - U^{\mu_2}$ is
real analytic on $\mathbb R$, since $V$ is a polynomial.
\end{proof}

\subsubsection{Proof of part (b)} % Lemma \ref{lem:convex} (b)}

\begin{proof}
 We have by \eqref{eq:sigma2withs1}
\[ \frac{d \sigma_2}{|dz|} (iy) = \frac{\tau}{\pi} \Re s_{1,-}(iy) \]

Since $s_1$ is a solution of $s^3 + \alpha s = \tau z$, we have
$(3s_1^2 + \alpha ) s_1' = \tau$, so that
\begin{equation} \label{eq:lemconvexb}
    \frac{d}{dy} \Re s_{1,-}(iy) =  -  \tau \Im  \frac{1}{3s_{1,-}^2(iy) + \alpha}.
    \end{equation}
For $y > 0$ we have $\Re s_{1,-}(iy) \geq 0$ by \eqref{eq:Res1}.
We also have $\Im s_{1,-}(iy) > 0$, so that  $\Im (s_{1,-}^2(iy)) \geq 0$. This implies
$\Im (1/ ( 3s_{1,-}^2(iy) + \alpha)) \leq 0$
and so indeed by \eqref{eq:lemconvexb}
\[ \frac{d}{dy} \Re s_{1,-}(iy) \geq 0, \qquad y > 0, \]
which proves part (b) of the lemma.
\end{proof}

\subsubsection{Proof of part (c)} % Lemma \ref{lem:convex} (c)}

\begin{proof}
Let $\alpha < 0$ and let $\nu_2$ be as in Lemma \ref{lem:convex} (c).
Since $\nu_2$ is a measure on $i \mathbb R$ we have for $x \in \mathbb R$, $x > 0$,
\[ - U^{\nu_2}(\sqrt{x}) = \frac{1}{2} \int_{i\mathbb R} \log(x-z^2) d\nu_2(z),
    \qquad x  > 0. \]
Hence
\[ \frac{d^2}{dx^2} \left(- U^{\nu_2}(\sqrt{x}) \right)
    = -\frac{1}{2} \int_{i\mathbb R} \frac{1}{(x-z^2)^2} d\nu_2(z),
    \qquad x > 0. \]
Since $\nu_2 \leq \sigma_2$ and the integrand is positive
for every $z \in i \mathbb R$, we have
\begin{align*}
    \frac{d^2}{dx^2} \left(- U^{\nu_2}(\sqrt{x}) \right)
    & > - \frac{1}{2} \int_{i\mathbb R} \frac{1}{(x-z^2)^2} d\sigma_2(z) \\
    & = \frac{1}{2} \frac{d}{dx}
        \left(\int_{i\mathbb R} \frac{1}{x-z^2} d\sigma_2(z) \right) \\
        & = \frac{\tau}{2} \frac{d}{dx}
        \left( \frac{s_1(\sqrt{x})}{\sqrt{x}} \right).
\end{align*}
where we used the second integral in \eqref{eq:integral1}.

Since $\theta_1' = \tau s_1$ we see that
\begin{align*}
    \frac{d^2}{dx^2} \left(- U^{\nu_2}(\sqrt{x}) \right)
    & > \frac{d^2}{dx^2} \left( \theta_1(\sqrt{x}) \right)
    \end{align*}
Since $V_3 = \theta_2 - \theta_3$, see \eqref{eq:V1bis}, we have for
$0 < x < (x^*(\alpha))^2$,
\begin{multline*}
    \frac{d^2}{dx^2} \left( V_3(\sqrt{x}) - U^{\nu_2}(\sqrt{x}) \right)
    >
    \frac{d^2}{dx^2} \left( \theta_2(\sqrt{x}) - \theta_3(\sqrt{x}) +
    \theta_1(\sqrt{x}) \right), \\
        \qquad 0 < x \leq (x^*(\alpha))^2.
    \end{multline*}
Since $\theta_1 + \theta_2 + \theta_3 = \frac{1}{2} \alpha^2$,
it now also follows that
\begin{align} \label{eq:estimateonV3sqrtx}
    \frac{d^2}{dx^2} \left( V_3(\sqrt{x}) - U^{\nu_2}(\sqrt{x}) \right)
    & > - 2
    \frac{d^2}{dx^2} \left( \theta_3(\sqrt{x}) \right),
        \qquad 0 < x \leq (x^*(\alpha))^2.
    \end{align}

Recall that $s_3(z)$ is the solution of $s^3 + \alpha s = \tau z$
with $s_3(0) = 0$. Since $\alpha < 0$ we have
that $s_3$ is an odd function which is analytic in a neighborhood of $0$.
Inserting the Taylor series
\[ s_3(z) = - \sum_{k=0}^{\infty} c_k z^{2k+1}, \qquad
    |z| < x^*(\alpha), \]
with $c_0 = - \frac{\tau}{\alpha} > 0$ into $\left[s_3(z) \right]^3
=  - \alpha s_3(z) + \tau z$ and comparing coefficients of $z$, it
is easy to show inductively that $c_k > 0$ for every $k$. Since
$\theta_3'(z) = \tau s_3(z)$ with $\theta_3(0) = 0$, we then also
have
\[ - \theta_3(z) = \tau
    \sum_{k=0}^{\infty} \frac{c_k}{2k+2} z^{2k+2},
        \qquad |z| < x^*(\alpha) \]
and
\begin{equation} \label{eq:estimateontheta3sqrtx}
    -2 \frac{d^2}{dx^2} \theta_3(\sqrt{x})
    = \tau
        \sum_{k=1}^{\infty} k c_k x^k > 0,
            \qquad 0 < x < (x^*(\alpha))^2
            \end{equation}
since all $c_k > 0$. The two inequalities  \eqref{eq:estimateonV3sqrtx}
and \eqref{eq:estimateontheta3sqrtx} give the convexity
of $V_3(\sqrt{x})-U^{\nu_2}(\sqrt{x})$, which completes the proof of
part (c) of Lemma \ref{lem:convex}.
\end{proof}

\section{Proof of Theorem \ref{theo:theo1}} \label{sec:proofTheorem1}

We basically follow Section 4 of \cite{DuiKu2} where Theorem \ref{theo:theo1}
was proved for the case $\alpha = 0$. However, we need more
additional results
from potential theory.

\subsection{Results from potential theory}

We use a number of results and notions from logarithmic potential
theory. The main reference is \cite{SafTot}. Most results
in \cite{SafTot} are stated for measures with compact
support, while we are also dealing with measures with
unbounded support, namely the real line or the
imaginary axis. Therefore we need a number of results from \cite{SafTot}
in a slightly stronger form that allows for measures with
unbounded supports.

The following theorem is known as the principle
of domination, and it is stated in  \cite[Theorem II.3.2]{SafTot} for the case where $\mu$
and $\nu$ have compact supports.

\begin{theorem} \label{thm:domination}
Suppose $\mu$ and $\nu$ are finite Borel measures
with $\int d\nu \leq \int d\mu$. Suppose that
$\mu$ has finite logarithmic energy and that $S(\mu) \neq \mathbb C$.
If for some constant $c$ the inequality
\begin{equation} \label{eq:dominationinequality}
    U^{\mu}(z) \leq U^{\nu}(z) + c
    \end{equation}
holds $\mu$-almost everywhere, then it holds for all $z \in \mathbb C$.
\end{theorem}
\begin{proof}
Let us first establish Theorem~\ref{thm:domination} under the assumption
that $\mu$ has compact support, say $S(\mu) \subset D_R = \{ z \mid |z| \leq R \}$.
Let $\hat \nu$ be the balayage of $\nu$ onto $D_R$. By the properties
of balayage, we then have $\int d\hat{\nu} = \int d\nu$ and for certain constant $\ell \geq 0$,
\begin{equation} \label{eq:balayagedom}
\begin{aligned}
    U^{\hat{\nu}}(z) & = U^{\nu}(z) + \ell, \qquad z \in D_R, \\
    U^{\hat{\nu}}(z) &\leq U^{\nu}(z) + \ell, \qquad z \in \mathbb C.
    \end{aligned}
    \end{equation}
Then if \eqref{eq:dominationinequality} holds $\mu$-a.e., we
find from the equality in \eqref{eq:balayagedom} and the fact that $S(\mu) \subset D_R$ that
\[ U^{\mu}(z) \leq U^{\hat \nu}(z) + c-\ell \qquad \mu\text{-a.e.}. \]
Thus by the principle of domination for measures with compact supports,
see \cite[Theorem II.3.2]{SafTot}, we find
$U^{\mu} \leq U^{\hat \nu} + c - \ell$ on $\mathbb C$, which in view of the
inequality in \eqref{eq:balayagedom} leads to $U^{\mu} \leq U^{\nu} + c$ on $\mathbb C$, as required.

We next assume that $\mu$ is as in the statement of the theorem. As $S(\mu) \neq \mathbb C$,
there is some disk $D(z_0, r) = \{z \in \mathbb C \mid |z-z_0| < r \}$ with $r > 0$ that is disjoint from $S(\mu)$.
By translation and dilation invariance of the statement in Theorem~\ref{thm:domination}
we may assume that $D(0,1)$ is disjoint from $S(\mu)$. We may also assume
that $U^{\nu}(0) < \infty$.

Let $\tilde{\mu}$ be the image of $\mu$ under the inversion $z \mapsto 1/z$.
Then $\tilde{\mu}$ has compact support and straightforward calculations shows that
\begin{equation} \label{Umutilde}
    U^{\tilde \mu}(z) = U^{\mu}(1/z) - \log |z| \int d\mu - U^{\mu}(0)
    \end{equation}
    and $I(\tilde \mu) = I(\mu)$.
If $\tilde{\nu}$ is the image of $\nu$ under the inversion $z \mapsto 1/z$,
then we likewise have
\begin{equation} \label{Unutilde}
    U^{\tilde \nu}(z) = U^{\nu}(1/z) - \log |z| \int d\nu - U^{\nu}(0).
    \end{equation}
Then from \eqref{eq:dominationinequality}, \eqref{Umutilde} and \eqref{Unutilde}, we get
\[ U^{\tilde \mu} \leq U^{\tilde \nu} -  \log |z| \left( \int d\mu - \int d \nu\right) + c - U^{\mu}(0) + U^{\nu}(0),
    \qquad \tilde \mu-\text{a.e.} \]
Thus $U^{\tilde \mu} \leq U^{\nu_1} + c_1$, $\tilde \mu$-a.e.\ where $c_1 =  c - U^{\mu}(0) + U^{\nu}(0)$ and
$\nu_1 =  \tilde \nu + (\int d\mu - \int d\nu) \delta_0$ is a finite
positive measure with the same total mass as $\tilde \mu$. By the principle
of domination for compactly supported measures $\mu$ and arbitrary $\nu$, that we just proved,
we find $U^{\tilde \mu} \leq U^{\nu_1} + c_1$ everywhere,
which in turn by \eqref{Umutilde} and \eqref{Unutilde} leads to $U^{\mu} \leq U^{\nu} + c$.
This proves the theorem.
\end{proof}

The following result is stated for compactly supported measures in
\cite[Theorem IV.4.5]{SafTot}, see also \cite{TotUll} where the result is attributed
to de la Vall\'ee Poussin \cite{dlVP}.

\begin{theorem} \label{thm:delaValleePoussin}
Let $\mu$ and $\nu$ be measures on $\mathbb C$ with $\int d\nu \leq \int d\mu$,
$S(\mu) \neq \mathbb C$, and finite logarithmic potentials $U^{\mu}$ and $U^{\nu}$. Suppose
that for some $c \in \mathbb R$, we have
\begin{equation} \label{Umudomination}
    U^{\mu}(z) \leq U^{\nu}(z) + c, \qquad z \in S(\mu).
    \end{equation}
Let
\[ A = \{ z \in \mathbb C \mid U^{\mu}(z) = U^{\nu}(z) + c \}. \]
Then
\[ \nu \mid_A \leq \mu \mid_A \]
in the sense that $\nu(B) \leq \mu(B)$ for every Borel set $B \subset A$.
\end{theorem}
\begin{proof}
By  Theorem \ref{thm:domination} we obtain from \eqref{Umudomination} that
\begin{equation} \label{eq:proofdlVP1}
    U^{\mu}(z) \leq U^{\nu}(z) + c, \qquad z \in \mathbb C.
    \end{equation}

It is enough to consider bounded Borel sets $B \subset A$.
Given such a $B$ we choose $R > 0$ such that $|z| < R/2$ for every
$z \in B$.
Let $\hat{\mu}$ and $\hat{\nu}$ be the balayages of $\mu$
and $\nu$ onto the closed disk $D_R := \{ z \in \mathbb C \mid |z| \leq |R|\}$.
By the properties of balayage we have, for certain constants $\ell_1$ and $\ell_2$,
\[ U^{\hat{\mu}}(z) = U^{\mu}(z) + \ell_1, \qquad U^{\hat{\nu}}(z) = U^{\nu}(z) + \ell_2,
    \qquad z \in D_R. \]
It then follows from \eqref{eq:proofdlVP1} that
\begin{equation} \label{eq:proofdlVP2}
    U^{\hat{\mu}}(z) \leq U^{\hat{\nu}}(z) + c + \ell_1 - \ell_2, \qquad z \in D_R
    \end{equation}
and again by Theorem \ref{thm:domination}
the inequality extends to all of $\mathbb C$, since $S(\hat{\mu}) \subset D_R$.
Equality holds in \eqref{eq:proofdlVP2} for $z \in D_R \cap A$,
and so in particular for $z \in B$.

Then by  \cite[Theorem IV.4.5]{SafTot}, we have that
$\hat{\nu}(B) \leq \hat{\mu}(B)$. Then also $\nu(B) \leq \mu(B)$
since $B$ is contained in the interior of $D_R$ and on $D_R$ the balayage
measures $\hat{\mu}$ and $\hat{\nu}$ differ
from $\mu$ and $\nu$ only on the boundary $\partial D_R =  \{ z \in \mathbb C \mid |z| = R\}$.
\end{proof}

We do not know if the condition $S(\mu) \neq \mathbb C$ is necessary in
Theorems~\ref{thm:domination} and \ref{thm:delaValleePoussin}. The condition is more
than sufficient for the purposes of this paper, since we will only
be dealing with  measures that are supported on either the real line or
the imaginary axis.

\subsection{Equilibrium problem for $\nu_3$} \label{subsec:equilibriumnu3}

Given a measure $\nu_2 \leq \sigma_2$ on $i \mathbb R$ with finite logarithmic energy
and $\int d\nu_2 = 2/3$,
the equilibrium problem for $\nu_3$ is to minimize
\begin{equation} \label{eq:nu3problem}
    I(\nu) + \int (V_3(x) - U^{\nu_2}(x)) \, d\nu(x)
    \end{equation}
among all measures $\nu$ on $\mathbb R$ with
$\int d\nu = 1/3$. In case $\alpha > 0$ we have
$V_3 \equiv 0$ and then we have that the minimizer $\nu_3$
of \eqref{eq:nu3problem} is equal to
\[ \nu_3 = \frac{1}{2} \hat{\nu}_2 = \frac{1}{2} \Bal(\nu_2, \mathbb R) \]
where $\hat{\nu}_2$ denotes the balayage of $\nu_2$
onto $\mathbb R$. Then $\nu_3$ has the density
\[ \frac{d\nu_3}{dx} = \frac{1}{\pi} \int_{i \mathbb R} \frac{|z|}{x^2 + |z|^2} d\nu_2(z) \]
and the support of $\nu_3$ is the full real line.
This is similar to what is happening
for the case $\alpha = 0$ in \cite[Section 4.2]{DuiKu2}.

In case $\alpha < 0$, the external field $V_3$ is positive on the interval
$(-x^*(\alpha), x^*(\alpha))$ and zero outside. We use this fact
to prove the following inequalities for $\nu_3$.

\begin{lemma} \label{lem:nu3balayage}
Let $\nu_3$ be the minimizer for \eqref{eq:nu3problem} among measures
on $\mathbb R$ with $\int d\nu = 1/3$.
Let $c \geq x^*(\alpha)$, and let
\begin{equation} \label{eq:defAc}
    A_c = \mathbb R \setminus (-c, c).
    \end{equation}
Then we have
\begin{equation} \label{eq:nu3estimate1}
    (\nu_3) \mid_{A_c} \,  \geq \frac{1}{2} \Bal(\nu_2, \mathbb R) \mid_{A_c}.
    \end{equation}
and
\begin{equation} \label{eq:nu3estimate2}
    (\nu_3) \mid_{A_c} \, \leq \frac{1}{2} \Bal(\nu_2, A_c).
    \end{equation}
\end{lemma}
\begin{proof}
The variational conditions associated with the minimization
problem for \eqref{eq:nu3problem} are
\begin{equation} \label{eq:nu3estimate3}
\left\{
\begin{aligned} 2 U^{\nu_3}(x)  + V_3(x) - U^{\nu_2}(x) & = \ell,
    \qquad x \in S(\nu_3), \\
    2 U^{\nu_3}(x)  + V_3(x) - U^{\nu_2}(x) & \geq \ell,
    \qquad x \in \mathbb R,
    \end{aligned}
    \right.
\end{equation}
where $\ell$ is a constant.
Let $\hat{\nu}_2 = \Bal(\nu_2, \mathbb R)$ be the balayage of $\nu_2$ onto $\mathbb R$.
Then $U^{\hat{\nu}_2} = U^{\nu_2}$ on $\mathbb R$,
so that it follows from \eqref{eq:nu3estimate3} that
\[ 2 U^{\nu_3}(x) = U^{\hat{\nu}_2}(x) - V_3(x) + \ell, \qquad x \in S(\nu_3). \]
Since $V_3(x) \geq 0$, we conclude that
\[
    2 U^{\nu_3}(x) \leq U^{\hat{\nu}_2}(x) + \ell, \qquad x \in S(\nu_3),
    \]
By the principle of domination, see Theorem \ref{thm:domination} (note that
the total masses of $2 \nu_3$ and $\hat{\nu}_2$ are equal),
we have that the inequality holds for every $x \in \mathbb C$,
\begin{equation} \label{eq:nu3estimate4}
    2 U^{\nu_3}(x) \leq U^{\hat{\nu}_2}(x) + \ell, \qquad x \in \mathbb C.
    \end{equation}

For $x \in A_c$, we have $V_3(x) = 0$ because of the definition \eqref{eq:defAc} with $c \geq x^*(\alpha)$.
Hence
\[ 2 U^{\nu_3}(x) \geq U^{\nu_2}(x) + \ell = U^{\hat{\nu}_2}(x) + \ell, \qquad x \in A_c, \]
because of the inequality in \eqref{eq:nu3estimate3}.
Then by inequality \eqref{eq:nu3estimate4}, we find that equality holds.
Thus
\[ A_c \subset \{ x \mid 2 U^{\nu_3}(x) = U^{\hat{\nu}_2}(x) + \ell \}, \]
and the inequality \eqref{eq:nu3estimate1} follows because of
\eqref{eq:nu3estimate4} and Theorem \ref{thm:delaValleePoussin}.

\medskip

The second inequality \eqref{eq:nu3estimate2} follows in a similar (even simpler) way.
Now we redefine $\hat{\nu}_2$ as the balayage of $\nu_2$ onto $A_c$:
\[ \hat{\nu}_2 = \Bal(\nu_2, A_c). \]

Let $x \in A_c$. Then $x \in S(\nu_3)$ because of \eqref{eq:nu3estimate1}, which has
already been proved.
Then we have, by the property of balayage and \eqref{eq:nu3estimate3}, since $V_3(x) = 0$ for $x \in A_c$,
\[ U^{\hat{\nu}_2}(x) = U^{\nu_2}(x)  = 2 U^{\nu_3}(x) - \ell, \qquad x \in A_c = S(\hat{\nu}_2). \]
Thus by another application
of Theorem \ref{thm:delaValleePoussin} we find \eqref{eq:nu3estimate2}.
\end{proof}

We note that it follows from \eqref{eq:nu3estimate1} with $c = x^*(\alpha)$, that
\begin{equation} \label{eq:Snu3inclusion}
    (-\infty, - x^*(\alpha)] \cup [x^*(\alpha), \infty) \subset S(\nu_3).
    \end{equation}
For every $c \geq x^*(\alpha)$, we find from \eqref{eq:nu3estimate2}
and the explicit expression for
the balayage onto $A_c$, that
\begin{equation} \label{eq:nu3densityestimate}
    \frac{d\nu_3}{dx} \leq \frac{1}{2 \pi} \frac{|x|}{\sqrt{x^2-c^2}}
        \int_{i\mathbb R} \frac{\sqrt{c^2+|z|}}{x^2 + |z|^2} \, d\nu_2(z),
            \qquad x \in \mathbb R, \, |x| > c.
        \end{equation}

In the case $\alpha < 0$ we make use of Lemma \ref{lem:convex} (c)
and Lemma \ref{lem:nu3balayage}
to conclude that the support of the minimizer $\nu_3$
is of the desired form.
This is done in the next lemma.

\begin{proposition} \label{prop:nu3}
Let $\nu_2$ be a measure on $i \mathbb R$ with
$\int d\nu_2 = 2/3$ and $\nu_2 \leq \sigma_2$. Assume $\nu_2$
has finite logarithmic energy. Let $\nu_3$ be the minimizer
of \eqref{eq:nu3problem} among all measures on $\mathbb R$
with total mass $1/3$. Then the support of  $\nu_3$ is of the form
\begin{equation} \label{eq:Snu3identity}
    S(\nu_3) = \mathbb R \setminus (-c_3, c_3)
    \end{equation}
for some $c_3 \geq 0$.

If $\alpha < 0$ then $c_3 < x^*(\alpha)$, and if $c_3 > 0$
then the density of $\nu_3$ vanishes as a square root at $\pm c_3$.
\end{proposition}

\begin{proof}
If $\alpha \geq 0$, then $S(\nu_3) = \mathbb R$ and we have \eqref{eq:Snu3identity}
with $c_3 = 0$.

For $\alpha < 0$, we use the fact that the external field $V_3(x) - U^{\nu_2}(x)$
is even. So it follows from
\cite[Theorem IV.1.10 (f)]{SafTot} that
$d\nu_3(x) = d \tilde{\nu}_3(x^2)/2$ defines a measure $\tilde{\nu}_3$
which is
the minimizer of the functional
\[ I(\nu) + 2 \int (V_3(\sqrt{x}) - U^{\nu_2}(\sqrt{x})) d\nu(x) \]
among measures $\nu$ on $[0,\infty)$ with $\int d\nu = 1/3$.
Note that
\begin{equation} \label{eq:Stildenu3}
    S(\tilde{\nu}_3) = \{ x^2 \mid x \in S(\nu_3) \}.
    \end{equation}
The external field $V_3(\sqrt{x}) - U^{\nu_2}(\sqrt{x})$ is convex
on $[0,(x^*(\alpha))^2]$ by Lemma \ref{lem:convex} (c),  which implies
by \cite[Theorem IV.1.10 (b)]{SafTot}
that $S(\tilde{\nu}_3) \cap [0,(x^*(\alpha))^2]$ is an interval.

By \eqref{eq:Snu3inclusion} we already know that $x^*(\alpha) \in
S(\nu_3)$. Then $x^*(\alpha)^2 \in S(\tilde{\nu}_3)$, and it follows
that
\[ S(\tilde{\nu}_3) \cap [0, x^*(\alpha)^2] = [c_3^2, x^*(\alpha)^2],
    \qquad \text{for some } c_3 \in [0, x^*(\alpha)]. \]
Combining this with \eqref{eq:Snu3inclusion}, \eqref{eq:Stildenu3},
we find \eqref{eq:Snu3identity}.

Now suppose $c_3 = x^*(\alpha) > 0$.
The minimizer $\nu_3$ is characterized by the variational conditions
\begin{equation} \label{eq:nu3variational}
\left\{
\begin{aligned}
    2U^{\nu_3}(x) & = U^{\nu_2}(x) - V_3(x), \qquad x\in S(\nu_3) = (-\infty, -c_3] \cup [c_3, \infty), \\
    2U^{\nu_3}(x) & > U^{\nu_2}(x) - V_3(x), \qquad x \in \mathbb R \setminus S(\nu_3),
\end{aligned}
\right.
\end{equation}
where the inequality on $\mathbb R \setminus S(\nu_3)$ is indeed strict
due to the convexity of $V_3(\sqrt{x}) - U^{\nu_2}(\sqrt{x})$.
Since $c_3 = x^*(\alpha)$, we have that $2 U^{\nu_3} = U^{\nu_2}$ on $S(\nu_3)$ and so
$2 \nu_3$ is the balayage of $\nu_2$ onto $S(\nu_3)$. Then the density of $\nu_3$
has a square root singularity at $\pm c_3$ and then it easily follows that
\[ \frac{d}{dx} U^{\nu_3}(x) = \int \frac{d\nu_3(s)}{s-x} \to + \infty
    \qquad  \textrm{as } x \to c_3+. \]
This is not compatible with \eqref{eq:nu3variational}, since $U^{\nu_2}(x) - V_3(x)$
is differentiable at $x=c_3$. Therefore $c_3 < x^*(\alpha)$ in case $x^*(\alpha) > 0$.

Finally, if $c_3 > 0$ then the density of $\nu_3$ vanishes as a square root
at $\pm c_3$ as a result of the convexity of  $V_3(\sqrt{x}) - U^{\nu_2}(\sqrt{x})$ again.
The proposition is proved.
\end{proof}

\subsection{Equilibrium problem for $\nu_1$} \label{subsec:equilibriumnu1}

Given $\nu_2$ and $\nu_3$ the equilibrium problem for $\nu_1$
is to minimize
\begin{equation} \label{eq:nu1problem}
    I(\nu) + \int (V(x) - \theta_1(x) - U^{\nu_2}(x)) d\nu(x)
    \end{equation}
among all probability measures $\nu$ on $\mathbb R$.
The minimizer exists and its support is contained
in an interval $[-X,X]$ that is independent of $\nu_2$.

This can be proved as in \cite{DuiKu2}, where weighted polynomials
were used.
Here we give a  proof using  potential theory.

\begin{lemma} \label{lem:nu1support}
Let $\tilde{\nu}_1$ be the minimizer of the weighted energy
\begin{equation} \label{eq:tildenu1problem}
    I(\nu) + \int (V(x) - \theta_1(x)) d\nu(x)
    \end{equation}
among probability measures on $\mathbb R$, and suppose
that
\begin{equation} \label{eq:tildenu1support}
    S(\tilde{\nu}_1) \subset [-X_1,X_1]
    \end{equation}
for some $X_1 > 0$. Let $\nu_2$ be a measure on
$i\mathbb R$, with finite potential $U^{\nu_2}$ and suppose
that $\nu_1$ is the minimizer for \eqref{eq:nu1problem}.
Then also
\begin{equation} \label{eq:nu1support}
    S(\nu_1) \subset [-X_1,X_1].
    \end{equation}
\end{lemma}
\begin{proof}
The variational conditions associated with the equilibrium
problems for \eqref{eq:tildenu1problem} and \eqref{eq:nu1problem}
are
\begin{equation} \label{eq:conditionstildenu1}
\left\{
\begin{aligned}
    2 U^{\tilde{\nu}_1}(x) + V_1(x) & = \tilde{\ell}_1, && x \in S(\tilde{\nu}_1), \\
    2 U^{\tilde{\nu}_1}(x) + V_1(x) & \geq \tilde{\ell}_1, &&  x \in \mathbb R,
    \end{aligned}
    \right.
    \end{equation}
and
\begin{equation} \label{eq:conditionsnu1}
\left\{
\begin{aligned}
    2 U^{\nu_1}(x) + V_1(x) - U^{\nu_2}(x) & = \ell_1, &&  x \in S(\nu_1), \\
    2 U^{\nu_1}(x) + V_1(x) - U^{\nu_2}(x) & \geq \ell_1, && x \in \mathbb R,
    \end{aligned}
    \right.
    \end{equation}
where $\tilde{\ell}_1$ and $\ell_1$ are certain constants.
Combining the relatons \eqref{eq:conditionstildenu1} and \eqref{eq:conditionsnu1}, we find
\begin{equation} \label{eq:conditionsnu1andtildenu1}
    \left\{
\begin{aligned}  2 U^{\tilde{\nu}_1}(x) - 2 U^{\nu_1}(x) + U^{\nu_2}(x)  & \leq  \tilde{\ell}_{1} - \ell_1,
    && x \in S(\tilde{\nu}_1), \\
    2 U^{\nu_1}(x) -  2 U^{\tilde{\nu}_1}(x) - U^{\nu_2}(x) & \leq \ell_1 - \tilde{\ell}_1,
    && x \in S(\nu_1).
    \end{aligned}
    \right.
\end{equation}

Since $ x \mapsto -U^{\nu_2}(x) = \int \log |x-z| \, d\nu_2(z)$ is strictly increasing
as $|x|$ increases, and since $S(\tilde{\nu}_1) \subset [-X_1,X_1]$, it follows
from the first inequality in \eqref{eq:conditionsnu1andtildenu1}
that
\[ 2 U^{\tilde{\nu}_1}(x) \leq 2 U^{\nu_1}(x) + \tilde{\ell}_{1} - \ell_1 - U^{\nu_2}(X_1),
    \qquad x \in S(\tilde{\nu}_1). \]
By the principle of domination, Theorem \ref{thm:domination}, the inequality holds everywhere
\[ 2 U^{\tilde{\nu}_1}(x) \leq 2 U^{\nu_1}(x) + \tilde{\ell}_{1} - \ell_1 - U^{\nu_2}(X_1),
    \qquad x \in \mathbb C. \]
Then for $x \in S(\nu_1)$ we find by combining this with the second inequality
in \eqref{eq:conditionsnu1andtildenu1}
\[ - U^{\nu_2}(x)  \leq  2 U^{\tilde{\nu_1}}(x) - 2 U^{\nu_1}(x) + \ell_1 - \tilde{\ell}_1
\leq - U^{\nu_2}(X_1), \]
which implies that $|x| \leq X_1$, since $x\mapsto -U^{\nu_2}(x)$ is even and strictly
increasing as $|x|$ increases.
\end{proof}

\subsection{Equilibrium problem for $\nu_2$}  \label{subsec:equilibriumnu2}

Given $\nu_1$ and $\nu_3$ on $\mathbb R$ with total masses
$\int d\nu_1 = 1$, $\int d\nu_3 = 1/3$, the equilibrium problem for $\nu_2$
is to minimize
\begin{equation} \label{eq:nu2problem}
    I(\nu) - \int (U^{\nu_1} + U^{\nu_3}) d\nu
\end{equation}
among all measures on $i \mathbb R$ with $\nu \leq \sigma_2$
and $\int d \nu_2 = 2/3$.

Recall that by Lemma \ref{lem:convex} (b) the density $d\sigma_2(iy)/|dz|$ increases
as $y > 0$ increases. Then we can use exactly the same
arguments as in Lemma 4.5 of \cite{DuiKu2} to conclude that
\begin{equation} \label{eq:Snu2}
    S(\nu_2) = S(\sigma_2)
    \end{equation}
    and
\begin{equation} \label{eq:Ssigma2minnu2}
    S(\sigma_2 - \nu_2) = (-i\infty, - ic_2] \cup [ic_2, i\infty)
    \end{equation}
for some $c_2 \geq 0$. The proof is based on iterated balayage introduced in \cite{KuiDra}.
This proof also shows the following analogue of Lemma \ref{lem:nu3balayage}.

\begin{lemma} \label{lem:nu2balayage}
Let $\nu_1$ and $\nu_3$ be measures on $\mathbb R$ with
$\int d\nu_1 = 1$ and $\int d\nu_3 = 1/3$, and having
finite logarithmic energy. Let $\nu_2$ be the minimizer
of \eqref{eq:nu2problem} among all measures $\nu$ on $i \mathbb R$
with total mass $2/3$ and $\nu \leq \sigma_2$.

Let $c \geq c_2$ and put
\begin{equation} \label{eq:defBc}
    B_c = (-i\infty, -ic] \cup [ic, i\infty).
    \end{equation}
Then we have
\begin{equation} \label{eq:nu2estimate1}
    (\nu_2) \mid_{B_c} \,  \geq \frac{1}{2} \Bal(\nu_1 + \nu_3, i \mathbb R) \mid_{B_c}.
    \end{equation}
and
\begin{equation} \label{eq:nu2estimate2}
    (\nu_2) \mid_{B_c} \, \leq \frac{1}{2} \Bal(\nu_1 + \nu_3, B_c).
    \end{equation}
\end{lemma}
\begin{proof}
This follows from the iterated balayage.
Alternatively, it could be proved from the variational
conditions associated with the minimization problems, as
we did for Lemma \ref{lem:nu3balayage}. We omit details.
\end{proof}

\subsection{Uniqueness of the minimizer} \label{subsec:uniqueness}
We write the energy functional \eqref{eq:energy} as
\begin{multline} \label{eq:energyfunctional2}
    E(\nu_1, \nu_2, \nu_3) =
    \frac{2}{3} I(\nu_1) + \frac{1}{12} I(2\nu_1 - 3 \nu_2) + \frac{1}{4} I(\nu_2 - 2 \nu_3) \\
    + \int V_1(x) d\nu_1(x) + \int V_3(x) d\nu_3(x).
    \end{multline}
    where $2 \nu_1 - 3 \nu_2$ and $\nu_2 - 2 \nu_3$ are signed
    measures with vanishing integral.
From this the uniqueness follows as in Section 4.5 of \cite{DuiKu2}.
Note that there is a mistake in formula (4.21) of \cite{DuiKu2},
since the coefficients of $I(\nu_1)$ and $I(2\nu_1 - 3\nu_2)$
are incorrect. However, the only important issue to
establish uniqueness is that
the coefficients are positive.

\subsection{Existence of the minimizer} \label{subsec:existence}

After these preparations we are able to show that the
minimizer exists. The proof follows along the lines of Section 4.6 of \cite{DuiKu2}.

We fix $p \in (1, 5/3)$. We are going to minimize the energy functional
\eqref{eq:energyfunctional2}
among all measures $\nu_1$, $\nu_2$, $\nu_3$ as before, but with
the additional restrictions that for certain given $X > 0$ and $K > 0$,
\begin{equation} \label{eq:nu1restriction}
    S(\nu_1) \subset [-X,X],
    \end{equation}
\begin{equation} \label{eq:nu2restriction}
    \frac{d\nu_2}{|dz|} \leq \frac{K}{|z|^p} \qquad \text{for } z \in i \mathbb R, \, |z| \geq X,
    \end{equation}
\begin{equation} \label{eq:nu3restriction}
    \frac{d\nu_3}{dx} \leq \frac{K}{x^p} \qquad \text{for } x \in \mathbb R, \, |x| \geq X.
    \end{equation}
The constants $X$ and $K$ are at our disposal, and later we will choose them large enough.

For any choice of $X$ and $K$ there is a unique vector of measures
that minimizes the energy functional subject to the usual constraints
as well as the additional restrictions \eqref{eq:nu1restriction}, \eqref{eq:nu1restriction},
\eqref{eq:nu1restriction}. Indeed, the additional restrictions yield
that the measures are restricted to a tight sets of measures. The
energy functional \eqref{eq:energyfunctional2} is strictly convex,
and so there is indeed a unique minimizer.

Our strategy of proof is now to show that for large enough $X$ and $K$
the additional restrictions are not effective.

\paragraph{Restriction \eqref{eq:nu1restriction}}
This is easy to do for \eqref{eq:nu1restriction}. Indeed because of
Lemma~\ref{eq:nu1support} it suffices to choose
\[ X > X_1 \]
where $X_1$ is as in the lemma.
Then it is easy to see that \eqref{eq:nu1restriction}
provides no extra restriction.

\paragraph{Restriction \eqref{eq:nu3restriction}}

The assumption $p \in (1, 5/3)$ ensures that
\begin{equation} \label{eq:Cp}
    C_p = \frac{1}{\pi} \int_0^{\infty} \frac{s^{1-p}}{1+s^2} ds  = \frac{1}{2 \sin (p\pi/2)} < 1.
    \end{equation}
Choose $\varepsilon_p > 0$ such that
\begin{equation} \label{eq:epsilonp}
    (1+ \varepsilon_p)^2 C_p < 1.
    \end{equation}
Pick $c \geq 0$ such that
\[ c  \geq x^*(\alpha) \qquad \text{ and } \qquad \sigma_2([-ic, ic]) \geq 2/3. \]
We take  $X > X_1$ such that
\begin{equation} \label{eq:Xbound1}
    \frac{|x|}{\sqrt{x^2 - c^2}} \leq 1 + \varepsilon_p, \qquad |x| \geq X,
    \end{equation}
Then $X > c$ and also
\begin{equation} \label{eq:Xbound2}
    \frac{\sqrt{c^2 + x^2}}{|x|} \leq 1 + \varepsilon_p, \qquad |x| \geq X.
\end{equation}

Assume that $\nu_2$ is a measure on $i\mathbb R$, symmetric around
the origin with $\int d\nu_2$ and satisfying the restriction
\eqref{eq:nu2restriction} as well as $\nu_2 \leq \sigma_2$. Let
$\nu_3$ be the minimizer of \eqref{eq:nu3problem} among measures
$\nu$ on $\mathbb R$ with $\int d\nu = 1/3$. Note that we do not
impose the restriction \eqref{eq:nu3restriction}.

Since $c \geq x^*(\alpha)$ we have the inequality \eqref{eq:nu3densityestimate}
which due to the symmetry of $\nu_2$ may be written as
\begin{equation} \label{eq:nu3estimate5}
    \frac{d\nu_3(x)}{dx} \leq
         \frac{1}{\pi} \frac{|x|}{\sqrt{x^2 - c^2}}
        \int_{0}^{i\infty} \frac{\sqrt{c^2 + |z|^2}}{x^2 + |z|^2} \, d\nu_2(z),
        \qquad x \in \mathbb R, \, |x| > c.
        \end{equation}

Let $x \geq X > c$. We split the integral in \eqref{eq:nu3estimate5}
into an integral from $0$ to $iX$ and from $iX$ to $i\infty$, and we
estimate using \eqref{eq:Xbound1}
\begin{align*}
 \frac{1}{\pi} \frac{x}{\sqrt{x^2-c^2}}
    \int_0^{iX} \frac{\sqrt{c^2 + |z|^2}}{x^2 + |z|^2} d\nu_2(z) & \leq
        \frac{1+\varepsilon_p}{\pi} \left(
    \max_{z \in [0,iX]} \frac{\sqrt{c^2 + |z|^2}}{x^2 + |z|^2} \right) \nu_2([0,iX]) \\
    & \leq \frac{1+\varepsilon_p}{3 \pi}  \frac{\sqrt{c^2+X^2}}{x^2},
    \end{align*}
and  using \eqref{eq:Xbound1} and \eqref{eq:Xbound2}
\begin{align*}
 \frac{1}{\pi} \frac{x}{\sqrt{x^2-c^2}}
    \int_{iX}^{i \infty} \frac{\sqrt{c^2 + |z|^2}}{x^2 + |z|^2} d\nu_2(z)
    & \leq \frac{(1+\varepsilon_p)^2}{\pi} \int_{iX}^{i\infty} \frac{|z|}{x^2 + |z|^2} d\nu_2(z) \\
    & \leq \frac{(1+\varepsilon_p)^2}{\pi} \int_{iX}^{i\infty} \frac{|z|}{x^2 + |z|^2}
    \frac{K}{|z|^p} |dz|
    \end{align*}
    where the last inequality holds since $\nu_2$ satisfies \eqref{eq:nu2restriction}.
This leads to
\begin{align*}
 \frac{1}{\pi} \frac{x}{\sqrt{x^2-c^2}}
    \int_{iX}^{i \infty} \frac{\sqrt{c^2 + |z|^2}}{x^2 + |z|^2} d\nu_2(z)
    & \leq \frac{(1+\varepsilon_p)^2 K}{\pi}
        \int_{0}^{i\infty} \frac{|z|^{1-p}}{x^2 + |z|^2}  |dz| \\
    & = \frac{(1+\varepsilon_p)^2 K}{\pi |x|^p} \int_0^{\infty} \frac{s^{1-p}}{1 + s^2} ds
    \end{align*}
where we made the change of variables $z = is$, $s \geq 0$.
Thus by \eqref{eq:Cp} we have
\begin{align*}
 \frac{1}{\pi} \frac{x}{\sqrt{x^2-c^2}}
    \int_{iX}^{i \infty} \frac{\sqrt{c^2 + |z|^2}}{x^2 + |z|^2} d\nu_2(z)
    \leq \frac{(1+\varepsilon_p)^2 C_p K}{|x|^p}, \qquad x \in \mathbb R, \, |x| \geq X.
    \end{align*}

In total we get
\begin{equation}
    \frac{d\nu_3(x)}{dx}
    \leq  \frac{1+\varepsilon_p}{3 \pi}  \frac{\sqrt{c^2+X^2}}{x^2}
      + \frac{(1+\varepsilon_p)^2 C_p K}{|x|^p}, \qquad x \in \mathbb R, \, |x| \geq X.
      \end{equation}

Since $(1+\varepsilon_p)^2 C_p < 1$ and $p < 5/3 < 2$, it is possible to take now $K$ sufficiently large,
say $K \geq K_3$, such that
\[ \frac{(1+\varepsilon_p)^2}{3 \pi}  \frac{\sqrt{c^2+X^2}}{x^2}
      + \frac{(1+\varepsilon_p)^2 C_p K}{|x|^p} \leq \frac{K}{|x|^p},
        \qquad x \in \mathbb R, \, |x| \geq X. \]
Then the additional restriction  \eqref{eq:nu3restriction} is satisfied.

\paragraph{Restriction \eqref{eq:nu2restriction}}

A similar argument shows the following. Choose
$\nu_1$ and $\nu_3$ on $\mathbb R$ with $\int d\nu_1 = 1$, $\int d\nu_3 = 1/3$
and satisfying the additional restrictions \eqref{eq:nu1restriction}
and \eqref{eq:nu3restriction}. Let $\nu_2$ be the minimizer
for \eqref{eq:nu2problem} for $\nu_2$ on $i \mathbb R$ with $\int d\nu_2 = 2/3$
and $\nu_2 \leq \sigma_2$. However, we do not impose \eqref{eq:nu2restriction}.

Then with the same choice for $X > X_1$ as above (so that
\eqref{eq:Xbound1} holds), we will find that for $K$ large enough, say $K \geq K_2$
we have
that
\[ \frac{d\nu_2}{|dz|} \leq \frac{K}{|z|^p}, \qquad z \in i \mathbb R, \, |z| \geq X. \]
That is, the restriction \eqref{eq:nu2restriction} is satisfied.

This completes the proof of existence of the minimizer for the vector
equilibrium problem.

\subsection{Proof of Theorem \ref{theo:theo1}}

After all this work the proof of Theorem \ref{theo:theo1} is short.
\begin{proof}
Existence and uniqueness of the minimizer is proved in Sections
\ref{subsec:uniqueness} and \ref{subsec:existence}.
We denote the minimizer by $(\mu_1, \mu_2, \mu_3)$.

The measure $\mu_1$ is the equilibrium measure in external field
$V_1 - U^{\mu_2}$,  which is real analytic on $\mathbb R \setminus \{0\}$.
If $0 \not\in S(\mu_1)$, then this implies by a result of Deift, Kriecherbauer and McLaughlin \cite{DeKrMc}
that $S(\mu_1)$ is a finite union of intervals with a density that has
the form \eqref{eq:rho1}. If $0\not\in S(\sigma_2-\mu_2)$ then $V_1 - U^{\mu_2}$
is real analytic on $\mathbb R$ (also at $0$) by Lemma \ref{lem:convex} (b), and again by \cite{DeKrMc}
we find that $S(\mu_1)$ is a finite union of intervals with a density \eqref{eq:rho1}.
The conditions \eqref{eq:varmu1} are the Euler-Lagrange conditions associated
with the minimization in external field, and they are valid in all cases.
This proves part (a).

The statements \eqref{eq:Smu2} about the supports of $\mu_2$ and $\sigma_2 - \mu_2$ were already proved
in Section \ref{subsec:equilibriumnu2}, see \eqref{eq:Snu2} and \eqref{eq:Ssigma2minnu2}.
The Euler-Lagrange  conditions \eqref{eq:varmu2} also follow from this.
The fact that $\rho_2$ vanishes as a square root at $\pm i c_2$ in case $c_2 > 0$
follows as in the proof of Lemma 4.5 of \cite{DuiKu2}.
The other statements in part (b) are obvious.

Part (c) follows from Proposition \ref{prop:nu3}, see also
\eqref{eq:nu3variational}. This completes the proof of Theorem
\ref{theo:theo1}.
\end{proof}

\section{A Riemann surface} \label{sec:RiemannSurface}

The rest of the paper is aimed at the proof of Theorem \ref{theo:theo2}. We
assume from now on that $(V,W, \tau)$ is regular, which means in particular that
$S(\mu_1)$ and $S(\sigma_2 - \mu_2)$ are disjoint. Then by part (a) of
Theorem \ref{theo:theo1} we have that $S(\mu_1)$ consists of a finite
union of disjoint intervals. We use this structure as well as that of the
supports of the other measures to build a Riemann surface in this section.

We start by collecting consequences of the vector equilibrium problem
in the form of properties of the $g$-functions. These will be used in
the construction of a meromorphic function $\xi$ on the Riemann surface.

\subsection{The $g$-functions}

The Euler-Lagrange variational conditions \eqref{eq:varmu1},
\eqref{eq:varmu2}, \eqref{eq:varmu3}, can be rewritten in terms of the $g$-functions
\begin{equation} \label{eq:defgj}
    g_j(z) = \int \log(z-s) \, d\mu_j(s), \qquad j=1,2,3
    \end{equation}
    that are defined as follows.

\begin{definition} For $j =1,2,3$ we define $g_j$ by the formula \eqref{eq:defgj}
with the following choice of branch for $\log(z-s)$ with $s \in S(\mu_j)$.
\begin{enumerate}
\item[\rm (a)] For $j=1,3$, we define $\log(z-s)$ for $s\in \mathbb R$
with a branch cut along $(-\infty, s]$ on the real line.
\item[\rm (b)] For $j=2$, we define $\log (z-s)$ for $s \in i\mathbb R$
with a branch cut along $(-\infty, 0] \cup [0,s]$, which is partly on the
real line and partly on the imaginary axis.
\end{enumerate}
In all cases the definition of $\log(z-s)$ is such that
\[ \log(z-s) \sim \log |z| + i \arg z, \qquad -\pi < \arg z < \pi \]
as $z \to \infty$.
\end{definition}

As a result we have that $g_1$ is defined and analytic
on $\mathbb C \setminus (-\infty, b_N]$, $g_2$ on $\mathbb C \setminus (i \mathbb R \cup \mathbb R^-)$,
and $g_3$ on $\mathbb C \setminus \mathbb R$.

In what follows we will frequently use the numbers
$\alpha_k$ given in the following definition.
\begin{definition}
We define
\begin{equation} \label{eq:defalphak}
    \alpha_k = \mu_1([a_{k+1}, +\infty)), \qquad k = 0, \ldots, N-1,
    \end{equation}
  and  $\alpha_{N}=0$.
\end{definition}
Recall that $\mu_1$ is supported on $\bigcup_{k=1}^N [a_k,b_k]$ with
$a_1 < b_1 < a_2 < \cdots < a_N < b_N$, and so
\[ \alpha_1 = 1 > \alpha_2 > \cdots > \alpha_{N-1} > \alpha_N = 0. \]

The above definitions are
such that the following hold.

\begin{lemma} \label{lem:jumpsg1}
\begin{enumerate}
\item[\rm (a)]
We have
\begin{equation} \label{eq:jumpg1onR0}
g_{1,+}(x)-g_{1,-}(x)=2\pi i \mu_1([x,\infty)) \quad \text{ for } x \in \mathbb R.
\end{equation}
In particular
\begin{equation} \label{eq:jumpg1ongap}
    g_{1,+}- g_{1,-} = 2 \pi i \alpha_k, \quad \text{ on } (b_k, a_{k+1})
\end{equation}
with $\alpha_k$ as in \eqref{eq:defalphak}.
Here we have put $b_0 = -\infty$ and $a_{N+1} = +\infty$.
\item[\rm (b)] We have
\[ g_{2,\pm}(x) = \int \log |x-s| \, d\mu_2(s) \pm \frac{2}{3} \pi i, \qquad x \in \mathbb R^-, \]
so that
    \begin{align} \label{eq:jumpg2onRmin}
        g_{2,+} - g_{2,-} = \frac{4}{3} \pi i, \qquad \text{on } \mathbb R^-.
        \end{align}
\item[\rm (c)] If $c_2 > 0$ then
    \begin{equation} \label{eq:jumpg2ongap}
    \left\{
    \begin{aligned}
        g_{2,+}(z) - g_{2,-}(z) & = \frac{2}{3} \pi i - 2\pi i \sigma_2([0,z]),
        && z \in (0, ic_2), \\
        g_{2,+}(z) - g_{2,-}(z) & = - \frac{2}{3} \pi i + 2 \pi i \sigma_2([z,0]),
        && z \in (-ic_2, 0).
        \end{aligned} \right.
        \end{equation}
\item[\rm (d)] We have \begin{align} \label{eq:jumpg3onR0} g_{3,+}(x)-g_{3,-}(x)=2\pi i \mu_3([x,\infty)), \qquad x\in \mathbb R. \end{align}
\item[\rm (e)] If $c_3 > 0$, then
\[ g_{3,\pm}(x) = \int \log |x-s| \, d\mu_3(s) \pm \frac{1}{6} \pi i, \qquad x \in (-c_3,c_3), \]
so that
    \begin{align} \label{eq:jumpg3ongap}
        g_{3,+} - g_{3,-} = \frac{1}{3} \pi i, \qquad \text{on } (-c_3,c_3).
        \end{align}

    \end{enumerate}
\end{lemma}
\begin{proof}
Parts (a) and (d) are immediate from \eqref{eq:defgj}.

Part (b) follows from the definition of $g_2$, the symmetry
of the measure $\mu_2$ on $i \mathbb R$, and the fact that $\int d \mu_2 = 2/3$.

For part (c), we note that for $z \in i \mathbb R^+$,
\begin{align*}
 g_{2,\pm}(z) = - U^{\mu_2}(z) + \frac{1}{3} \pi i \pm  \pi i \mu_2([z, i \infty)), \end{align*}
such that
\begin{align}\label{eq:jumpg2oniR0}
 g_{2,+}(z) - g_{2,-}(z) = 2 \pi i \mu_2([z, i \infty))
    = \frac{2}{3} \pi i - 2 \pi i \mu_2([0, z])
    \end{align}
since $\mu_2$ has total mass $1/3$ on $i \mathbb R^+$.
If $z \in [0, ic_2]$, then $\mu_2([0,z]) = \sigma_2([0,z])$
and the first equation in \eqref{eq:jumpg2ongap} follows.
The second equation follows in a similar way.

Part (e) follows from the definition of $g_3$, together with the fact
that $\mu_3$ is symmetric on $\mathbb R$ with total mass $1/3$, so
that $\mu_3([c_3, \infty)) = 1/6$.
\end{proof}

We have the following jump properties.
\begin{lemma} \label{lem:jumpsg2}
\begin{enumerate}
\item[\rm (a)] We have
\begin{equation} \label{eq:jumpg1onR}
\left\{
\begin{aligned}
    g_{1,+} + g_{1,-} - g_2 & = V_1 - \ell_1, && \text{ on } S(\mu_1) \cap \mathbb R^+, \\
    g_{1,+} + g_{1,-} - g_{2,\pm} & = V_1 - \ell_1 \mp \frac{2}{3}\pi i, &&
        \text{ on } S(\mu_1) \cap \mathbb R^-, \\
    \Re \left( g_{1,+} + g_{1,-} - g_2 \right) & \leq V_1 - \ell_1,
        && \text{ on } \mathbb R \setminus S(\mu_1).
        \end{aligned} \right.
        \end{equation}
\item[\rm (b)] On $i \mathbb R$ we have
\begin{equation} \label{eq:jumpg2oniR}
\left\{
\begin{aligned}
    g_{2,+} + g_{2,-}  & = g_1 + g_3, &&   \text{ on } S(\sigma_2-\mu_2), \\
   \Re \left(g_{2,+} + g_{2,-} \right) & > \Re \left(g_1 + g_3  \right),
    && \text{ on } i \mathbb R \setminus S(\sigma_2 - \mu_2).
    \end{aligned} \right.
    \end{equation}
\item[\rm (c)] We have on $\mathbb R$,
\begin{equation} \label{eq:jumpg3onR}
\left\{
\begin{aligned}
    g_{3,+} + g_{3,-} - g_2 & = V_3,  && \text{ on } S(\mu_3) \cap \mathbb R^+,  \\
    g_{3,+} + g_{3,-} - g_{2,\pm}  & = V_3 \mp \frac{2}{3}\pi i,  &&   \text{ on }  S(\mu_3) \cap \mathbb R^-, \\
    \Re \left(g_{3,+} + g_{3,-} - g_2 \right)  & < V_3,  && \text{ on } \mathbb R \setminus S(\mu_3).
    \end{aligned} \right.
    \end{equation}
\end{enumerate}
\end{lemma}

\begin{proof}
Part (a) follows from \eqref{eq:varmu1}. Indeed, we have
\begin{align*}
    2 U^{\mu_1}(x) & = - g_{1,+}(x) - g_{1,-}(x), \quad \text{ for } x \in \mathbb R, \\
     U^{\mu_2}(x) & =
    \begin{cases}
    - g_2(x), & \quad \text{ for } x \in \mathbb R^+, \\
    - g_{2,\pm}(x) \pm \frac{2}{3} \pi i, & \quad \text{ for } x \in \mathbb R^-,
    \end{cases}  \end{align*}
so that \eqref{eq:varmu1} indeed leads to \eqref{eq:jumpg1onR}.

For part (b) we note that
\[ 2 U^{\mu_2}(z) = -g_{2,+}(z) - g_{2,-}(z) \pm \frac{2}{3} \pi i,
    \qquad z \in i \mathbb R^{\pm}. \]
Also because of symmetry of $g_1$ and $g_3$ around $0$,
\[
\left\{
\begin{aligned} U^{\mu_1}(z) & = -g_1(z) \pm \frac{1}{2} \pi i, \\
    U^{\mu_3}(z) & = - g_3(z) \pm \frac{1}{6} \pi i,
    \end{aligned} \right.
    \qquad \text{for } z \in i \mathbb R^{\pm}. \]
Using this in \eqref{eq:varmu2} we obtain \eqref{eq:jumpg2oniR}.

Part (c) follows from \eqref{eq:varmu3} in the same way
that we obtained part (a) from \eqref{eq:varmu1}.
\end{proof}

In what follows we also use the derivatives of the $g$-function,
which we denote by $F_1, F_2, F_3$.

\begin{definition} \label{def:Fj}
We define
\begin{equation} \label{eq:defFj}
    F_j(z) = g_j'(z) = \int \frac{d\mu_j(s)}{z-s} , \qquad j=1,2,3,
    \end{equation}
which is defined and analytic for  $z \in \mathbb C \setminus S(\mu_j)$.
\end{definition}

The jump of $F_j$ gives us the density of $\mu_j$, since we have
\begin{equation} \label{eq:jumpFj}
\begin{aligned}
    \frac{d\mu_1}{dx} & = - \frac{1}{2\pi i} \left( F_{1,+}(x) - F_{1,-}(x) \right), \qquad x \in \mathbb R, \\
    \frac{d\mu_2}{dz} & = - \frac{1}{2\pi i} \left( F_{2,+}(z) - F_{2,-}(z) \right), \qquad z \in i \mathbb R, \\
    \frac{d\mu_3}{dx} & = - \frac{1}{2\pi i} \left( F_{3,+}(x) - F_{3,-}(x) \right), \qquad x \in \mathbb R.
    \end{aligned}
    \end{equation}

\subsection{Riemann surface $\mathcal R$ and $\xi$-functions}
We construct a four sheeted Riemann surface $\mathcal R$ in the following way.
Four sheets $\mathcal R_j$ defined as
\begin{equation} \label{eq:Riemannsurface}
\left\{
\begin{aligned}
 \mathcal R_1 & = \mathbb C \setminus S(\mu_1), \\
 \mathcal R_2 & = \mathbb C \setminus (S(\mu_1) \cup S(\sigma_2-\mu_2)), \\
 \mathcal R_3 & = \mathbb C \setminus (S(\sigma_2-\mu_2)\cup S(\mu_3)), \\
 \mathcal R_4 & = \mathbb C \setminus S(\mu_3),
\end{aligned} \right. \end{equation}
are connected as follows: $\mathcal R_1$ is
connected to $\mathcal R_2$ via $S(\mu_1)$, $\mathcal R_2$ is
connected to $\mathcal R_3$ via $S(\sigma_2-\mu_2)$ and $\mathcal
R_3$ is connected to $\mathcal R_4$ via $S(\mu_3)$. Every
connection is  in the usual crosswise manner. We compactify
the Riemann surface by adding a point at infinity to the
first sheet $\mathcal R_1$, and a second point at infinity
which is common to the other three sheets.

The genus of the Riemann surface is (recall that $S(\mu_1)$
consists of $N$ intervals and recall the classification of the cases
in Section \ref{subsec:cases})
\begin{equation} \label{eq:genusR}
    g(\mathcal R) =
    \begin{cases} N-1, &   \quad \text{in Cases I, II, and III,} \\ %  \text{in Cases I, III, and IV,}
        N,  &  \quad \text{in Cases IV and V.} % \text{ in Cases II and V.}
        \end{cases}
        \end{equation}

Using the functions $F_j$ defined in Definition \ref{def:Fj} we
define the $\xi$ functions.
\begin{definition}
We define functions $\xi_j$, $j=1,2,3,4$ by
\begin{equation} \label{eq:defxij}
\left\{
\begin{aligned}
    \xi_1 & = V' - F_1, && \text{on } \mathcal R_1, \\
    \xi_2 & = F_1 - F_2 + \theta_1', && \text{on } \mathcal R_2 \setminus i \mathbb R, \\
    \xi_3 & = F_2 - F_3 + \theta_2', && \text{on } \mathcal R_3 \setminus (\mathbb R \cup i \mathbb R), \\
    \xi_4 & = F_3 + \theta_3', &&  \text{on } \mathcal R_4 \setminus \mathbb R.
\end{aligned} \right.
\end{equation}
\end{definition}

We first prove that $\xi_2$, $\xi_3$, and $\xi_4$ are analytic on their full respective sheets.

\begin{lemma}
\begin{enumerate}
\item[\rm (a)] We have $\xi_{2,+} = \xi_{2,-}$ on $(-ic_2, ic_2)$,
    and so $\xi_2$ has an analytic extension to $\mathcal R_2$.
\item[\rm (b)] We have $\xi_{3,+} = \xi_{3,-}$ on $(-ic_2, ic_2)$ and on $(-c_3,c_3)$,
    and so $\xi_3$ has an analytic extension to $\mathcal R_3$.
\item[\rm (c)] We have $\xi_{4,+} = \xi_{4,-}$ on $(-c_3,c_3)$,
    and so $\xi_4$ has an analytic extension to $\mathcal R_4$.
\end{enumerate}
\end{lemma}

\begin{proof}
On $(-ic_2, ic_2)$ we have by \eqref{eq:jumpg2ongap} and \eqref{eq:defFj}
that
\begin{align*}
    F_{2,+}(z) - F_{2,-}(z) & = \frac{d}{dz} \left(g_{2,+}(z) - g_{2,-}(z) \right) \\
        & = -2\pi i \frac{d\sigma_2}{dz} \\
        & =   \tau (s_{1,+}(z) - s_{1,-}(z)), \qquad z \in (-ic_2,ic_2),
        \end{align*}
see \eqref{eq:sigma2withs1}.
Since $\tau s_1 = \theta_1'$, we get
\[ F_{2,+} - \theta_{1,+} = F_{2,-} - \theta'_{1,-}, \qquad \text{ on } (-ic_2, ic_2), \]
and also
\[ \xi_{2,+} = \xi_{2,-}, \qquad \text{ on } (-ic_2, ic_2) \]
since $F_1$ is analytic on the imaginary axis. This proves part (a).

We also have
\[ s_{1,+}(z) - s_{1,-}(z) = s_{2,-}(z) - s_{2,+}(z), \qquad z \in i \mathbb R. \]
Then above argument also shows that
\[ \xi_{3,+} = \xi_{3,-}, \qquad \text{ on } (-ic_2, ic_2) \]
since $F_1$ and $F_3$ are analytic on the imaginary axis.

If $c_3 > 0$ (which can only happen if $\alpha < 0$), then  $F_3$
is analytic across $(-c_3,c_3)$. Both $\xi_2$ and $\xi_3$ are analytic
across $(-x^*(\alpha), x^*(\alpha))$, and so a fortiori across $(-c_3, c_3)$,
since $c_3 < x^*(\alpha)$. This proves part (c) and the remaining statement of part (b).
\end{proof}

We continue to denote the analytic extension by $\xi_j$,
$j=2,3,4$.

\begin{proposition} \label{prop:globalmeromorphic}
The function
\[ \xi : \bigcup_{j=1}^4 \mathcal R_j \to \mathbb C \] given
by $\xi(z) = \xi_j(z)$ for $z \in \mathcal R_j$
extends to a meromorphic function (also denoted by $\xi$)
on $\mathcal R$. The meromorphic function has a pole of order
$\deg V - 1$ at infinity on the first sheet, and a simple pole at
the other point at infinity.
\end{proposition}

\begin{proof}
From \eqref{eq:jumpg1onR} and \eqref{eq:defFj} it follows that
\begin{equation} \label{eq:jumpF1onR}
    F_{1,+}(x) + F_{1,-}(x) - F_2(x) = V_1'(x), \qquad x \in S(\mu_1),
    \end{equation}
since $F_2$ is analytic on $\mathbb R \setminus \{0\}$, so that
$F_{2,\pm} = F_2$. Since $V_1 = V - \theta_1$, we obtain by
the definition \eqref{eq:defxij} that
\begin{equation} \label{eq:jumpxi1onR}
    \xi_{1,\pm}(x) = \xi_{2,\mp}(x), \qquad x \in S(\mu_1).
    \end{equation}

From the first equation in  \eqref{eq:thetacontinuations} and
the fact that $c_2 \geq y^*(\alpha)$, we obtain that
\[ \theta_{1,\pm} = \theta_{2,\mp}, \qquad \text{on } S(\sigma_2-\mu_2)
    = (-i\infty, -ic_2] \cup [ic_2, \infty). \]
Then we obtain from \eqref{eq:jumpg2oniR} that
\begin{equation} \label{eq:jumpF2oniR}
    F_{2,+}(z) + F_{2,-}(z) = F_1(z) + F_3(z), \qquad z \in S(\sigma_2 -\mu_2),
    \end{equation}
and it follows from \eqref{eq:defxij} that
\begin{equation} \label{eq:jumpxi2oniR}
    \xi_{2,\pm}(z) = \xi_{3,\mp}(z), \qquad z \in S(\sigma_2 - \mu_2).
    \end{equation}

From \eqref{eq:jumpg3onR} we similarly find
\begin{equation} \label{eq:jumpF3onR}
    F_{3,+}(x) + F_{3,-}(x) - F_2(x) = V_3'(x), \qquad x \in S(\mu_3).
    \end{equation}
We next claim that
\begin{equation} \label{eq:V3claim}
    V_3(x) =  \theta_{2,\pm}(x) - \theta_{3,\mp}(x), \qquad x \in \mathbb R.
\end{equation}
Indeed, for $x \in (-x^*(\alpha), x^*(\alpha))$ we have by \eqref{eq:V3bis} that
$V_3(x) = \theta_2(x) - \theta_3(x)$
and \eqref{eq:V3claim} holds, while for $x \in\mathbb R$ with $|x| \geq x^*$,
we have by \eqref{eq:thetacontinuations} and \eqref{eq:V3bis} that
both sides of \eqref{eq:V3claim} are equal to zero.
Using \eqref{eq:V3claim} in \eqref{eq:jumpF3onR}
we obtain from the definition \eqref{eq:defxij} that
\begin{equation} \label{eq:jumpxi3onR}
    \xi_{3,\pm}(x) = \xi_{4,\mp}(x), \qquad x \in S(\mu_3).
    \end{equation}

The analyticity of $\xi$ across the three cuts $S(\mu_1)$,
$S(\sigma_2-\mu_2)$ and $S(\mu_3)$ is now established  by
\eqref{eq:jumpxi1onR}, \eqref{eq:jumpxi2oniR}, and \eqref{eq:jumpxi3onR}.

As $z \to \infty$, we have by
\eqref{eq:defxij} and \eqref{eq:defFj} that
\[ \xi_1(z) = V'(z) + \OO\left( z^{-1} \right), \]
which implies that $\xi$ has a pole of order $\deg V - 1$ at the point
at infinity on the first sheet.

From \eqref{eq:defxij}, \eqref{eq:defFj} and
Lemma \ref{lem:thetajasymptotics} it also follows that for $j=2,3,4$,
\[ \xi_j(z) = \omega^k \tau^{4/3} z^{1/3} + \OO(z^{-1/3}),
     \qquad \text{as } z \to \infty, \]
where the value of $k$ depends on $j$
and on the quadrant in which $z \to \infty$. Since the
other point
at infinity (which is common to the second, third and fourth sheets)
is a double branch point, we have that $z^{-1/3}$ is a local
coordinate, so that $\xi$ indeed has a simple pole.
\end{proof}

\begin{remark} \label{rem:spectralcurve}
It follows from Proposition \ref{prop:globalmeromorphic} that $\xi_j$, $j=1,2,3,4$
are solutions of a quartic equation
\begin{equation} \label{eq:spectralcurve}
    \xi^4 + p_3(z) \xi^3 + p_2(z) \xi^2 + p_1(z) \xi + p_0(z)  = 0
    \end{equation}
with coefficients that are polynomial in $z$. This algebraic equation
is known as the spectral curve \cite{BeEyHa1}. The degrees of the polynomial
coefficients are determined by the degree of $V$.

Using this fact, we indicate how to remove the condition that $S(\mu_1)$
and $S(\sigma_2 - \mu_2)$ are disjoint in  part (a) of Theorem \ref{theo:theo1}.
This condition was included in order to be able to conclude that $S(\mu_1)$ is
a finite union of intervals. In case the condition does not hold, we can now
argue as follows.

Given $\varepsilon > 0$ we modify the vector equilibrium problem by requiring
that $S(\mu_1) \subset (-\infty, -\varepsilon] \cup [\varepsilon, \infty)$. Then the
proof of existence and uniqueness of the minimizer follows in the same way
as in Section \ref{sec:proofTheorem1}. Let us denote the minimizer by $(\mu_1^{\varepsilon},
\mu_2^{\varepsilon}, \mu_3^{\varepsilon})$. Because the external field $V_1 - U^{\mu_2^{\varepsilon}}$
is real analytic on $\mathbb R \setminus (-\varepsilon, \varepsilon)$ we have that
the support of $\mu_1^{\varepsilon}$ is a finite union of intervals. The
further structure of the minimizers is the same, which means
that we can construct a Riemann surface with a globally meromorphic function on it,
in the same way as we did in this section. The only difference is that the meromorphic
function may have a pole in $\pm \varepsilon$. The algebraic equation \eqref{eq:spectralcurve}
has coefficients that are rational in $z$, with $\pm \varepsilon$ as the only (simple) poles,
and so we may write the spectral curve as
\begin{equation} \label{eq:spectralcurve2}
    (z^2 - \varepsilon^2) \xi^4 + q_3(z) \xi^3 + q_2(z) \xi^2 + q_1(z) \xi + q_0(z)  = 0
    \end{equation}
with polynomial coefficients $q_j(z)$, whose degrees is determined by the degree of $V$.
In particular, the degrees do not depend on $\varepsilon$.

As $\varepsilon \to 0$ we have that $\mu_j^{\varepsilon}$ converges weakly to $\mu_j$ for $j=1,2,3$,
and it is not difficult to show that the spectral curve \eqref{eq:spectralcurve2} has a limit
as well. Then $V_1'(z) - \int \frac{d\mu_1(s)}{z-s}$ is the solution of an algebraic equation and therefore
the support of $\mu_1$ is a finite union of intervals.
\end{remark}

\subsection{Properties of the $\xi$ functions}

From \eqref{eq:jumpFj} and the definition \eqref{eq:defxij} we find that
\begin{equation}  \label{eq:densityofmuj}
     \begin{aligned}
    \frac{d\mu_1}{dx} & = \frac{1}{2\pi i} \left( \xi_{1,+}(x) - \xi_{1,-}(x) \right), \qquad x \in \mathbb R, \\
    \frac{d\mu_2}{dz} & = \frac{1}{2\pi i} \left( \xi_{2,+}(z) - \xi_{2,-}(z) \right)
        - \frac{1}{2\pi i} \left( \theta'_{1,+}(z) - \theta'_{1,-}(z) \right) \\
    & = -\frac{1}{2\pi i} \left( \xi_{3,+}(z) - \xi_{3,-}(z) \right)
        + \frac{1}{2\pi i} \left( \theta'_{2,+}(z) - \theta'_{2,-}(z) \right), \quad z \in i \mathbb R, \\
    \frac{d\mu_3}{dx} & = \frac{1}{2\pi i} \left( \xi_{3,+}(x) - \xi_{3,-}(x) \right)
        - \frac{1}{2\pi i} \left( \theta'_{2,+}(x) - \theta'_{2,-}(x) \right) \\
    & = - \frac{1}{2\pi i} \left( \xi_{4,+}(x) - \xi_{4,-}(x) \right)
        + \frac{1}{2\pi i} \left( \theta'_{3,+}(x) - \theta'_{3,-}(x) \right), \quad x \in \mathbb R.
        \end{aligned}
\end{equation}

Since $\xi$ is an odd function with a simple pole at infinity on sheets $2$, $3$ and $4$,
there is an expansion in the local coordinate $z^{-1/3}$:
\[ \xi(z) = c_{-1} z^{1/3} + c_1 z^{-1/3} + c_3 z^{-1} + c_5 z^{-5/3} + \cdots \]
with certain real constants $c_j$.
Keeping track of the principal branches, we obtain the following
for the asymptotic behavior of $\xi_2$, $\xi_3$ and $\xi_4$ as $z \to \infty$.
Note that the structure of the formulas in terms of the factors $\omega$
and $\omega^2$ is the same as in Lemma \ref{lem:sjasymptotics}
and Lemma \ref{lem:thetajasymptotics}.

\begin{lemma} \label{lem:xijasymptotics}
We have as $z \to \infty$
\begin{align*}
   \xi_2(z) & = \begin{cases}
   c_{-1} z^{1/3} + c_1 z^{-1/3} + c_3 z^{-1} + c_5 z^{-5/3} + \cdots, & \text{in } I \cup IV, \\
   c_{-1} \omega z^{1/3} + c_1 \omega^2 z^{-1/3} + c_3 z^{-1} + c_5 \omega z^{-5/3} + \cdots,
        & \text{in } II, \\
   c_{-1} \omega^2 z^{1/3} + c_1 \omega z^{-1/3} + c_3 z^{-1} + c_5 \omega^2 z^{-5/3} + \cdots,
        & \text{in } III,
        \end{cases} \\
   \xi_3(z) & = \begin{cases}
   c_{-1} \omega z^{1/3} + c_1 \omega^2 z^{-1/3} + c_3 z^{-1} + c_5 \omega z^{-5/3} + \cdots,
        & \text{in } I, \\
   c_{-1} z^{1/3} + c_1 z^{-1/3} + c_3 z^{-1} + c_5 z^{-5/3} + \cdots, & \text{in } II \cup III, \\
   c_{-1} \omega^2 z^{1/3} + c_1 \omega z^{-1/3} + c_3 z^{-1} + c_5 \omega^2 z^{-5/3} + \cdots,
        & \text{in } IV,
        \end{cases} \\
   \xi_4(z) & = \begin{cases}
   c_{-1} \omega^2 z^{1/3} + c_1 \omega z^{-1/3} + c_3 z^{-1} + c_5 \omega^2 z^{-5/3} + \cdots,
        & \text{in } I \cup II, \\
   c_{-1} \omega z^{1/3} + c_1 \omega^2 z^{-1/3} + c_3 z^{-1} + c_5 \omega z^{-5/3} + \cdots,
        & \text{in } III \cup IV.
        \end{cases}
\end{align*}
with constants
\begin{equation} \label{eq:constants}
    c_{-1} = \tau^{4/3}, \qquad c_1 = - \frac{\alpha}{3} \tau^{2/3},
    \qquad c_3 = 1/3. \end{equation}
\end{lemma}
\begin{proof}
The form of the asymptotics follows from the definition \eqref{eq:defxij}
snd Lemmas \ref{lem:sjasymptotics} and \ref{lem:thetajasymptotics}.
Since $\mu_1$ is a probability measure with compact support, and symmetric
with respect to the origin, we have
\begin{equation} \label{eq:F1expansion}
    F_1(z) = \frac{1}{z} + \OO(z^{-3}).
    \end{equation}
From
\begin{equation} \label{eq:F2formula}
    F_2(z) = \theta_1'(z) - \xi_2(z) + F_1(z)
        = \tau s_1(z) - \xi_2(z) + F_1(z)
        \end{equation}
we see that $F_2$ has a series representation around
$z = \infty$ in powers of $z^{1/3}$. It should start with
$F_2(z) = \frac{2}{3} z^{-1} + \cdots$ since $\mu_2$
has total mass $2/3$, see \eqref{eq:defFj}.
Then the coefficients for $z^{1/3}$ and $z^{-1/3}$ should vanish
which by Lemma \ref{lem:sjasymptotics},
and \eqref{eq:F2formula} leads to the expressions in \eqref{eq:constants} for $c_{-1}$ and $c_{1}$.
The coefficient of $z^{-1}$ should be $2/3$, which by \eqref{eq:F1expansion}
and \eqref{eq:F2formula} leads to $c_3 = 1/3$.
\end{proof}

Combining \eqref{eq:F1expansion} and \eqref{eq:F2formula} with
Lemmas \ref{lem:sjasymptotics} and \ref{lem:xijasymptotics}
we also find that there is a real constant
\begin{equation} \label{eq:constantC}
    C = c_5 - \frac{\alpha^3}{81} \tau^{-2/3}
    \end{equation}
such that
\begin{align} \label{eq:F2expansion}
 F_2(z) & = \begin{cases}
    \frac{2}{3} z^{-1} - C z^{-5/3} + \OO(z^{-2}), &\qquad z \in I \cup IV, \\
    \frac{2}{3} z^{-1} - C \omega z^{-5/3} + \OO(z^{-2}), & \qquad z \in II, \\
    \frac{2}{3} z^{-1} - C \omega^2 z^{-5/3} + \OO(z^{-2}), & \qquad z \in III
    \end{cases}
\end{align}
as $z \to \infty$.
 Then by \eqref{eq:defxij}, and Lemmas \ref{lem:thetajasymptotics} and  \ref{lem:xijasymptotics},
 we also get
 \begin{align} \label{eq:F3expansion}
F_3(z) & = \begin{cases}
    \frac{1}{3} z^{-1} + C \omega^2 z^{-5/3} + \OO(z^{-2}), & \qquad z \in I \cup II, \\
    \frac{1}{3} z^{-1} + C \omega z^{-5/3} + \OO(z^{-2}), & \qquad z \in III \cup IV,
    \end{cases}
    \end{align}
 as $z \to \infty$, with the same constant $C$.
Using this and Lemma \ref{lem:thetajasymptotics} in \eqref{eq:densityofmuj}
we find that
\begin{equation} \label{eq:mu23decay}
\begin{aligned}
    \frac{d\mu_2}{|dz|} & = \frac{\sqrt{3}}{2\pi} C |z|^{-5/3} + \OO(z^{-2}),
        && \text{ as } |z| \to \infty, \, z \in i \mathbb R \\
    \frac{d\mu_3}{dx} & = \frac{\sqrt{3}}{2\pi} C  x^{-5/3} + \OO(x^{-2}),
        && \text{ as } x \to \infty, \, x \in \mathbb R.
        \end{aligned}
       \end{equation}
which gives the precise rate of decay of the densities of $\mu_2$ and $\mu_3$
along the imaginary and real axis, respectively.
It also follows from \eqref{eq:mu23decay} that the constant is positive, $C > 0$.

Furthermore, after integration, we find from \eqref{eq:F1expansion}, \eqref{eq:F2expansion},
and \eqref{eq:F3expansion}
\begin{align*}
    g_1(z) & = \log z + \OO(z^{-1}) \\
    g_2(z) & = \begin{cases}
    \frac{2}{3} \log z + \frac{3}{2} C z^{-2/3} + \OO(z^{-1}), & \qquad z \in I \cup IV, \\
    \frac{2}{3} \log z + \frac{3}{2} C \omega z^{-2/3} + \OO(z^{-1}), & \qquad z \in II, \\
    \frac{2}{3} \log z + \frac{3}{2} C \omega^2 z^{-2/3} + \OO(z^{-1}), & \qquad z \in III,
    \end{cases} \\
    g_3(z) & = \begin{cases}
    \frac{1}{3} \log z - \frac{3}{2} C \omega^2 z^{-2/3} + \OO(z^{-1}), & \qquad z \in I \cup II, \\
    \frac{1}{3} \log z - \frac{3}{2} C \omega z^{-2/3} + \OO(z^{-1}),  &\qquad z \in III \cup IV
    \end{cases}
    \end{align*}
as $z \to \infty$. It may be verified from \eqref{eq:defgj} that the constant of integration
indeed vanishes.

\subsection{The $\lambda$ functions} \label{subsec:lambda}

The $\lambda$ functions are primitive functions of the $\xi$ functions
(that is, Abelian integrals).
They are defined as follows.
\begin{definition} We define
\begin{equation} \label{eq:deflambdaj}
\left\{
\begin{aligned}
    \lambda_1 & = V - g_1 - \ell_1, \\
    \lambda_2 & = g_1 - g_2 + \theta_1, \\
    \lambda_3 & = g_2 - g_3 + \theta_2, \\
    \lambda_4 & = g_3 + \theta_3.
    \end{aligned} \right.
    \end{equation}
\end{definition}

It is convenient to consider each $\lambda_j$ function as defined on $\mathcal R_j$
with an extra cut to ensure single-valuedness.
Thus, $\lambda_1$ is defined and analytic on $\mathbb C \setminus (-\infty, b_N]$,
$\lambda_2$  on $\mathbb C \setminus ( (-\infty, b_N] \cup i \mathbb R)$,
$\lambda_3$ on $\mathbb C \setminus (\mathbb R \cup i \mathbb R)$,
and $\lambda_4$ on $\mathbb C \setminus \mathbb R$.

The $\lambda$-functions satisfy the following jump conditions
that follow by combining the jumps for the $g$-functions given in Lemmas \ref{lem:jumpsg1}
and \ref{lem:jumpsg2} with the properties \eqref{eq:thetacontinuations} of
the functions $\theta_j$.

\begin{lemma} \label{lem:jumpslambda}
\begin{enumerate}
\item[\rm (a)] We have
\begin{equation} \label{eq:lambdajump1}
\left\{
\begin{aligned}
    \lambda_{1,\pm} - \lambda_{2,\mp}  & = 0,  \qquad \quad  \text{ on } S(\mu_1) \cap \mathbb R^+, \\
    \lambda_{1,\pm} - \lambda_{2,\mp}  & = \mp \frac{2}{3} \pi i, \quad   \text{ on } S(\mu_1) \cap \mathbb R^-.
    \end{aligned} \right.
    \end{equation}
and for $k = 0, \ldots, N$,
\begin{equation} \label{eq:lambdajump2}
\left\{
\begin{aligned}
    \lambda_{1,+} - \lambda_{1,-} & = - 2 \pi i \alpha_k, \\
    \lambda_{2,+} - \lambda_{2,-} & =  2 \pi i \alpha_k,
    \end{aligned} \right.
    \quad \text{ on } (b_k, a_{k+1}) \cap \mathbb R^+,
    \end{equation}
\begin{equation} \label{eq:lambdajump2b}
\left\{
\begin{aligned}
    \lambda_{1,+} - \lambda_{1,-} & = - 2 \pi i \alpha_k, \\
    \lambda_{2,+} - \lambda_{2,-} & =  2 \pi i \alpha_k - \frac{4}{3} \pi i,
    \end{aligned} \right.
    \quad \text{ on } (b_k, a_{k+1}) \cap \mathbb R^-,
    \end{equation}
where $b_0 = - \infty$, $a_{N+1} = +\infty$, and $\alpha_k$
is given by \eqref{eq:defalphak} for $k=0, \ldots, N$.
\item[\rm (b)] On the imaginary axis we have
\begin{align} \label{eq:lambdajump3}
    \lambda_{2,\pm} - \lambda_{3,\mp} & = 0, \qquad \text{on } S(\sigma_2 - \mu_2),
    \end{align}
and
\begin{equation} \label{eq:lambdajump4}
\left\{
\begin{aligned}
    \lambda_{2,+}- \lambda_{2,-} & = \mp \frac{2}{3} \pi i, \\
    \lambda_{3,+}- \lambda_{3,-} & = \pm \frac{2}{3} \pi i,
    \end{aligned} \right.  \quad \text{ on } (-ic_2, ic_2) \cap i \mathbb R^{\pm}.
    \end{equation}
\item[\rm (c)] Finally,
\begin{equation} \label{eq:lambdajump5}
\left\{
\begin{aligned}
    \lambda_{3,\pm} - \lambda_{4,\mp} & = 0, \qquad \quad   \text{ on }  (c_3, \infty), \\
    \lambda_{3,\pm} - \lambda_{4,\mp} & = \mp \frac{2}{3} \pi i, \quad  \text{ on }  (-\infty, -c_3),
    \end{aligned} \right.
    \end{equation}
and
\begin{equation} \label{eq:lambdajump6}
    \left\{
    \begin{aligned}
    \lambda_{3,+}- \lambda_{3,-} & = - \frac{1}{3} \pi i, \\
    \lambda_{4,+}- \lambda_{4,-} & = \frac{1}{3} \pi i,
    \end{aligned} \right. \qquad \text{on } (-c_3, c_3).
    \end{equation}
\end{enumerate}
\end{lemma}

\begin{proof}
The equations follow by combining the jumps for the $g$-functions given in Lemmas \ref{lem:jumpsg1}
and \ref{lem:jumpsg2} with the properties \eqref{eq:thetacontinuations} of
the functions $\theta_j$.
Only \eqref{eq:lambdajump4} requires more explanation.

From \eqref{eq:jumpg2ongap} and \eqref{eq:deflambdaj} we get for $z \in (0, ic_2)$,
\[ \lambda_{2,+}(z) - \lambda_{2,-}(z)
     = - \frac{2}{3} \pi i + 2 \pi i \sigma_2([0,z]) + (\theta_{1,+}(z) - \theta_{1,-}(z)) \]
where we used the fact that $g_1$ and $g_3$ are analytic on the imaginary axis.
Since $\sigma_2$ has density \eqref{eq:sigma2withs1} and $\theta_1' = \tau s_1$,
we have
\begin{align*} \sigma_2([0,z]) = \int_0^z \frac{d\sigma_2}{dz}(z') dz'
    & = \frac{1}{2\pi i} \int_0^z (\theta_{1,-}'(z') - \theta_{1,+}'(z')) dz' \\
    &  = \frac{1}{2\pi i} \left( \theta_{1,-}(z) - \theta_{1,+}(z) - \theta_{1,-}(0) + \theta_{1,+}(0) \right)
    \end{align*}
and the first equality in \eqref{eq:lambdajump4} follows for $z \in (0,ic_2)$
since $\theta_{1,-}(0) = \theta_{1,+}(0)$.
The other equalities contained in \eqref{eq:lambdajump4} follow in the same way.
\end{proof}

The $\lambda$ functions also satisfy a number of inequalities.

\begin{lemma} We have the following inequalities
\begin{align}\label{eq:varlambda21}
    \Re(\lambda_{2,+}-\lambda_{1,-})<0 & \quad \text{on } \mathbb R \setminus S(\mu_1), \\
    \label{eq:varlambda23}
    \Re(\lambda_{2,+}-\lambda_{3,-})<0 & \quad \text{on } i\R\setminus S(\sigma_2-\mu_2), \\
    \label{eq:varlambda43}
    \Re(\lambda_{4,+}-\lambda_{3,-})<0 & \quad \text{on } \R\setminus S(\mu_3).
\end{align}
\end{lemma}
\begin{proof}
The inequalities follow from the inequalities in Lemma \ref{lem:jumpsg2}.

For example, to obtain \eqref{eq:varlambda21} we note that by \eqref{eq:deflambdaj}
and \eqref{eq:V1bis}
\begin{align*}
    \lambda_{2,+} - \lambda_{1,-} & = (g_1 - g_2 + \theta_1)_+ - (V-g_1 - \ell_1)_- \\
        & = g_{1,+} + g_{1,-} - g_{2,+} - V_1 + \ell_1
        \end{align*}
which indeed has a negative real part on $\mathbb R \setminus S(\mu_1)$ because
of the inequality in \eqref{eq:jumpg1onR}. The inequality is strict because of the
regularity assumption.

The other inequalities \eqref{eq:varlambda23} and \eqref{eq:varlambda43} follow in a similar way.
\end{proof}

The asymptotic behavior of the $\lambda$-functions follows
by combining the definition \eqref{eq:deflambdaj}, which we state for ease of future reference.
\begin{lemma}
We have as $z \to \infty$
\begin{align}\label{eq:asymlambda1}
    \lambda_1(z)  & = V(z) - \log z - \ell_1 + \OO(z^{-1}),\\
    \label{eq:asymlambda2}\lambda_2(z) & = \theta_1(z) + \begin{cases}
    \frac{1}{3} \log z - \frac{3}{2} C z^{-2/3} + \OO(z^{-1}),  &  z \in I \cup IV, \\
    \frac{1}{3} \log z - \frac{3}{2} C \omega z^{-2/3} + \OO(z^{-1}), & z \in II, \\
    \frac{1}{3} \log z - \frac{3}{2} C \omega^2 z^{-2/3} + \OO(z^{-1}), & z \in III,
    \end{cases} \\
    \label{eq:asymlambda3}\lambda_3(z) & = \theta_2(z) + \begin{cases}
    \frac{1}{3} \log z - \frac{3}{2} C \omega z^{-2/3} + \OO(z^{-1}), & z \in I, \\
    \frac{1}{3} \log z - \frac{3}{2} C z^{-2/3} + \OO(z^{-1}), & z \in II \cup III, \\
    \frac{1}{3} \log z - \frac{3}{2} C \omega^2 z^{-2/3} + \OO(z^{-1}), & z \in IV,
    \end{cases} \\
    \label{eq:asymlambda4}\lambda_4(z) & = \theta_3(z) +  \begin{cases}
    \frac{1}{3} \log z - \frac{3}{2} C \omega^2 z^{-2/3} + \OO(z^{-1}), &  z \in I \cup II, \\
    \frac{1}{3} \log z - \frac{3}{2} C \omega z^{-2/3} + \OO(z^{-1}), & z \in III \cup IV,
    \end{cases}
    \end{align}
where $C$ is the constant from \eqref{eq:constantC}.
\end{lemma}

\section{Pearcey integrals and the first transformation of the
RH problem} \label{sec:Pearcey}

Now we come to the first transformation of the RH problem \eqref{eq:RHforY} which
as in \cite{DuiKu2} will be done with the help of Pearcey integrals. The present
setup is however slightly different from the one in \cite{DuiKu2}, since we
will construct a RH problem for the Pearcey integrals that
is $n$-dependent.

\subsection{Definitions}
We note that the weights from \eqref{eq:weightswj} can be written as
\begin{equation} \label{eq:weightsinp0}
\left\{
\begin{aligned}
    w_{0,n}(x) & = { e}^{-n V(x)} p_{0,n}(x), \\
    w_{1,n}(x) & = (n \tau)^{-1} { e}^{-n V(x)} p'_{0,n}(x), \\
    w_{2,n}(x) & = (n \tau)^{-2} { e}^{-n V(x)} p''_{0,n}(x),
\end{aligned} \right.
\end{equation}
where
\begin{equation} \label{eq:defp0n}
    p_{0,n}(x) = \int_{-\infty}^{\infty} { e}^{-n(W(s) - \tau xs)} ds,
        \qquad W(s) = \frac{1}{4} s^4 + \frac{\alpha}{2} s^2,
    \end{equation}
is a Pearcey integral that satisfies the third order linear ODE
\begin{equation} \label{eq:PearceyODE}
    p''' + n^2 \tau^2 \alpha p' - n^3 \tau^4 x p = 0.
    \end{equation}

Other solutions to the same ODE are given by similar integrals.
\begin{definition}
For $j=0, \ldots, 5$ and $n \in \mathbb N$, we define
\begin{equation} \label{Pearcey-integrals}
      p_{j,n}(z) = \int_{\Gamma_j} { e}^{-n (W(s) - \tau zs)} ds,
        \qquad z \in \mathbb C,
\end{equation}
where the contours $\Gamma_j$ are
\begin{equation} \left\{
        \begin{array}{ll}
            \Gamma_0= (-\infty,\infty),  & \Gamma_1=(i \infty,0]\cup [0,\infty), \\
            \Gamma_2= (i \infty,0]\cup [0,-\infty), \qquad &\Gamma_3= (-i \infty,0]\cup [0,-\infty),\\
            \Gamma_4=(-i \infty,0]\cup [0,\infty), \qquad &\Gamma_5=i  \mathbb R,
        \end{array} \right.
   \end{equation}
or homotopic deformations such as the ones shown in Figure \ref{fig:PearceyC},
and with the orientation as  also shown in
Figure \ref{fig:PearceyC}.
\end{definition}

All Pearcey integrals \eqref{Pearcey-integrals} are entire functions in
the complex plane.

    \begin{figure}
        \centering
          \includegraphics[angle=270,scale=0.4]{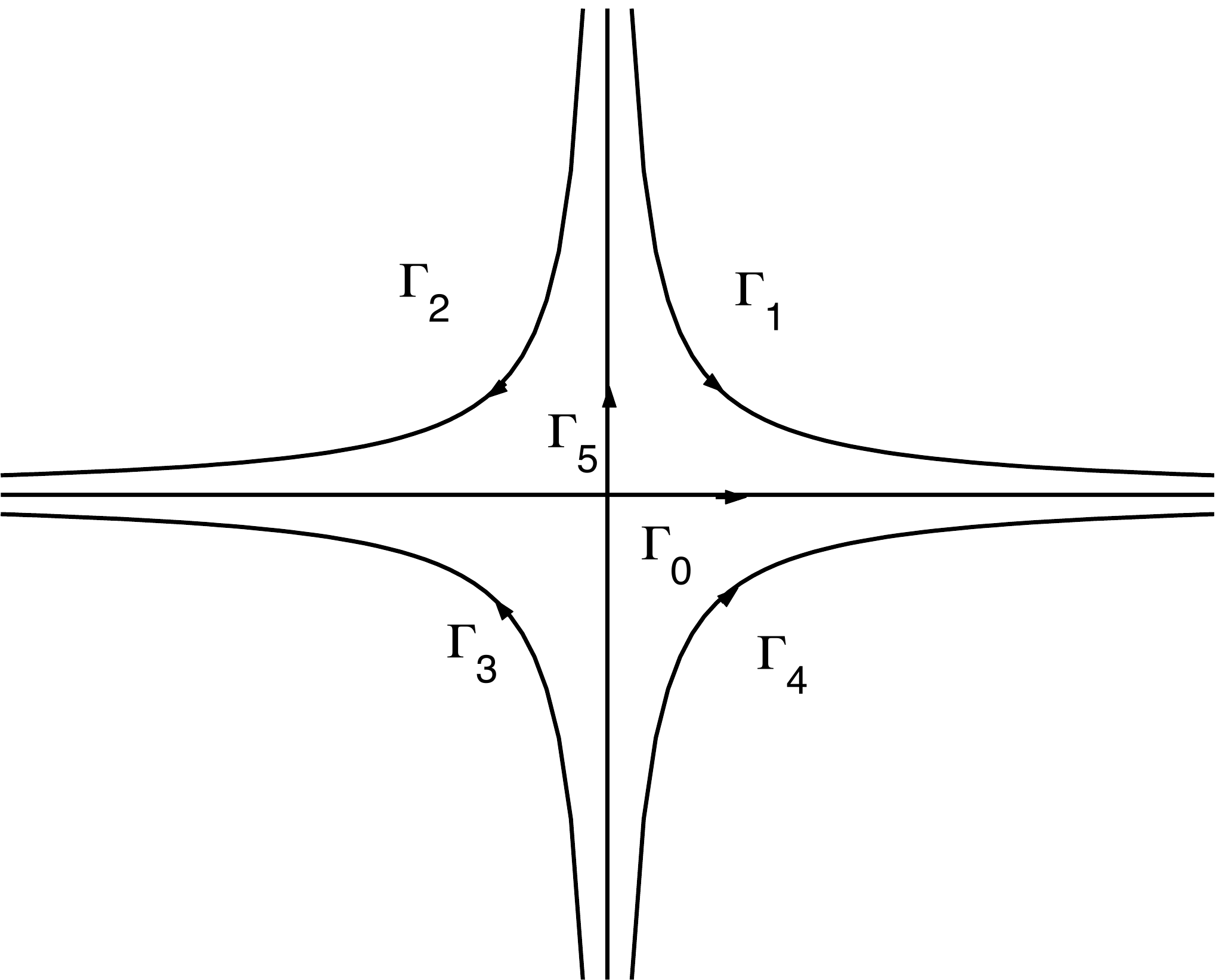}
         \caption{Contours $\Gamma_j$ in the definition of the Pearcey integrals}
         \label{fig:PearceyC}
    \end{figure}

Certain combinations of these functions are used to build a
$3 \times 3$ matrix valued $P_n : \mathbb C \setminus (\mathbb R \cup i \mathbb R)
\to \mathbb C^{3 \times 3}$ as follows.

\begin{definition}
In each of the four quadrants (denoted $I$, $II$, $III$ and $IV$)
we define
\begin{align} \label{eq:defPn}
P_n = \begin{cases} \begin{pmatrix}
    p_{0,n}   &   -p_{2,n}   &   -p_{5,n}  \\
    p_{0,n}'  &   -p_{2,n}'  &   -p_{5,n}' \\
    p_{0,n}'' &   -p_{2,n}'' &   -p_{5,n}''
\end{pmatrix} & \textrm{ in } I, \\
     \begin{pmatrix}
    p_{0,n}   &   -p_{1,n}   &   -p_{5,n}  \\
    p_{0,n}'  &   -p_{1,n}'  &   -p_{5,n}' \\
    p_{0,n}'' &   -p_{1,n}'' &   -p_{5,n}''
\end{pmatrix} & \textrm{ in } II, \\
     \begin{pmatrix}
    p_{0,n}   &   -p_{4,n}   &   -p_{5,n}  \\
    p_{0,n}'  &   -p_{4,n}'  &   -p_{5,n}' \\
    p_{0,n}'' &   -p_{4,n}'' &   -p_{5,n}''
\end{pmatrix} & \textrm{ in } III, \\
 \begin{pmatrix}
    p_{0,n}   &   -p_{3,n}   &   -p_{5,n}  \\
    p_{0,n}'  &   -p_{3,n}'  &   -p_{5,n}' \\
    p_{0,n}'' &   -p_{3,n}'' &   -p_{5,n}''
\end{pmatrix} & \textrm{ in } IV.
\end{cases}
\end{align}
\end{definition}

Then $P_n$ is indeed defined and analytic in
$\mathbb C  \setminus \left(\mathbb R \cup i \mathbb R\right)$
with the following jump properties on the real line (oriented
from left to right) and the imaginary axis (oriented from bottom to top).
As usual the orientation determines the $+$ and $-$ side of the curve,
where the $+$ side is on the left and the $-$ side on the right.

\begin{lemma} \label{lem:Pnjumps}
The matrix valued function $P_n$ satisfies the jump relations
\begin{align} \label{eq:jumpPnonR}
P_{n,+}(z) &= P_{n,-}(z) \begin{pmatrix}
                            1 & 0 & 0\\
                            0 & 1 & 0\\
                            0 & -1 & 1
                            \end{pmatrix}, &&   z\in  \mathbb R, \\
      \label{eq:jumpPnoniR}
P_{n,+}(z) & = P_{n,-}(z) \begin{pmatrix}
                            1 & -1 & 0\\
                            0 & 1 & 0\\
                            0 & 0 & 1
                            \end{pmatrix}, && z\in i \mathbb R.
\end{align}
\end{lemma}
\begin{proof}
The jumps easily follow from the definitions. It follows for
example from the definition of the Pearcey integrals
and the orientation of the contours $\Gamma_j$, see Figure~\ref{fig:PearceyC}, that
\[ -p_{2,n} = - p_{3,n} + p_{5,n}. \]
This relation implies the jump relation \eqref{eq:jumpPnonR} for $z \in \mathbb R^+$
in view of the definition of $P_n$ in the first and fourth quadrant given in
\eqref{eq:defPn}.

The other jumps are proved in a similar way.
\end{proof}

\subsection{Large $z$ asymptotics}

The large $z$ behavior of $P_n(z)$ is obtained from a classical saddle
point analysis of the integrals \eqref{Pearcey-integrals} that defines
each of its entries.
Recall that the saddle point equation \eqref{eq:saddlepoint2} has
the three solutions $s_1(z)$, $s_2(z)$ and $s_3(z)$, see Section
\ref{subsec:saddlepoint}. The value
of $-W(s) + \tau zs$ at the saddle $s_j(z)$ is denoted by $\theta_j(z)$,
as in Section \ref{subsec:saddlevalues}.

\begin{lemma} \label{lem:Pnasymptotics0}
We have as $z \to \infty$,
\begin{multline} \label{eq:Pnasymptotics}
    P_n(z) = \sqrt{\frac{2\pi n}{3}} \tau
    \begin{pmatrix}  n^{-1} \tau^{-4/3} z^{-1/3} & 0 & 0 \\ 0 & 1 & 0 \\ 0 & 0 & n \tau^{4/3} z^{1/3} \end{pmatrix} \\
    \times
    \left(I + \frac{\alpha}{6} (\tau z)^{-2/3} \begin{pmatrix} 0 & 1 & 0 \\ 0 & 0 & -1 \\ -3 & 0  & 0 \end{pmatrix}
        + \OO(z^{-4/3})\right) \\
        \times
        \begin{cases}
        \begin{pmatrix} 1 & -\omega^2 & -\omega \\
            1 &  -1 & -1 \\
            1 & -\omega & -\omega^2 \end{pmatrix}
            { e}^{n \Theta(z)}, &  \text{$z$ in $I$}, \\
        \begin{pmatrix}  -\omega^2 & -1 & -\omega \\
            -1 &  -1 & -1 \\
             -\omega & -1 & -\omega^2 \end{pmatrix}
            { e}^{n \Theta(z)}, & \text{$z$ in $II$}, \\
        \begin{pmatrix}  -\omega & -1 & \omega^2 \\
            -1 &  -1 & 1 \\
             -\omega^2 & -1 & \omega \end{pmatrix}
            { e}^{n \Theta(z)}, & \text{$z$ in $III$}, \\
        \begin{pmatrix}  1 & -\omega & \omega^2 \\
            1 &  -1 & 1 \\
             1 & -\omega^2 & \omega \end{pmatrix}
            { e}^{n \Theta(z)}, & \text{$z$ in $IV$},
            \end{cases}
                \end{multline}
where $\Theta$ is defined as in \eqref{eq:defTheta}.
\end{lemma}
\begin{proof}
The proof is a tedious  saddle point analysis for all integrals
that define $P_n$ in the respective quadrants, see also \cite{BleKu3}.
The ODE \eqref{eq:PearceyODE} can
be used to find the form of the asymptotic expansion. Indeed, putting
$p = { e}^{nF}$ in \eqref{eq:PearceyODE} we obtain the following
differential equation for $f = F'$,
\[ f^3 -  \tau^4 z + \tau^2 \alpha f +  \tfrac{3}{n} f f' + \tfrac{1}{n^2} f'' = 0. \]
This nonlinear ODE has solutions $f$ with expansions
\begin{multline*}
    f(z) =  \omega^k \tau^{4/3} z^{1/3} -  \frac{\alpha}{3} \omega^{2k} \tau^{2/3} z^{-1/3} -
    \frac{1}{3n} z^{-1} \\
    + \left(\frac{\alpha^3}{81} - \frac{\alpha}{9n}\right) \omega^k \tau^{-2/3} z^{-5/3} + \OO(z^{-7/3}),
    \qquad k = 0,1,2,
    \end{multline*}
as $z \to \infty$ in one of the quadrants. Here $\omega = { e}^{2\pi i/3}$, as before.
Thus after integration and by Proposition \ref{lem:thetajasymptotics},
we have for a certain $j = 1,2,3$,
\[ F(z) = \theta_{j}(z) - \frac{1}{3n} \log z + \frac{1}{n} \log C +
    \frac{\alpha}{6n} \omega^k (\tau z)^{-2/3} + \OO(z^{-4/3}) \]
where $C$ is a constant, and so there are solutions of the Pearcey ODE
\eqref{eq:PearceyODE} that behave like
\begin{equation} \label{eq:enFasympform}
    p(z) = C z^{-1/3} { e}^{n \theta_j(z)}
    \left(1 + \frac{\alpha}{6}  \omega^k (\tau z)^{-2/3} + \OO(z^{-4/3}) \right)
    \end{equation}
as $z \to \infty$ in one of the quadrants.

Each of the functions $p_{0,n}, \ldots, p_{5,n}$
that appears in the definition of $P_n$ in a certain quadrant
has an asymptotic expansion as in \eqref{eq:enFasympform}. We have
to associate with each such function a value of $j$ and the
corresponding value of $k$.
To do this, we have to perform a saddle point analysis on the
integral representation \eqref{Pearcey-integrals}.

The saddle point equation
\begin{equation} \label{eq:speq}
    W'(s) - \tau z = s^3 + \alpha s - \tau z = 0
    \end{equation}
for \eqref{Pearcey-integrals} has three solutions $s_j(z)$.
It turns out that in each quadrant and for each $j=1,2,3$, we have that $s_j(z)$
is the relevant saddle for the Pearcey integral that defines the
function in the $j$th column.
Thus the initial contour can be deformed into the steepest descent path through
$s_j(z)$, or into a union of steepest descent paths with $s_j(z)$ as
the determining saddle.

Then by a classical saddle point analysis, see e.g.~\cite{Mil}, we obtain
the following asymptotic behavior
\begin{align} \nonumber
   p(z) & = \pm \sqrt{\frac{2\pi}{n W''(s_j(z))}}
    { e}^{-n \left(W(s_j(z)) - \tau z s_j(z)\right)} \left( 1 + \OO(z^{-2/3})\right) \\
    \label{eq:p0nasympt}
    & = \pm \sqrt{\frac{2\pi}{n W''(s_j(z))}}
    { e}^{n \theta_j(z)} \left( 1 + \OO(z^{-2/3})\right)
    \end{align}
as $z \to \infty$. Since $W''(s) = 3s^2 + \alpha$, and $s_j(z) =
\OO(z^{1/3})$ by Lemma \ref{lem:sjasymptotics}, we find from
\eqref{eq:p0nasympt} the behavior
\[ p(z) = \pm \sqrt{\frac{2\pi}{3n}} s_j(z)^{-1}  { e}^{n \theta_j(z)} \left( 1 + \OO(z^{-2/3}) \right) \]
which then by \eqref{eq:enFasympform} takes the form
\[ p(z) = \pm \sqrt{\frac{2\pi}{3n}} \omega^{2k} (\tau z)^{-1/3} { e}^{n \theta_j(z)}
    \left( 1 + \frac{\alpha}{6}  \omega^k (\tau z)^{-2/3} + \OO(z^{-4/3}) \right) \]
for a certain value of $k$, depending on $j$ and depending on the particular quadrant.
In the first quadrant, for example, we have $k=j-1$ as follows from Lemma \ref{lem:sjasymptotics}.

We may differentiate the expansions and we obtain for the derivative
\begin{align*}
    p'(z) & = \pm \sqrt{\frac{2\pi}{3n}} \omega^{2k} (\tau z)^{-1/3} n \tau s_j(z) {e}^{n \theta_j(z)}
    \left( 1 + \frac{\alpha}{6}  \omega^k (\tau z)^{-2/3} + \OO(z^{-4/3}) \right) \\
    & = \pm \sqrt{\frac{2\pi}{3n}} n \tau  {e}^{n \theta_j(z)}
        \left( 1 - \frac{\alpha}{6} \omega^k (\tau z)^{-2/3} + \OO(z^{-4/3}) \right)
        \end{align*}
and for the second derivative
\[ p''(z) = \pm \sqrt{\frac{2\pi}{3n}} (n \tau)^2 \omega^k (\tau z)^{1/3} {e}^{n \theta_j(z)}
        \left( 1 - \frac{\alpha}{2} \omega^k (\tau z)^{-2/3} + \OO(z^{-4/3}) \right). \]

In each quadrant we have to take the correct sign (determined by the
orientation of the contours) and the correct value of $k$.
The formulas can then be written in the form \eqref{eq:Pnasymptotics}.
\end{proof}

Let us define the constant matrices $A_j$ as follows.
\begin{definition}
We define
\begin{align} \label{eq:A1A2}
    A_1=\frac{i}{\sqrt{3}}\begin{pmatrix}
    1& -\omega & -\omega^2 \\
    1& -1 & -1\\
    1& -\omega^2 & -\omega
\end{pmatrix}, & \qquad
    A_2=\frac{i}{\sqrt{3}}\begin{pmatrix}
    -\omega& -1 & -\omega^2\\
    -1& -1 & -1\\
    -\omega^2& -1 & -\omega
\end{pmatrix},\\ \label{eq:A3A4}
    A_3=\frac{i}{\sqrt{3}}\begin{pmatrix}
    -\omega^2& -1 & \omega\\
    -1& -1 & 1\\
    -\omega& -1& \omega^2
\end{pmatrix}, & \qquad
    A_4=\frac{i}{\sqrt{3}}\begin{pmatrix}
    1& -\omega^2 & \omega\\
    1& -1 & 1\\
    1& -\omega & \omega^2
\end{pmatrix}.
\end{align}

%\begin{align} \label{eq:B1B2}
%    B_1=\frac{1}{\sqrt{3i}}\begin{pmatrix}
%   1& -\omega^2 & -\omega\\
%    1& -1 & -1\\
%    1& -\omega & -\omega^2
%\end{pmatrix}, & \qquad
%    B_2=\frac{1}{\sqrt{3} i}\begin{pmatrix}
%   -\omega^2& -1 & -\omega\\
%   -1& -1 & -1\\
%    -\omega& -1 & -\omega^2
%\end{pmatrix},\\ \label{eq:B3B4}
%    B_3=\frac{1}{\sqrt{3} i}\begin{pmatrix}
%    -\omega& -1 & \omega^2\\
%    -1& -1 & 1\\
%    -\omega^2& -1& \omega
% \end{pmatrix}, & \qquad
%    B_4=\frac{1}{\sqrt{3} i}\begin{pmatrix}
%    1& -\omega & \omega^2\\
%    1& -1 & 1\\
%    1& -\omega^2 & \omega
%\end{pmatrix}.
%\end{align}
\end{definition}

The prefactor $\frac{i}{\sqrt{3}}$ is chosen such that $\det A_j = 1$ for $j=1,2,3,4$.
Then we can reformulate Lemma \ref{lem:Pnasymptotics0} as follows.

\begin{corollary} \label{cor:Pnasymptotics}
We have as $z \to \infty$  in the $j$th quadrant,
\begin{equation} \label{eq:Pnasymptotics3}
    P_n(z) = P_{n,0} \left(I + \OO(z^{-2/3}) \right)
    \diag \begin{pmatrix} z^{-1/3} & 1 & z^{1/3} \end{pmatrix} A_j^{-t} {e}^{n \Theta(z)}
    \end{equation}
where $P_{n,0}$ is the invertible matrix
\begin{equation} \label{eq:defPn0}
    P_{n, 0} = \sqrt{2\pi n} \tau i \begin{pmatrix}
    n^{-1} \tau^{-4/3} & 0 & 0 \\ 0 & 1 & 0 \\ - \frac{\alpha}{2} n \tau^{2/3} & 0 & n \tau^{4/3} \end{pmatrix}.
    \end{equation}
\end{corollary}
The asymptotic formula \eqref{eq:Pnasymptotics3} will be
the most convenient to work with in what follows.
Note that we still have an error term $I + \OO(z^{-2/3})$ which is somewhat
remarkable.

\subsection{First transformation: $Y \mapsto X$}

With the $3 \times 3$ matrix-valued $P_n$ we can now
perform the first transformation of the Riemann-Hilbert problem.
Recall that $Y$ is the solution of the RH problem \eqref{eq:RHforY}.

\begin{definition}
We define $X$ by
\begin{equation}\label{eq:YtoX}
    X(z)= \begin{pmatrix} 1 & 0 \\ 0 & C_n \end{pmatrix}
    Y(z)\begin{pmatrix} 1&0\\
    0& D_n P_n^{-t}(z) {e}^{n \Theta(z)}
    \end{pmatrix},
\end{equation}
for $z$ in the $j$th quadrant.
Here $C_n$ and $D_n$ are the constant matrices
\begin{equation} \label{eq:CnDndef}
    C_n = P_{n,0}^t D_n^{-1}, \qquad
    D_n = \diag \begin{pmatrix} 1 & n \tau & (n \tau)^2 \end{pmatrix},
    \end{equation}
with $P_{n,0}$ given by \eqref{eq:defPn0}, and
$P_n$ is given by \eqref{eq:defPn},
and $\Theta$ is given by \eqref{eq:defTheta}.
\end{definition}

The matrices in the right-hand side of \eqref{eq:YtoX}
are $4\times 4$ matrices written in block form, where the right lower
block has size $3 \times 3$. The transformation \eqref{eq:YtoX}
does not affect the $(1,1)$ entry. The factor ${e}^{n \Theta}(z)$ is included in the
definition of $X(z)$ in order to simplify the asymptotic
behavior of $X$. However, it will complicate the jump matrices,
as we will see.

The matrix valued function $X$ is defined and analytic in each
of the four quadrants.

The asymptotic behavior for $X$ is as follows.
\begin{lemma}
We have
\begin{equation} \label{eq:Xasymptotics}
 X(z)=(I+\OO(z^{-2/3})) \begin{pmatrix}
    z^n&0 &0&0\\
    0&  z^{-\frac{n}{3} + \frac{1}{3}} & 0 & 0\\
    0&  0 & z^{-\frac{n}{3}} & 0\\
    0&  0 & 0 & z^{-\frac{n}{3} -\frac{1}{3}}
                \end{pmatrix}
                \begin{pmatrix}1 &  0 \\
                0& A_j\end{pmatrix}
\end{equation}
as $z \to \infty$ in the $j$th quadrant with matrices
%\begin{equation} \label{eq:Ajdef}
%    A_j = B_j^{-t}, \qquad j=1, \ldots, 4,
%    \end{equation}
 where the matrices $A_j$ are given by \eqref{eq:A1A2} and \eqref{eq:A3A4}.
\end{lemma}
As in \cite{DuiKu2} we have that the asymptotic formula \eqref{eq:Xasymptotics}
for $X$ is not uniform up to the axis.

%[THERE WAS A MISTAKE IN \cite{DuKu} IN THE ASYMPTOTIC FORMULA IN (3.39)  !!]

\begin{proof}
By Corollary \ref{cor:Pnasymptotics} and
the definitions \eqref{eq:CnDndef} of $C_n$ and $D_n$ we have
\[ C_n D_n P_n^{-t}(z) {e}^{n \Theta(z)} =
        \left( I + \OO(z^{-2/3})\right) \begin{pmatrix} z^{1/3} & 0 & 0 \\
            0 & 1 & 0 \\
            0 & 0 & z^{-1/3} \end{pmatrix}
             A_j
            \] as $z \to \infty$ in the $j$th quadrant.
            Using this in \eqref{eq:YtoX}, together with
the asymptotics of $Y$ as given in \eqref{eq:RHforY},
we obtain the lemma.
\end{proof}

The jumps for $X$ on the real and imaginary axis are as follows.
\begin{lemma}
We have
\[ X_+(z)= X_-(z) J_{X}(z), \quad z\in \mathbb R\cup i\mathbb R, \]
where the jump matrices $J_{X}$ are given as follows.
\begin{itemize}
\item  On the real line we have
for $x \in (-\infty, -x^*(\alpha)] \cup [x^*(\alpha), \infty)$
\begin{align} \label{eq:jumpXonR1}
J_X(x)&= \begin{pmatrix}
    1& {e}^{-n (V(x)-\theta_{1}(x))} & 0 & 0\\
    0& 1 & 0 & 0 \\
    0& 0 & {e}^{n (\theta_{2,+}(x) - \theta_{3,+}(x))} & 1 \\
    0& 0 & 0 & {e}^{n (\theta_{3,+}(x) - \theta_{2,+}(x))}
    \end{pmatrix},
\end{align}
and for $x \in (-x^*(\alpha), x^*(\alpha))$ (only relevant in case $\alpha < 0$),
\begin{align} \label{eq:jumpXonR2}
J_X(x)&= \begin{pmatrix}
    1& {e}^{-n (V(x)-\theta_1(x))} & 0 & 0\\
    0& 1 & 0 & 0 \\
    0& 0 & 1 & {e}^{-n (\theta_{2}(x) - \theta_{3}(x))} \\
    0& 0 & 0 & 1
    \end{pmatrix}.
\end{align}
\item On the imaginary axis, we have for
$z \in (-i \infty, -i y^*(\alpha)] \cup [iy^*(\alpha), i \infty)$,
\begin{align}\label{eq:jumpXoniR1}
J_{X}(z)&= \begin{pmatrix}
    1&0 & 0 & 0\\
    0 & {e}^{n (\theta_{1,+}(z) - \theta_{2,+}(z))} & 0 & 0\\
    0&  1 & {e}^{n (\theta_{2,+}(z) - \theta_{1,+}(z))}  & 0\\
    0&  0 & 0 & 1
    \end{pmatrix},
\end{align}
and for $z \in (-i y^*(\alpha), i y^*(\alpha))$ (only relevant in case $\alpha > 0$),
\begin{align}\label{eq:jumpXoniR2}
J_{X}(z)&= \begin{pmatrix}
    1&0 & 0 & 0\\
    0 & 1 & 0 & 0\\
    0&  {e}^{-n(\theta_{2}(z) - \theta_{1}(z))} & 1  & 0 \\
    0&  0 & 0 & 1
    \end{pmatrix}.
\end{align}
\end{itemize}
\end{lemma}

\begin{proof}
The jump matrix for $x \in \mathbb R$ is by \eqref{eq:YtoX} and the
jump condition in \eqref{eq:RHforY}
\begin{align} \nonumber
    J_X(x) & = \begin{pmatrix} 1 & 0 \\ 0 & {e}^{-n \Theta_-(x)} P_{n,-}^t(x) D_n^{-1} \end{pmatrix}
    \begin{pmatrix} 1 & w_{0,n}(x) & w_{1,n}(x) & w_{2,n}(x) \\
        0 & 1 & 0 & 0 \\
        0 & 0 & 1 & 0 \\
        0 & 0 & 0 & 1 \end{pmatrix} \\ \nonumber
        & \hspace*{6cm} \times
        \begin{pmatrix} 1 & 0 \\ 0 & D_n  P_{n,+}^{-t}(x) {e}^{n \Theta_+(x)} \end{pmatrix}
        \\
        & = \begin{pmatrix}
        1 & \begin{pmatrix} w_{0,n}(x) & w_{1,n}(x) & w_{2,n}(x) \end{pmatrix}
            D_n  P_{n,+}^{-t}(x) {e}^{n \Theta_+(x)} \\[5pt]
        0 & {e}^{-n \Theta_-(x)} P_{n,-}^t(x) P_{n,+}^{-t}(x) {e}^{n \Theta_+(x)} \end{pmatrix}
    \label{eq:JXeval0}
\end{align}

The row vector in the right upper corner of \eqref{eq:JXeval0} is by
\eqref{eq:weightsinp0} and \eqref{eq:CnDndef}
\begin{multline} \label{eq:JXeval1}
    \begin{pmatrix} w_{0,n}(x) & w_{1,n}(x) & w_{2,n}(x) \end{pmatrix} D_n  P_{n,+}^{-t}(x) {e}^{n \Theta_+(x)} \\
     = {e}^{-n V(x)} \begin{pmatrix} p_{0,n}(x) & p_{0,n}'(x) & p_{0,n}''(x) \end{pmatrix}
        P_{n,+}^{-t}(x) {e}^{n \Theta_+(x)}
        \end{multline}
Since $\begin{pmatrix} p_{0,n}(x) & p_{0,n}'(x) & p_{0,n}''(x) \end{pmatrix}^t$ is the
first column of $P_n(x)$, see \eqref{eq:defPn},
it follows that \eqref{eq:JXeval1} is equal to
${e}^{-n V(x)} \begin{pmatrix} 1 & 0 & 0  \end{pmatrix} {e}^{n \Theta_+(x)} $
which by \eqref{eq:defTheta} leads to
\begin{multline} \label{eq:JXeval2}
    \begin{pmatrix} w_{0,n}(x) & w_{1,n}(x) & w_{2,n}(x) \end{pmatrix} D_n  P_{n,+}^{-t}(x) {e}^{n \Theta_+(x)} \\
     =     \begin{pmatrix} {e}^{-n (V(x) -\theta_1(x))} & 0 & 0 \end{pmatrix}.
    \end{multline}
This leads to the first row of the jump matrices \eqref{eq:jumpXonR1}--\eqref{eq:jumpXonR2}.

To evaluate the $3 \times 3$ block in the right lower corner of \eqref{eq:JXeval0}
we note that  we have by \eqref{eq:jumpPnonR}
\[ P_{n,-}^t(x) P_{n,+}^{-t}(x)  =
    \left(P_{n,-}^{-1}(x) P_{n,+}(x) \right)^{-t} = \begin{pmatrix} 1 & 0 & 0 \\ 0 & 1 & 1 \\ 0 & 0 & 1 \end{pmatrix}
    \]
so that
\begin{multline} \label{eq:JXeval3}
    {e}^{-n \Theta_-(x)} P_{n,-}^t(x) P_{n,+}^{-t}(x)  {e}^{n \Theta_+(x)}
    = {e}^{-n \Theta_-(x)} \begin{pmatrix} 1 & 0 & 0 \\ 0 & 1 & 1 \\ 0 & 0 & 1 \end{pmatrix}
    {e}^{n \Theta_+(x)}.
    \end{multline}
Then using \eqref{eq:jumpThetaonR} to write ${e}^{-n \Theta_-(x)}$
in terms of ${e}^{-n \Theta_+(x)}$, and then using the  explicit expressions
\eqref{eq:defTheta} for $\Theta$, we see that \eqref{eq:JXeval3} indeed reduces to the $3 \times 3$ right
lower block in \eqref{eq:jumpXonR1}--\eqref{eq:jumpXonR2}.

The proof of the other expressions for $J_X$
follows in a similar way.
\end{proof}

\subsection{RH problem for $X$}

To summarize, we have found the following RH problem
for $X$

\begin{equation}\label{eq:RHforX}
\left\{\begin{array}{l}
    X  \textrm{ is analytic in } \mathbb C\setminus (\mathbb R \cup i \mathbb R), \\[5pt]
    X_+ = X_- J_X, \qquad \text{on } \mathbb R \cup i \mathbb R, \\
    X(z)= (I+\OO(z^{-\frac{2}{3}}))
    \diag \begin{pmatrix} z^n & z^{-\frac{n-1}{3}}& z^{-\frac{n}{3}}& z^{-\frac{n+1}{3}} \end{pmatrix}
    \begin{pmatrix} 1 & 0 \\ 0 & A_j \end{pmatrix} \\
     \hfill \textrm{as $z\to \infty$ in the $j$th quadrant},
    \end{array}\right.
\end{equation}
where $J_X$ is given by \eqref{eq:jumpXonR1}--\eqref{eq:jumpXoniR2}.

Each of the jump matrices $J_X$ is nontrivial only in certain
$2 \times 2$ blocks. The nontrivial blocks are triangular and assume one of
the forms
\[ \begin{pmatrix} 1 & * \\ 0 & 1 \end{pmatrix} \quad \text{or} \quad
    \begin{pmatrix} 1 & 0 \\ * & 1 \end{pmatrix} \]
with a real off-diagonal entry $*$, or
\[ \begin{pmatrix} * & 1 \\ 0 & * \end{pmatrix}
    \quad \text{or} \quad \begin{pmatrix} * & 0 \\ 1 & * \end{pmatrix} \]
with oscillatory diagonal entries of absolute value $1$.
The first form indicates that an external field is acting
and the second form indicates the presence of an upper constraint.
In this way we can already see the connection with the vector
equilibrium.

Let us examine this in more detail.

\paragraph{Jump $J_X$ on the real line}
The jump matrix $J_X$ on the real line, see \eqref{eq:jumpXonR1} and \eqref{eq:jumpXonR2},
takes the block form
\[ J_X = \begin{pmatrix} (J_X)_1 & 0 \\ 0 & (J_X)_3 \end{pmatrix}, \qquad \text{on } \mathbb R, \]
where $(J_X)_1$ and $(J_X)_3$ are $2 \times 2$ matrices.

We have
\begin{equation} \label{eq:defJX1}
    (J_X)_1(x) = \begin{pmatrix} 1 & {e}^{-n V_1(x)} \\ 0 & 1 \end{pmatrix}, \qquad x \in \mathbb R,
    \end{equation}
where
\[ V_1(x) =  V(x) - \theta_1(x) \]
is indeed the external field that acts on the first measure in the vector
equilibrium problem, see \eqref{eq:V1bis}.

Furthermore, we have by \eqref{eq:jumpXonR2},
\begin{equation} \label{eq:defJX3ongap}
    (J_X)_3(x) =
     \begin{pmatrix} 1 & {e}^{-n V_3(x)} \\ 0 & 1 \end{pmatrix},
    \qquad x \in (-x^*(\alpha),x^*(\alpha)),
    \end{equation}
with
\[ V_3(x) = \theta_{2}(x) - \theta_{3}(x) \quad \text{for } x \in (-x^*(\alpha), x^*(\alpha)), \]
which by \eqref{eq:V3bis} is indeed the non-zero part of the external field $V_3$
that acts on the third measure in the vector equilibrium problem.
The external field $V_3$ plays a role only in case $x^*(\alpha) > 0$, that is, in case $\alpha < 0$.

The right lower block in \eqref{eq:jumpXonR1}
has oscillatory diagonal entries. We define $\psi_3$ by
\begin{equation} \label{eq:defpsi3}
    \psi_3(z) = \theta_2(z) - \theta_3(z) \end{equation}
so that
\begin{equation} \label{eq:defJX3}
    (J_X)_3(x) = \begin{pmatrix}
    {e}^{n \psi_{3,+}(x)} & 1 \\ 0 & {e}^{n \psi_{3,-}(x)} \end{pmatrix},
    \qquad  x \in \mathbb R, \, |x| > x^*(\alpha).
    \end{equation}

Then $\psi_{3,\pm}$ is purely imaginary for $|x| > x^*(\alpha)$ with
\begin{align*}
    \frac{d}{dx} \psi_{3,\pm}(x) & =
    \frac{d}{dx} \left[ \theta_{2,+}(x) - \theta_{2,-}(x) \right] \\
    & = \tau (s_{2,+}(x) - s_{2,-}(x)) = 2 i \tau \Im s_{2,+}(x),
    \end{align*}
    which is purely imaginary with positive imaginary part.
Thus we can associate with $\psi_3$ a measure $\sigma_3$ on $\mathbb R$ by
putting
\begin{align}
    \frac{d \sigma_3}{dx} =
        \frac{1}{2\pi i} \frac{d}{dx} \left(\theta_{2,+}(x) - \theta_{3,+}(x)\right)
            = \frac{\tau}{\pi} \Im s_{2,+}(x), & \quad \text{for } x \in \mathbb R.
     \label{eq:defsigma3}
        \end{align}
Then $S(\sigma_3) = \mathbb R$ if $\alpha > 0$
and $S(\sigma_3) = \mathbb R\setminus (-x^*(\alpha), x^*(\alpha))$ if $\alpha < 0$.

Because of the upper triangular form of \eqref{eq:defJX3} it will turn
out that $\sigma_3$ acts as a lower constraint on the third measure in
the sense that
\begin{equation} \label{eq:constraintsigma3}
    \mu_3 + \sigma_3 \geq 0
    \end{equation}
and then we could allow signed measures $\mu_3$ in the vector equilibrium
problem. However, in this more general vector equilibrium problem we would
still find $\mu_3 \geq 0$ so that the constraint \eqref{eq:constraintsigma3}
does not play a role after all.

We will not use the measure $\sigma_3$ anymore.

\begin{figure}[t]
\begin{center}
   \setlength{\unitlength}{1truemm}
   \begin{picture}(100,70)(-5,2)
       \put(0,30){\line(1,0){95}}
%       \put(47.5,20){\line(0,1){40}}
       \put(47.5,0){\line(0,1){60}}
       \put(93,26.6){$\mathbb R $}
       \put(48.5,58){$i\mathbb R $}
       \put(95,30){\thicklines\vector(1,0){.0001}}
       \put(47.5,60){\thicklines\vector(0,1){.0001}}
       \put(47.5,15){\thicklines\circle*{1}}
       \put(35,14){$-y^*(\alpha)$}
       \put(47.5,45){\thicklines\circle*{1}}
       \put(38,44){$y^*(\alpha)$}

       \put(70,10){$\begin{pmatrix} 1 & {e}^{-nV_1} & 0 & 0 \\ 0 & 1 & 0 & 0 \\
            0 & 0 & {e}^{n \psi_{3,+}} & 1 \\ 0 & 0 & 0 & {e}^{n \psi_{3,-}} \end{pmatrix}$}
        \put(80,20){\line(0,1){10}}
        \put(80,30){\thicklines\vector(0,1){0.0001}}
        \put(70,15){\line(-4,1){60}}
        \put(10,30){\thicklines\vector(-4,1){0.0001}}

       \put(60,45){$\begin{pmatrix} 1 & 0 & 0 & 0 \\ 0 & 1 & 0 & 0 \\
            0 & {e}^{-nV_2} & 1 & 0  \\ 0 & 0 & 0 & 1 \end{pmatrix}$}
        \put(60,40){\line(-1,0){12.5}}
        \put(47.5,40){\thicklines\vector(-1,0){0.0001}}
        \put(62.5,38){\line(-1,-1){15}}
        \put(47.5,23){\thicklines\vector(-1,-1){0.0001}}

        \put(-6,49){$\begin{pmatrix} 1 & 0 & 0 & 0 \\ 0 & {e}^{n \psi_{2,-}} & 0 & 0 \\
            0 & 1 & {e}^{n \psi_{2,+}} & 0 \\ 0 & 0 & 0 & 1 \end{pmatrix} $}
        \put(37.5,50){\line(1,0){10}}
        \put(47.5,50){\thicklines\vector(1,0){0.0001}}
        \put(32.5,40){\line(1,-2){15}}
        \put(47.5,10){\thicklines\vector(1,-2){0.0001}}

   \end{picture}
% \vspace{-20mm}
   \caption{Jump matrices $J_X$ in case $\alpha > 0$}
   \label{fig:jumpsJX1}
\end{center}
%\vspace{-5mm}
\end{figure}

\paragraph{Jump $J_X$ on the imaginary axis}
The jump matrix $J_X$ on the imaginary axis, see \eqref{eq:jumpXoniR1} and \eqref{eq:jumpXoniR2},
takes the block form
\[ J_X = \begin{pmatrix} 1 & 0 & 0 \\ 0 &  (J_X)_2 & 0 \\ 0 & 0 & 1 \end{pmatrix}, \qquad \text{on } i\mathbb R, \]
where $(J_X)_2$ is a $2 \times 2$ matrix.

We define
\begin{equation} \label{eq:defpsi2}
    \psi_2(z) = \theta_1(z) - \theta_2(z),  \qquad z \in i \mathbb R, \, |z| > y^*(\alpha),
    \end{equation}
so that we have by \eqref{eq:jumpXoniR1}
\begin{equation} \label{eq:defJX2}
    (J_X)_2(z) = \begin{pmatrix} {e}^{n \psi_{2,+}(z)} & 0 \\ 1 & {e}^{n \psi_{2,-}(z)} \end{pmatrix},
    \qquad z \in i \mathbb R, \, |z| > y^*(\alpha).
\end{equation}

Then $\psi_2$ is associated with the measure $\sigma_2$,
since
\begin{align*}
    \frac{d}{dz} \psi_{2,-}(z) & =
    \frac{d}{dz} \left[ \theta_{1,-}(z) - \theta_{1,+}(z) \right] \\
    & = \tau (s_{1,-}(z) - s_{1,+}(z)) = 2 \pi i  \frac{d\sigma_2}{dz}(z),
    \end{align*}
by \eqref{eq:sigma2withs1}, and also
\[ \psi_{2,\pm}(z) =
    \begin{cases}
        \mp 2 \pi i \sigma_2([0,z]), & \qquad z \in i \mathbb R^+, \\
        \mp 2 \pi i \sigma_2([z,0]), & \qquad z \in i \mathbb R^-.
        \end{cases} \]

We also identify an external field $V_2$ on the imaginary axis,
which is only there in case $\alpha > 0$. We have by \eqref{eq:jumpXoniR2}
\begin{equation} \label{eq:defJX2ongap}
    (J_X)_2(z) = \begin{pmatrix} 1 & 0 \\  {e}^{-n V_2(z)} & 1 \end{pmatrix},
    \qquad z \in (-iy^*(\alpha), iy^*(\alpha)),
    \end{equation}
with
\begin{equation} \label{eq:defV2}
    V_2(z) = \begin{cases}
    \theta_{2}(z) - \theta_{1}(z), & \quad \text{for }  z \in (-iy^*(\alpha), iy^*(\alpha)), \\
    0, & \quad \text{elsewhere on $i\mathbb R$}.
    \end{cases}
    \end{equation}
The external field $V_2$ will not be active, since it acts only
on the part of $\mu_2$ that is in $(-iy^*(\alpha), iy^*(\alpha))$,
and this part is zero, since $S(\mu_2) = S(\sigma_2) = i \mathbb R \setminus (-iy^*(\alpha), iy^*(\alpha))$.

See Figures \ref{fig:jumpsJX1} and \ref{fig:jumpsJX2} for
the jump matrices $J_X$ in the two cases $\alpha > 0$ and $\alpha < 0$.

\begin{figure}[t]
\begin{center}
   \setlength{\unitlength}{1truemm}
   \begin{picture}(100,70)(-5,2)
       \put(0,30){\line(1,0){95}}
%       \put(47.5,20){\line(0,1){40}}
       \put(47.5,0){\line(0,1){60}}
       \put(93,26.6){$\mathbb R $}
       \put(48.5,58){$i\mathbb R $}
       \put(95,30){\thicklines\vector(1,0){.0001}}
       \put(47.5,60){\thicklines\vector(0,1){.0001}}
       \put(67.5,30){\thicklines\circle*{1}}
       \put(63,26.6){$x^*(\alpha)$}
       \put(27.5,30){\thicklines\circle*{1}}
       \put(23.5,26.6){$-x^*(\alpha)$}

       \put(70,10){$\begin{pmatrix} 1 & {e}^{-nV_1} & 0 & 0 \\ 0 & 1 & 0 & 0 \\
            0 & 0 & {e}^{n \psi_{3,+}} & 1 \\ 0 & 0 & 0 & {e}^{n \psi_{3,-}} \end{pmatrix}$}
        \put(80,20){\line(0,1){10}}
        \put(80,30){\thicklines\vector(0,1){0.0001}}
        \put(70,15){\line(-4,1){60}}
        \put(10,30){\thicklines\vector(-4,1){0.0001}}

       \put(60,45){$\begin{pmatrix} 1 & {e}^{-nV_1} & 0 & 0 \\ 0 & 1 & 0 & 0 \\
            0 & 0 & 1 & {e}^{-n V_3}  \\ 0 & 0 & 0 & 1 \end{pmatrix}$}
        \put(60,40){\line(-1,-2){5}}
        \put(55,30){\thicklines\vector(-1,-2){0.0001}}
        \put(60,40){\line(-2,-1){20}}
        \put(40,30){\thicklines\vector(-2,-1){0.0001}}

        \put(-6,49){$\begin{pmatrix} 1 & 0 & 0 & 0 \\ 0 & {e}^{n \psi_{2,-}} & 0 & 0 \\
            0 & 1 & {e}^{n \psi_{2,+}} & 0 \\ 0 & 0 & 0 & 1 \end{pmatrix} $}
        \put(37.5,50){\line(1,0){10}}
        \put(47.5,50){\thicklines\vector(1,0){0.0001}}
        \put(32.5,40){\line(1,-2){15}}
        \put(47.5,10){\thicklines\vector(1,-2){0.0001}}

   \end{picture}
% \vspace{-20mm}
   \caption{Jump matrices $J_X$ in case $\alpha < 0$}
   \label{fig:jumpsJX2}
\end{center}
%\vspace{-5mm}
\end{figure}

\section{Second transformation $X \mapsto U$}

\subsection{Definition of second transformation}
The second transformation of the RH problem uses the functions
that come from the vector equilibrium problem. It is possible to
state the transformation in terms of either the $g$-functions, or the $\lambda$-functions.

\begin{definition}
We define the $4 \times 4$ matrix valued function $U$ by
\begin{equation} \label{eq:defU}
    U(z)= \left(I + \frac{3}{2} n C E_{2,4} \right) {e}^{nL} X(z) {e}^{-n G(z)},
\end{equation}
where $C$ is the constant from \eqref{eq:constantC}, $G$ is given by
\begin{align} \nonumber
    G & = \diag \begin{pmatrix} g_1 + \ell_1 & g_2 - g_1 & g_3 - g_2 & -g_3 \end{pmatrix} \\
      & = -\diag \begin{pmatrix} \lambda_1 - V & \lambda_2-\theta_1 & \lambda_3 - \theta_2 & \lambda_4 - \theta_3
\end{pmatrix} \label{eq:defG}
\end{align}
and $L$ is a constant diagonal matrix
\begin{equation} \label{eq:defL}
    L= \diag \begin{pmatrix} \ell_1 & 0 & 0 & 0
\end{pmatrix},
\end{equation}
with $\ell_1$ the variational constant in the Euler-Lagrange condition on $\mu_1$.
\end{definition}
Note that the equality of the two diagonal matrices in \eqref{eq:defG} follows from
the definition \eqref{eq:deflambdaj} of the $\lambda$-functions.

Then $U$ is defined and analytic in $\mathbb C \setminus (\mathbb R \cup i \mathbb R)$.

\subsection{Asymptotic behavior of $U$}
\begin{lemma} We have
\begin{equation} \label{eq:Uasymptotics}
    U(z)= \left(I +\OO(z^{-1/3})\right)\begin{pmatrix}    1 & 0 & 0 &0\\
    0 & z^{1/3} & 0 & 0\\
    0 & 0 & 1  & 0 \\
    0 & 0 & 0 & z^{-1/3}
    \end{pmatrix} \begin{pmatrix} 1 & 0 \\
    0 & A_j
    \end{pmatrix}
    \end{equation}
as $z\to \infty$ in the $j$th quadrant.
\end{lemma}
\begin{proof}
We have because of the asymptotic behavior of the $\lambda$ functions that
\begin{multline*}
    {e}^{-nG(z)} = \diag \begin{pmatrix} z^{-n} {e}^{-n \ell_1} & z^{n/3} & z^{n/3} & z^{n/3} \end{pmatrix}
    \\ \times \left(I - \frac{3}{2}nC z^{-2/3} \begin{pmatrix} 0 & 0 \\ 0 &  \Omega_j \end{pmatrix} + \OO(z^{-1})\right)
    \end{multline*}
as $z \to \infty$ in the $j$th quadrant, where
\[ \left\{
    \begin{aligned}
        \Omega_1 & = \diag \begin{pmatrix} 1 & \omega & \omega^2 \end{pmatrix} \\
        \Omega_2 & = \diag \begin{pmatrix} \omega & 1 & \omega^2 \end{pmatrix} \\
        \Omega_3 & = \diag \begin{pmatrix} \omega^2 & 1 & \omega \end{pmatrix} \\
        \Omega_4 & = \diag \begin{pmatrix} 1 & \omega^2  & \omega \end{pmatrix}
        \end{aligned} \right.
\]

Then by the asymptotic behavior of $X$, and the definition of $U$, we get
\begin{multline*} U(z) = \left(I + \frac{3}{2} n C E_{2,4} + \OO(z^{-2/3})\right)
        \diag \begin{pmatrix} 1 & z^{1/3} & 1 & z^{-1/3} \end{pmatrix}
        \begin{pmatrix} 1 & 0 \\ 0 & A_j \end{pmatrix}
    \\ \times \left(I - \frac{3}{2}nC z^{-2/3} \begin{pmatrix} 0 & 0 \\ 0 &  \Omega_j \end{pmatrix} + \OO(z^{-1})\right)
    \end{multline*}
as $z \to \infty$ in the $j$ th quadrant.
We can move the $\OO(z^{-1})$ to the front, but then the $\OO(z^{-2/3})$
reduces to $\OO(z^{-1/3})$:
\begin{multline*} U(z) = \left(I + \frac{3}{2} n C E_{2,4} + \OO(z^{-1/3})\right)
        \diag \begin{pmatrix} 1 & z^{1/3} & 1 & z^{-1/3} \end{pmatrix}
        \begin{pmatrix} 1 & 0 \\ 0 & A_j \end{pmatrix}
    \\ \times \left(I - \frac{3}{2}nC z^{-2/3} \begin{pmatrix} 0 & 0 \\ 0 &  \Omega_j \end{pmatrix} \right)
    \end{multline*}

We also want to move the $z^{-2/3}$ term to the left.
Then we pick up an $\OO(1)$ contribution in the $(2,4)$ entry.
Indeed we have
\begin{multline*} z^{-2/3}
        \diag \begin{pmatrix} 1 & z^{1/3} & 1 & z^{-1/3} \end{pmatrix}
        \begin{pmatrix} 1 & 0 \\ 0 & A_j \Omega_j A_j^{-1} \end{pmatrix}
        \diag \begin{pmatrix} 1 & z^{-1/3} & 1 & z^{1/3} \end{pmatrix} \\
        = E_{2,4} + \OO(z^{-1/3}) \end{multline*}
as $z \to \infty$ in the $j$th quadrant, since
\[ A_j \Omega_j A_j^{-1} = \begin{pmatrix} 0 & 0 & 1 \\ 1 & 0 & 0 \\ 0 & 1 & 0 \end{pmatrix} \]
for every $j$. The lemma follows.
\end{proof}

\begin{remark}
As a consistency check, we compute the jumps of the
matrix valued function $A(z)$ defined by
\[ A(z) = \diag \begin{pmatrix} 1 & z^{1/3} & 1 & z^{-1/3} \end{pmatrix}
    \begin{pmatrix} 1 & 0 \\ 0 & A_j \end{pmatrix} \]
for $z$ in the $j$th quadrant.
Then we have by the definition \eqref{eq:A1A2}--\eqref{eq:A3A4} of the matrices $A_j$,
\begin{equation} \label{eq:jumpA}
    A_+(z) = A_-(z) J_A(z), \qquad z \in \mathbb R \cup i \mathbb R,
    \end{equation}
with jump matrix
\begin{equation} \label{eq:defJA}
    J_A(z) =
    \begin{cases}
    \begin{pmatrix} 1 & 0 & 0 & 0 \\ 0 & 1 & 0 & 0 \\ 0 & 0 & 0 & 1 \\ 0 & 0 & -1 & 0 \end{pmatrix}, &
    \qquad z\in \mathbb R, \\
    \begin{pmatrix} 1 & 0 & 0 & 0 \\ 0 & 0 & -1 & 0 \\ 0 & 1 & 0 & 0 \\ 0 & 0 & 0 & 1 \end{pmatrix}, &
    \qquad z\in i\mathbb R,
    \end{cases}
    \end{equation}
which is indeed what we expect.
\end{remark}

\subsection{Jump matrices for $U$}
The jump  $U_+ = U_- J_U$ with jump matrix $J_U$ takes a different
form on the various parts of the real and imaginary axis.
Since $G$ and $L$ are diagonal matrices, the jump matrix $J_U$
has the same block structure as $J_X$.
In terms of the $\lambda$ functions the jumps take on a very nice form.

\begin{figure}[t]
\begin{center}
   \setlength{\unitlength}{1truemm}
   \begin{picture}(100,70)(-5,2)
       \put(0,30){\line(1,0){95}}
%       \put(47.5,20){\line(0,1){40}}
       \put(47.5,0){\line(0,1){60}}
       \put(93,26.6){$\mathbb R $}
       \put(48.5,58){$i\mathbb R $}
       \put(95,30){\thicklines\vector(1,0){.0001}}
       \put(47.5,60){\thicklines\vector(0,1){.0001}}
       \put(47.5,5){\thicklines\circle*{1}}
       \put(48,4){$-ic_2$}
       \put(47.5,50){\thicklines\circle*{1}}
       \put(42,49){$ic_2$}

       \put(20,30){\thicklines\circle*{1}}
       \put(40,30){\thicklines\circle*{1}}
       \put(55,30){\thicklines\circle*{1}}
       \put(75,30){\thicklines\circle*{1}}
       \put(18,26.6){$a_1$}
       \put(38,26.6){$b_1$}
       \put(53,26.6){$a_2$}
       \put(73,26.6){$b_2$}

       \put(55,7){\small{$\begin{pmatrix} {e}^{-2n \pi i \alpha_k} &  \hspace*{-0.2cm} {e}^{n(\lambda_2 - \lambda_1)} &  \hspace*{-0.7cm} 0 &  \hspace*{-0.7cm} 0 \\
       0 & \hspace*{-0.2cm}  {e}^{2n\pi i \alpha_k} &  \hspace*{-0.7cm} 0 &  \hspace*{-0.7cm}  0 \\
            0 & \hspace*{-0.2cm}  0 & \hspace*{-0.7cm}  {e}^{n(\lambda_{3,+}-\lambda_{3,-})} &  \hspace*{-0.7cm} 1 \\ 0 &  \hspace*{-0.2cm} 0 & \hspace*{-0.7cm}  0 & \hspace*{-0.7cm}  {e}^{n( \lambda_{4,+}-\lambda_{4,-})}
            \end{pmatrix}$}}
        \put(80,17){\line(0,1){13}}
        \put(80,30){\thicklines\vector(0,1){0.0001}}
        \put(72,17){\line(-2,1){25}}
        \put(46,30){\thicklines\vector(-2,1){0.0001}}
        \put(55,17){\line(-4,1){52}}
        \put(3,30){\thicklines\vector(-4,1){0.0001}}

       \put(-30,7){\small{$\begin{pmatrix} {e}^{n(\lambda_{1,+}-\lambda_{1,-})} & \hspace*{-0.7cm}  1 & \hspace*{-0.7cm}  0 &\hspace*{-0.7cm}  0 \\ 0 & \hspace*{-0.7cm} {e}^{n(\lambda_{2,+}-\lambda_{2,-})} &\hspace*{-0.7cm}  0 & \hspace*{-0.7cm} 0 \\
            0 & \hspace*{-0.7cm} 0 & \hspace*{-0.7cm} {e}^{n(\lambda_{3,+}-\lambda_{3,-})} & \hspace*{-0.7cm} 1 \\ 0 &\hspace*{-0.7cm}  0 &\hspace*{-0.7cm} 0 &\hspace*{-0.7cm}  {e}^{n(\lambda_{4,+}-\lambda_{4,-})} \end{pmatrix}$}}
        \put(32,20){\line(0,1){10}}
        \put(32,30){\thicklines\vector(0,1){0.0001}}
        \put(42,10){\line(1,1){20}}
        \put(62,30){\thicklines\vector(1,1){0.0001}}

       \put(63,42){\small{$\begin{pmatrix} 1 & 0 & 0 & 0 \\ 0 & 1 & 0 & 0 \\
            0 & {e}^{n(\lambda_{2}-\lambda_{3})} & 1 & 0  \\
            0 & 0 & 0 & 1 \end{pmatrix}$}}
        \put(60,40){\line(-1,0){12.5}}
        \put(47.5,40){\thicklines\vector(-1,0){0.0001}}

        \put(-20,49){\small{$ \begin{pmatrix} 1 & 0 & 0 & 0 \\ 0 & {e}^{n(\lambda_{2,+}-\lambda_{2,-})} & 0 & 0 \\
            0 & 1 & {e}^{n(\lambda_{3,+}-\lambda_{3,-})} & 0 \\ 0 & 0 & 0 & 1 \end{pmatrix} $}}
        \put(37.5,55){\line(1,0){10}}
        \put(47.5,55){\thicklines\vector(1,0){0.0001}}

   \end{picture}
% \vspace{-20mm}
   \caption{Jump matrices $J_U$ in case $\alpha > 0$}
   \label{fig:jumpsJU1}
\end{center}
%\vspace{-5mm}
\end{figure}

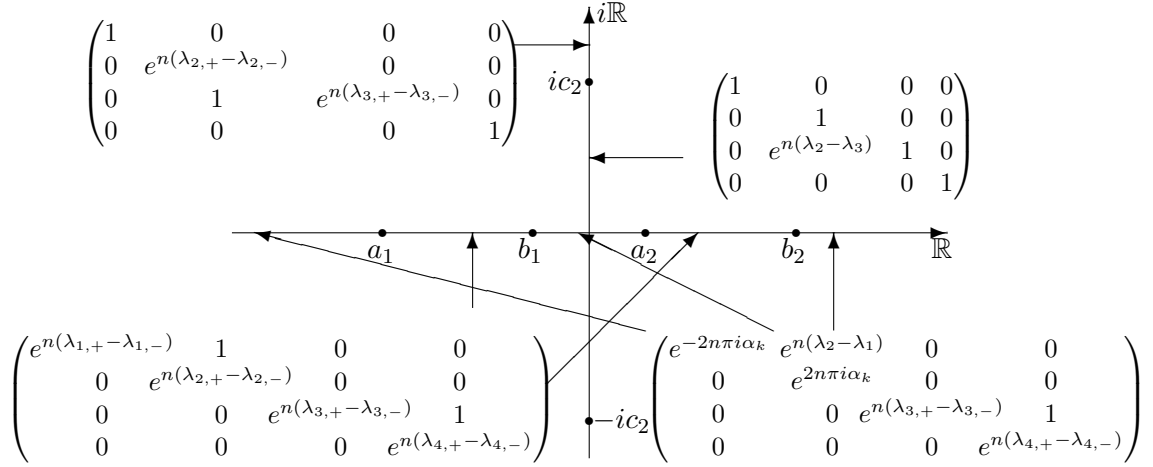
\begin{figure}[t]
\begin{center}
   \setlength{\unitlength}{1truemm}
   \begin{picture}(100,70)(-5,2)
       \put(0,30){\line(1,0){95}}
%       \put(47.5,20){\line(0,1){40}}
       \put(47.5,0){\line(0,1){60}}
       \put(93,26.6){$\mathbb R $}
       \put(48.5,58){$i\mathbb R $}
       \put(95,30){\thicklines\vector(1,0){.0001}}
       \put(47.5,60){\thicklines\vector(0,1){.0001}}
       \put(40,30){\thicklines\circle*{1}}
       \put(38,26.6){$-c_3$}
       \put(55,30){\thicklines\circle*{1}}
       \put(53,26.6){$c_3$}

       \put(20,30){\thicklines\circle*{1}}
       \put(35,30){\thicklines\circle*{1}}
       \put(60,30){\thicklines\circle*{1}}
       \put(75,30){\thicklines\circle*{1}}
       \put(18,26.6){$a_1$}
       \put(33,26.6){$b_1$}
       \put(58,26.6){$a_2$}
       \put(73,26.6){$b_2$}

       \put(52,7){\small{$\begin{pmatrix} {e}^{-2n \pi i \alpha_k} & {e}^{n(\lambda_{2}-\lambda_{1})} &  \hspace*{-0.7cm} 0 & \hspace*{-0.7cm}  0 \\
       0 & {e}^{2n\pi i \alpha_k} & \hspace*{-0.7cm}  0 & \hspace*{-0.7cm}  0 \\
            0 & 0 & \hspace*{-0.7cm}  {e}^{n(\lambda_{3,+}-\lambda_{3,-})} &  \hspace*{-0.7cm} 1 \\
            0 & 0 &  \hspace*{-0.7cm}  0 &  \hspace*{-0.7cm} {e}^{n(\lambda_{4,+}-\lambda_{4,-})} \end{pmatrix}$}}
        \put(80,20){\line(0,1){10}}
        \put(80,30){\thicklines\vector(0,1){0.0001}}
        \put(51,18){\line(-4,1){48}}
        \put(3,30){\thicklines\vector(-4,1){0.0001}}

       \put(-30,7){\small{$\begin{pmatrix} {e}^{n(\lambda_{1,+}-\lambda_{1,-})} & \hspace*{-0.7cm}  1 & \hspace*{-0.7cm}  0 &  \hspace*{-0.7cm}  0 \\
       0 &  \hspace*{-0.7cm} {e}^{n(\lambda_{2,+}-\lambda_{2,-})} & \hspace*{-0.7cm}  0 &   \hspace*{-0.7cm}  0 \\
            0 & \hspace*{-0.7cm}  0 &  \hspace*{-0.7cm}  {e}^{n(\lambda_{3,+}-\lambda_{3,-})} &  \hspace*{-0.7cm} 1 \\
            0 &  \hspace*{-0.7cm} 0 & \hspace*{-0.7cm}   0 & \hspace*{-0.7cm}  {e}^{n(\lambda_{4,+}-\lambda_{4,-})} \end{pmatrix}$}}
        \put(30,20){\line(0,1){10}}
        \put(30,30){\thicklines\vector(0,1){0.0001}}
        \put(45,10){\line(1,1){20}}
        \put(65,30){\thicklines\vector(1,1){0.0001}}

       \put(52,45){\small{$\begin{pmatrix} {e}^{-2n \pi i \alpha_k} & {e}^{n(\lambda_{2}-\lambda_{1})} & 0 & 0 \\
            0 & {e}^{2n \pi i \alpha_k} & 0 & 0 \\
            0 & 0 & 1 & {e}^{n(\lambda_{4}-\lambda_{3})}  \\
            0 & 0 & 0 & 1 \end{pmatrix}$}}
        \put(60,40){\line(-1,-1){10}}
        \put(50,30){\thicklines\vector(-1,-1){0.0001}}

        \put(-20,49){\small{$\begin{pmatrix} 1 & 0 & 0 & 0 \\
            0 & {e}^{n(\lambda_{2,+}-\lambda_{2,-})} & 0 & 0 \\
            0 & 1 & {e}^{n(\lambda_{3,+}-\lambda_{3,-})} & 0 \\ 0 & 0 & 0 & 1 \end{pmatrix} $}}
        \put(37.5,55){\line(1,0){10}}
        \put(47.5,55){\thicklines\vector(1,0){0.0001}}

   \end{picture}
% \vspace{-20mm}
   \caption{Jump matrices $J_U$ in case $\alpha < 0$}
   \label{fig:jumpsJU2}
\end{center}
%\vspace{-5mm}
\end{figure}

The jump matrix $J_U$ on the real line has the block form
\begin{equation} \label{eq:JUonR}
    J_U(x) = \begin{pmatrix} \left(J_U\right)_1(x) & 0 \\ 0 & \left(J_U\right)_3(x)
    \end{pmatrix}, \qquad x \in \mathbb R,
    \end{equation}
with $2 \times 2$ blocks $(J_U)_1$ and $(J_U)_3$.
On the imaginary axis it takes the form
\begin{equation} \label{eq:JUoniR}
    J_U(z) = \begin{pmatrix} 1 & 0 & 0 \\ 0 & \left(J_U\right)_2(x) & 0 \\ 0 & 0 & 1 \end{pmatrix},
    \qquad z \in i \mathbb R,
    \end{equation}
with a $2 \times 2$ block $(J_U)_2$.

\begin{lemma}
We have
\begin{equation} \label{eq:defJU1}
    \left(J_U\right)_1 = \begin{pmatrix}
    {e}^{n(\lambda_{1,+}- \lambda_{1,-})} & {e}^{n(\lambda_{2,+} - \lambda_{1,-})} \\
    0 & {e}^{n(\lambda_{2,+}- \lambda_{2,-})} \end{pmatrix},
        \qquad \text{on } \mathbb R,
    \end{equation}
\begin{equation} \label{eq:defJU2}
    \left(J_U\right)_2 = \begin{pmatrix}
    {e}^{n(\lambda_{2,+}- \lambda_{2,-})} & 0 \\
    {e}^{n(\lambda_{2,+} - \lambda_{3,-})} & {e}^{n(\lambda_{3,+}- \lambda_{3,-})} \end{pmatrix},
        \qquad \text{on } i \mathbb R,
    \end{equation}
and
\begin{equation} \label{eq:defJU3}
     \left(J_U\right)_3 = \begin{pmatrix}
    {e}^{n(\lambda_{3,+}- \lambda_{3,-})} & {e}^{n(\lambda_{4,+} - \lambda_{3,-})} \\
    0 & {e}^{n(\lambda_{4,+}- \lambda_{4,-})} \end{pmatrix},
        \qquad \text{on } \mathbb R.
    \end{equation}
\end{lemma}

\begin{proof}
By \eqref{eq:defU}, \eqref{eq:defG} we get
\[ \left(J_U\right)_1 =
    \begin{pmatrix} {e}^{-n (\lambda_{1,-} - V)} & 0 \\ 0 & {e}^{-n(\lambda_{2,-}-\theta_{1,-})} \end{pmatrix}
    \left(J_X\right)_1
    \begin{pmatrix} {e}^{n (\lambda_{1,+} - V)} & 0 \\ 0 & {e}^{n(\lambda_{2,+}-\theta_{1,+})} \end{pmatrix}.
    \]
Now $\theta_1$ is analytic across $\mathbb R \setminus \{0\}$, and $V_1 = V - \theta_1$,
so that by \eqref{eq:defJX1} we indeed obtain \eqref{eq:defJU1}.

For $\left(J_U\right)_2$ we find in a similar way
\begin{equation} \label{eq:JU2formula}
    \left(J_U\right)_2 =
    \begin{pmatrix} {e}^{-n (\lambda_{2,-} - \theta_{1,-})} & 0 \\ 0 & {e}^{-n(\lambda_{3,-}-\theta_{2,-})} \end{pmatrix}
\left(J_X\right)_2
    \begin{pmatrix} {e}^{n (\lambda_{2,+} - \theta_{1,+})} & 0 \\ 0 & {e}^{n(\lambda_{3,+}-\theta_{2,+})} \end{pmatrix}.
    \end{equation}
For $z \in i \mathbb R$ with $|z| > y^*(\alpha)$ we have $\theta_{1,\pm}(z) = \theta_{2,\mp}(z)$
by \eqref{eq:thetacontinuations} and then  \eqref{eq:defJX2}, \eqref{eq:defpsi2},
and \eqref{eq:JU2formula} give us that \eqref{eq:defJU2} holds on
$i \mathbb R \setminus (-y^*(\alpha), i y^*(\alpha))$.
On $(-iy^*(\alpha), iy^*(\alpha))$ (which is only relevant in case $\alpha > 0$)
we also obtain \eqref{eq:defJU2}, but now we use the fact that
$\theta_1$ and $\theta_2$ are both analytic on $(-iy^*(\alpha), iy^*(\alpha))$, together
with \eqref{eq:defJX2ongap}, \eqref{eq:defV2}, and \eqref{eq:JU2formula}.

For $\left(J_U\right)_3$ we obtain from \eqref{eq:defU} and \eqref{eq:defG}
\begin{equation} \label{eq:JU3formula}
    \left(J_U\right)_3 =
    \begin{pmatrix} {e}^{-n (\lambda_{3,-} - \theta_{2,-})} & 0 \\ 0 & {e}^{-n(\lambda_{4,-}-\theta_{3,-})} \end{pmatrix}
\left(J_X\right)_3
    \begin{pmatrix} {e}^{n (\lambda_{3,+} - \theta_{2,+})} & 0 \\ 0 & {e}^{n(\lambda_{4,+}-\theta_{3,+})} \end{pmatrix}.
    \end{equation}
For $\left(J_X\right)_3$ we have the two expressions
\eqref{eq:defJX3ongap} and \eqref{eq:defJX3}. Using this in \eqref{eq:JU3formula}
we obtain \eqref{eq:defJU3} in both cases.
\end{proof}

The expressions in \eqref{eq:defJU1}--\eqref{eq:defJU3} are
valid over the full (real or imaginary) axis. Observe in particular
that the two expressions \eqref{eq:defJX2} and \eqref{eq:defJX2ongap}
for $\left(J_X\right)_2$ both lead to \eqref{eq:defJU2},
and the two expressions \eqref{eq:defJX3ongap} and \eqref{eq:defJX3}
for $\left(J_X \right)_3$ both lead to \eqref{eq:defJU3}.
Hence the special roles that $\pm x^*(\alpha)$ (in case $\alpha < 0$)
and $\pm i y^*(\alpha)$ (in case $\alpha > 0$) played in the jump matrix $J_X$
for  $X$ have disappeared in the jump matrix $J_U$ for $U$.

\subsection{RH problem for $U$}

We have found the following RH problem
for $U$.

\begin{equation}\label{eq:RHforU}
\left\{\begin{array}{l}
    U \textrm{ is analytic in } \mathbb C\setminus (\mathbb R \cup i \mathbb R), \\[5pt]
    U_+ = U_- J_U, \qquad \text{on } \mathbb R \cup i \mathbb R, \\
    U(z)= (I+ \OO(z^{-1/3}))
    \diag \begin{pmatrix} 1 & z^{\frac{1}{3}}& 1 & z^{-\frac{1}{3}} \end{pmatrix}
    \begin{pmatrix} 1 & 0 \\ 0 & A_j \end{pmatrix} \\
    \hfill \textrm{as $z\to \infty$ in the $j$th quadrant},
    \end{array}\right.
\end{equation}
where $J_U$ is given by \eqref{eq:JUonR}--\eqref{eq:JUoniR}
with $\left(J_U\right)_k$ for $k=1,2,3$ given by \eqref{eq:defJU1}--\eqref{eq:defJU3}.

The parts $\left(J_U\right)_k$ in the jump matrix $J_U$ have
different expressions in the various parts of the real and imaginary axis.
This follows from \eqref{eq:defJU1}--\eqref{eq:defJU3} and the jump
properties of the $\lambda$-functions as given
in Lemma \ref{lem:jumpslambda}. Also recall that $n$ is a multiple
of three.

\begin{lemma}
\begin{enumerate}
\item[\rm (a)] For $(J_U)_1$ we have
\begin{equation} \label{eq:JU1explicit}
\left(J_U\right)_1 = \begin{cases}
    \begin{pmatrix}
    {e}^{n(\lambda_{1,+} - \lambda_{1,-})} & 1 \\
    0 & {e}^{n(\lambda_{2,+}-\lambda_{2,-})}
    \end{pmatrix},  & \text{ on } S(\mu_1), \\
    \begin{pmatrix}
    {e}^{- 2n\pi i \alpha_k} & {e}^{n(\lambda_{2,+}-\lambda_{1,-})} \\
    0 & {e}^{2n\pi i \alpha_k}
    \end{pmatrix},
    & \begin{array}{l} \text{on }  (b_k, a_{k+1}), \\
        \text{for } k =0, 1, \ldots, N. \end{array}
\end{cases} \end{equation}
\item[\rm (b)] For $(J_U)_2$ we have
\begin{equation} \label{eq:JU2explicit}
\left(J_U\right)_2 = \begin{cases}
    \begin{pmatrix}
    {e}^{n(\lambda_{2,+}-\lambda_{2,-})} & 0  \\
    1 & {e}^{n(\lambda_{3,+} - \lambda_{3,-})}
    \end{pmatrix}, & \text{ on } S(\sigma_2-\mu_2),\\
    \begin{pmatrix}
    1 & 0 \\
    {e}^{n(\lambda_{2,+}-\lambda_{3,-})} & 1
    \end{pmatrix}, & \text{ on } i \mathbb R \setminus S(\sigma_2-\mu_2),
\end{cases} \end{equation}
\item[\rm (c)] For $(J_U)_3$ we have
\begin{equation} \label{eq:JU3explicit}
\left(J_U\right)_3 = \begin{cases}
    \begin{pmatrix}
     {e}^{n(\lambda_{3,+}-\lambda_{3,-})} & 1 \\
     0 & {e}^{n(\lambda_{4,+}-\lambda_{4,-})}
    \end{pmatrix}, & \text{ on } S(\mu_3), \\
    \begin{pmatrix}
    (-1)^n & {e}^{n(\lambda_{4,+}-\lambda_{3,-})}  \\
    0 & (-1)^n
    \end{pmatrix}, & \text{ on }  \mathbb R \setminus S(\mu_3).
\end{cases} \end{equation}
\end{enumerate}
\end{lemma}

\begin{proof}
All expressions follow from \eqref{eq:defJU1}--\eqref{eq:defJU3} and
Lemma \ref{lem:jumpslambda}.

In particular, note that  on $(-c_3,c_3)$ (which can only be non-empty if $\alpha < 0$),
we have by \eqref{eq:lambdajump6}
\[ {e}^{n(\lambda_{3,+}- \lambda_{3,-})}
    = {e}^{n(\lambda_{4,+}- \lambda_{4,-})}
    = {e}^{n (-\frac{1}{3} \pi i)} = (-1)^n \]
    since $n$ is a multiple of three. This explains
    the entries $(-1)^n$ in \eqref{eq:JU3explicit} on $(-c_3,c_3)$.
\end{proof}

\section{Opening of lenses}

The next step in the steepest descent analysis is the opening of
lenses around $S(\mu_1)$, $S(\sigma_2-\mu_2)$ and $S(\mu_3)$. This will be done
in the third and fourth  transformations $U\mapsto T$ and $T\mapsto S$.

\subsection{Third transformation $U \mapsto T$}

In the third transformation we open the lenses around $S(\sigma_2 -\mu_2)$
and $S(\mu_3)$ that will be denoted by $L_2$ and $L_3$, respectively.
The lenses will be closed unbounded sets that do not intersect.
There are three situations, depending on whether $c_2$ and $c_3$ are positive or zero.
We recall that by regularity we can not have  $c_2=c_3=0$.
The three different cases differ in the shapes of the lens, which is due to the
different supports of $S(\sigma_2-\mu_2)$ and $S(\mu_3)$ as illustrated
in Figures \ref{fig:lenses1}, \ref{fig:lenses2} and \ref{fig:lenses3}.

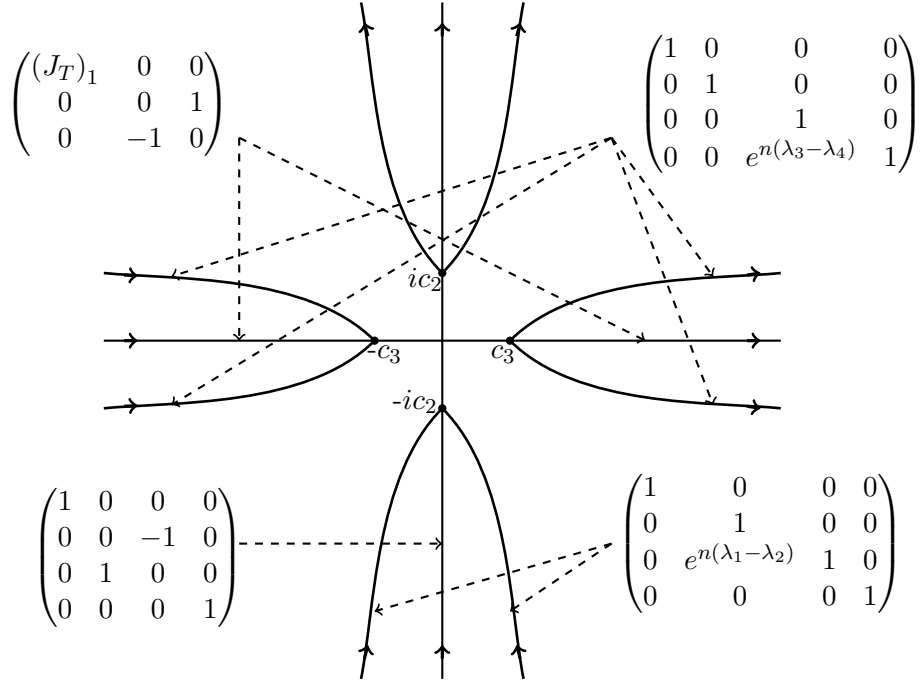
\begin{figure}[t]
\begin{center}
\begin{tikzpicture}[scale=0.9]
\draw[-,thick] (-5,0) -- (5,0);
\draw[-,thick] (0,-5) -- (0,5);
\filldraw [black] (1,0) circle (1.5pt)
    (-1,0) circle (1.5pt)
    (0,-1) circle (1.5pt)
    (0,1) circle (1.5pt);

\draw[very thick](0.9,-.2) node{$c_3$};
\draw[very thick](-0.85,-.2) node{-$c_3$};
\draw[very thick](-0.25,0.9) node{$ i c_2$};
\draw[very thick](-0.4,-0.9) node{-$ i c_2$};

\draw[->,very thick](0,4.5) -- (0,4.7);
\draw[->,very thick](-1.1,4.5) -- (-1.15,4.7);
\draw[->,very thick](1.1,4.5) -- (1.15,4.7);
\draw[->,very thick](0,-4.7) -- (0,-4.5);
\draw[->,very thick](-1.15,-4.7) -- (-1.1,-4.5);
\draw[->,very thick](1.15,-4.7) -- (1.1,-4.5);

%\draw[->,very thick](-0.6,0)--(-0.5,0);
%\draw[->,very thick](0,0.5)--(0,0.6);

\draw[->,very thick](-4.7,0) -- (-4.5,0);
\draw[->,very thick](-4.7,0.97) -- (-4.5,0.96);
\draw[->,very thick](-4.7,-0.97) -- (-4.5,-0.96);

\draw[->,very thick](4.5,0) -- (4.7,0);
\draw[->,very thick](4.5,0.97) -- (4.7,1);
\draw[->,very thick](4.5,-0.97) -- (4.7,-1);

    \draw[line width=1pt] (0,1) .. controls (1,2) and (1,4) .. (1.2,5);
    \draw[line width=1pt] (0,1) .. controls (-1,2) and (-1,4) .. (-1.2,5);
    \draw[line width=1pt] (0,-1) .. controls (1,-2) and (1,-4) .. (1.2,-5);
    \draw[line width=1pt] (0,-1) .. controls (-1,-2) and (-1,-4) .. (-1.2,-5);

    \draw[line width=1pt] (1,0) .. controls (2,1) and (4,0.9) .. (5,1);
    \draw[line width=1pt] (1,0) .. controls (2,-1) and (4,-0.9) .. (5,-1);
    \draw[line width=1pt] (-1,0) .. controls (-2,1) and (-4,0.9) .. (-5,1);
    \draw[line width=1pt] (-1,0) .. controls (-2,-1) and (-4,-0.9) .. (-5,-1);

\draw[->,thick,dashed](-3,3)--(3,0);
\draw[->,thick,dashed](-3,3) --(-3,0);

%\draw[->,thick,dashed](-3,-3)--(0,3);
\draw[->,thick,dashed](-3,-3) --(0,-3);

\draw[->,thick,dashed](2.5,3)--(4,0.95);
\draw[->,thick,dashed](2.5,3)--(4,-0.95);
\draw[->,thick,dashed](2.5,3)--(-4,0.95);
\draw[->,thick,dashed](2.5,3)--(-4,-0.95);

\draw[->,thick,dashed](2.5,-3)--(1,-4);
\draw[->,thick,dashed](2.5,-3)--(-1,-4);
%\draw[->,thick,dashed](2.5,-3)--(-1,4);
%\draw[->,thick,dashed](2.5,-3)--(1,4);

\draw[very thick]
        (5,3.5) node{$\begin{pmatrix} 1 & 0 & 0 & 0 \\ 0 & 1 & 0 & 0 \\
                                0 & 0 & 1 & 0 \\ 0 & 0 & {e}^{n(\lambda_3 - \lambda_4)} & 1 \end{pmatrix}$}
    (-4.8,3.5) node{$\begin{pmatrix} \left(J_T\right)_1 & 0 & 0 \\
                                0 & 0 & 1\\ 0  & -1 & 0 \end{pmatrix}$}
    (4.7,-3) node{$\begin{pmatrix} 1 & 0 & 0 & 0 \\ 0 & 1 & 0 & 0 \\
                                0 & {e}^{n(\lambda_1-\lambda_2)} & 1 & 0 \\ 0 & 0 & 0 & 1 \end{pmatrix}$}
    (-4.5,-3.2) node{$\begin{pmatrix} 1 & 0 & 0 & 0 \\ 0 & 0 & -1 & 0 \\
                                0 & 1 & 0 & 0 \\ 0 & 0 & 0 & 1 \end{pmatrix}$};
\end{tikzpicture}
\end{center}
\caption{Unbounded lenses in case $c_2>0$ and $c_3>0$, corresponding to Cases IV and V.
The figure also shows the jump contour $\Sigma_T$ and some of the jump matrices in
the RH problem \eqref{eq:RHforT} for $T$.}
\label{fig:lenses1}
\end{figure}

\begin{itemize}
\item
In the Cases IV and V we have that both  $c_2$ and $c_3$ are  positive and
we choose  the lips of the lenses such that they have $\pm  i c_2$ and $\pm c_3$ as endpoints,
as shown in Figure~\ref{fig:lenses1}.
\item
In Case III we have $c_2=0$ and $c_3 > 0$, and now we open
the lens such that the lips around $S(\sigma_2-\mu_2)= i \mathbb R$ stay away from the imaginary axis
and intersecting the real line at two points $\pm \gamma_3$ as in Figure~\ref{fig:lenses2}.
We recall that in Case III we have $0\notin S(\mu_1)$ and hence the number  $N$ of intervals
in $S(\mu_1)$ is  even. We then choose $\gamma_3$ such that
\begin{align}\label{eq:gamma3bound}
    0<\gamma_3<\min(a_{N/2+1},c_3), \qquad \text{in Case III}.
\end{align}
The lens around $S(\mu_3)$ is as in the Cases IV and V.
\item In the cases  I and II we have $c_2 > 0$ and $c_3=0$. We then take the
lens around $S(\sigma_2-\mu_2)$   as in the Cases IV and V above, but we choose the
lips of the lens around $S(\mu_3)=\mathbb R$ such that it is away from the real axis
and intersects the imaginary axis at two points $\pm i \gamma_2$ with
\begin{align} \label{eq:gamma2bound}
        0<\gamma_2<c_2, \qquad \text{in Cases I and II}.
\end{align}
\end{itemize}
We choose all lenses to be symmetric with respect to both the real and imaginary axes.

Note that we have $\gamma_3$ in Case III and $\gamma_2$ in Cases I and II.
For ease of presentation we also define
\begin{equation} \label{eq:gamma23are0}
\begin{aligned}
    \gamma_2 & = 0 \qquad \text{in Cases III, IV, and V}, \\
    \gamma_3 & = 0 \qquad \text{in Cases I, II, IV, and V}.
\end{aligned}
\end{equation}

There are further requirements on the lenses that are important for the
steepest descent analysis. These are formulated in the next two lemmas.

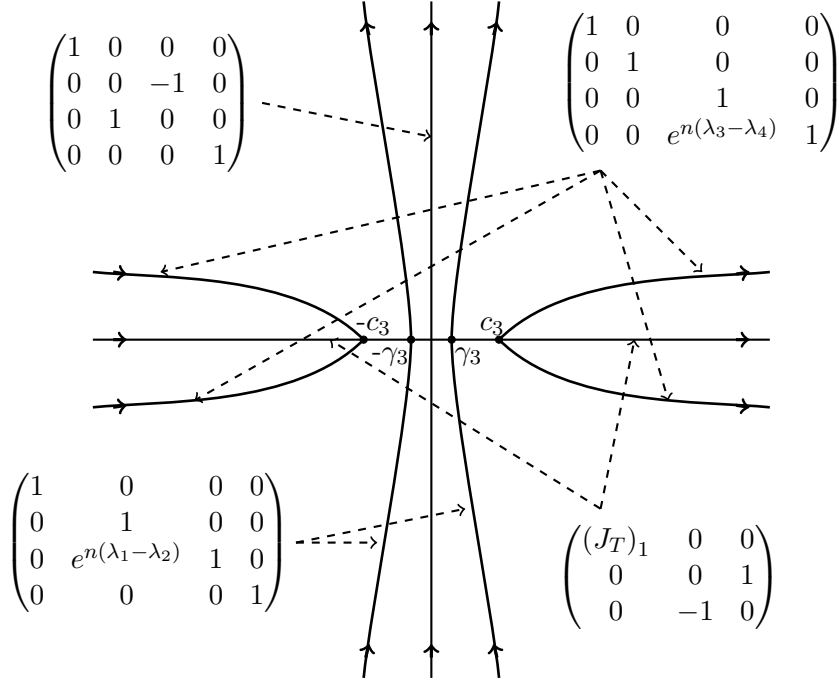
\begin{figure}[t]
\begin{center}
\begin{tikzpicture}[scale=0.9 ]

 \draw[very thick](0.9,.2) node{$c_3$};
 \draw[very thick](-0.85,.2) node{-$c_3$};

\draw [-,thick] (-5,0) -- (5,0);
\draw[-,thick] (0,-5) -- (0,5);
\filldraw [black] (-1,0) circle (1.5pt)
    (1,0) circle (1.5pt);

\filldraw [black] (-0.3,0) circle (1.5pt)     (0.3,0) circle (1.5pt) ;

\draw[->,very thick](0,4.5) -- (0,4.7);
\draw[->,very thick](-0.92,4.5) -- (-0.97,4.7);
\draw[->,very thick](0.92,4.5) -- (0.97,4.7);
\draw[->,very thick](0,-4.7) -- (0,-4.5);
\draw[->,very thick](-0.97,-4.7) -- (-0.92,-4.5);
\draw[->,very thick](0.97,-4.7) -- (0.92,-4.5);

%\draw[->,very thick](-0.6,0)--(-0.5,0);
%\draw[->,very thick](0,0.5)--(0,0.6);

\draw[->,very thick](-4.7,0) -- (-4.5,0);
\draw[->,very thick](-4.7,1) -- (-4.5,0.97);
\draw[->,very thick](-4.7,-1) -- (-4.5,-0.97);

\draw[->,very thick](4.5,0) -- (4.7,0);
\draw[->,very thick](4.5,0.97) -- (4.7,1);
\draw[->,very thick](4.5,-0.97) -- (4.7,-1);

\draw[very thick](.55,-.25) node{$\gamma_3$};
\draw[very thick](-.6,-.25) node{-$\gamma_3$};

    \draw[line width=1pt] (5,1) .. controls (4,0.9) and (2,1) .. (1,0);

    \draw[line width=1pt] (1,0) .. controls (2,-1) and (4,-0.9) .. (5,-1);
    \draw[line width=1pt] (-1,0) .. controls (-2,1) and (-4,0.9) .. (-5,1);
    \draw[line width=1pt] (-1,0) .. controls (-2,-1) and (-4,-0.9) .. (-5,-1);

    \draw[line width=1pt] (0.3,0) .. controls (0.3,1) and (0.9,4) .. (1,5);
    \draw[line width=1pt] (0.3,0) .. controls (0.3,-1) and (0.9,-4) .. (1,-5);
    \draw[line width=1pt] (-0.3,0) .. controls (-0.3,1) and (-0.9,4) .. (-1,5);
    \draw[line width=1pt] (-0.3,0) .. controls (-0.3,-1) and (-0.9,-4) .. (-1,-5);

    \draw[->,thick,dashed](-2.5,3.5) -- (0,3);

\draw[->,thick,dashed](2.5,-2.5)--(3,0);
\draw[->,thick,dashed](2.5,-2.5) --(-1.5,0);

\draw[->,thick,dashed](2.5,2.5)--(4,1);
\draw[->,thick,dashed](2.5,2.5)--(3.5,-0.9);
\draw[->,thick,dashed](2.5,2.5)--(-4,1);
\draw[->,thick,dashed](2.5,2.5)--(-3.5,-0.9);

\draw[->,thick,dashed](-2,-3)--(-0.8,-3);
\draw[->,thick,dashed](-2,-3)--(0.5,-2.5);

\draw[very thick]
        (4,3.8) node{$\begin{pmatrix} 1 & 0 & 0 & 0 \\ 0 & 1 & 0 & 0 \\
                        0 & 0 & 1 & 0 \\ 0 & 0 & {e}^{n(\lambda_3 - \lambda_4)} & 1 \end{pmatrix}$}
    (3.5,-3.5) node{$\begin{pmatrix} \left(J_T\right)_1 & 0 & 0 \\
                        0 & 0 & 1\\ 0  & -1 & 0 \end{pmatrix}$}
    (-4.2,-3) node{$\begin{pmatrix} 1 & 0 & 0 & 0 \\ 0 & 1 & 0 & 0 \\
                        0 & {e}^{n(\lambda_1-\lambda_2)} & 1 & 0 \\ 0 & 0 & 0 & 1 \end{pmatrix}$}
    (-4.2,3.5) node{$\begin{pmatrix} 1 & 0 & 0 & 0 \\ 0 & 0 & -1 & 0 \\
                        0 & 1 & 0 & 0 \\ 0 & 0 & 0 & 1 \end{pmatrix}$};
\end{tikzpicture}
   \caption{Unbounded lenses in case $c_2 = 0$ and $c_3 > 0$, corresponding to Case III.
   The figure also shows the jump contour $\Sigma_T$ and some
   of the jump matrices in the RH problem \eqref{eq:RHforT} for $T$.}
   \label{fig:lenses2}
\end{center}
\end{figure}

\begin{lemma}\label{lem:ineqlensmu2}
We  can and do choose the lens $L_2$ around $S(\sigma_2-\mu_2)$ such that
\[ \Re (\lambda_3 - \lambda_2) < 0 \quad \textrm{ in } L_2\setminus i\mathbb R \]  and
such that
\[ \{z\in \mathbb C \mid |z|>R, \, | \Re(z)|< \eps |\Im(z)|\}\subset L_2\]
for some $\eps>0$ and $R>0$.
\end{lemma}

\begin{proof}
From \eqref{eq:deflambdaj} we obtain
\[ \lambda_3 - \lambda_2 = 2g_2 - g_1 - g_3 + \theta_2 - \theta_1 \]
and for $z \in S(\sigma_2-\mu_2)$ we have by \eqref{eq:jumpg2oniR} and \eqref{eq:thetacontinuations} that
\begin{align*}
    (\lambda_3 - \lambda_2)_{\pm}(z)
    & = \pm \left(g_{2,+}(z) - g_{2,-}(z) + \theta_{1,-}(z) - \theta_{1,+}(z) \right).
\end{align*}
Using \eqref{eq:jumpg2oniR0}, \eqref{eq:sigma2withs1} and
\eqref{eq:thetajder} we can further rewrite this to
   \begin{align*} (\lambda_3 - \lambda_2)_{\pm}(z)
 =\begin{cases}
  \pm \frac{2}{3} \pi i \pm 2\pi i (\sigma_2 - \mu_2)([0,z])& \textrm{for }\Im z>0,\\
  \mp \frac{2}{3} \pi i \mp 2\pi i (\sigma_2 - \mu_2)([z,0])& \textrm{for }\Im z<0,
  \end{cases}
    \end{align*}
which is purely imaginary.  Then
\[ \frac{d}{dy} \Im \left( \pm (\lambda_3-\lambda_2)_{\pm} (iy)\right)
    = 2\pi \frac{d(\sigma_2 - \mu_2)}{|dz|} (iy)  \]
which is positive for  $y \in \mathbb R$ in Case III, and for $|y| > c_2 > 0$ in
the other cases.

The first statement of the lemma then follows from the Cauchy Riemann equations
since $\Re (\lambda_3 - \lambda_2)_{\pm} = 0$ on the imaginary axis.
The second statement follows then from the asymptotic behavior
of $\lambda_2$ and $\lambda_3$ as given in \eqref{eq:asymlambda2}, \eqref{eq:asymlambda3}
and Lemma \ref{lem:thetajasymptotics}.
\end{proof}

\begin{figure}[t]
\begin{center}

\begin{tikzpicture}[scale=0.9]
\draw[-,,thick] (-5,0) -- (5,0);
\draw[-,thick] (0,-5) -- (0,5);
\filldraw [black] (0,-1) circle (1.5pt)
    (0,1) circle (1.5pt);
\filldraw [black] (0,0.3) circle (1.5pt)     (0,-0.3) circle (1.5pt) ;

\draw[very thick](0.3,0.55) node{$i\gamma_2$};
\draw[very thick](0.35,-0.55) node{-$i\gamma_2$};

\draw[very thick](-0.3,0.9) node{$ i c_2$};
\draw[very thick](-0.4,-0.9) node{-$ i c_2$};

\draw[->,very thick](0,4.5) -- (0,4.7);
\draw[->,very thick](-1.1,4.5) -- (-1.15,4.7);
\draw[->,very thick](1.1,4.5) -- (1.15,4.7);
\draw[->,very thick](0,-4.7) -- (0,-4.5);
\draw[->,very thick](-1.15,-4.7) -- (-1.1,-4.5);
\draw[->,very thick](1.15,-4.7) -- (1.1,-4.5);

%\draw[->,very thick](-0.6,0)--(-0.5,0);
%draw[->,very thick](0,0.5)--(0,0.6);

\draw[->,very thick](-4.7,0) -- (-4.5,0);
\draw[->,very thick](-4.7,1) -- (-4.5,0.97);
\draw[->,very thick](-4.7,-1) -- (-4.5,-0.97);

\draw[->,very thick](4.5,0) -- (4.7,0);
\draw[->,very thick](4.5,0.97) -- (4.7,1);
\draw[->,very thick](4.5,-0.97) -- (4.7,-1);

\draw[->,thick,dashed](3.5,-2.5)--(2,0);
%\draw[->,thick,dashed](2.5,-2.5) --(-1.5,0);

\draw[->,thick,dashed](-2.5,-2.5)--(0.95,4);
\draw[->,thick,dashed](-2.5,-2.5)--(-0.95,4);
\draw[->,thick,dashed](-2.5,-2.5)--(0.95,-4);
\draw[->,thick,dashed](-2.5,-2.5)--(-0.95,-4);

\draw[->,thick,dashed](-2.5,3.5)--(0,3);
\draw[->,thick,dashed](-2.5,3.5)--(0,-3);

\draw[->,thick,dashed](3.5,2.5)--(4,0.9);
\draw[->,thick,dashed](3.5,2.5)--(4,-0.9);

\draw[very thick]
        (4,3.8) node{$\begin{pmatrix} 1 & 0 & 0 & 0 \\ 0 & 1 & 0 & 0 \\
                        0 & 0 & 1 & 0 \\ 0 & 0 & {e}^{n(\lambda_3 - \lambda_4)} & 1 \end{pmatrix}$}
    (4,-3.5) node{$\begin{pmatrix} \left(J_T\right)_1 & 0 & 0 \\
                        0 & 0 & 1\\ 0  & -1 & 0 \end{pmatrix}$}
    (-4.8,-3) node{$\begin{pmatrix} 1 & 0 & 0 & 0 \\ 0 & 1 & 0 & 0 \\
                        0 & {e}^{n(\lambda_1-\lambda_2)} & 1 & 0 \\ 0 & 0 & 0 & 1 \end{pmatrix}$}
    (-4.2,3.5) node{$\begin{pmatrix} 1 & 0 & 0 & 0 \\ 0 & 0 & -1 & 0 \\
                        0 & 1 & 0 & 0 \\ 0 & 0 & 0 & 1 \end{pmatrix}$};

    \draw[line width=1pt] (0,1) .. controls (1,2) and (1,4) .. (1.2,5);
    \draw[line width=1pt] (0,1) .. controls (-1,2) and (-1,4) .. (-1.2,5);
    \draw[line width=1pt] (0,-1) .. controls (1,-2) and (1,-4) .. (1.2,-5);
    \draw[line width=1pt] (0,-1) .. controls (-1,-2) and (-1,-4) .. (-1.2,-5);

    \draw[line width=1pt] (0,0.3) .. controls (1,0.3) and (4,0.9) .. (5,1);
    \draw[line width=1pt] (-0,0.3) .. controls (-1,0.3) and (-4,0.9) .. (-5,1);
    \draw[line width=1pt] (0,-0.3) .. controls (1,-0.3) and (4,-0.9) .. (5,-1);
    \draw[line width=1pt] (-0,-0.3) .. controls (-1,-0.3) and (-4,-0.9) .. (-5,-1);
\end{tikzpicture}
\end{center}
\caption{Unbounded lenses in case $c_2>0$ and $c_3=0$, corresponding to Cases I and II.
The figure also shows the jump contour $\Sigma_T$ and some of the jump matrices
in the RH problem \eqref{eq:RHforT} for $T$.}
\label{fig:lenses3}
\end{figure}
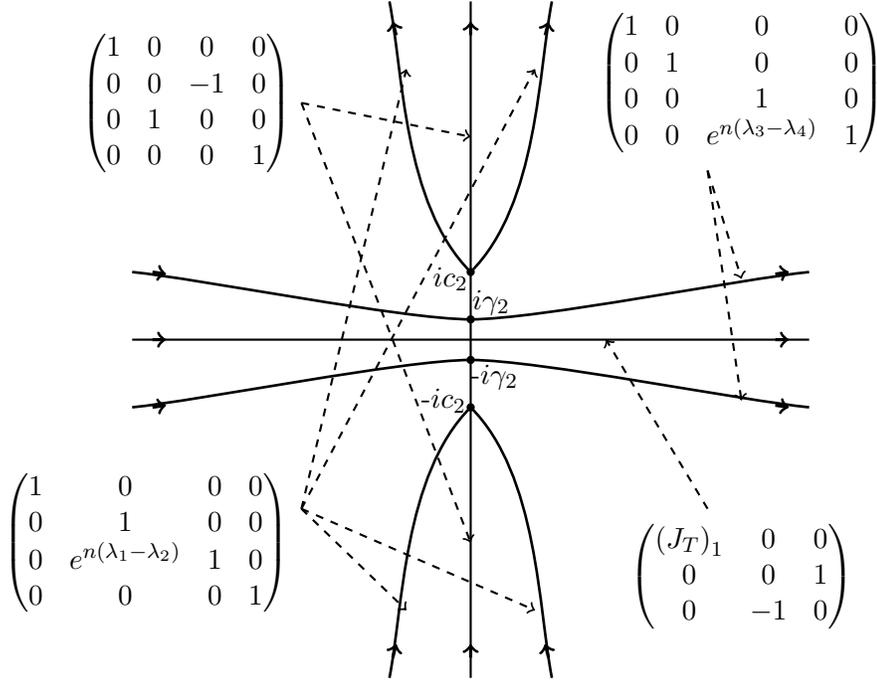

\begin{lemma} \label{lem:ineqlensmu3}
We can and do choose the lens $L_3$ around $S(\mu_3)$ such that
\[ \Re (\lambda_3 - \lambda_4) < 0 \quad \text{ in } L_3\setminus \mathbb R\]
and such that
\[ \{z\in \mathbb C \mid |z|>R, \, | \Im(z) |<\eps |\Re(z)|\} \subset L_3\]
for some $\eps>0$ and $R>0$.
\end{lemma}
\begin{proof}
The proof is similar to the proof of  Lemma \ref{lem:ineqlensmu2}.

Using \eqref{eq:deflambdaj} we obtain
\begin{align*}
\lambda_3-\lambda_4= -2g_3+g_2+\theta_2-\theta_3.
\end{align*}
and hence
\begin{align*}
\lambda_{3,\pm}-\lambda_{4,\pm}= -g_{3,+}-g_{3,-}+g_{2,\pm}+\theta_{2,\pm}-\theta_{3,\pm}\pm(g_{3,-}-g_{3,+}).
\end{align*}
Now by \eqref{eq:V3bis} and \eqref{eq:defsigma3}
\[  \theta_{2,\pm}(x) -\theta_{3,\pm}(x) =
    \begin{cases}
    V_3(x) \pm 2\pi i \sigma_3([x^*(\alpha),x]), & x > 0, \\
    V_3(x) \mp 2 \pi i \sigma_3([x, -x^*(\alpha)]), & x < 0.
    \end{cases}
\]
Combining this with \eqref{eq:jumpg3onR} and \eqref{eq:jumpg3onR0} leads to
\begin{align*}
    \lambda_{3,\pm}(x)-\lambda_{4,\pm}(x)=
        \pm 2\pi i \times \begin{cases} \sigma_3([x^*(\alpha),x]) - \mu_3([x,\infty)), & x>c_3,\\
            -\sigma_3([x,-x^*(\alpha)) - \mu_3([x,\infty)) - \frac{1}{3},  & x<-c_3,
            \end{cases}
\end{align*}
which is purely imaginary.  We find
\[ \frac{d}{dx} \Im \left( \pm (\lambda_3-\lambda_4)_{\pm} (x)\right)
    = 2\pi \frac{d(\sigma_3 + \mu_3)}{dx} (x)  \]
and this is positive for  $x \in \mathbb R$ in Cases I and II, and
for $|x| > c_3 > 0$ in the other cases.

As in Lemma \ref{lem:ineqlensmu2} the first statement of the lemma now follows
by the Cauchy Riemann equations.
The second statement follows by \eqref{eq:asymlambda3}, \eqref{eq:asymlambda4},
and Lemma \ref{lem:thetajasymptotics}.
\end{proof}

Now that we have defined the lenses $L_2$ and $L_3$  we can come to the actual definition of
the transformation $U \mapsto T$. The transformation   is based on the
following factorization of $(J_U)_2$ on $S(\sigma_2 - \mu_2)$ as given by \eqref{eq:JU2explicit} (recall
that $\lambda_{2,\pm} = \lambda_{3,\mp}$)
\begin{align} \nonumber
(J_U)_2 & = \begin{pmatrix} {e}^{n(\lambda_{3}- \lambda_{2})_-} & 0 \\
    1 & {e}^{n(\lambda_{3} - \lambda_{2})_+} \end{pmatrix} \\
    & = \begin{pmatrix} 1 & {e}^{n(\lambda_{3}- \lambda_{2})_-} \\ 0 & 1 \end{pmatrix}
    \begin{pmatrix} 0 & -1 \\ 1 & 0 \end{pmatrix}
    \begin{pmatrix} 1 & {e}^{n(\lambda_{3}- \lambda_{2})_+} \\ 0 & 1 \end{pmatrix}
    \label{eq:JU2factorization}
\end{align}
and  the factorization
of $(J_U)_3$ on $S(\mu_3)$ as given by \eqref{eq:JU3explicit} (recall
that $\lambda_{3,\pm} = \lambda_{4,\mp}$ or $\lambda_{3,\pm} = \lambda_{4,\mp} \pm \frac{2}{3} \pi i$
and $n$ is a multiple of three)
\begin{align} \nonumber
(J_U)_3 & = \begin{pmatrix} {e}^{n(\lambda_{3}- \lambda_{4})_+} & 1 \\
    0 & {e}^{n(\lambda_{3} - \lambda_{4})_-} \end{pmatrix} \\
    & = \begin{pmatrix} 1 & 0 \\ {e}^{n(\lambda_{3}- \lambda_{4})_-}  & 1 \end{pmatrix}
    \begin{pmatrix} 0 & 1 \\ -1 & 0 \end{pmatrix}
    \begin{pmatrix} 1 & 0 \\ {e}^{n(\lambda_{3}- \lambda_{4})_+} & 1 \end{pmatrix}.
    \label{eq:JU3factorization}
\end{align}
This leads to the following definition of $T$.

\begin{definition} \label{def:T}
We define the $4 \times 4$ matrix valued function $T$ by
\begin{equation} \label{eq:defT1}
    T = U \times
    \begin{cases}
    \begin{pmatrix} 1 & 0 & 0 & 0 \\
            0 & 1 & -{e}^{n(\lambda_{3}- \lambda_{2})} & 0 \\ 0 & 0 & 1 & 0 \\ 0 & 0 & 0 & 1\end{pmatrix}
        &  \begin{array}{l} \text{in the left part}\\
    \text{of the lens}\\ \text{around } S(\sigma_2 - \mu_2),\end{array} \\
    \begin{pmatrix} 1 & 0 & 0 & 0 \\
    0 & 1 & {e}^{n(\lambda_{3}- \lambda_{2})} & 0  \\ 0 & 0 & 1 & 0 \\ 0 & 0 & 0 & 1 \end{pmatrix}
    &  \begin{array}{l} \text{in the right part} \\ \text{of the lens } \\ \text{around } S(\sigma_2 - \mu_2).\end{array}
    \end{cases} \end{equation}

\begin{equation}\label{eq:defT2}
    T = U \times
    \begin{cases}
    \begin{pmatrix} 1 & 0 & 0 & 0 \\
    0 & 1 & 0 & 0 \\
    0 & 0 & 1 & 0 \\
    0 & 0 & -{e}^{n(\lambda_{3}- \lambda_{4})} & 1 \end{pmatrix}
    & \begin{array}{l} \text{in the upper part} \\ \text{of the lens}\\ \text{around } S(\mu_3),\end{array} \\
    \begin{pmatrix} 1 & 0 & 0 & 0 \\
    0 & 1 & 0 & 0 \\
    0 & 0 & 1 & 0 \\
    0 & 0 & {e}^{n(\lambda_{3}- \lambda_{4})} & 1 \end{pmatrix}
    &  \begin{array}{l} \text{in the lower part} \\ \text{of the lens}\\ \text{around } S(\mu_3),\end{array}\end{cases}
\end{equation}
and
\begin{equation} \label{eq:defT3}
    T= U \ \text{ elsewhere}.
    \end{equation}
  \end{definition}

 Then $T$ is defined and analytic in $\mathbb C \setminus \Sigma_T$ where $\Sigma_T$ is the contour
 consisting of the real and imaginary axes and the lips of the lenses around $S(\sigma_2 - \mu_2)$
 and $S(\mu_3)$.

\subsection{RH problem for $T$}

Now $T$ solves the following RH problem
\begin{equation}\label{eq:RHforT}
\left\{\begin{array}{l}
    T \textrm{ is analytic in } \mathbb C\setminus \Sigma_T, \\[5pt]
    T_+ = T_- J_T, \qquad \text{on }\Sigma_T, \\
    T(z)= (I+ \OO(z^{-1/3}))
    \diag \begin{pmatrix} 1 & z^{\frac{1}{3}}& 1 & z^{-\frac{1}{3}} \end{pmatrix}
    \begin{pmatrix} 1 & 0 \\ 0 & A_j \end{pmatrix} \\
     \hfill \textrm{as $z\to \infty$ in the $j$th quadrant},
    \end{array}\right.
\end{equation}
with certain jump matrices $J_T$ that will be described in the next subsection.
The matrices $A_j$ are given in \eqref{eq:A1A2}--\eqref{eq:A3A4}. The contour $\Sigma_T$ and some of the jump
matrices $J_T$ are shown in Figures \ref{fig:lenses1}, \ref{fig:lenses2} and \ref{fig:lenses3}.

Since the lenses around $S(\sigma_2-\mu_2)$ and $S(\mu_3)$ are unbounded,
we have to be careful about the asymptotic behavior of $T(z)$ as $z\to \infty$, since it could be
different from the asymptotic behavior of $U(z)$. However, a simple check using
\eqref{eq:asymlambda2}-\eqref{eq:asymlambda4}, \eqref{eq:defT1} and \eqref{eq:defT2}
shows that the asymptotic behavior is the same.

However, it is good to note the following. The asymptotic behavior of the Pearcey functions
given in \eqref{eq:Pnasymptotics} is not uniform up to the real and imaginary axes.
By following the transformations $Y\mapsto X\mapsto U$ we see that the same
is true for the asymptotic behavior of $U$. It requires an independent check
that after opening of the unbounded lenses the asymptotic behavior of $T$ is in fact uniform
in each of the quadrants.  This phenomenon also appeared in \cite{DuiKu2}.

\subsection{Jump matrices for $T$}

Our next task is to compute the jump matrices $J_T$ in the Riemann-Hilbert problem \eqref{eq:RHforT} for $T$.
The definitions \eqref{eq:defT1} and \eqref{eq:defT2} of $T$  and the structure
of the jump matrix $J_U$ as given in \eqref{eq:JUonR} and \eqref{eq:JUoniR} yield that
\[ T_+ = T_- J_T \]
where $J_T$ has again the structure
\begin{align} \label{eq:jumpTstructure}
 J_T = \begin{cases} \begin{pmatrix} (J_T)_1 & 0 \\ 0 & (J_T)_3 \end{pmatrix} &
    \begin{array}{l} \text{on $\mathbb R \setminus (-\gamma_3, \gamma_3)$ and} \\
        \text{on the lips of the lens} \\ \text{around $S(\mu_3)$}, \end{array} \\[20pt]
      \begin{pmatrix} 1 & 0 & 0 \\ 0 & (J_T)_2 & 0 \\ 0 & 0 & 1 \end{pmatrix} &
    \begin{array}{l} \text{on $i\mathbb R \setminus (-i\gamma_2, i\gamma_2)$ and} \\
        \text{on the lips of the lens} \\ \text{around $S(\sigma_2 -\mu_2)$}, \end{array}
        \end{cases}
     \end{align}
with $2 \times 2$ blocks $(J_T)_k$ for $k=1,2,3$.  The block structure as in \eqref{eq:jumpTstructure}
is not valid on the intervals $(-\gamma_3,\gamma_3)$ and $(-i \gamma_2,i \gamma_2)$ (if non-empty).
On these intervals the block structure changes to
\begin{equation} \label{eq:JToffdiagonal}
\begin{aligned}
    J_T & = \begin{pmatrix} (J_T)_1 & * \\ 0 & (J_T)_3 \end{pmatrix} \qquad  \text{on } (-\gamma_3, \gamma_3), \\
  J_T & = \begin{pmatrix} 1 & 0 & 0 \\ 0 & (J_T)_2 & 0 \\ 0 & * & 1 \end{pmatrix} \qquad \text{on } (-i \gamma_2, i \gamma_2),
     \end{aligned}
\end{equation}
with some non-zero entries that are denoted by $*$, see \eqref{eq:jumpTcase2a} and \eqref{eq:jumpTcase3a} below.

The diagonal blocks are given in the next lemma.
\begin{lemma}
\begin{enumerate}
\item[\rm (a)] For $(J_T)_1$ we have
\begin{align}\label{eq:JT1explicit}
    (J_T)_1 =\begin{cases} (J_U)_1  &\text{ on } \mathbb R,\\
    I_2 &\begin{array}{l} \text {on  the lips of the lens around }S(\mu_3),\end{array} \end{cases}
 \end{align}
  where $(J_U)_1$ is given by \eqref{eq:JU1explicit}.
\item[\rm (b)] For $(J_T)_2$ we have
\begin{align}   \label{eq:JT2explicit}
    (J_T)_2 =\begin{cases} \begin{pmatrix} 0 & -1 \\ 1 & 0 \end{pmatrix}
            &\text{ on } S(\sigma_2 - \mu_2), \\
  \begin{pmatrix}    1 & {e}^{n(\lambda_{3}- \lambda_{2})}  \\ 0 & 1 \end{pmatrix}
   & \begin{array}{l} \text{on the lips of the lens}\\ \text{around } S(\sigma_2 - \mu_2), \end{array} \\
    \begin{pmatrix}
    1 & 0 \\
    {e}^{n(\lambda_{2,+}-\lambda_{3,-})} & 1
    \end{pmatrix},& \text{ on } (-ic_2, ic_2).
    \end{cases}
     \end{align}
\item[\rm (c)]  For $(J_T)_3$ we have
\begin{align} \label{eq:JT3explicit}
 (J_T)_3 = \begin{cases} \begin{pmatrix} 0 & 1 \\ -1 & 0 \end{pmatrix}
    &\text{ on } S(\mu_3), \\
 \begin{pmatrix}  1 & 0 \\
     {e}^{n(\lambda_{3}- \lambda_{4})} & 1  \end{pmatrix}
    &   \begin{array}{l} \text{on the lips of the lens}\\ \text{around } S(\mu_3), \end{array} \\
    \begin{pmatrix}
    (-1)^n & {e}^{n(\lambda_{4,+}-\lambda_{3,-})}  \\
    0 & (-1)^n
    \end{pmatrix}& \text{ on } (-c_3, c_3).
  \end{cases}
    \end{align}
 \end{enumerate}
 \end{lemma}

\begin{proof}
The jumps follow from the jumps $J_U$ for $U$ in \eqref{eq:JU1explicit}--\eqref{eq:JU3explicit}
and the transformation \eqref{eq:defT1}--\eqref{eq:defT3}, where we use the
factorizations \eqref{eq:JU2factorization} and \eqref{eq:JU3factorization}.

Outside the lenses the blocks $(J_T)_2$ and $(J_T)_3$ have not changed,
and so on $(-ic_2, ic_2) = i \mathbb R \setminus S(\sigma_2-\mu_2)$
we have $(J_T)_2 = (J_U)_2$ and on $(-c_3,c_3) = \mathbb R \setminus S(\mu_3)$
we have $(J_T)_3 = (J_U)_3$.
See also \eqref{eq:jumpTcase3a} and \eqref{eq:jumpTcase2a} for the result
of the calculation of the full jump matrix $J_T$ on the  intervals $(-i \gamma_2, i\gamma_2)$
and $(-\gamma_3, \gamma_3)$.
\end{proof}

It remains to describe the off-diagonal entries in \eqref{eq:JToffdiagonal}
that only occur in case $\gamma_3 > 0$ or $\gamma_2 > 0$, that is in Cases III or
I/II.

\paragraph{Case III: $\gamma_3 > 0$.} In Case III we have $c_2 = 0$
and the lips of the lens  around $S(\sigma_2-\mu_2)$ intersect the real line at $\pm \gamma_3$ with $\gamma_3 > 0$.

In order to compute the jump matrix $J_T$ on $(-\gamma_3,\gamma_3)$, we have to note that by the regularity
assumption we have that $0\notin S(\mu_1)\cup S(\mu_3)$, and by the choice of $\gamma_3$ in \eqref{eq:gamma3bound}
we have that $[-\gamma_3,\gamma_3]$ is disjoint from $S(\mu_1)\cap S(\mu_3)$.
Also, the measure $\mu_1$ is symmetric so that $\alpha_{N/2}=1/2$.
Hence the jump matrix $J_U$ given in \eqref{eq:JUonR}, \eqref{eq:JU1explicit} and \eqref{eq:JU3explicit}
takes the form
\[ J_U = \begin{pmatrix}
(-1)^n & {e}^{n(\lambda_{2,+}-\lambda_{1,-})}& 0 & 0\\
0 & (-1)^n & 0 & 0 \\
0 & 0 & (-1)^n & {e}^{n(\lambda_{4,+}-\lambda_{3,-})}\\
0 & 0 & 0 & (-1)^n \end{pmatrix} \quad \text{on }(-\gamma_3,\gamma_3).\]
From \eqref{eq:defT1} we obtain
\begin{multline*}
J_T= \begin{pmatrix}
1 & 0 & 0 & 0 \\
0 & 1 & \mp {e}^{n(\lambda_3-\lambda_2)} & 0 \\
0 & 0 & 1 & 0 \\
0 & 0 & 0 & 1\end{pmatrix}
 J_U
\begin{pmatrix}
1 & 0 & 0 & 0 \\
0 & 1 & \pm {e}^{n(\lambda_3-\lambda_2)} & 0 \\
0 & 0 & 1 & 0 \\
0 & 0 & 0 & 1\end{pmatrix} \\
\text{ on  }(-\gamma_3,\gamma_3)\cap \mathbb R^\pm.
\end{multline*}

After some calculations using  \eqref{eq:lambdajump2}, \eqref{eq:lambdajump2b}, \eqref{eq:lambdajump6}
and the fact that $n$ is a multiple of $3$ we obtain from the last two expressions that
\begin{multline}\label{eq:jumpTcase2a}J_T=\begin{pmatrix}
(-1)^n & {e}^{n (\lambda_{2,+}-\lambda_{1,-})} & \pm {e}^{n(\lambda_{3,+}-\lambda_{1,-})} & 0\\
0 & (-1)^n  & 0 & \mp{e}^{n(\lambda_{4,+}-\lambda_{2,-})} \\
 0 & 0 & (-1)^n & {e}^{n(\lambda_{4,+}-\lambda_{3,-})}\\
0 & 0 & 0 & (-1)^n \end{pmatrix}\\
\text{ on  }(-\gamma_3,\gamma_3)\cap \mathbb R^\pm.
\end{multline}

\paragraph{Cases I and II: $\gamma_2 > 0$.}
In Cases I and II we have $c_3 = 0$ and the lips of the lens  around $S(\mu_3)$ intersect the imaginary axis
at $\pm i\gamma_2$ with $\gamma_2 > 0$.  By \eqref{eq:JUoniR} and \eqref{eq:JU2explicit} we have
\[J_U=\begin{pmatrix}
1 & 0& 0 & 0\\
0 & 1 & 0 & 0 \\
0 & {e}^{n(\lambda_{2,+}-\lambda_{3,-})} & 1 & 0\\
0 & 0 & 0 & 1\end{pmatrix} \quad
\text{ on }(-i \gamma_2,i \gamma_2).
\]
Combining this with \eqref{eq:defT2} and using \eqref{eq:lambdajump4},
the fact that $\lambda_4$ is analytic on $i\mathbb R\setminus \{0\}$ and the fact
that $n$ is a multiple of $3$  leads to \begin{align}\label{eq:jumpTcase3a}
J_T=\begin{pmatrix}
1 & 0 & 0& 0\\
0 & 1  &0 & 0 \\
 0 & {e}^{n (\lambda_{2,+}-\lambda_{3,-})} & 1 & 0\\
0 & \pm {e}^{n (\lambda_{2,+}-\lambda_{4,-})} & 0 &1 \end{pmatrix} \quad
\text{ on } (-i \gamma_2,i\gamma_2)\cap i \mathbb{R}^\pm.\end{align}

\subsection{Fourth transformation $T \mapsto S$}
In the fourth transformation we open up a lens around each interval $[a_k, b_k]$
of $S(\mu_1)$. The union of these lenses will be denoted by $L_1$ and $L_1$ is a closed bounded set.
This is done in a standard way based on the factorization

\begin{equation}     \label{eq:JT1factorization}
\begin{aligned}
(J_T)_1 = (J_U)_1 & = \begin{pmatrix} {e}^{n(\lambda_{1}- \lambda_{2})_+} & 1 \\
    0 & {e}^{n(\lambda_{1} - \lambda_{2})_-} \end{pmatrix} \\
    & = \begin{pmatrix} 1 & 0 \\ {e}^{n(\lambda_{1}- \lambda_{2})_-}  & 1 \end{pmatrix}
    \begin{pmatrix} 0 & 1 \\ -1 & 0 \end{pmatrix}
    \begin{pmatrix} 1 & 0 \\ {e}^{n(\lambda_{1}- \lambda_{2})_+} & 1 \end{pmatrix}.
\end{aligned}
\end{equation}

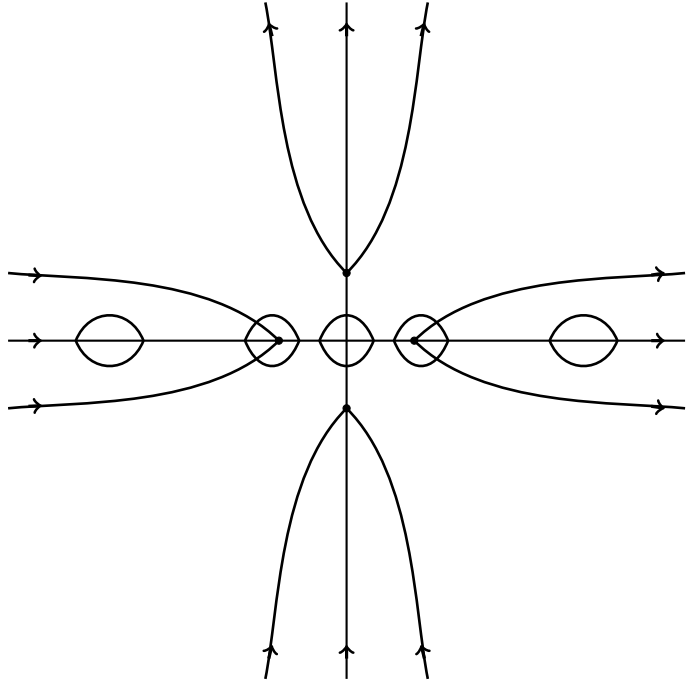
\begin{figure}[t]
\begin{center}
\begin{tikzpicture}[scale=0.9]
\draw[-,thick] (-5,0) -- (5,0);
\draw[-,thick] (0,-5) -- (0,5);
\filldraw [black] (1,0) circle (1.5pt)
    (-1,0) circle (1.5pt)
    (0,-1) circle (1.5pt)
    (0,1) circle (1.5pt);

% \draw[very thick](0.9,-.2) node{$c_3$};
%\draw[very thick](-0.85,-.2) node{-$c_3$};
%\draw[very thick](-0.25,0.9) node{$ i c_2$};
%\draw[very thick](-0.4,-0.9) node{-$ i c_2$};

\draw[->,very thick](0,4.5) -- (0,4.7);
\draw[->,very thick](-1.1,4.5) -- (-1.15,4.7);
\draw[->,very thick](1.1,4.5) -- (1.15,4.7);
\draw[->,very thick](0,-4.7) -- (0,-4.5);
\draw[->,very thick](-1.15,-4.7) -- (-1.1,-4.5);
\draw[->,very thick](1.15,-4.7) -- (1.1,-4.5);

%\draw[->,very thick](-0.6,0)--(-0.5,0);
%\draw[->,very thick](0,0.5)--(0,0.6);

\draw[->,very thick](-4.7,0) -- (-4.5,0);
\draw[->,very thick](-4.7,0.97) -- (-4.5,0.96);
\draw[->,very thick](-4.7,-0.97) -- (-4.5,-0.96);

\draw[->,very thick](4.5,0) -- (4.7,0);
\draw[->,very thick](4.5,0.97) -- (4.7,1);
\draw[->,very thick](4.5,-0.97) -- (4.7,-1);

\draw[line width=1pt] (0,1) .. controls (1,2) and (1,4) .. (1.2,5);
\draw[line width=1pt] (0,1) .. controls (-1,2) and (-1,4) .. (-1.2,5);
\draw[line width=1pt] (0,-1) .. controls (1,-2) and (1,-4) .. (1.2,-5);
\draw[line width=1pt] (0,-1) .. controls (-1,-2) and (-1,-4) .. (-1.2,-5);

\draw[line width=1pt] (1,0) .. controls (2,1) and (4,0.9) .. (5,1);
\draw[line width=1pt] (1,0) .. controls (2,-1) and (4,-0.9) .. (5,-1);
\draw[line width=1pt] (-1,0) .. controls (-2,1) and (-4,0.9) .. (-5,1);
\draw[line width=1pt] (-1,0) .. controls (-2,-1) and (-4,-0.9) .. (-5,-1);

\draw[line width=1pt] (-4,0) .. controls (-3.8,0.5) and (-3.2,0.5) .. (-3,0);
\draw[line width=1pt] (-4,0) .. controls (-3.8,-0.5) and (-3.2,-0.5) .. (-3,0);
\draw[line width=1pt] (4,0) .. controls (3.8,0.5) and (3.2,0.5) .. (3,0);
\draw[line width=1pt] (4,0) .. controls (3.8,-0.5) and (3.2,-0.5) .. (3,0);

\draw[line width=1pt] (-1.5,0) .. controls (-1.3,0.5) and (-0.9,0.5) .. (-0.7,0);
\draw[line width=1pt] (-1.5,0) .. controls (-1.3,-0.5) and (-0.9,-0.5) .. (-0.7,0);
\draw[line width=1pt] (1.5,0) .. controls (1.3,0.5) and (0.9,0.5) .. (0.7,0);
\draw[line width=1pt] (1.5,0) .. controls (1.3,-0.5) and (0.9,-0.5) .. (0.7,0);

\draw[line width=1pt] (-0.4,0) .. controls (-0.2,0.5) and (0.2,0.5) .. (0.4,0);
\draw[line width=1pt] (-0.4,0) .. controls (-0.2,-0.5) and (0.2,-0.5) .. (0.4,0);

\end{tikzpicture}
\end{center}
\caption{The opening of the lenses around $S(\mu_1)$. The lens around $(a_k,b_k)$
intersects the lens around $S(\mu_3)$ if and only if $\pm c_3\in (a_k,b_k)$ as is the case in the figure.  }
\label{fig:lensesmu1}
\end{figure}

\begin{lemma}\label{lem:ineqlensmu1}
We can and do choose the lenses  $L_1$ around $S(\mu_1)$ such that
\[
 \Re (\lambda_1 - \lambda_2) < 0 \quad\text{ in } L_1\setminus \mathbb R.\]
\end{lemma}
\begin{proof}
By \eqref{eq:deflambdaj} we obtain
\[(\lambda_1-\lambda_2)_\pm=V-l_1-2g_{1,\pm}+g_{2,\pm}-\theta_{1}.\]
By \eqref{eq:jumpg1onR}, \eqref{eq:V1bis} and \eqref{eq:jumpg1onR0} this leads to
\[
        (\lambda_1-\lambda_2)_\pm(x) =
        \begin{cases} \mp(g_{1,+}(x)-g_{1,-}(x)) = \pm 2\pi i \mu_1([x,\infty)), &   x>0, \\
        \mp(g_{1,+}(x)-g_{1,-}(x)) \pm \frac{2}{3}\pi i =\pm 2\pi i \mu_1([x,\infty)) \pm \frac{2}{3}\pi i, & x<0.
        \end{cases}
        \]
Thus $\pm(\lambda_1-\lambda_2)_\pm$ is purely imaginary and
\[ \frac{d}{dx} \Im \left(  \pm(\lambda_1-\lambda_2)_\pm(x) \right) = 2 \pi \frac{d \mu_1}{dx} \]
which is strictly positive on each of the intervals $(a_k,b_k)$ by the regularity assumption on $\mu_1$.
The lemma follows by the Cauchy Riemann equations.
\end{proof}

In addition to the condition described in the lemma we also make sure that the lenses around $S(\mu_1)$
do not intersect the lens around $S(\sigma_2-\mu_2)$ and the lips of the lens around $S(\mu_3)$,
except for the case when $\pm c_3\in S(\mu_1)$. In that case we choose the lips around the interval(s)
containing $\pm c_3$ in such a way that they intersect the lips of the lens around $S(\mu_3)$ exactly
once in each quadrant as it is shown in Figure~\ref{fig:lensesmu1}.

If $0\in S(\mu_1)$ then the lips of the lens around the interval containing $0$
intersect the imaginary axis in the two points $\pm i \widehat \gamma_2$ where
\[ 0 < \widehat \gamma_2 < c_2.\]
See also Figure \ref{fig:lensesmu1}.
 If in addition $0\in S(\mu_3)$ then also
\[ 0 < \widehat \gamma_2 < \gamma_2.\]
If $0 \not\in S(\mu_1)$ then we put
\[ \widehat \gamma_2 = 0. \]

We also take the lenses so that they are symmetric in both the real and imaginary axis.

We now define $S$ as follows.
\begin{definition} \label{def:S}
We define
\begin{equation} \label{eq:defS1}
    S = T\times
    \begin{cases}
      \begin{pmatrix} 1 & 0 & 0 & 0 \\
    -{e}^{n(\lambda_{1}- \lambda_{2})} & 1 & 0 & 0 \\
    0 & 0 & 1 & 0 \\ 0 & 0 & 0 & 1 \end{pmatrix}
    & \begin{array}{l}\text{in the upper part}\\ \text{of the lenses}\\ \text{around } S(\mu_1),\end{array} \\
      \begin{pmatrix} 1 & 0 & 0 & 0 \\
    {e}^{n(\lambda_{1}- \lambda_{2})} & 1 & 0 & 0 \\
    0 & 0 & 1 & 0 \\ 0 & 0 & 0 & 1 \end{pmatrix}
    & \begin{array}{l}\text{in the lower part}\\ \text{of the lenses}\\ \text{around } S(\mu_1),\end{array}
    \end{cases}
    \end{equation}
and
\begin{equation} \label{eq:defS2}
    S = T \qquad \text{elsewhere}.
    \end{equation}
\end{definition}
Then $S$ is defined and analytic in $\mathbb C\setminus \Sigma_S$ where $\Sigma_S$
is the contour consisting of the real and imaginary axis and the lips of the lenses
around $S(\mu_1)$, $S(\sigma_2-\mu_2)$ and $S(\mu_3)$.

\subsection{RH problem for $S$}

The asymptotic behavior as $z \to \infty$ clearly has not changed, and
so $S$ satisfies the following RH problem.

\begin{equation}\label{eq:RHforS}
\left\{\begin{array}{l}
    S \textrm{ is analytic in } \mathbb C\setminus \Sigma_S, \\[5pt]
    S_+ = S_- J_S, \qquad \text{on }\Sigma_S, \\
    S(z)= (I+ \OO(z^{-1/3}))
    \diag \begin{pmatrix} 1 & z^{\frac{1}{3}}& 1 & z^{-\frac{1}{3}} \end{pmatrix}
    \begin{pmatrix} 1 & 0 \\ 0 & A_j \end{pmatrix} \\
     \hfill \textrm{as $z\to \infty$ in the $j$th quadrant}
    \end{array}\right.
\end{equation}
with jump matrices $J_S$ that are described next.

The jump matrices  $J_S$ have again the block structure
\begin{align} \label{eq:jumpSstructure}
 J_S = \begin{cases} \begin{pmatrix} (J_S)_1 & 0 \\ 0 & (J_S)_3 \end{pmatrix} &
    \begin{array}{l} \text{on $\mathbb R \setminus (-\gamma_3, \gamma_3)$ and} \\
        \text{on the lips of the lenses} \\ \text{around $S(\mu_1)$ and $S(\mu_3)$}, \end{array} \\[20pt]
     \begin{pmatrix} 1 & 0 & 0 \\ 0 & (J_S)_2 & 0 \\ 0 & 0 & 1 \end{pmatrix} &
    \begin{array}{l} \text{on $i\mathbb R \setminus ((-i\gamma_2, i\gamma_2) \cup (-i\widehat{\gamma}_2, i \widehat{\gamma}_2))$} \\
        \text{and on the lips of the lens} \\ \text{around $S(\sigma_2 -\mu_2)$}, \end{array}
        \end{cases}
     \end{align}
with $2 \times 2$ blocks $(J_S)_k$ for $k=1,2,3$. The block structure is different
on the intervals  $(-\gamma_3,\gamma_3)$,
$(-i \gamma_2,i\gamma_2)$ and $(-i \widehat \gamma_2,i \widehat \gamma_2)$ if non-empty.

On these intervals we have
\begin{equation} \label{eq:JSoffdiagonal}
\begin{aligned}
 J_S & = \begin{pmatrix} (J_S)_1 & * \\ 0 & (J_S)_3 \end{pmatrix} && \text{on } (-\gamma_3, \gamma_3), \\
 J_S & = \begin{pmatrix} 1 & 0 & 0 \\ * & (J_S)_2 & 0 \\ * & * & 1 \end{pmatrix}
                && \text{on } (-i \gamma_2, i\gamma_2) \text{ or } (-i \widehat{\gamma}_2, i \widehat{\gamma}_2),
     \end{aligned}
\end{equation}
with possible non-zero entries that are denoted by $*$.

The diagonal blocks are given in the following lemma.

\begin{lemma}
\begin{enumerate}
\item[\rm (a)]
For $(J_S)_1$ we have
\begin{align} \label{eq:JS1explicit}
(J_S)_1 = \begin{cases}
    \begin{pmatrix} 0 & 1 & \\ -1 & 0  \end{pmatrix} & \text{on } S(\mu_1),\\
    \begin{pmatrix} {e}^{-2n \pi i \alpha_k} & {e}^{n(\lambda_{2,+}-\lambda_{1,-})} \\ 0 & {e}^{2n\pi i \alpha_k} \end{pmatrix} &
    \begin{array}{l} \text{on } (b_k, a_{k+1}) \\ \text{for } k =0, \ldots, N, \end{array} \\
        \begin{pmatrix} 1 & 0  \\
    {e}^{n(\lambda_{1}- \lambda_{2})(z)} & 1  \end{pmatrix} &
    \begin{array}{l} \text{on the lips of the lenses} \\ \text{around } S(\mu_1), \end{array} \\
    I_2 & \begin{array}{l} \text{on  the lips of the lens} \\ \text{around } S(\mu_3). \end{array}
\end{cases}\end{align}
\item[\rm (b)] For $(J_S)_2$ we have
\begin{align} \label{eq:JS2explicit}
    (J_S)_2 & = \begin{cases} \begin{pmatrix} 0 & -1 \\ 1 & 0 \end{pmatrix} & \text{on } S(\sigma_2-\mu_2), \\
        \begin{pmatrix} 1 & {e}^{n(\lambda_3-\lambda_2)} \\ 0 & 1 \end{pmatrix} &
            \begin{array}{l} \text{on the lips of the lens} \\ \text{around } S(\sigma_2-\mu_2), \end{array} \\
    \begin{pmatrix} 1 & 0 \\ {e}^{n(\lambda_{2,+}-\lambda_{3,-})} & 1 \end{pmatrix} &
        \text{on } (-ic_2, ic_2).
        \end{cases}
        \end{align}
\item[\rm (c)] For $(J_S)_3$ we have
\begin{align} \label{eq:JS3explicit}
(J_S)_3= \begin{cases} \begin{pmatrix} 0 & 1 \\ -1 & 0 \end{pmatrix} & \text{on } S(\mu_3), \\
    \begin{pmatrix} (-1)^n & {e}^{n(\lambda_{4,+}-\lambda_{3,-})} \\ 0 & (-1)^n \end{pmatrix} & \text{on } (-c_3, c_3), \\
    I_2&  \begin{array}{l} \text{on the lips of the lenses} \\ \text{around } S(\mu_1), \end{array} \\
    \begin{pmatrix} 1 & 0 \\ {e}^{n(\lambda_3-\lambda_4)} & 1 \end{pmatrix} &
    \begin{array}{l} \text{on the lips of the lens} \\ \text{around }S(\mu_3). \end{array}
    \end{cases}
\end{align}
\end{enumerate}
\end{lemma}

\begin{proof}
The expressions for $(J_S)_1$ on $S(\mu_1)$ and on the lips of the lenses around $S(\mu_1)$
follow from the factorization \eqref{eq:JT1factorization} and the transformation
\eqref{eq:defS1}--\eqref{eq:defS2}.
On the other contours we have $(J_S)_1 = (J_T)_1$. Also $(J_S)_k = (J_T)_k$ for $k=2,3$,
and so we find all further expressions in the lemma from \eqref{eq:JT1explicit}--\eqref{eq:JT3explicit}.
\end{proof}

We next give the jump matrices $J_S$ on the intervals $(-\gamma_3, \gamma_3)$, $(-i\gamma_2, i\gamma_2)$
and $(-i \widehat \gamma_2, i \widehat \gamma_2)$ and in particular  the off-diagonal entries
 that were denoted by $*$ in \eqref{eq:JSoffdiagonal}. They depend on the Cases I-V. We simply
 present the formulas without further comment. Of course they follow from the jump matrices $J_T$
 and the tranformation  \eqref{eq:defS1}--\eqref{eq:defS2}.
Note that
\[ \widehat{\gamma}_2 > 0 \Leftrightarrow 0 \in S(\mu_1), \qquad \gamma_3 > 0 \Leftrightarrow c_2 = 0,
    \qquad \gamma_2 > 0 \Leftrightarrow c_3 = 0. \]

\paragraph{Case I.}
In Case I we have $0 \in S(\mu_1)$, $c_2 > 0$ and $c_3 = 0$ so that
$\widehat{\gamma}_2 > 0$ and $\gamma_2 > 0$ and $ \gamma_3 = 0$. By
construction $\widehat{\gamma}_2 < \gamma_2$ and we find in the Case
I
\begin{align} \label{eq:JScaseI}
J_S= \begin{cases}
\begin{pmatrix}
 1 & 0 & 0 & 0 \\
 0 & 1 & 0 & 0 \\
 \mp {e}^{n(\lambda_1-\lambda_{3,-})} & {e}^{n(\lambda_{2,+}-\lambda_{3,-})} & 1 & 0 \\
 -{e}^{n(\lambda_{1}-\lambda_4)} & \pm {e}^{n(\lambda_{2,+}-\lambda_4)} & 0 & 1
\end{pmatrix} & \begin{array}{l}\text{on } (-i \widehat \gamma_2,i \widehat\gamma_2)\cap i \mathbb R^\pm\end{array} \\
     J_T & \text{ on } (-i \gamma_2, -i \widehat \gamma_2) \cup (i \widehat\gamma_2,i\gamma_2),
     \end{cases}
\end{align}
where $J_T$ is given by \eqref{eq:jumpTcase3a}.

\paragraph{Case II.} In Case II we have $0 \not\in S(\mu_1)$, $c_2 > 0$ and $c_3 = 0$ so that
$\widehat{\gamma}_2 = 0$, $\gamma_2 > 0$ and  $\gamma_3 = 0$. In this case
\begin{align} \label{eq:JScaseII}
J_S= J_T \qquad \text{on } (-i \gamma_2, i \gamma_2) \end{align}
where $J_T$ is given by \eqref{eq:jumpTcase3a}.

\paragraph{Case III.} In Case III we have $0 \not\in S(\mu_1)$, $c_2 = 0$ and $c_3 > 0$ so that
$\widehat{\gamma}_2 = 0$, $\gamma_2 = 0$ and $\gamma_3 > 0$. We have in this case
\begin{align} \label{eq:JScaseIII}
J_S= J_T \qquad \text{on } (-\gamma_3,  \gamma_3) \end{align}
where $J_T$ is given by \eqref{eq:jumpTcase2a}.

\paragraph{Case IV.} In Case IV we have $0 \in S(\mu_1)$, $c_2 > 0$ and $c_3 > 0$ so that
$\widehat{\gamma}_2 > 0$, $\gamma_2 = 0$ and $\gamma_3 = 0$. We have in this case
\begin{align} \label{eq:JScaseIV}
J_S=
\begin{pmatrix}
 1 & 0 & 0 & 0 \\
 0 & 1 & 0 & 0 \\
 \mp {e}^{n(\lambda_1-\lambda_{3,-})} & {e}^{n(\lambda_{2,+}-\lambda_{3,-})} & 1 & 0 \\
 0 & 0& 0 & 1
\end{pmatrix} & \begin{array}{l}\text{on } (-i \widehat \gamma_2,i \widehat\gamma_2)\cap i \mathbb R^\pm\end{array}.
\end{align}

\paragraph{Case V.} In Case V we have $0 \not\in S(\mu_1)$, $c_2 > 0$ and $c_3 > 0$ so that
$\widehat{\gamma}_2 = 0$, $\gamma_2 = 0$ and $\gamma_3 = 0$. There are no exceptional
intervals in this case.

\subsection{Behavior of jumps as $n \to \infty$}

Having collected all the jump matrices $J_S$ we may study their behavior as $n \to \infty$.
It turns out that all off-diagonal entries of the form $\pm e^{n(\lambda_j - \lambda_k)}$
are such that $\Re (\lambda_j - \lambda_k) < 0$ and therefore they are exponentially decaying
as $n \to \infty$.

For $(J_S)_1$ in \eqref{eq:JS1explicit} we have off diagonal entries
\begin{align} \begin{cases}
    e^{n(\lambda_{2,+}-\lambda_{1,-})} & \text{ on } \mathbb R \setminus S(\mu_1), \\
    e^{n(\lambda_1 - \lambda_2)} &  \text{ on the lips of the lenses around } S(\mu_1),
    \end{cases}
    \end{align}
    which are indeed exponentially decaying because of \eqref{eq:varlambda21} and Lemma \ref{lem:ineqlensmu1}.

For $(J_S)_2$ in \eqref{eq:JS2explicit} we have off diagonal entries
\begin{align} \begin{cases}
    e^{n(\lambda_{2,+}-\lambda_{3,-})} & \text{ on } i \mathbb R \setminus S(\sigma_2 - \mu_2), \\
    e^{n(\lambda_3 - \lambda_2)} &  \text{ on the lips of the lens $L_2$ around } S(\sigma_2 - \mu_2),
        \end{cases} \end{align}
    which are  exponentially decaying because of \eqref{eq:varlambda23} and Lemma \ref{lem:ineqlensmu2}.

For $(J_S)_3$ in  \eqref{eq:JS3explicit} we have off diagonal entries
\begin{align} \begin{cases}
    e^{n(\lambda_{4,+}-\lambda_{3,-})} &  \text{ on } \mathbb R \setminus S(\mu_3), \\
    e^{n(\lambda_3 - \lambda_4)} &  \text{ on the lips of the lens $L_3$ around } S(\mu_3),
        \end{cases} \end{align}
    which are  exponentially decaying because of \eqref{eq:varlambda43} and Lemma \ref{lem:ineqlensmu3}.

The remaining off-diagonal entries appear in the $*$ entries in \eqref{eq:JSoffdiagonal} on
the special intervals $(-\gamma_3, \gamma_3)$ or $(-i \gamma_2, i \gamma_2)$ and $(-i \widehat \gamma_2, i \widehat \gamma_2)$.
They are explicitly given in the formulas \eqref{eq:jumpTcase3a}, \eqref{eq:jumpTcase2a}, \eqref{eq:JScaseI}, and \eqref{eq:JScaseIV}.
It turns out that all these entries are exponentially decaying as $n \to \infty$.
We will not verify all the cases here, but let us check the jump matrix $J_S$ on
$(-i \widehat \gamma_2, i \widehat \gamma_2)$ in Case I as given in \eqref{eq:JScaseI}.
Here there are four off diagonal entries
\begin{equation} \label{eq:entriesJScaseI}
    \mp {e}^{n(\lambda_1 - \lambda_{3,-})}, \quad   {e}^{n(\lambda_{2,+} - \lambda_{3,-})},
    \quad  - {e}^{n(\lambda_1 - \lambda_4)}, \quad  \pm {e}^{n(\lambda_{2,+} - \lambda_{4})}.
    \end{equation}
Since we are in Case I, the inequalities from \eqref{eq:varlambda23} and Lemmas \ref{lem:ineqlensmu3} and
\ref{lem:ineqlensmu1} apply on $[-i\widehat{\gamma}_2, i \widehat{\gamma_2}]$ since this interval
is contained in $i \mathbb R \setminus S(\sigma_2-\mu_2)$ and it belongs to the two lenses $L_1$ and $L_3$.
So we have
\begin{equation} \label{eq:lambdainequalities}
    \Re(\lambda_1 - \lambda_2) \leq 0, \qquad \Re( \lambda_{2,+} - \lambda_{3,-}) < 0, \qquad \Re(\lambda_3 - \lambda_4) \leq 0
    \end{equation}
on $[-i\widehat{\gamma}_2, i \widehat{\gamma_2}]$. In fact, equality in the first and third
inequalities of \eqref{eq:lambdainequalities} holds only at $0$.
Then indeed all entries in \eqref{eq:entriesJScaseI} are exponentially decaying as $n\to \infty$,
uniformly on $[-i\widehat{\gamma}_2, i \widehat{\gamma}_2]$.

In the next step of the steepest descent analysis we will ignore all exponentially small
entries in the jump matrices $J_S$. This will lead to matrices $J_M$
that we use as jump matrices for the so-called global parametrix.
The matrices $J_M$ are given in \eqref{eq:jumpM1}--\eqref{eq:jumpM2} below.

\section{Global parametrix} \label{sec:global}

If we ignore all entries in the jump matrices $J_S$ that are exponentially
small as $n \to \infty$, we find the following model Riemann-Hilbert problem
for $M : \mathbb C \setminus (\mathbb R \cup i \mathbb R) \to \mathbb C^{4 \times 4}$.

\begin{equation}\label{eq:RHforM}
\left\{\begin{array}{l}
    M  \textrm{ is analytic in } \mathbb C\setminus (\mathbb R \cup i \mathbb R), \\[5pt]
    M_+ = M_- J_M, \qquad \text{on } \mathbb R \cup i \mathbb R, \\
    M(z)= (I+\OO(z^{-1}))
     \begin{pmatrix} 1 & 0 & 0 & 0 \\
     0 & z^{1/3} & 0 & 0 \\
     0 & 0 & 1 & 0 \\
     0 & 0 & 0 & z^{-1/3} \end{pmatrix}
     \begin{pmatrix} 1 & 0 \\ 0 & A_j \end{pmatrix}  \\
     \qquad \textrm{as $z\to \infty$ in the $j$th quadrant},
    \end{array}\right.
\end{equation}
where $J_M$ is given as follows.
On the real line, the jump matrix has the block form
\begin{align} \label{eq:jumpMonR}
    J_M(x) = \begin{pmatrix} (J_M)_1(x) & 0 \\ 0 & (J_M)_3(x) \end{pmatrix}, \qquad
     x \in \mathbb R,
\end{align}
with
\begin{align} \label{eq:jumpM1}
    (J_M)_1 & =  \begin{cases}
    \begin{pmatrix} 0 & 1 \\ -1 & 0 \end{pmatrix} &
        \text{ on } S(\mu_1), \\
    \begin{pmatrix} {e}^{-2n \pi i \alpha_k} & 0 \\ 0 & {e}^{2n \pi i \alpha_k} \end{pmatrix} &
        \begin{array}{l} \text{on } (b_k, a_{k+1}), \\
            \text{for } k = 1, \ldots, N-1, \end{array} \\
        I_2 & \text{ on } (-\infty, a_1) \cup (b_N, +\infty),
        \end{cases} \\ \label{eq:jumpM3}
   (J_M)_3 & = \begin{cases}
    \begin{pmatrix} 0 & 1 \\ -1 & 0 \end{pmatrix} &
        \text{ on } S(\mu_3), \\
        (-1)^n I_2 & \text{ on }  (-c_3,c_3).
    \end{cases}
\end{align}
On the imaginary axis the jump matrix has the block form
\begin{align} \label{eq:jumpMoniR}
    J_M(x) & = \begin{pmatrix} 1 & 0 & 0  \\ 0 & (J_M)_2(x) & 0 \\ 0 & 0 & 1 \end{pmatrix},
     \qquad x \in i \mathbb R,
    \end{align}
with
\begin{align} \label{eq:jumpM2}
    (J_M)_2 & =  \begin{cases}
    \begin{pmatrix} 0 & -1 \\ 1 & 0 \end{pmatrix} &
        \text{ on } S(\sigma_2 - \mu_2), \\
        I_2 & \text{ on } (-ic_2,ic_2).
        \end{cases}
        \end{align}

Note that we have strenghtened the $\OO$ term in the asymptotic condition
in \eqref{eq:RHforM} from $\OO(z^{-1/3})$ to $\OO(z^{-1})$.

The solution has fourth root singularities at
all branch points $a_k, b_k$, for $k=1, \ldots, N$, and at $\pm c_2$, $\pm c_3$
if $c_2, c_3 > 0$. We give more details on the construction in the
rest of this section.

\subsection{Riemann surface as an $M$-curve}

We follow the approach of \cite{KuiMo} in using meromorphic differentials
on the Riemann surface as the main ingredient in the construction of the
global parametrix. Recall that $\mathcal R$ is a four sheeted cover
of the Riemann sphere. We use
\[ \pi : \mathcal R \to \overline{\mathbb C} \]
to denote the canonical projection, and we let $\pi_j$ be its
restriction to the $j$th sheet. Then for $z \in \mathbb C$ we have
that
\[ \pi_j^{-1}(z) \]
denotes the point on the $j$th sheet that projects onto $z$. It
is well-defined for $z \in \mathbb C \setminus (\mathbb R \cup i \mathbb R)$.

The Riemann surface $\mathcal R$ has the structure of an $M$-curve.
It has an anti-holomorphic involution
\[ \phi : \mathcal R \to \mathcal R : P \mapsto \bar{P} \]
with $\bar{P}$ on the same sheet as $P$. The fixed point set
of $\phi$ consists of $g+1$ connected components ($g$ is the genus of $\mathcal R$)
\[ \{ P \in \mathcal R : \phi(P) = P \} = \Sigma_0 \cup \Sigma_1 \cup \cdots \cup \Sigma_g \]
where
\begin{itemize}
\item $\Sigma_0$ contains $\infty_1$ and consists of the
intervals $(-\infty, a_1]$ and $[b_N, \infty)$ on the first and
second sheets.
\item $\Sigma_i$ for $i =1,\ldots, N-1$ contains the intervals $[b_i, a_{i+1}]$ on the first
and second sheets.
\end{itemize}

A full description of the curves depends on the case we are in. Recall
that there are five cases.
\begin{description}
\item[Cases I and II:] In Cases I and II we have $c_3 = 0$ and the genus is $N-1$.
In these cases we have only the above curves $\Sigma_0, \ldots, \Sigma_{N-1}$ and they
are on the first and second sheets only.
\item[Case III:] In Case III we still have $g = N-1$. Since $0 \not\in \supp(\mu_1)$
the number $N$ is even (by symmetry). The curve $\Sigma_{N/2}$ is now of
a different character, since it visits all four sheets. Indeed
we have that $\Sigma_{N/2}$ consists of $[b_{N/2}, a_{N/2 + 1}]$ on the first
and second sheets, and $[-c_3, c_3]$ on the third and fourth seets.
The component $\Sigma_{N/2}$ is also unusual in that it contains four branch points,
in contrast to the other curves that have two branch points.
\item[Cases IV and V:] In Cases IV and V the genus is $N$. In these cases
there is an additional component $\Sigma_N$ that consists of the
intervals $[-c_3,c_3]$ on the third and fourth sheets. The other components
are on the first and second sheets only, as in Cases I and II.
\end{description}

We need the following result on non-special divisors on $M$-curves.

\begin{lemma} \label{lem:nonspecial}
If $P_j \in \Sigma_j$ for $j=1, \ldots, g$ then the divisor $\sum\limits_{j=1}^g P_j$
is non-special.
\end{lemma}
\begin{proof}
This can be found in \cite[Theorem 2.4]{Hui}.
\end{proof}

Recall that the divisor $\sum\limits_{j=1}^g P_j$ is non-special if and only
if there are no non-constant analytic functions on $\mathcal R \setminus \{P_1, \ldots, P_g\}$
with only simple poles at the points $P_1, \ldots, P_g$.

\subsection{Canonical homology basis}

The Riemann surface has a canonical homology basis $\{ A_1, \ldots, A_g; B_1, \ldots, B_g\}$
that we choose such that the cycle $A_j$ is homologous to $\Sigma_j$ but disjoint from it.
See Figures \ref{fig: Case II}--\ref{fig: Case V} for an illustration of the curves for
the cycles with their orientation in the various cases.
In all cases we choose the $B$-cycles such that as sets they are invariant under the involution $\phi$.

The Cases I and II are very similar and we only show the Figure \ref{fig: Case II} for
Case II. The $B$-cycles are on the first sheet only.
The cycle $B_j$  surrounds the interval $[a_1, b_j]$ once in the negative (counterclockwise)
direction, and it intersects the real line  in $(-\infty, a_1)$ and in $(b_j, a_{j+1})$.
The $A$-cycles are partly on the upper half of the first sheet and partly on the lower half of
the second sheet. The cycle $A_j$ passes through the cuts $[a_j,b_j]$ and $[a_{j+1}, b_{j+1}]$
with orientation as indicated in Figures \ref{fig: Case II}.
In Case II there is an $A$-cycle that intersects with the imaginary axis. We make
sure that on the second sheet it does so in the interval $(-ic_2, 0)$.

% ******************************

\begin{figure}[t] % Case II
\begin{center}
\begin{tikzpicture}[scale=1.5]  % First sheet
\draw[-] (-3,0) node[below]{$a_1$}--(-2,0) node[below]{$b_1$};
\draw[-] (-1.5,0) node[below]{}--(-0.6,0) node[below]{};
\draw[-] (0.6,0) node[below]{}--(1.5,0) node[below]{};
\draw[-] (2.0,0) node[below]{$a_N$}--(3.0,0) node[below]{$b_N$};
\filldraw  (-3,0) circle(1pt) (-2,0) circle(1pt) (-1.5,0) circle(1pt) (-0.6,0) circle(1pt);
\filldraw  (0.6,0) circle(1pt) (1.5,0) circle(1pt) (2,0) circle(1pt) (3,0) circle(1pt);
\draw[very thick,->] (-2.5,0) ellipse (20pt and 10pt) (-1.84,0.3) node{$B_1$};
\draw[very thick] (-1.8,0) ellipse (50pt and 15pt) (-0.25,0.4) node{$B_2$};
\draw[very thick] (-1.2,0) ellipse (85pt and 20pt) (1.4,0.5) node{$B_{N-1}$};
\draw[->,very thick] (-1.1,-0.7)--(-1.11,-0.7);
\draw[->,very thick] (-2,-0.53)--(-2.01,-0.53);
\draw[->,very thick] (-2.6,-0.35)--(-2.61,-0.35);
\draw[very thick,dashed] (-2.3,0)..controls (-2.3,0.2) and (-1.2,0.2).. (-1.2,0) node[above]{$A_1$};
\draw[very thick,dashed] (-0.9,0)..controls (-0.9,0.2) and (0.9,0.2).. (0.9,0) node[above]{};
\draw[very thick,dashed] (1.2,0)..controls (1.2,0.2) and (2.3,0.2).. (2.3,0) node[above]{$A_{N-1}$};
\draw[->,very thick] (-1.72,0.15)--(-1.73,0.15);
\draw[->,very thick] (0.01,0.16)--(0.0,0.16);
\draw[->,very thick] (1.79,0.16)--(1.78,0.16);
\draw[-] (-4.5,0.8)--(4.5,0.8)--(4.5,-0.8)--(-4.5,-0.8)--(-4.5,0.8);
\end{tikzpicture}

\begin{tikzpicture}[scale=1.5] % Second sheet
\draw[-] (-3,0) node[above]{$a_1$}--(-2,0) node[above]{$b_1$};
\draw[-] (-1.5,0) node[above]{}--(-0.6,0) node[above]{};
\draw[-] (0.6,0) node[above]{}--(1.5,0) node[above]{};
\draw[-] (2.0,0) node[above]{$a_N$}--(3.0,0) node[above]{$b_N$};
\filldraw  (-3,0) circle(1pt) (-2,0) circle(1pt) (-1.5,0) circle(1pt) (-0.6,0) circle(1pt);
\filldraw  (0.6,0) circle(1pt) (1.5,0) circle(1pt) (2,0) circle(1pt) (3,0) circle(1pt);
\draw[-] (0,0.4) node[left]{$ic_2$}--(0,0.8);
\draw[-] (0,-0.4) node[left]{$-ic_2$}--(0,-0.8);
\filldraw (0,0.4) circle(1pt) (0,-0.4) circle(1pt);
\draw[very thick,dashed] (-2.3,0)..controls (-2.3,-0.2) and (-1.2,-0.2).. (-1.2,-0) node[below]{$A_1$};
\draw[very thick,dashed] (-0.9,0)..controls (-0.9,-0.2) and (0.9,-0.2).. (0.9,0) node[below]{};
\draw[very thick,dashed] (1.2,0)..controls (1.2,-0.2) and (2.3,-0.2).. (2.3,-0) node[below]{$A_{N-1}$};
\draw[<-,very thick] (-1.72,-0.15)--(-1.73,-0.15);
\draw[<-,very thick] (0.01,-0.16)--(0.0,-0.16);
\draw[<-,very thick] (1.79,-0.16)--(1.78,-0.16);
\draw[-] (-4.5,0.8)--(4.5,0.8)--(4.5,-0.8)--(-4.5,-0.8)--(-4.5,0.8);
\end{tikzpicture}

\begin{tikzpicture}[scale=1.5] % Third sheet
\draw[-] (0,0.4) node[left]{$ic_2$}--(0,0.8);
\draw[-] (0,-0.4) node[left]{$-ic_2$}--(0,-0.8);
\filldraw (0,0.4) circle(1pt) (0,-0.4) circle(1pt);
\draw[-] (-4.5,0)--(4.5,0);
\draw[-] (-4.5,0.8)--(4.5,0.8)--(4.5,-0.8)--(-4.5,-0.8)--(-4.5,0.8);
\end{tikzpicture}

\begin{tikzpicture}[scale=1.5] % Fourth sheet
\draw[-] (-4.5,0)--(4.5,0);
\draw[-] (-4.5,0.8)--(4.5,0.8)--(4.5,-0.8)--(-4.5,-0.8)--(-4.5,0.8);
\end{tikzpicture}
\end{center}
\caption{Case II}
\label{fig: Case II}
\end{figure}

% *******************

\begin{figure}[t] % Case III
\begin{center}
\begin{tikzpicture}[scale=1.5]  % First sheet
\draw[-] (-3,0) node[below]{$a_1$}--(-2,0) node[below]{$b_1$};
\draw[-] (-1.5,0) node[below]{}--(-0.6,0) node[below]{};
\draw[-] (0.6,0) node[below]{}--(1.5,0) node[below]{};
\draw[-] (2.0,0) node[below]{$a_N$}--(3.0,0) node[below]{$b_N$};
\filldraw  (-3,0) circle(1pt) (-2,0) circle(1pt) (-1.5,0) circle(1pt) (-0.6,0) circle(1pt);
\filldraw  (0.6,0) circle(1pt) (1.5,0) circle(1pt) (2,0) circle(1pt) (3,0) circle(1pt);
\draw[very thick,->] (-2.5,0) ellipse (20pt and 10pt) (-1.84,0.3) node{$B_1$};
\draw[very thick] (-1.8,0) ellipse (50pt and 15pt) (-0.25,0.4) node{};
\draw[very thick] (-1.2,0) ellipse (85pt and 20pt) (1.4,0.5) node{$B_{N-1}$};
\draw[->,very thick] (-1.1,-0.7)--(-1.11,-0.7);
\draw[->,very thick] (-2,-0.53)--(-2.01,-0.53);
\draw[->,very thick] (-2.6,-0.35)--(-2.61,-0.35);
\draw[very thick,dashed] (-2.3,0)..controls (-2.3,0.2) and (-1.2,0.2).. (-1.2,0) node[above]{$A_1$};
\draw[very thick,dashed] (-0.9,0)..controls (-0.9,0.2) and (0.9,0.2).. (0.9,0) node[above]{};
\draw[very thick,dashed] (1.2,0)..controls (1.2,0.2) and (2.3,0.2).. (2.3,-0) node[above]{$A_{N-1}$};
\draw[->,very thick] (-1.72,0.15)--(-1.73,0.15);
\draw[->,very thick] (0.01,0.16)--(0.0,0.16);
\draw[->,very thick] (1.79,0.16)--(1.78,0.16);
\draw[-] (-4.5,0.8)--(4.5,0.8)--(4.5,-0.8)--(-4.5,-0.8)--(-4.5,0.8);
\end{tikzpicture}

\begin{tikzpicture}[scale=1.5] % Second sheet
\draw[-] (-3,0) node[above]{$a_1$}--(-2,0) node[above]{$b_1$};
\draw[-] (-1.5,0) node[above]{}--(-0.6,0) node[above]{};
\draw[-] (0.6,0) node[above]{}--(1.5,0) node[above]{};
\draw[-] (2.0,0) node[above]{$a_N$}--(3.0,0) node[above]{$b_N$};
\filldraw  (-3,0) circle(1pt) (-2,0) circle(1pt) (-1.5,0) circle(1pt) (-0.6,0) circle(1pt);
\filldraw  (0.6,0) circle(1pt) (1.5,0) circle(1pt) (2,0) circle(1pt) (3,0) circle(1pt);
\draw[-] (0,-0.8)--(0,0.8);
\draw[very thick,dashed] (-2.3,0).. controls (-2.3,-0.2) and (-1.2,-0.2).. (-1.2,0) (-2.3,0) node[below]{$A_1$};
\draw[very thick,dashed] (-0.9,0).. controls (-0.9,-0.3).. (0,-0.3) (-0.7,-0.25) node[below]{$A_{N/2}$};
\draw[very thick,dashed] (0.9,0).. controls (0.9,-0.4) .. (0,-0.4) (0.8,-0.31) node[below]{$A_{N/2}$};
\draw[very thick,dashed] (1.2,0).. controls (1.2,-0.2) and (2.3,-0.2).. (2.3,-0) node[below]{$A_{N-1}$};
\draw[<-,very thick] (-1.72,-0.15)--(-1.73,-0.15);
\draw[<-,very thick] (1.79,-0.16)--(1.78,-0.16);
\draw[->,very thick] (-0.41,-0.3)--(-0.4,-0.3);
\draw[->,very thick] (0.4,-0.4)--(0.41,-0.4);
\draw[-] (-4.5,0.8)--(4.5,0.8)--(4.5,-0.8)--(-4.5,-0.8)--(-4.5,0.8);
\end{tikzpicture}

\begin{tikzpicture}[scale=1.5] % Third sheet
\draw[-] (0,-0.8)--(0,0.8);
\draw[-] (-4.5,0)--(-1.0,0) (1.0,0)--(4.5,0);
\filldraw (-1,0) circle(1pt) node[above]{$-c_3$} (1,0) circle(1pt) node[above]{$c_3$};
\draw[very thick,dashed] (0,-0.4).. controls (-1.5,-0.4) .. (-1.5,0) (-1.5,-0.28) node[below]{$A_{N/2}$};
\draw[very thick,dashed] (0,-0.3).. controls (1.5,-0.3) .. (1.5,0) (1.5,-0.26) node[below]{$A_{N/2}$};
\draw[->,very thick] (-0.61,-0.4)--(-0.6,-0.4);
\draw[->,very thick] (0.6,-0.3)--(0.61,-0.3);
\draw[-] (-4.5,0.8)--(4.5,0.8)--(4.5,-0.8)--(-4.5,-0.8)--(-4.5,0.8);
\end{tikzpicture}

\begin{tikzpicture}[scale=1.5] %Fourth sheet
\draw[-] (-4.5,0)--(-1.0,0) (1.0,0)--(4.5,0);
\filldraw (-1,0) circle(1pt) node[below]{$-c_3$} (1,0) circle(1pt) node[below]{$c_3$};
\draw[very thick,dashed] (-1.5,0)..controls (-1.5,0.4) and (1.5,0.4).. (1.5,0)   node[above]{$A_{N/2}$};
\draw[<-,very thick] (0,0.3)--(0.01,0.3);
\draw[-] (-4.5,0.8)--(4.5,0.8)--(4.5,-0.8)--(-4.5,-0.8)--(-4.5,0.8);
\end{tikzpicture}
\caption{Case III}
\label{fig: Case III}
\end{center}
\end{figure}

In Case III we have  $c_2 = 0$ and $c_3 > 0$.
The genus of $\mathcal{R}$ is $g=N-1$. Note that
in this case, the number of intervals in $S(\mu_1)$ is even and the
origin does not belong to $S(\mu_1)$. The canonical homology basis
is chosen as in Figure \ref{fig: Case III}. The $B$-cycles are the same as in Case II. In particular
they are only on the first sheet. The $A$-cylces are also the same as in Case II, except for the
cycle $A_{N/2}$ which crosses the imaginary axis. This cycle visits all four sheets as
indicated in Figure \ref{fig: Case III}.

\begin{figure}[t] % Case V
\begin{center}
\begin{tikzpicture}[scale=1.5]  % First sheet
\draw[-] (-3,0) node[below]{$a_1$}--(-2,0) node[below]{$b_1$};
\draw[-] (-1.5,0) node[below]{}--(-0.6,0) node[below]{};
\draw[-] (0.6,0) node[below]{}--(1.5,0) node[below]{};
\draw[-] (2.0,0) node[below]{$a_N$}--(3.0,0) node[below]{$b_N$};
\filldraw  (-3,0) circle(1pt) (-2,0) circle(1pt) (-1.5,0) circle(1pt) (-0.6,0) circle(1pt);
\filldraw  (0.6,0) circle(1pt) (1.5,0) circle(1pt) (2,0) circle(1pt) (3,0) circle(1pt);
\draw[very thick,->] (-2.5,0) ellipse (20pt and 10pt) (-1.84,0.3) node{$B_1$};
\draw[very thick] (-1.8,0) ellipse (50pt and 15pt) (-0.25,0.4) node{$B_2$};
\draw[very thick] (-1.2,0) ellipse (85pt and 20pt) (1.4,0.5) node{$B_{N-1}$};
\draw[->,very thick] (-1.1,-0.7)--(-1.11,-0.7);
\draw[->,very thick] (-2,-0.53)--(-2.01,-0.53);
\draw[->,very thick] (-2.6,-0.35)--(-2.61,-0.35);
\draw[very thick,dashed] (-2.3,0)..controls (-2.3,0.2) and (-1.2,0.2).. (-1.2,0) node[above]{$A_1$};
\draw[very thick,dashed] (-0.9,0)..controls (-0.9,0.2) and (0.9,0.2).. (0.9,0) node[above]{};
\draw[very thick,dashed] (1.2,0)..controls (1.2,0.2) and (2.3,0.2).. (2.3,-0) node[above]{$A_{N-1}$};
\draw[->,very thick] (-1.72,0.15)--(-1.73,0.15);
\draw[->,very thick] (0.01,0.16)--(0.0,0.16);
\draw[->,very thick] (1.79,0.16)--(1.78,0.16);
\draw[-] (-4.5,0.8)--(4.5,0.8)--(4.5,-0.8)--(-4.5,-0.8)--(-4.5,0.8);
\end{tikzpicture}

\begin{tikzpicture}[scale=1.5]  % Second sheet
\draw[-] (-3,0) node[above]{$a_1$}--(-2,0) node[above]{$b_1$};
\draw[-] (-1.5,0) node[above]{}--(-0.6,0) node[above]{};
\draw[-] (0.6,0) node[above]{}--(1.5,0) node[above]{};
\draw[-] (2.0,0) node[above]{$a_N$}--(3.0,0) node[above]{$b_N$};
\filldraw  (-3,0) circle(1pt) (-2,0) circle(1pt) (-1.5,0) circle(1pt) (-0.6,0) circle(1pt);
\filldraw  (0.6,0) circle(1pt) (1.5,0) circle(1pt) (2,0) circle(1pt) (3,0) circle(1pt);
\draw[-] (0,0.4)--(0,0.8) (0,-0.4)--(0,-0.8);
\filldraw (0,0.4) circle(1pt) node[left]{$ic_2$} (0,-0.4) node[left]{$-ic_2$} circle(1pt);
\draw[very thick,dashed] (-2.3,0)..controls (-2.3,-0.2) and (-1.2,-0.2).. (-1.2,-0) node[below]{$A_1$};
\draw[very thick,dashed] (-0.9,0)..controls (-0.9,-0.2) and (0.9,-0.2).. (0.9,0) node[below]{};
\draw[very thick,dashed] (1.2,0)..controls (1.2,-0.2) and (2.3,-0.2).. (2.3,-0) node[below]{$A_{N-1}$};
\draw[<-,very thick] (-1.72,-0.15)--(-1.73,-0.15);
\draw[<-,very thick] (0.01,-0.16)--(0.0,-0.16);
\draw[<-,very thick] (1.79,-0.16)--(1.78,-0.16);
\draw[very thick] (0,0.7).. controls (-4.5,0.7) and (-4.5, -0.7)..(0,-0.7)  (-3.5,0.5) node[right]{$B_N$};
\draw[->,very thick] (-3.36,0)--(-3.36,0.01);
\draw[-] (-4.5,0.8)--(4.5,0.8)--(4.5,-0.8)--(-4.5,-0.8)--(-4.5,0.8);
\end{tikzpicture}

\begin{tikzpicture}[scale=1.5] % Third sheet
\draw[-] (0,0.4)--(0,0.8) (0,-0.4)--(0,-0.8);
\filldraw (0,0.4) circle(1pt) node[left]{$ic_2$} (0,-0.4) node[left]{$-ic_2$} circle(1pt);
\draw[-] (-4.5,0)--(-1.0,0) (1.0,0)--(4.5,0);
\filldraw (-1,0) circle(1pt) node[below]{$-c_3$} (1,0) circle(1pt) node[below]{$c_3$};
\draw[very thick,dashed] (-1.5,0)..controls (-1.5,0.2) and (1.5,0.2).. (1.5,0)  node[above]{$A_N$};
\draw[->,very thick] (0,0.15)--(-0.01,0.15);
\draw[very thick] (0,0.7).. controls (1,0.7) and (1, -0.7)..(0,-0.7)  (0.5,0.5) node[right]{$B_N$};
\draw[->,very thick] (0.75,0)--(0.75,-0.01);
\draw[-] (-4.5,0.8)--(4.5,0.8)--(4.5,-0.8)--(-4.5,-0.8)--(-4.5,0.8);
\end{tikzpicture}

\begin{tikzpicture}[scale=1.5]  % Fourth sheet
\draw[-] (-4.5,0)--(-1.0,0) (1.0,0)--(4.5,0);
\filldraw (-1,0) circle(1pt) node[above]{$-c_3$} (1,0) circle(1pt) node[above]{$c_3$};
\draw[very thick,dashed] (-1.5,0)..controls (-1.5,-0.2) and (1.5,-0.2).. (1.5,0)  node[below]{$A_N$};
\draw[<-,very thick] (0,-0.15)--(-0.01,-0.15);

\draw[-] (-4.5,0.8)--(4.5,0.8)--(4.5,-0.8)--(-4.5,-0.8)--(-4.5,0.8);
\end{tikzpicture}
\caption{Case V}
\label{fig: Case V}
\end{center}
\end{figure}

The two Cases IV and V are again very similar and we only show the Figure \ref{fig: Case V}
for Case V. Here we have $c_2 > 0$ and $c_3 > 0$,
and the Riemann surface $\mathcal R$ has genus $g=N$.
The cycles $A_1, \ldots, A_{N-1}$ and $B_1, \ldots, B_{N-1}$ are as in
the previous Cases I and II.
There are two extra cycles $A_N$ and $B_N$. The cycle $A_N$ is on
the third and fourth sheets. It consists of a part in the upper half
plane of $\mathcal R_3$  from a point in  $(c_3,\infty)$ to a point
in $(-\infty,-c_3)$, together with a part in the lower half plane of
$\mathcal{R}_4$ that we choose to be the mirror image in the real
line of the part on the third sheet. The cycle $B_N$ is on the second
and third sheets. The part on $\mathcal R_2$ goes from a point in $(-i\infty, -ic_2)$
to a point in $(ic_2, \infty)$ and it goes around all intervals $[a_j, b_j]$
that are in the left-half plane. The part on $\mathcal R_3$ intersect the
real line somewhere in $(0, c_3)$.

By construction we have
\begin{equation} \label{eq:involcyc}
    \phi(B_j) = -B_j, \qquad \phi(A_j) \sim A_j, \qquad j=1, \ldots, g,
    \end{equation}
where the symbol $\sim$ means that $\phi(A_j)$ is homologous to $A_j$ in $\mathcal R$.

\subsection{Meromorphic differentials}

Let us now recall some facts about meromorphic differentials on the
Riemann surface. Most of the results that we will be using can be
found in \cite{FarKra}.

A meromorphic differential with simple poles only is called a
meromorphic differential of the third kind. A meromorphic differential
of the third kind is uniquely determined by its $A$-periods
and the location and residues at its poles, provided that
the residues add up to zero.

We pick points $P_j$, $j=1, \ldots, g$ with $P_j \in \Sigma_j$. For each such choice
we define a meromorphic differential $\omega_P$ of the third kind as follows.

We use $\infty_1$ to denote the point at infinity on the first sheet
and $\infty_2$ to denote the other point at infinity which is common to all
three other sheets.

\begin{definition} \label{def: omegaP}
Let $P_j \in \Sigma_j$ for $j=1, \ldots, g$. Then
$\omega_P$ is the meromorphic differential of third kind on $\mathcal R$
which is uniquely determined by the following conditions
\begin{enumerate}
\item[\rm (a)] The meromorphic differential has simple poles at
$a_j, b_j$, $j=1, \ldots, N$, at $\pm ic_2$ (if $c_2 > 0$),
at $\pm c_3$ (if $c_3 > 0$), at the points $P_j$, $j=1, \ldots, g$ and at
$\infty_2$. The residues at the finite branch points  are equal to $-1/2$:
\begin{equation} \label{eq:residues}
\begin{aligned}
      \Res_{z = a_j} \omega_P & = \Res_{z = b_j} \omega_P = - \tfrac{1}{2}, \qquad j =1, \ldots, N, \\
      \Res_{z = \pm ic_2} \omega_P & = -\tfrac{1}{2}, \qquad \text{(only if $c_2 > 0$)}, \\
      \Res_{z = \pm c_3} \omega_P & = -\tfrac{1}{2}, \qquad \text{(only if $c_3 > 0$)},
\end{aligned}
\end{equation}
the residue at $\infty_2$ is equal to $2$,
\begin{equation} \label{eq:residuesinfty}
      \Res_{z = \infty_2} \omega_P  = 2
\end{equation}
and the residue at the points $P_j$ is equal to $1$:
\begin{equation} \label{eq:residueP}
      \Res_{z = P_j} \omega_P  = 1, \qquad j=1, \ldots, g.
\end{equation}
\item[\rm (b)]
The meromorphic differential has vanishing $A$ periods in Cases I, II, IV, and V:
\begin{align} \label{eq:Aperiod}
    \oint_{A_k} \omega_P = 0, \quad k =1, \ldots, g, \qquad \text{(in Cases I, II, IV, and V)}.
\end{align}
In Case III all $A$-periods are vanishing, except the one of $A_{N/2}$:
\begin{align} \label{eq:AperiodIII}
       \oint_{A_k} \omega_P  & = - \pi i \delta_{k,N/2}, \quad k=1, \ldots, g, \qquad \text{(only in Case III)}.
\end{align}
\end{enumerate}
\end{definition}
A simple count shows that the residues of $\omega_P$ add up to $0$
and therefore the meromorphic differential $\omega_P$ is indeed uniquely
defined by the pole conditions and the $A$-period conditions.

If one or more of the $P_j$'s coincide with a branch point,
then the residue conditions \eqref{eq:residues} have to be modified
appropriately. For example, if $P_j = a_{j}$ then
\[ \Res_{P_j} \omega_P = \tfrac{1}{2}. \]
In this way, the meromorphic differential $\omega_P$
depends continuously on the $P_j$'s and is well defined for each choice
of $P_j \in \Gamma_j$, $j =1, \ldots, g$.

The anti-holomorphic involution $\phi$ is used to map a meromorphic
differential $\omega$
to a meromorphic differential $\phi^\#(\omega)$ in an obvious way.
If $\omega$ is equal to $f_j(z) dz$ for a meromorphic function $f_j$
on sheet $j$, then $\phi^\#(\omega)$ is equal to
\[ \overline{f_j(\bar{z})} dz \]
on sheet $j$.
A crucial property is that $\omega_P$ is invariant under $\phi^\#$.

\begin{lemma} Let $P_j \in \Sigma_j$, $j=1, \ldots, g$. Then
\[ \omega_P = \phi^{\#}(\omega_P). \]
\end{lemma}

\begin{proof}
Since all the poles of $\omega_P$ are invariant under the involution $\phi$,
the meromorphic differential $\phi^\#(\omega_P)$ has the
same poles and residues as $\omega_P$. We have to show that their $A$-periods are
the same. We have
\begin{equation}\label{eq:deltaperi}
    \oint_{A_k} \phi^{\#}(\omega_P) = \overline{ \oint_{\phi(A_k)} \omega_P }.
\end{equation}

We have that $\phi(A_k)$ is homologous to $A_k$ in $\mathcal R$, but in the
process of deforming $\phi(A_k)$ to $A_k$ we pick up residue contributions
from the poles of $\omega_P$.

For the cycles $A_k$ that are only on two sheets (which is the
typical situation) we pick up a residue contribution from the two
endpoints of a gap in the support of either $\mu_1$ or $\mu_3$ and
from $P_k$. As the deformation from $\phi(A_k)$ to $A_k$ will result
in clockwise loops around these points, the total residue is
$\tfrac{1}{2} + \tfrac{1}{2} - 1 = 0$. It follows that
\[ \oint_{A_k} \phi^{\#}(\omega_P) = \overline{ \oint_{A_k} \omega_P} = 0, \]
since the $A$-period of $\omega_P$ is zero for $A$-cycles that are only on
two sheets.

In Case III there is a cycle $A_{N/2}$ which is on all four sheets.
If we deform $\phi(A_{N/2})$ into $A_{N/2}$ we pick up residue
contributions from the endpoints $b_{N/2}, a_{N/2+1}, -c_3, c_3$ and
from $P_{N/2}$. Then we have residue $\tfrac{1}{2}$ four times and
$-1$ once, so that the total residue contribution is $1$ (As the
residue contribution comes from clockwise loops around these
points). It follows that
\[ \oint_{A_{N/2}} \phi^{\#}(\omega_P) = \overline{ \oint_{A_{N/2}} \omega_P + 2 \pi i} =
    \overline{-\pi i + 2 \pi i} = - \pi i, \]
since the $A_{N/2}$ period is $-\pi i$ by definition \eqref{eq:AperiodIII}.

So the $A$-periods of $\phi^{\#}(\omega_P)$ and $\omega_P$ agree, and
the lemma follows.
\end{proof}

\begin{proposition} \label{prop:mapPsi}
The $B$-periods of $\omega_P$ are purely imaginary
and the map
\begin{multline} \label{eq:mapPsi}
    \Psi : \Sigma_1 \times \cdots \times \Sigma_g \to (\mathbb R \slash \mathbb Z )^{g} : \\
        (P_1, \ldots, P_g) \mapsto
        \frac{1}{2\pi i} \left( \int_{B_1} \omega_P, \ldots, \int_{B_g} \omega_P \right)
        \end{multline}
is a well-defined, continuous bijection.
\end{proposition}
\begin{proof} This follows as in \cite{KuiMo}. The fact that the map is well-defined and continuous is proved as in
\cite[Proposition 2.3]{KuiMo}. Due to the fact that the divisors
$\sum\limits_{j=1}^g P_j$ with $P_j \in \Sigma_j$ are non-special, see Lemma \ref{lem:nonspecial},
the argument in the proof of \cite[Theorem 2.6]{KuiMo} gives first the injectivity of $\Psi$. Then
the invariance of domain argument of the same proof yields the surjectivity of~$\Psi$.
\end{proof}

\subsection{Definition and properties of functions $u_j$}

Due to Proposition \ref{prop:mapPsi} there exists a choice of points $P_j$, $j=1, \ldots, g$
such that the corresponding meromorphic differential $\omega_P$ satisfies
\begin{equation}\label{eq:Bperiods}
\begin{aligned}
    \oint_{B_k}\omega_P & = -2 n \pi i \alpha_k, \qquad k=1,\ldots, N-1, \\
    \oint_{B_N} \omega_P & =- n \pi i, \qquad \text{(only in Cases IV and V),}
\end{aligned}
\end{equation}
where the equalities hold modulo $2 \pi i \mathbb Z$. Note that $\omega_P$
is varying with $n$. We consider $n$ as fixed and work with $\omega_P$
satisfying \eqref{eq:Bperiods} throughout this section. Of course, $\omega_P$
also satisfies the conditions given in Definition \ref{def: omegaP}.

We are going to integrate $\omega_P$ along paths on $\mathcal R$ that
start from $\infty_1$ (the point at infinity on the first sheet) and
that on each sheet remain in the same quadrant. So the paths do not cross the
contours $\Sigma_j$, $j=0, \ldots, g$ and also do not intersect
the imaginary axis except along the cut $S(\sigma_2 - \mu_2)$ that
connects the second and third sheets. There may be a choice
in the cut that one takes when passing from the first to the
second sheet. However, this will lead to the same value for
the integral
because of the vanishing of the integral of $\omega_P$
over the cycles $A_k$ and $\phi(A_k)$, see \eqref{eq:Aperiod}.
Note that the exceptional $A$-period \eqref{eq:AperiodIII} in Case III
does not play a role here.

With this convention for the paths we define functions $u_j$ as follows.

\begin{definition}
For $j=1,2,3,4$ and $z \in \mathbb C \setminus (\mathbb R \cup i \mathbb R)$ we define
$u_j(z)$ as the Abelian integral
\[ u_j(z)=\int_{\infty_1}^{\pi_j^{-1}(z)} \omega_P. \]
\end{definition}

Abusing the notation, we also write
\[ u_j(z) = \int_{\infty_1}^z \omega_P \]
where $z$ is considered as a point on the $j$th sheet. The path from $\infty_1$ to $z \in \mathbb R_j$
follows the convention described above, namely that on each sheet it stays in the same
quadrant.

Then the functions $u_j$ are defined and analytic on $\mathbb C \setminus (\mathbb R \cup i \mathbb R)$
with the following jump properties.
In what follows we write $\equiv$ to denote equality up to integer multiples of $2\pi i$.

\begin{lemma} \label{lem:jumpuj}
\begin{enumerate}
\item[\rm (a)] For $x \in (a_k,b_k)$ with $k = 1, \ldots, N$, we have
\begin{align}
    u_{1,\pm}(x) = u_{2,\mp}(x).
    \end{align}
\item[\rm (b)] For $x \in (b_k, a_{k+1})$ with $k=0, \ldots, N$, we have
\begin{align}
    u_{1,+}(x) & \equiv u_{1,-}(x) - 2 n \pi i \alpha_k, \\ \nonumber
        &  \qquad \text{unless $P_k$ is on the first sheet and $x = \pi(P_k)$} \\
    u_{2,+}(x) & \equiv u_{2,-}(x) + 2 n \pi i \alpha_k + \pi i, \\ \nonumber
        & \qquad  \text{unless $P_k$ is on the second sheet and $x = \pi(P_k)$}
    \end{align}
where $b_0 = - \infty$, $a_{N+1} = + \infty$, $\alpha_0 = 1$ and $\alpha_N = 0$.
\item[\rm (c)] For $x \in (-c_3, c_3)$ we have
\begin{align}
    u_{3,+}(x) & \equiv u_{3,-}(x) - (n+1) \pi i, \\ \nonumber
        & \qquad \text{unless $P_{k}$ is on the third sheet and $x = \pi(P_{k})$} \\
    u_{4,+}(x) & \equiv u_{4,-}(x) + n \pi i, \\ \nonumber
        & \qquad \text{unless $P_{k}$ is on the fourth sheet and $x = \pi(P_{k})$}
    \end{align}
where $k = N/2$ in Case III and $k = N$ in Cases IV and V.
\item[\rm (d)] For $x \in (-\infty, -c_3) \cup (c_3,\infty)$ we have
\begin{align} u_{3,\pm}(x) = u_{4,\mp}(x).
\end{align}
\item[\rm (e)] On the imaginary axis we have
\begin{align}
    \label{eq:uoniR1}
    u_{1,+}(x) & = u_{1,-}(x) && \text{ for } x \in  i \mathbb R, \\
    \label{eq:uoniR2}
    u_{2,+}(x) & = u_{2,-}(x) && \text{ for } x\in (-ic_2, ic_2), \\
    \label{eq:uoniR3}
    u_{2,\pm}(x) & = u_{3,\mp}(x) && \text{ for } x \in  (-i\infty, -ic_2] \cup [ic_2, i\infty), \\
    \label{eq:uoniR4}
    u_{3,+}(x) & \equiv u_{3,-}(x) + \pi i && \text{ for } x\in (-ic_2, ic_2), \\
    \label{eq:uoniR5}
    u_{4,+}(x) & \equiv u_{4,-}(x) + \pi i && \text{ for } x\in i \mathbb R.
    \end{align}
\end{enumerate}
\end{lemma}

\begin{proof}
This is a straightforward but rather tedious verification. All identities
or equivalences come down to the calculation of a period of a closed loop on $\mathcal R$
for the meromorphic differential $\omega_P$.

For example, to prove \eqref{eq:uoniR5} for $x \in i \mathbb R^+$, we note that
\[ u_{4,+}(x) - u_{4,-}(x) = \int_C \omega_P \]
where $C$ is the cycle on $\mathcal R$ that is shown in Figure
\ref{fig: cycleC} for the Cases IV and V. This cycle can be deformed
to a sum of $-\phi(A_{N})$, the cycles $-A_j$, $j=1, \ldots, N-1$
and a closed loop around $-ic_2$. This indeed leads to
\eqref{eq:uoniR5} since $\omega_P$ has vanishing $A$-periods and
$\phi(A)$ periods, and the only contribution comes from $-ic_2$
which gives us $\pm \pi i$.

In the Case III (not shown in the figure), the cycle $C$ can be deformed into
the cycles, $-A_j$, $j=1, \ldots N-1$, that include the exceptional cycle $-A_{N/2}$.
It is now because of \eqref{eq:AperiodIII} that we obtain \eqref{eq:uoniR5}
for $x \in i \mathbb R^+$.

The other relations follows in a similar way.
\end{proof}

\begin{figure}[t] % Cycle C
\begin{center}
\begin{tikzpicture}[scale=1.5] % First sheet
\draw[-] (-3,0) node[below]{$a_1$}--(-2,0) node[below]{$b_1$};
\draw[-] (-1.5,0) node[below]{}--(-1.0,0) node[below]{};
\draw[-] (-0.5,0)--(0.5,0);
\draw[-] (1,0) node[below]{}--(1.5,0) node[below]{};
\draw[-] (2.0,0) node[below]{$a_N$}--(3.0,0) node[below]{$b_N$};
\filldraw  (-3,0) circle(1pt) (-2,0) circle(1pt) (-1.5,0) circle(1pt) (-1,0) circle(1pt) (-0.5,0) circle(1pt);
\filldraw  (0.5,0) circle(1pt) (1,0) circle(1pt) (1.5,0) circle(1pt) (2,0) circle(1pt) (3,0) circle(1pt);
%
% B-cyles
%\draw[very thick,->] (-2.5,0) ellipse (20pt and 8pt) (-1.84,0.26) node{$B_1$};
%\draw[very thick] (-2.0,0) ellipse (36pt and 12pt) (-1.1,0.4) node{};
%\draw[very thick] (-1.4,0) ellipse (60pt and 16pt) (1.4,0.6) node{};
%\draw[very thick] (-1.0,0) ellipse (80pt and 20pt) (1.4,0.6) node{$B_{N-1}$};
%\draw[->,very thick] (-2.5,-0.28)--(-2.51,-0.28);
%\draw[->,very thick] (-2.0,-0.41)--(-2.01,-0.41);
%\draw[->,very thick] (-1.4,-0.56)--(-1.41,-0.56);
%\draw[->,very thick] (-1.0,-0.7)--(-1.01,-0.7);
%
% A-cyles on first sheet
%\draw[very thick,dashed] (-2.3,0)..controls (-2.3,0.2) and (-1.3,0.2).. (-1.3,0) node[above]{$A_1$};
%\draw[very thick,dashed] (-1.2,0)..controls (-1.2,0.2) and (-0.3,0.2).. (-0.3,0) node[above]{};
%\draw[very thick,dashed] (0.3,0)..controls (0.3,0.2) and (1.2,0.2).. (1.2,0) node[above]{};
%\draw[very thick,dashed] (1.3,0)..controls (1.3,0.2) and (2.3,0.2).. (2.3,0) node[above]{$A_{N-1}$};
%\draw[->,very thick] (-1.82,0.15)--(-1.83,0.15);
%\draw[->,very thick] (-0.7,0.16)--(-0.71,0.16);
%\draw[->,very thick] (0.7,0.16)--(0.69,0.16);
%\draw[->,very thick] (1.86,0.16)--(1.85,0.16);
%
\draw[very thick] (-2.3,0) .. controls (-2.3,0.5) and (-4.5,0.5) .. (-4.5,0.3) (-3,0.3) node[above]{$C$};
\draw[->,very thick] (-3.4,0.4)--(-3.41,0.4);
\draw[very thick] (2.3,0) .. controls (2.3,0.5) and (4.5,0.5) .. (4.5,0.3) (3,0.3) node[above]{$C$};
\draw[->,very thick] (3.41,0.4)--(3.4,0.4);
\draw[-] (-4.5,0.8)--(4.5,0.8)--(4.5,-0.8)--(-4.5,-0.8)--(-4.5,0.8);
\end{tikzpicture}

\begin{tikzpicture}[scale=1.5] % Second sheet
\draw[-] (-3,0) node[above]{$a_1$}--(-2,0) node[above]{$b_1$};
\draw[-] (-1.5,0) node[above]{}--(-1.0,0) node[above]{};
\draw[-] (-0.5,0)--(0.5,0);
\draw[-] (1,0) node[above]{}--(1.5,0) node[above]{};
\draw[-] (2.0,0) node[above]{$a_N$}--(3.0,0) node[above]{$b_N$};
\filldraw  (-3,0) circle(1pt) (-2,0) circle(1pt) (-1.5,0) circle(1pt) (-1,0) circle(1pt) (-0.5,0) circle(1pt);
\filldraw  (0.5,0) circle(1pt) (1,0) circle(1pt) (1.5,0) circle(1pt) (2,0) circle(1pt) (3,0) circle(1pt);
\draw[-] (0,0.4)--(0,0.8) (0,-0.4)--(0,-0.8);
\filldraw (0,0.4) circle(1pt) node[left]{$ic_2$} (0,-0.4) node[left]{$-ic_2$} circle(1pt);
%
%\draw[very thick,dashed] (-2.3,0)..controls (-2.3,-0.2) and (-1.3,-0.2).. (-1.3,0) (-2.3,0) node[below]{$A_1$};
%\draw[very thick,dashed] (-1.2,0)..controls (-1.2,-0.2) and (-0.3,-0.2).. (-0.3,0) node[below]{};
%\draw[very thick,dashed] (0.3,0)..controls (0.3,-0.2) and (1.2,-0.2).. (1.2,0) node[below]{};
%\draw[very thick,dashed] (1.3,0)..controls (1.3,-0.2) and (2.3,-0.2).. (2.3,0) node[below]{$A_{N-1}$};
%\draw[<-,very thick] (-1.82,-0.15)--(-1.83,-0.15);
%\draw[<-,very thick] (-0.7,-0.16)--(-0.71,-0.16);
%\draw[<-,very thick] (0.7,-0.16)--(0.69,-0.16);
%\draw[<-,very thick] (1.86,-0.16)--(1.85,-0.16);
%
%\draw[very thick] (0,0.7).. controls (-4.5,0.7) and (-4.5, -0.7)..(0,-0.7)  (-3.5,0.5) node[right]{$B_N$};
%\draw[->,very thick] (-3.36,0)--(-3.36,0.01);
%
\draw[very thick] (-2.3,0) .. controls (-2.3,-0.7)  .. (0,-0.7) (-2.5,-0.3) node[below]{$C$};
\draw[->,very thick] (-1,-0.7)--(-1.01,-0.7);
\draw[very thick] (2.3,0) .. controls (2.3,-0.6).. (0,-0.6) (2.5,-0.3) node[below]{$C$};
\draw[->,very thick] (1.01,-0.6)--(1,-0.6);
\draw[-] (-4.5,0.8)--(4.5,0.8)--(4.5,-0.8)--(-4.5,-0.8)--(-4.5,0.8);
\end{tikzpicture}

\begin{tikzpicture}[scale=1.5]  % Third sheet
\draw[-] (0,0.4)--(0,0.8) (0,-0.4)--(0,-0.8);
\filldraw (0,0.4) circle(1pt) node[left]{$ic_2$} (0,-0.4) node[left]{$-ic_2$} circle(1pt);
\draw[-] (-4.5,0)--(-1.0,0) (1.0,0)--(4.5,0);
\filldraw (-1,0) circle(1pt) node[below]{$-c_3$} (1,0) circle(1pt) node[below]{$c_3$};
%
%\draw[very thick,dashed] (-1.5,0)..controls (-1.5,0.2) and (1.5,0.2).. (1.5,0)  node[above]{$A_N$};
%\draw[->,very thick] (0,0.15)--(-0.01,0.15);
%
%\draw[very thick] (0,0.7).. controls (1,0.7) and (1, -0.7)..(0,-0.7)  (0.5,0.5) node[right]{$B_N$};
%\draw[->,very thick] (0.75,0)--(0.75,-0.01);
%
\draw[very thick] (-1.8,0) .. controls (-1.8,-0.6)  .. (0,-0.6) (-2.0,-0.3) node[below]{$C$};
\draw[->,very thick] (-1,-0.6)--(-1.01,-0.6);
\draw[very thick] (1.8,0) .. controls (1.8,-0.7).. (0,-0.7) (2.0,-0.3) node[below]{$C$};
\draw[->,very thick] (1.01,-0.7)--(1,-0.7);
\draw[-] (-4.5,0.8)--(4.5,0.8)--(4.5,-0.8)--(-4.5,-0.8)--(-4.5,0.8);
\end{tikzpicture}

\begin{tikzpicture}[scale=1.5]  % Fourth sheet
\draw[-] (-4.5,0)--(-1.0,0) (1.0,0)--(4.5,0);
\filldraw (-1,0) circle(1pt) node[above]{$-c_3$} (1,0) circle(1pt) node[above]{$c_3$};
%
%\draw[very thick,dashed] (-1.5,0)..controls (-1.5,-0.2) and (1.5,-0.2).. (1.5,0)  node[below]{$A_N$};
%\draw[<-,very thick] (0,-0.15)--(-0.01,-0.15);
%
\draw[very thick] (-1.8,0) .. controls (-1.8,0.6)  .. (0,0.6) (-2.0,0.3) node{$C$};
\draw[->,very thick] (-1.01,0.6)--(-1.0,0.6);
\draw[very thick] (1.8,0) .. controls (1.8,0.6).. (0,0.6) (2.0,0.3) node{$C$};
\draw[->,very thick] (1.0,0.6)--(1.01,0.6);
\draw[dashed] (0,-0.8)--(0,0.8);
\draw[-] (-4.5,0.8)--(4.5,0.8)--(4.5,-0.8)--(-4.5,-0.8)--(-4.5,0.8);
\end{tikzpicture}
\end{center}
\caption{Cycle $C$}
\label{fig: cycleC}
\end{figure}

We state without proof the behavior of the functions $u_j$ near the branch points.
They follow from residue conditions \eqref{eq:residues}.
\begin{lemma} \label{lem:uatbranch}
\begin{enumerate}
\item[\rm (a)] For $j=1,2$ and $k=1, \ldots, N$, we have
\[ u_j(z) = - \tfrac{1}{4} \log(z-a_k) + \OO(1), \qquad \text{as } z \to a_k, \]
\[ u_j(z) = - \tfrac{1}{4} \log(z-b_k) + \OO(1), \qquad \text{as } z \to b_k. \]
\item[\rm (b)] For $j=2,3$ and $c_2 > 0$,  we have
\[ u_j(z) = - \tfrac{1}{4} \log(z \mp ic_2) + \OO(1), \qquad \text{as } z \to \pm ic_2. \]
\item[\rm (c)] For $j=3,4$ and $c_3 > 0$, we have
\[ u_j(z) = - \tfrac{1}{4} \log(z \mp c_3) + \OO(1), \qquad \text{as } z \to \pm c_3. \]
\item[\rm (d)] If $P_k$ is on the $j$th sheet of the Riemann surface, then
\[ u_j(z) = \log(z- \pi(P_k)) + \OO(1), \qquad \text{as } z \to \pi(P_k), \]
where $k =1, \ldots, g$.
\item[\rm (e)] As $z \to \infty$,
\[ u_1(z) = \OO(1/z) \qquad \text{ and } \qquad
    u_j(z) = - \tfrac{2}{3} \log z  + \OO(1), \quad \text{ for } j=2,3,4. \]
\end{enumerate}
\end{lemma}

If $P_k$ coincides with one of the branch points then the behavior should be modified.
This should be obvious and we do not give details here.

\subsection{Definition and properties of functions $v_j$}

The functions $v_j$ are the exponentials of the functions $u_j$
which we again consider as functions on $\mathbb C \setminus (\mathbb R \cup i \mathbb R)$.

\begin{definition} \label{def:vj}
For $j=1,2,3,4$ and $z \in \mathbb C \setminus (\mathbb R \cup i \mathbb R)$ we
define
\[ v_j(z) = {e}^{u_j(z)}  \]
\end{definition}

The jump properties of $v_j$ follow from Lemma \ref{lem:jumpuj}.
We state them in a vector form.
\begin{corollary} \label{lem:jumpvj}
    The vector-valued function $\begin{pmatrix} v_1 & v_2 & v_3 & v_4 \end{pmatrix}$
    is analytic in $\mathbb C \setminus (\mathbb R \cup i \mathbb R)$ with
    jump property
    \[ \begin{pmatrix} v_1 & v_2 & v_3 & v_4 \end{pmatrix}_+
        = \begin{pmatrix} v_1 & v_2 & v_3 & v_4 \end{pmatrix}_- J_v
            \qquad \text{on } \mathbb R \cup i \mathbb R \]
        with a jump matrix $J_v$ that takes the following form.
    \begin{enumerate}
    \item[\rm (a)]
    On the real axis the jump matrix has the block form
    \[ J_v = \begin{pmatrix} (J_v)_1 & 0 \\ 0 & (J_v)_3 \end{pmatrix}
        \qquad \text{on } \mathbb R \]
    with $2 \times 2$ blocks
    \begin{equation}
    \begin{aligned}
        (J_v)_1 & = \begin{pmatrix} 0 & 1 \\ 1 & 0 \end{pmatrix}
            && \text{ on } S(\mu_1), \\
        (J_v)_1 & = \begin{pmatrix} {e}^{-2n \pi i \alpha_k} & 0 \\ 0 & - {e}^{2n \pi i \alpha_k} \end{pmatrix}
            && \text{ on } (b_k, a_{k+1}),
            \end{aligned}
            \end{equation}
    for $k=0, \ldots, N$,   and
    \begin{equation}
    \begin{aligned}
        (J_v)_3  & = \begin{pmatrix} 0 & 1 \\ 1 & 0 \end{pmatrix}
            && \text{ on } S(\mu_3), \\
        (J_v)_3 & = \begin{pmatrix} (-1)^{n+1} & 0 \\ 0 & (-1)^n \end{pmatrix}
          && \text{ on } (-c_3, c_3).
          \end{aligned}
          \end{equation}
          \item[\rm (b)]    On the imaginary axis the jump matrix has the block form
    \[ J_v = \begin{pmatrix} 1 & 0 & 0 \\ 0 & (J_v)_2 & 0 \\ 0 & 0 & -1 \end{pmatrix}
        \qquad \text{on } i \mathbb R, \]
     with the $2 \times 2$ block
     \begin{equation}
    \begin{aligned}
        (J_v)_2  & = \begin{pmatrix} 0 & 1 \\ 1 & 0 \end{pmatrix}
            && \text{ on } S(\sigma_2-\mu_2), \\
        (J_v)_2 & = \begin{pmatrix} 1 & 0 \\ 0 & -1 \end{pmatrix}
          && \text{ on } (-ic_2, ic_2).
          \end{aligned}
          \end{equation}
          \end{enumerate}
\end{corollary}
Note that each of the blocks $(J_v)_j$, $j=1,2,3$ has determinant $-1$.

From Lemma \ref{lem:uatbranch} and Definition \ref{def:vj} we obtain the
behavior of the $v_j$ functions near the branch points and other special points.
\begin{corollary} \label{cor:vatbranch}
\begin{enumerate}
\item[\rm (a)] For $j=1,2$ and $k=1, \ldots, N$, we have
\[ v_j(z) = \OO \left((z-a_k)^{-1/4}\right), \qquad \text{as } z \to a_k, \]
\[ u_j(z) = \OO \left((z-b_k)^{-1/4}\right), \qquad \text{as } z \to b_k. \]
\item[\rm (b)] For $j=2,3$ and $c_2 > 0$,  we have
\[ u_j(z) = \OO \left((z\mp ic_2)^{-1/4}\right), \qquad \text{as } z \to \pm ic_2. \]
\item[\rm (c)] For $j=3,4$ and $c_3 > 0$, we have
\[ u_j(z) = \OO \left((z\mp c_3)^{-1/4}\right), \qquad \text{as } z \to \pm c_3. \]
\item[\rm (d)] If $P_k$ is on the $j$th sheet of the Riemann surface, then
\[ v_j(z) = \OO(z- \pi(P_k)), \qquad \text{as } z \to \pi(P_k), \]
where $k =1, \ldots, g$.
\item[\rm (e)] As $z \to \infty$,
\[ v_1(z) = 1 +  \OO(1/z) \qquad \text{ and } \qquad v_j(z) = \OO(z^{-2/3}), \quad \text{ for } j=2,3,4. \]
\end{enumerate}
\end{corollary}

\subsection{The first row of $M$}

We now define the entries in the first row of $M$.

\begin{definition} \label{def:Mfirstrow}
We define
\begin{align}
    M_{11}(z) & = v_1(z), \\
    M_{12}(z) & =
        \begin{cases} v_2(z) & \text{for } \Im z > 0, \\
                                -v_2(z)  & \text{for } \Im z < 0,
                                \end{cases} \\
    M_{13}(z) & =
        \begin{cases} v_3(z) & \text{for $z$ in the first and third quadrants}, \\
                                -v_3(z)  & \text{for $z$ in the second and fourth quadrants},
                                \end{cases} \\
    M_{14}(z) & =
        \begin{cases} -v_4(z) & \text{for } \Re z > 0, \\
                                v_4(z)  & \text{for } \Re z < 0.
                                \end{cases}
\end{align}
\end{definition}

It turns out that with this distribution of $\pm$ signs the row
vector $(M_{1j})$ has exactly the correct jump properties that are required
in the RH problem for $M$.

\begin{proposition}
The row vector $( M_{1j})_{j=1, \ldots 4}$ satisfies the conditions
that are necessary for the first row of the solution of the RH problem for $M$.

That is, the entries $M_{1j}$ are analytic in $\mathbb C \setminus (\mathbb R \cup i \mathbb R)$
with jump property
\begin{equation} \begin{pmatrix} M_{11} & M_{12} & M_{13} & M_{14} \end{pmatrix}_+
    = \begin{pmatrix} M_{11} & M_{12} & M_{13} & M_{14} \end{pmatrix}_- J_M
\end{equation}
on $\mathbb R \cup i \mathbb R$, where $J_M$ is the jump matrix in the
RH problem for $M$, see \eqref{eq:RHforM},
and also
\begin{multline}
    \begin{pmatrix} M_{11}(z) & M_{12}(z) & M_{13}(z) & M_{14}(z) \end{pmatrix} \\
    = \begin{pmatrix} 1 + \OO(z^{-1}) & \OO(z^{-2/3}) & \OO(z^{-2/3}) & \OO(z^{-2/3}) \end{pmatrix}
\end{multline}
as $z \to \infty$.
\end{proposition}
\begin{proof}
This follows directly from Lemma \ref{lem:jumpvj}, Corollary \ref{cor:vatbranch} (e), and Definition \ref{def:Mfirstrow}.
\end{proof}

\subsection{The other rows of $M$}

We will now construct the other rows of $M$ out of the first row.

\begin{lemma} \label{lem:dime}
The vector space of meromorphic functions on $\mathcal R$ (including the constant
functions)  whose divisor is greater than or equal to
\[ -\sum_{j=1}^g P_j-3\infty_2 \]
 is of dimension $4$.
\end{lemma}
\begin{proof}
Let us denote, for a positive divisor $D^{\prime}$,
the space of meromorphic functions on $\mathcal{R}$ whose divisor is
greater than or equal to $-D^{\prime}$ by $L(D^{\prime})$.

Let $D = \sum_{j=1}^g P_j$. Any $F\in L(D+ 3\infty_2)$
has a pole at $\infty_2$ of order at most $3$. Therefore, if $w$ is a local
coordinate near $\infty_2$, we have
\[ F(w) = f_3 w^{-3} + f_2 w^{-2} + f_1 w^{-1} + \OO(1) \qquad \text{ as } w \to 0, \]
for certain numbers $f_1, f_2, f_3$.
The kernel of the linear map
\[ L(D + 3 \infty_2) \to \mathbb C^3 : F \mapsto (f_1,f_2,f_3) \]
is equal to  the space $L(D)$. Since $D$ is non-special, the space
$L(D)$ contains only constant functions and therefore is of dimension $1$.
Therefore by the dimension theorem for linear maps
\begin{equation*}
\mathrm{dim}L(D +3\infty_2)\leq
\mathrm{dim}L(D)+3=4.
\end{equation*}
However, by the Riemann-Roch theorem (see e.g. \cite{FarKra}), we also
have
\begin{equation*}
    L(D+3\infty_2)\geq 4.
\end{equation*}
This proves the lemma.
\end{proof}

We continue to use
\[ D = \sum_{j=1}^g P_j. \]

Let $F_j$ be a function in $L(D + 3 \infty_2)$. We use $F_{jk}$ to denote the
restriction of $F_j$ to the $k$th sheet and we consider the row vector
\begin{equation} \label{eq:modifiedrow}
    \begin{pmatrix} M_{11} F_{j1} & M_{12} F_{j2} & M_{13} F_{j3} & M_{14} F_{j4} \end{pmatrix}
    \end{equation}
Note that the possible poles of $F_{jk}$ at the points $\pi(P_l)$ are cancelled
by the zeros of $M_{1k}$ at the same point. Therefore the row vector
remains bounded at these points and it has the same behavior near
the branch points as the first row.

It is easy to check that the row vector has a jump on $\mathbb R \cup i \mathbb R$
with jump matrix $J_M$. In order to construct the other rows we therefore
aim to find three functions $F_2$, $F_3$ and $F_4$ in  $L(D + 3 \infty_2)$
such that the corresponding row vectors also satisfy the asymptotic condition
for the respective rows $2$, $3$ and $4$ in the RH problem for $M$.
This will be done in the next proposition.

\begin{proposition}
The RH problem for $M$ has a unique solution that is constructed in
the form
\[ M = \begin{pmatrix} M_{11} & M_{12} & M_{13} & M_{14} \\
    M_{11} F_{21} & M_{12} F_{22} & M_{13} F_{23} & M_{14} F_{24} \\
    M_{11} F_{31} & M_{12} F_{32} & M_{13} F_{33} & M_{14} F_{34} \\
    M_{11} F_{41} & M_{12} F_{42} & M_{13} F_{43} & M_{14} F_{44}
    \end{pmatrix} \]
\end{proposition}
\begin{proof}

Let us do this for the second row. The asymptotic condition we get from
\eqref{eq:RHforM} for the second row  gives us as $z \to \infty$ in the first quadrant
\begin{align} \label{eq:condF21}
     M_{11}(z) F_{21}(z) & = \OO(z^{-1}) \\ \label{eq:condF22}
     M_{12}(z) F_{22}(z) & =  \frac{i}{\sqrt{3}}   z^{1/3} + \OO(z^{-2/3})  \\ \label{eq:condF23}
     M_{13}(z) F_{23}(z) & =  - \frac{i}{\sqrt{3}} \omega   z^{1/3} + \OO(z^{-2/3})  \\ \label{eq:condF24}
   M_{14}(z) F_{24}(z) & =  - \frac{i}{\sqrt{3}} \omega^2  z^{1/3} + \OO(z^{-2/3})
   \end{align}

The first condition \eqref{eq:condF21} is satisfied if and only if $F_2(\infty_1) = 0$.
The other conditions determine the behavior of $F_2$ near $\infty_2$.

There is an expansion
\[ F_2(w) = f_3 w^{-3} + f_2 w^{-2} + f_1 w^{-1} + \OO(1) \]
where $w$ is the local coordinate near $\infty_2$ which we
choose to be equal to $z^{-1/3}$ in
the first quadrant of the second sheet. Its behavior in other quadrants
and on other sheets is determined by analytic continuation.
Thus
\begin{equation} \label{eq:F22expansion}
    F_{22}(z) = f_3 z + f_2 z^{2/3} + f_1 z^{1/3} + \OO(1)
    \end{equation}
as $z \to \infty$ in the first quadrant.
We also have
\begin{equation} \label{eq:M12expansion}
    M_{12}(z) = v_2(z) = m_3 z^{-2/3} + m_2 z^{-1} + m_1 z^{-4/3} + \OO(z^{-5/3})
    \end{equation}
as $z \to \infty$ in the first quadrant with a nonzero first coefficient
\[ m_3 \neq 0. \]
Then inserting \eqref{eq:F22expansion} and \eqref{eq:M12expansion}
into the condition \eqref{eq:condF22} we obtain a linear system of
equations for the unknowns $f_3, f_2, f_1$ that has a unique solution.
These three conditions together with the fact that $F_2(\infty_1) = 0$
determines $F_2$ uniquely.

Now it requires an independent check that the conditions \eqref{eq:condF23}
and \eqref{eq:condF24} are satisfied as well with the same function $F_2$, and also the analogous
conditions that come from the asymptotic condition in the other quadrants.
This then completes the construction of the second row of $M$.

The functions $F_3$ and $F_4$ are found in a similar way and they are
used to construct the remaining rows of $M$.\end{proof}

From the construction it follows that the entries of $M$ and $M^{-1}$
are uniformly bounded in $n$ away from the branch points and infinity. More precisely, we have the following proposition.

\begin{proposition} \label{prop:columnsofM} We have that
\begin{enumerate}
\item[{\rm (a)}] the first columns of $M$ and $M^{-t}$ are bounded  away from the points $\{a_1,b_1,\ldots,a_N,b_N\}$, with a bound that is uniform in $n$.
\item[{\rm (b)}] the second columns of $M$ and $M^{-t}$ are bounded  away from the points $\{a_1,b_1,\ldots,a_N,b_N\}$, $\pm i c_2$ and $\infty$, with a bound that is uniform in $n$.
\item[{\rm (c)}] the third columns of $M$ and $M^{-t}$ are bounded  away from the points $\pm ic_2,\pm c_3$ and $\infty$, with a bound that is uniform in $n$.
\item[{\rm (d)}] the first columns of $M$ and $M^{-t}$ are bounded  away from the points $\pm c_3$ and $\infty$,  with a bound that is uniform in $n$.
\end{enumerate}
\end{proposition}
\begin{proof}
We will only proof part (a) as the others follow by similar arguments.

Let us start with the first column of $M$. From the structure of the
RH problem we see that the entries of $M$ in the first column  are
analytic in $\C\setminus [a_1,b_n]$, bounded near $\infty$ and have
an analytic continuation across each interval $(a_k,b_k)$ and
$(b_k,a_{k+1})$. Hence the entries are bounded if we stay away from
$\{a_1,b_1,\ldots,a_N,b_N\}$.  Moreover, $M$ depends continuously on
the parameters $n\alpha_k$ modulo integers. By compactness of the
parameter space, it then follows that we can choose the bound for
the entries such that they hold uniformly in $n$.

As for the first column of $M^{-t}$, we note that $M^{-t}$ satisfies
a RH problem that has the same structure as the RH problem for $M$.
Then the statement follows from the same arguments. Alternatively,
the statement for $M^{-t}$ follows from the identity
\begin{align} \label{eq:MinverseM}
M^{-t} =\begin{pmatrix} 1 & 0 & 0 & 0\\ 0 & 0 & 0 &-1\\
0 & 0 & -1 & 0 \\
0 & -1 & 0 & 0
\end{pmatrix}\widetilde{M},
\end{align}
where $\widetilde M$ is the solution of the RH problem \eqref{eq:RHforM} but with parameters $-\alpha_k$ instead of $\alpha_k$ for $k=1,\ldots,N$. The identity \eqref{eq:MinverseM}  follows from the fact that both sides solve the same, uniquely solvable, RH problem.
\end{proof}
The following corollary will be used in the next section.
\begin{corollary}\label{cor:MinvyMx}
Let $K\subset S(\mu_1)\setminus \{a_1,b_1,\ldots,a_N,b_N\}$ be compact. Then for every $x\in K$ we have
\begin{align}
M_+^{-1}(y)M_+(x)=\begin{pmatrix}
I_2+\OO(x-y) &*\\
* & *
\end{pmatrix} \quad \text{as } y\to x,
\end{align}
uniformly in $n$. The $*$ entries denote unimportant $2\times 2$ blocks.
\end{corollary}
\begin{proof}
By  Proposition \ref{prop:columnsofM} the first two columns in $M_+(x)-M_+(y)$ are uniformly bounded and since they are analytic it follows that they are of order $\OO(x-y)$ as $y\to x$
uniformly in $n$. Also  the first two rows of $M_+^{-1}(y)$ are uniformly bounded by Proposition \ref{prop:columnsofM}. Then  the corollary follows, since
\begin{align*}
M_+^{-1}(y)M_+(x)=I+M_+^{-1}(y)(M_+(x)-M_-(y)).
\end{align*}
\end{proof}
\section{Local parametrices and final transformation}

\subsection{Local parametrices}
The global parametrix $M$ will not be a good approximation to $S$ near
the branch points.
Around the branch points we construct a local parametrix in a fairly standard
way with the use of Airy functions. We will  not give full details about
the construct here but only give the relevant formulas with some comments.

\subsubsection{Statement of local RH problems}

Let
\[ BP = \{ a_k, b_k \mid k=1, \ldots, n\} \cup \left(\{ \pm ic_2, \pm c_3 \} \setminus \{ 0\} \right). \]
be the set of branch points. Note that $\pm ic_2$ an $\pm c_3$ are only branch
points if they are non-zero.

There is a possibility that $\pm c_3$ coincides with one of the end points $a_k, b_k$ of
the support of $\mu_1$. This is a case that could be handled just as well
but the formulas are slightly different and we prefer not to give full detail
in this case. Thus we assume
\begin{equation} \label{eq:c3different}
    \pm c_3 \not \in \{ a_k, b_k  \mid k =1, \ldots, n\}.
    \end{equation}

We take a sufficiently small disk $D_{p}$ around each of the branch points $p \in BP$.
The disks are mutually disjoint. We also make sure that the disks $D_{a_k}$ and $D_{b_k}$
are small enough such that they do not  intersect with the lips of the global lens around $S(\mu_3)$,
and similarly the disks $D_{\pm c_3}$  do not intersect with the lips of the lenses
around $S(\mu_1)$. Also the disks around $\pm ic_2$ are small enough
so that they do not intersect with the lenses around $S(\mu_1)$ and $S(\mu_3)$.

We use
\[ D = \bigcup_{p \in BP} D_p \]
to denote the union of the disks. Then we would like to find
a parametrix $P : \overline{D} \setminus \Sigma_S \to \mathbb C^{4\times 4}$
that has the same jumps as $S$ on the parts of $\Sigma_S$ in $D$.
There is a very minor complication here in case that $a_k$ or $b_k$ belongs
to $\mathbb R \setminus S(\mu_3) = (-c_3, c_3)$. Then the jump matrix $J_S$ for $S$
has in its right lower block $(J_S)_3$ an off-diagonal entry that is exponentially
small in a full neighborhood of $a_k$ or $b_k$, see the second formula in   \eqref{eq:JS3explicit}.
Being exponentially small this entry will play no role in what follows.
However, for the construction of the local parametrix it is more convenient
to set this entry equal to $0$.  This is also what we did when defining the part
$(J_M)_3$ in the jump matrix for $M$, see \eqref{eq:jumpM3}. Thus we use $(J_M)_3$
for the jump matrix $(J_P)_3$ on the real line in $D_{a_k}$ and $D_{b_k}$, see \eqref{eq:JPonDp3} below.

A similar thing happens in case $\pm c_3 \in \mathbb R \setminus S(\mu_1)$.
Then an exponentially small entry is in $(J_S)_1$, see the second formula in \eqref{eq:JS1explicit}
that we set equal to zero in the jump matrix for $P$.

Thus in $D_p$ where $p$ is one of $a_k, b_k$, for $k=1, \ldots, N$, we take
\begin{equation} \label{eq:JPblock13}
     J_P = \begin{pmatrix} (J_P)_1 & 0 \\ 0 & (J_P)_3 \end{pmatrix}
     \end{equation}
with
\begin{equation} \label{eq:JPonDp1}
\begin{aligned}
    (J_P)_1 & = (J_S)_1 && \quad \text{ on } \Sigma_S \cap D_p, \\
    (J_P)_3 & = (J_M)_3 && \quad \text{ on } \mathbb R \cap D_p, \\
    (J_P)_3 & = I_2 && \quad \text{ on } \left( \Sigma_S \setminus \mathbb R \right) \cap D_p.
    \end{aligned}
    \end{equation}

If $c_2 >0$ then in $D_{ic_2}$ and $D_{-ic_2}$  we put
\begin{equation}
     J_P =      \begin{pmatrix} 1 & 0 & 0 \\ 0 &  (J_P)_2 & 0 \\ 0 & 0 & 1 \end{pmatrix}
     \end{equation}
with
\begin{equation} \label{eq:JPonDp2}
    (J_P)_2 = (J_S)_2 \quad \text{ on } \Sigma_S \cap D_p.
    \end{equation}

If $c_3 > 0$ then in $D_{c_3}$ and $D_{-c_3}$, we have again the
block structure \eqref{eq:JPblock13}, but now we put
\begin{equation} \label{eq:JPonDp3}
\begin{aligned}
    (J_P)_1 & = (J_M)_1 && \quad \text{ on } \mathbb R \cap(D_{c_3} \cup D_{-c_3}), \\
    (J_P)_1 & = I_2     && \quad \text{ on } (\Sigma_S \setminus \mathbb R) \cap (D_{c_3} \cup D_{-c_3}), \\
    (J_P)_3 & = (J_S)_3 && \quad \text{ on } \Sigma_S \cap (D_{c_3} \cup D_{-c_3}),
    \end{aligned}
    \end{equation}

With this definition of $J_P$ we look for a parametrix
\[ P : \overline{D} \setminus \Sigma_S \to \mathbb C^{4\times 4} \]
that satisfies the following local RH problem.
\begin{equation}\label{eq:RHforP}
\left\{\begin{array}{l}
    P \textrm{ is analytic in } D \setminus \Sigma_S \text{ and continuous in } \overline{D} \setminus \Sigma_S, \\[5pt]
    P_+ = P_- J_P, \qquad \text{on }\Sigma_S \cap D, \\
    P= (I+ \OO(n^{-1})) M \quad \text{as } n \to \infty, \text{ uniformly on } \partial D \setminus \Sigma_S.
    \end{array} \right.
\end{equation}

\subsubsection{Airy functions}

Around each branch point the construction of $P$ is essentially a
$2 \times 2$ problem, that can be solved in a standard way using Airy functions,
see \cite{DKMVZ1} for the $2 \times 2$ case
and \cite{BleKu1,DuiKu2} for larger size RH problems.

The RH problem for Airy functions is the following.
It will be stated in terms of an auxiliary $\zeta$ variable.

\begin{equation}\label{eq:RHforA}
\left\{\begin{array}{l}
    A  :\mathbb C \setminus \Sigma_A \to \mathbb C^{2 \times 2} \textrm{ is analytic}, \\[5pt]
    A_+ = A_- J_A, \qquad \text{on } \Sigma_A, \\[5pt]
    A(\zeta)=
    \begin{pmatrix} \zeta^{-1/4} & 0 \\ 0 & \zeta^{1/4} \end{pmatrix}
    \frac{1}{\sqrt{2}} \begin{pmatrix} 1 & i \\ i & 1 \end{pmatrix} (I + \OO(\zeta^{-3/2})) \\
    \ \hfill{ \times
    \begin{pmatrix}
    e^{- \frac{2}{3} \zeta^{3/2}} & 0 \\
    0 & e^{\frac{2}{3} \zeta^{3/2}} \end{pmatrix}
    \textrm{ as } \zeta \to \infty},
    \end{array}
     \right.
\end{equation}
where the contour $\Sigma_A$ and the jump matrices $J_A$ are
shown in Figure \ref{fig:RHforA}.

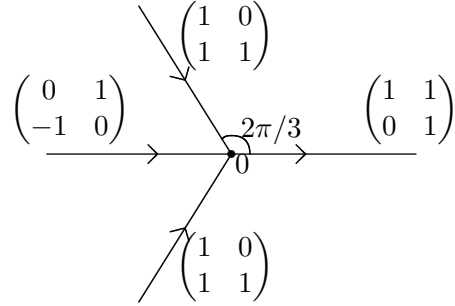
\begin{figure}[t]
\centering
\unitlength 0.7pt
\linethickness{0.5pt}
\begin{picture}(200,160)(-100,-80)
   % curves
   \put(0,0){\line(1,0){100}}
   \put(0,0){\line(-1,0){100}}
   \qbezier(0,0)(-25,40)(-50,80)
   \qbezier(0,0)(-25,-40)(-50,-80)
   % arrows on curves
   \qbezier(40,0)(39,1)(35,5)   \qbezier(40,0)(39,-1)(35,-5)
   \qbezier(-40,0)(-41,1)(-45,5)   \qbezier(-40,0)(-41,-1)(-45,-5)
   \qbezier(-25,40)(-24,43.5)(-23,47)   \qbezier(-25,40)(-29,41)(-33,42)
   \qbezier(-25,-40)(-24,-43.5)(-23,-47)   \qbezier(-25,-40)(-29,-41)(-33,-42)
   % origin with angle
   \put(2,-10){$0$}
   \put(0,0){\circle*{4}}
   \qbezier(10,0)(10,10)(0,10)
   \qbezier(0,10)(-4,10)(-5,8)
   \put(5,9){$2\pi/3$}
   % jump matrices
   \put(70,20){$\begin{pmatrix} 1 & 1  \\ 0 & 1\end{pmatrix}$}
   \put(-30,60){$\begin{pmatrix} 1 & 0 \\  1 & 1 \end{pmatrix}$}
   \put(-120,20){$\begin{pmatrix} 0& 1  \\ -1 & 0\end{pmatrix}$}
   \put(-30,-65){$\begin{pmatrix} 1 & 0 \\ 1 & 1 \end{pmatrix}$}
   \end{picture}
   \caption{Contour $\Sigma_A$ and jump matrices $J_A$ in the RH problem for $A$.  \label{fig:RHforA}}
\end{figure}

We have stated the RH problem in such a way that $\det A(\zeta) \to 1$ as $ \zeta \to \infty$,
which implies that also the solution has constant determinant $1$.
This explains the factors $\sqrt{2 \pi}$ and $\pm i$ that appear in the solution \eqref{eq:Airysol} below.
Define the three functions
\[ y_0(\zeta) = \Ai(\zeta), \qquad y_1(\zeta)  = \omega \Ai(\omega \zeta),
    \qquad y_2(\zeta) = \omega^2 \Ai(\omega^2 \zeta), \]
where $\Ai$ denotes the usual Airy function. That is, $\Ai$ is the unique solution of the Airy differential equation
$y''(\zeta) = \zeta y(\zeta)$ with asymptotic behavior
\[ \Ai(\zeta) = \frac{1}{2\sqrt{\pi} \zeta^{1/4}} e^{-\frac{2}{3} \zeta^{3/2}}\left(1 + \OO(\zeta^{-3/2})\right) \]
as $\zeta \to \infty$ with $-\pi + \varepsilon < \arg \zeta < \pi - \varepsilon$ for any $\varepsilon > 0$.
Also $y_1$ and $y_2$ are solutions of the Airy equation and the relation
\[ y_0 + y_1 + y_2 = 0 \]
is satisfied.

Then the solution of the RH problem \eqref{eq:RHforA} is given by
\begin{equation} \label{eq:Airysol}
    A(\zeta) = \sqrt{2\pi} \times \begin{cases}
        \begin{pmatrix} y_0(\zeta) & -y_2(\zeta) \\
    -i y_0(\zeta) & i y_2'(\zeta) \end{pmatrix} & \text{ for } 0 < \arg \zeta < 2\pi/3, \\
    \begin{pmatrix} -y_1(\zeta) & -y_2(\zeta) \\
    i y_1'(\zeta) & i y_2'(\zeta) \end{pmatrix} & \text{ for } 2\pi/3 < \arg \zeta < \pi, \\
    \begin{pmatrix} -y_2(\zeta) & y_1(\zeta) \\
    i y_2'(\zeta) & -i y_1'(\zeta) \end{pmatrix} & \text{ for } -\pi < \arg \zeta < -2\pi/3, \\
        \begin{pmatrix} y_0(\zeta) & y_1(\zeta) \\
    -i y_0'(\zeta) & -i y_1'(\zeta) \end{pmatrix} & \text{ for } -2 \pi/3 < \arg \zeta < 0.
    \end{cases}
\end{equation}

The solution $A$ of the Airy RH problem is the main building block for the local parametrix $P$.

\subsubsection{Parametrix $P$ in $D_{b_k}$}

In the neighborhood $D_{b_k}$ of a right-end point $b_k$ of
the support of $\mu_1$ the local parametrix $P$ takes the form
\begin{equation} \label{eq:defPnearbk}
    P(z) = M(z) \begin{pmatrix} P_1(z) & 0 \\ 0 & I_2 \end{pmatrix}, \qquad z \in D_{b_k} \setminus \Sigma_S.
    \end{equation}

To describe the $ 2 \times 2$ block $P_1(z)$ we use the solution $A(\zeta)$ of the
RH problem \eqref{eq:RHforA}, the functions $\lambda_1$ and $\lambda_2$ that come from
\eqref{eq:deflambdaj} (and that also appear in the jump matrix $(J_P)_1 = (J_S)_1$, see \eqref{eq:JS1explicit}),
and a function
\begin{equation} \label{eq:defmapf}
        f(z)=\begin{cases}
        \left(\frac{3}{4}(\lambda_{1}(z) - \lambda_{2}(z))  \pm \frac{3}{2} \pi i \alpha_k \right)^{2/3} & \text{ if }b_k>0,\\
           \left(\frac{3}{4}(\lambda_{1}(z) - \lambda_{2}(z))  \pm \frac{3}{2} \pi i \alpha_k  \mp \frac{1}{2}\pi i\right)^{2/3}
& \text{ if } b_k<0,
\end{cases}
        \end{equation}
 for $\pm \Im z > 0, \ z \in D_{b_k} \setminus \mathbb R$.

It turns out that $f$ has an analytic extension to $D_{b_k}$ which maps $b_k$ to $0$
(to check this one  uses the
jump properties \eqref{eq:lambdajump2}--\eqref{eq:lambdajump2b} for $\lambda_1$ and $\lambda_2$, among other
things).
Shrinking $D_{b_k}$ if necessary, one has that $\zeta = f(z)$ is a conformal map from $D_{b_k}$ to
a convex neighborhood of $\zeta = 0$ with $f(b_k) = 0$, $f'(b_k) >0$ and which is real for real $z$.
We adjust the lens around $(a_k,b_k)$ in such a way that the lips of the lens within $D_{b_k}$ are mapped by
$f$ into the rays $\arg \zeta = 2\pi/3$.

Then $P_1$ is given by
\begin{multline} \label{eq:defP1bk}
    P_1(z) =  \begin{pmatrix} e^{\pm n \pi i \alpha_k} & 0 \\ 0 & e^{\mp n \pi i \alpha_k} \end{pmatrix}
        \frac{1}{\sqrt{2}} \begin{pmatrix} 1 & -i \\ -i & 1 \end{pmatrix}
         \begin{pmatrix} n^{1/6} f(z)^{1/4} & 0 \\ 0 & n^{-1/6} f(z)^{-1/4} \end{pmatrix} \\
         \times A(n^{2/3} f(z))
            \begin{pmatrix} e^{\frac{1}{2} n(\lambda_1(z) - \lambda_2(z))} & 0 \\ 0 & e^{-\frac{1}{2} n(\lambda_1(z) - \lambda_2(z))} \end{pmatrix}, \\
              \hfill{\textrm{ for } z\in D_{b_k} \setminus \Sigma_S, \ \pm \Im z > 0,}
\end{multline}
with the principal branches of the fourth roots in $f(z)^{\pm 1/4}$.

\subsubsection{Parametrix $P$ in $D_{a_k}$}

The construction of $P$ in a neighborhood $D_{a_k}$ of a left endpoint $a_k$
of $S(\mu_1)$ is similar. Here we have a map
\begin{equation} \label{eq:defmapf2}
        f(z)=\begin{cases}
        \left(\frac{3}{4}(\lambda_{1}(z) - \lambda_{2}(z))  \pm \frac{3}{2} \pi i \alpha_{k-1} \right)^{2/3} &\text{if } a_k>0,\\
     \left(\frac{3}{4}(\lambda_{1}(z) - \lambda_{2}(z))  \pm \frac{3}{2} \pi i \alpha_{k-1}\mp \frac{1}{2}\pi i \right)^{2/3} &\text{if } a_k<0,\end{cases}
\end{equation}
 for  $\pm \Im z > 0, \, z \in D_{a_k} \setminus \mathbb R$,
with $f(a_k) = 0$, $f'(a_k) < 0$. If necessary, we shrink the disk $D_{a_k}$ and adjust
the lips of the lens around $(a_k,b_k)$ such that $\zeta = f(z)$ is a conformal map in $D_{a_k}$
onto a convex neighborhood of $\zeta  =0$ that maps $\Sigma_S \cap D_{a_k}$ into $\Sigma_A$.

Then the local parametrix takes the form
\begin{equation} \label{eq:defPnearak}
    P(z) = M(z) \begin{pmatrix} P_1(z) & 0 \\ 0 & I_2 \end{pmatrix}, \qquad z \in D_{a_k} \setminus \Sigma_S,
    \end{equation}
with
\begin{multline} \label{eq:defP1ak}
    P_1(z) =  \begin{pmatrix} e^{\pm n \pi i \alpha_{k-1}} & 0 \\ 0 & e^{\mp n \pi i \alpha_{k-1}} \end{pmatrix}
        \frac{1}{\sqrt{2}} \begin{pmatrix} 1 & -i \\ i & -1 \end{pmatrix}
         \begin{pmatrix} n^{1/6} f(z)^{1/4} & 0 \\ 0 & n^{-1/6} f(z)^{-1/4} \end{pmatrix} \\
         \times
         A(n^{2/3} f(z))  \begin{pmatrix} 1 & 0 \\ 0 & -1 \end{pmatrix}
            \begin{pmatrix} e^{\frac{1}{2} n(\lambda_1(z) - \lambda_2(z))} & 0 \\ 0 & e^{-\frac{1}{2} n(\lambda_1(z) - \lambda_2(z))} \end{pmatrix} \\
            \textrm{ for } z \in D_{a_k} \setminus \Sigma_S, \ \pm \Im z > 0.
\end{multline}
The fourth roots $f(z)^{\pm 1/4}$ are taken to be analytic in $D_{a_k} \setminus [a_k,b_k]$ and positive
for real $z < a_k$, $z \in D_{a_k}$.

\subsubsection{Parametrix $P$ in $D_{\pm ic_2}$}

This case is only relevant if $c_2 > 0$ and so we assume $c_2 > 0$.
The parametrix in $D_{\pm ic_2}$ takes the form
\begin{equation} \label{eq:defPnearc2}
    P(z) = M(z) \begin{pmatrix} 1 & 0 & 0 \\ 0 & P_2(z) & 0 \\ 0 & 0 & 1 \end{pmatrix}
    \end{equation}
where $P_2(z)$ is a $2 \times 2$ matrix valued function in $D_{\pm i c_2} \setminus \Sigma_2$.

In $D_{ic_2}$ it is constructed with the function
\begin{equation} \label{eq:defmapfc2} f(z) = \left( \frac{3}{4}(\lambda_2(z) - \lambda_3(z)) \mp \frac{1}{2} \pi i \right)^{3/2}
    \quad \text{for } z \in D_{ic_2} \setminus i \mathbb R, \ \pm \Re z > 0,
    \end{equation}
which (with an appropriate understanding of the $3/2$-power) has an extension to a conformal map
on $D_{ic_2}$ with $f(ic_2) = 0$ and $f'(ic_2) \in i \mathbb R^+$.
Then  $f(z) > 0$ for $z = iy \in D_{ic_2} \cap i\mathbb R$, $y < c_2$.
We take $P_2$ in $D_{ic_2}$ as
\begin{multline} \label{eq:defP2c2}
    P_2(z) =  (-1)^n
        \frac{1}{\sqrt{2}} \begin{pmatrix} -i & 1 \\ -1 & i \end{pmatrix}
        \begin{pmatrix} n^{1/6} f(z)^{1/4} & 0 \\ 0 & n^{-1/6} f(z)^{-1/4} \end{pmatrix} \\
         \times A(n^{2/3} f(z))  \begin{pmatrix} 0 & -1 \\ 1 & 0 \end{pmatrix}
            \begin{pmatrix} e^{-\frac{1}{2} n(\lambda_2(z) - \lambda_3(z))} & 0 \\ 0 & e^{\frac{1}{2} n(\lambda_2(z) - \lambda_3(z))} \end{pmatrix} \\
            \textrm{ for } z \in D_{ic_2} \setminus \Sigma_S.
\end{multline}
The fourth roots $f(z)^{\pm 1/4}$ are positive for $z = iy \in D_{ic_2} \cap i \mathbb R$, $y < c_2$.

The construction of $P$ in $D_{-ic_2}$ is very similar. By symmetry we
can also obtain it from $P$ in $D_{i c_2}$ by means of the formula
\[ P(z) = \begin{pmatrix} 1 & 0 & 0 & 0 \\ 0 & 1 & 0 & 0 \\ 0 & 0 & -1 & 0 \\ 0 & 0 & 0 & 1 \end{pmatrix}
    P(-z) \begin{pmatrix} 1 & 0 & 0 & 0 \\ 0 & 1 & 0 & 0 \\ 0 & 0 & -1 & 0 \\ 0 & 0 & 0 & 1 \end{pmatrix} \]
for $z \in D_{-ic_2} \cap \Sigma_S$.

\subsubsection{Parametrix $P$ in $D_{\pm c_3}$}

This case is only relevant if $c_3 > 0$ and so we assume $c_3 > 0$.
The parametrix $P$ in $D_{\pm c_3}$ takes the form

\begin{equation} \label{eq:defPnearc3}
    P(z) = M(z) \begin{pmatrix} I_2 & 0  \\ 0 & P_3(z)  \end{pmatrix}
    \end{equation}
where a $2 \times 2$ block $P_3(z)$.

It is constructed with the function
\begin{equation} \label{eq:defmapfc3}
        f(z)=
        \left(\frac{3}{4}(\lambda_{3}(z) - \lambda_{4}(z))  \pm \frac{1}{4} \pi i  \right)^{2/3}
        \quad \text{for } \pm \Im z > 0, \ z \in D_{-c_3} \setminus \mathbb R.
        \end{equation}
Then $f$ has an analytic extension to $D_{- c_3}$ which maps $\pm c_3$ to $0$.
Shrinking $D_{\pm c_3}$ if necessary, one has that $\zeta = f(z)$ is a conformal map from
$D_{-c_3}$ to a neighborhood of $\zeta = 0$ with $f(-c_3) =0$ and $f'(-c_3) > 0$.
We adjust the lens around $(-\infty,-c_3)$ in such a way that the lips of the lens within $D_{-c_3}$
are mapped into the rays $\arg \zeta = 2\pi/3$.
Then $P_3$ is given by
\begin{multline} \label{eq:defP3a}
    P_3(z) =  (-1)^n
        \frac{1}{\sqrt{2}} \begin{pmatrix} 1 & -i \\ -i & 1 \end{pmatrix}
         \begin{pmatrix} n^{1/6} f(z)^{1/4} & 0 \\ 0 & n^{-1/6} f(z)^{-1/4} \end{pmatrix} \\
         \times A(n^{2/3} f(z))
            \begin{pmatrix} e^{\frac{1}{2} n(\lambda_3(z) - \lambda_4(z))} & 0 \\ 0 & e^{-\frac{1}{2} n(\lambda_3(z) - \lambda_4(z))} \end{pmatrix}, \\
              \hfill{\textrm{ for } z\in D_{-c_3} \setminus \Sigma_S.}
\end{multline}

For the construction in $D_{c_3}$ we can use the conformal map $f$ with
the same definition as in \eqref{eq:defmapfc3}, but now considered in the
neighborhood of $c_3$. Then we have $f(c_3) = 0$ and $f'(c_3) < 0$,
and after adjusting of the lenses we then define $P_3$ as
\begin{multline} \label{eq:defP3b}
    P_3(z) =  (-1)^n
        \frac{1}{\sqrt{2}} \begin{pmatrix} 1 & -i \\ i & -1 \end{pmatrix}
         \begin{pmatrix} n^{1/6} f(z)^{1/4} & 0 \\ 0 & n^{-1/6} f(z)^{-1/4} \end{pmatrix} \\
         \times A(n^{2/3} f(z)) \begin{pmatrix} 1 & 0 \\ 0 & -1 \end{pmatrix}
            \begin{pmatrix} e^{\frac{1}{2} n(\lambda_3(z) - \lambda_4(z))} & 0 \\ 0 & e^{-\frac{1}{2} n(\lambda_3(z) - \lambda_4(z))} \end{pmatrix}, \\
              \hfill{\textrm{ for } z\in D_{c_3} \setminus \Sigma_S.}
\end{multline}

\subsection{Final transformation}

Having the global parametrix $M$ and the local parametrix $P$
we are ready for the fifth and also final transformation.

\begin{definition}
We define the $4 \times 4$ matrix valued function $R$ by
\begin{equation} \label{eq:defR}
    R(z) = \begin{cases}
    S(z) P(z)^{-1} & \text{ in the disks around each of the branch points,} \\
    S(z) M(z)^{-1}  & \text{ outside the disks.}
    \end{cases}
\end{equation}
\end{definition}

Then $R$ is defined and analytic outside the union of $\Sigma_S$ and
the boundaries of the disks around the branch points. However, since
the jumps of $S$ and $M$ agree on $S(\mu_1) \cap S(\mu_3)$ and on
$S(\sigma_2 - \mu_2)$, the matrix $R$ has analytic continuation
across the parts of these supports that are outside the disks. We
also have that the jumps of $S$ and $P$ agree inside the disks
$D_{\pm ic_2}$ and so $R$ has analytic continuation in the interior
of these two disks. In the disk $D_{a_k}, D_{b_k}$ and $D_{\pm c_3}$
the jumps of $S$ and $P$ may not be exactly the same. They could
differ by an exponentially small entry on the real line inside these
disks. The jumps on the lips of the lenses are the same inside the
disks, and so $R$ has an analytic continuation across these lenses,
but $R$ could have an exponentially small jump on the real line
inside the disks.

The result is that $R$ has an analytic continuation to $\mathbb C \setminus \Sigma_R$
for a certain contour $\Sigma_R$ and that $R$ satisfies the following RH problem.
\begin{equation}\label{eq:RHforR}
\left\{\begin{array}{l}
    R \textrm{ is analytic in } \mathbb C\setminus \Sigma_R, \\[5pt]
    R_+ = R_- J_R, \qquad \text{on }\Sigma_R, \\
    R(z)= I+ \OO(z^{-1}) \quad \text{as } z\to \infty.
    \end{array}\right.
\end{equation}
where
\begin{equation} \label{eq:defJR}
    J_R(z) =
    \begin{cases}
    M(z) P(z)^{-1}, &  z \in \bigcup_{p \in BP} \partial D_p, \\
  P_-(z) J_S(z) P_+(z)^{-1}, & z \in \Sigma_R \cap \bigcup_{p \in BP} D_p, \\
    M_-(z) J_S(z) M_+(z)^{-1}, & z \in \Sigma_R \setminus \overline{\bigcup_{p \in BP} D_p}.
    \end{cases}
    \end{equation}

All of the jump matrices are close to the identity matrix as $n \to \infty$.
Indeed because of the matching condition in \eqref{eq:RHforP} we have
\[ J_R(z) = I_4 + \OO(n^{-1}) \qquad \text{as } n \to \infty \]
uniformly for $z \in \bigcup_{p \in BP} \partial D_p$.
The other jumps are exponentially close to the identity matrix as $n \to \infty$
with a bound that improves as $z \to \infty$.
Indeed we have
\[ J_R(z) = I_4 + \OO(\exp(- c n (|z| + 1)) \qquad \text{as } n \to \infty \]
uniformly for $z \in \Sigma_R \setminus \left( \bigcup_{p \in BP} \partial D_p\right)$.
This follows as in \cite[section 8.4]{DuiKu2}.

Observe that from the definition \eqref{eq:defR} and the asymptotic behaviors
of $S$ and $M$ as given in \eqref{eq:RHforS} and \eqref{eq:RHforM}
we first find that $R(z) = I_4 + \OO(z^{-1/3})$ as $z \to \infty$.
However  $J_R - I_4$ is also exponentially decaying as $z \to \infty$ with $n$
fixed, and therefore the better bound $R(z) = I_4 + \OO(z^{-1})$ in \eqref{eq:RHforR} indeed holds.

From this we conclude as in \cite{DuiKu2}
\begin{proposition} \label{prop:finalRestimate}
    There is a constant $C > 0$ such that for every $n$,
    \[ \| R(z) - I_4 \| \leq \frac{C}{n (|z| + 1)} \]
    uniformly for $z \in \mathbb C \setminus \Sigma_R$.
\end{proposition}

This concludes the steepest descent analysis of the RH problem for $Y$.

\subsection{Proof of Theorem \ref{theo:theo2}}

We use the following consequence of Corollary \ref{cor:MinvyMx} and Proposition
\ref{prop:finalRestimate}

\begin{lemma} \label{lem:Sestimate}
For every $x \in S(\mu_1) \setminus \bigcup_{k=1}^N\left( D_{a_k}\cup D_{b_k}\right)$
we have
\[ S_+^{-1}(y)S_+(x) = \begin{pmatrix} I_2 + \OO(x-y)
 & * \\
 * & * \end{pmatrix} \qquad \text{as } y \to x, \]
uniformly in $n$. The $*$ entries denote unimportant $2\times 2$ blocks.
\end{lemma}
\begin{proof}
By \eqref{eq:defR} and the fact that $x$ is  outside the disks, we have
\begin{align}\label{eq:SinvySx}
S_+^{-1}(y)S_+(x)=M_+^{-1}(y) R_+^{-1}(y) R_+(x) M_+(x),\end{align}
if $y$ is close enough to $x$.
Note that $R$ may not be analytic in a neighborhood of $x$. However, by deforming contours  into the lower half plane we see that $R$ does  have an analytic continuation to a neighborhood of $x$ where the same estimate of Proposition \ref{prop:finalRestimate} is valid. Now write
\begin{align*}
R^{-1}_+(y) R_+(x)=I+R^{-1}_+(y)\left(R_+(x)-R_+(y)\right).
\end{align*}
 Then by Cauchy's Theorem
 \begin{align*}
R^{-1}_+(y) R_+(x)&=I+R^{-1}_+(y)\frac{1}{2\pi i} \oint_{|z-x|=r} R(z)\left(\frac{1}{z-x}
-\frac{1}{z-y} \right)\ {d}z\\
&=I+R^{-1}_+(y)\frac{(x-y)}{2\pi i} \oint_{|z-x|=r} \frac{R(z)-I}{(z-x)(z-y)} \ {d}z,
\end{align*}
for some $r>0$.  Combining this with Proposition \ref{prop:finalRestimate} leads to
\[R_+^{-1}(y) R_+(x)=I+\OO\left(\frac{x-y}{n}\right) \quad \text{as } y\to x.\]
Now inserting this in \eqref{eq:SinvySx} and using Corollary \ref{cor:MinvyMx} gives the statement.
\end{proof}

We follow the effect of the transformations $Y \mapsto X \mapsto U \mapsto T \mapsto S \mapsto R$
on the kernel $K_{11}^{(n)}$ given by \eqref{eq:kernelintermsofY}.

The transformation $Y \mapsto X$ as given by \eqref{eq:YtoX} then gives
\begin{equation} \label{eq:kernelintermsofXa}
\begin{aligned}
    K_{11}^{(n)}(x,y)  = & \frac{1}{2\pi i(x-y)} \begin{pmatrix} 0 & w_{0,n}(y) & w_{1,n}(y) & w_{2,n}(y) \end{pmatrix} \\
        & \times \begin{pmatrix} 1 & 0 \\ 0 & D_n P_{n,+}^{-t}(y) e^{n \Theta_+(y)} \end{pmatrix}
            X_{+}^{-1}(y) X_+(x) \\
            & \times \begin{pmatrix} 1 & 0 \\ 0 & e^{-n \Theta_+(x)} P_{n,+}^{t}(x)  D_n^{-1}  \end{pmatrix}
            \begin{pmatrix} 1 \\ 0 \\ 0 \\ 0 \end{pmatrix}.
            \end{aligned}
\end{equation}
By \eqref{eq:JXeval2} we then get
\begin{equation} \label{eq:kernelintermsofX}
    K_{11}^{(n)}(x,y) = \frac{1}{2\pi i(x-y)} \begin{pmatrix} 0 & e^{-n(V(y) - \theta_1(y))} & 0 & 0 \end{pmatrix}
            X_{+}^{-1}(y) X_+(x)
            \begin{pmatrix} 1 \\ 0 \\ 0 \\ 0 \end{pmatrix}.
            \end{equation}

Then by the transformation $X \mapsto U$ given in \eqref{eq:defU} we get
\begin{equation} \label{eq:kernelintermsofUa}
\begin{aligned}
    K_{11}^{(n)}(x,y)  = & \frac{e^{n(g_{1,+}(x)+ \ell_1)}}{2\pi i(x-y)}
        \begin{pmatrix} 0 & e^{-n(V(y)-\theta_1(y)-g_{1,+}(y)+g_{2,+}(y))} & 0 & 0 \end{pmatrix} \\
        & \times
            U_{+}^{-1}(y) U_+(x)
            \begin{pmatrix} 1 \\ 0 \\ 0 \\ 0 \end{pmatrix},
                \quad x, y \in \mathbb R.
            \end{aligned}
            \end{equation}
By \eqref{eq:deflambdaj} this leads to
\begin{equation} \label{eq:kernelintermsofU}
\begin{aligned}
    K_{11}^{(n)}(x,y)  = & \frac{e^{n(g_{1,+}(x) - g_{1,+}(y))}}{2\pi i(x-y)}
        \begin{pmatrix} 0 & e^{n(\lambda_{2,+}(y)-\lambda_{1,+}(y))} & 0 & 0 \end{pmatrix} \\
        & \times
            U_{+}^{-1}(y) U_+(x)
            \begin{pmatrix} 1 \\ 0 \\ 0 \\ 0 \end{pmatrix},
                \quad x, y \in \mathbb R.
            \end{aligned}
            \end{equation}

The transformation $U\mapsto T$ given by \eqref{eq:defT1}--\eqref{eq:defT3} only
acts on the lower right $2 \times 2$ block, and does not affect the
expression \eqref{eq:kernelintermsofU} for the correlation kernel. We get
\begin{equation} \label{eq:kernelintermsofT}
\begin{aligned}
    K_{11}^{(n)}(x,y)  = & \frac{e^{n(g_{1,+}(x) - g_{1,+}(y))}}{2\pi i(x-y)}
        \begin{pmatrix} 0 & e^{n(\lambda_{2,+}(y)-\lambda_{1,+}(y))} & 0 & 0 \end{pmatrix} \\
        & \times
            T_{+}^{-1}(y) T_+(x)
            \begin{pmatrix} 1 \\ 0 \\ 0 \\ 0 \end{pmatrix},
                \quad x, y \in \mathbb R.
            \end{aligned}
            \end{equation}

Next, by the transformation $T \mapsto S$ of \eqref{eq:defS1} and \eqref{eq:defS2}
we find
\begin{equation} \label{eq:kernelintermsofS}
\begin{aligned}
    K_{11}^{(n)}(x,y)  = & \frac{e^{n(g_{1,+}(x) - g_{1,+}(y))}}{2\pi i(x-y)}
        \begin{pmatrix} -\chi_{S(\mu_1)}(y) & e^{n(\lambda_{2,+}(y)-\lambda_{1,+}(y))} & 0 & 0 \end{pmatrix} \\
        & \times
            S_{+}^{-1}(y) S_+(x)
            \begin{pmatrix} 1 \\ \chi_{S(\mu_1)}(x) e^{n(\lambda_{1,+}(x)-\lambda_{2,+}(x))}  \\ 0 \\ 0 \end{pmatrix},
                \quad x, y \in \mathbb R
            \end{aligned}
            \end{equation}
where $\chi_{S(\mu_1)}$ denotes the characteristic function of the set $S(\mu_1)$.

Let $x \in S(\mu_1) \setminus \{a_k, b_k \mid k = 1, \ldots, N \}$.
The factor $e^{n(g_{1,+}(x) - g_{1,+}(y))}$ disappears as $y \to x$.
Then we use Lemma~\ref{lem:Sestimate} and we get
\begin{equation} \label{eq:kernelxx1}
\begin{aligned}
    K_{11}^{(n)}(x,x)  & = \lim_{y \to x} \frac{-1 + e^{n(\lambda_{2,+}(y)-\lambda_{1,+}(y))} e^{n(\lambda_{1,+}(x)-\lambda_{2,+}(x))}}{2\pi i(x-y)}
    + \OO(1)  \\
    & = \frac{n}{2\pi i} \frac{d}{dx} \left( \lambda_{1,+}(x) - \lambda_{2,+}(x) \right) + \OO(1)
    \end{aligned}
            \end{equation}
as $n \to \infty$.

Thus
\begin{align*}
    \lim_{n\to\infty} \frac{1}{n} K_{11}^{(n)}(x,x) & = \frac{1}{2\pi i} \frac{d}{dx} \left( \lambda_{1,+}(x) - \lambda_{2,+}(x) \right) \\
    & = \frac{1}{2\pi i} \left( F_{1,-}(x) - F_{1,+}(x) \right) \\
    & = \frac{d \mu_1}{dx},  \end{align*}
    where we used   \eqref{eq:deflambdaj},  \eqref{eq:jumpg1onR},   \eqref{eq:defFj}, and
    \eqref{eq:jumpFj}.
This completes the proof of Theorem \ref{theo:theo2}.

\section*{Acknowledgements}

M. Duits and A.B.J. Kuijlaars are grateful for the support and hospitality of MSRI in Berkeley in the fall of 2010.

A.B.J. Kuijlaars is supported by K.U. Leuven research grant
OT/08/33, FWO-Flanders project G.0427.09, by the Belgian
Interuniversity Attraction Pole P06/02, and by
grant MTM2008-06689-C02-01 of the Spanish Ministry of Science and
Innovation. 

M. Y. Mo acknowledges financial support by the EPSRC
grant EP/G019843/1.

\end{document}